\documentclass[12pt,reqno]{amsart}
\usepackage{adjustbox}
\usepackage{hyperref}
\usepackage{comment}
\usepackage{amsmath, amssymb, url, mathtools,centernot, tikz}
\usepackage{xcolor}
\usepackage{wrapfig}
\usepackage{mathrsfs}
\usepackage[utf8]{inputenc}
\DeclareUnicodeCharacter{2010}{-}
\DeclareUnicodeCharacter{2011}{-}
\DeclareUnicodeCharacter{2012}{--}
\DeclareUnicodeCharacter{2013}{--}
\DeclareUnicodeCharacter{2014}{---}
\DeclareUnicodeCharacter{2018}{`}
\DeclareUnicodeCharacter{2019}{'}
\DeclareUnicodeCharacter{201C}{``}
\DeclareUnicodeCharacter{201D}{''}
\newtheorem{proposition}{Proposition}
\newcommand{\norminf}[1]{\|#1\|_{\infty}}
\newcommand{\normmax}[1]{\|#1\|_{\max}}
\makeatletter
\newcommand*{\addFileDependency}[1]{% argument=file name and extension
  \typeout{(#1)}% latexmk will find this if $recorder=0 (however, in that case, it will ignore #1 if it is a .aux or .pdf file etc and it exists! if it doesn't exist, it will appear in the list of dependents regardless)
  \@addtofilelist{#1}% if you want it to appear in \listfiles, not really necessary and latexmk doesn't use this
  \IfFileExists{#1}{}{\typeout{No file #1.}}% latexmk will find this message if #1 doesn't exist (yet)
}
\makeatother

\usepackage[left=1in,right=1in,top=1in,bottom=1in]{geometry}

\usepackage[authoryear,round]{natbib}
\usepackage[utf8]{inputenc} % allow utf-8 input
\usepackage[T1]{fontenc}    % use 8-bit T1 fonts
\usepackage{url}            % simple URL typesetting
\usepackage{booktabs}       % professional-quality tables
\usepackage{amsfonts}       % blackboard math symbols
\usepackage{nicefrac}       % compact symbols for 1/2, etc.
\usepackage{setspace}
\usepackage{caption}
\usepackage{subcaption}
\usepackage{ragged2e}

\usepackage{comment}
\usepackage{adjustbox}      % General-purpose boxes and floats
\usepackage{algpseudocode, algorithm}
\usepackage{amsmath, amssymb, amsthm}
\usepackage{array}
\usepackage{blkarray}       % Extended block matrices
\usepackage{appendix}
\usepackage{booktabs}       % Prettier tables
\usepackage{bm}             % Bold symbolds in simpler fashion
\usepackage[justification = centerlast]{caption}
\usepackage{color}
\usepackage{graphicx}
\usepackage{multirow}
\usepackage{mathtools}      % Additional math commands and coloneqq symbol
\usepackage{natbib}         % Citations
\usepackage{paralist}       % Additional control of lists
\usepackage{parskip}        % Spacing between paragraphs
\usepackage{setspace}       % Singlespacing and doublespacing
\usepackage{subcaption}     % Loads subfigure and subcaption commands
\DeclareMathOperator{\exv}{\mathbb{E}}

\DeclareMathOperator{\diag}{diag}
\DeclareMathOperator{\tr}{tr}

\DeclareMathOperator{\SVT}{SVT}

\newcommand{\e}{\varepsilon}
\newcommand{\I}{\mathbb{I}}

\newcommand{\T}{^{\top}}  % General ranspose symbol

\newcommand{\bvL}{\mathbf{v}_{L^*_0}}
\newcommand{\bvS}{\mathbf{v}_{S^*_0}}
\newcommand{\bvW}{\mathbf{v}_{W_0}}

\renewcommand{\P}{\mathbb{P}}

\theoremstyle{plain}
\newtheorem{theorem}{Theorem}
\newtheorem{corollary}{Corollary}
\newtheorem{lemma}{Lemma}

\theoremstyle{definition}
\newtheorem{assumption}{Assumption}

\theoremstyle{remark}

\newtheorem{example}{Example}
\newtheorem{remark}{Remark}
\newcommand{\E}{\mathbb{E}}
\newcommand{\bL}{L^*}
\newcommand{\bS}{S^*}

\newcolumntype{H}{>{\setbox0=\hbox\bgroup}c<{\egroup}@{}} % "Hidden" column type
 \title[]{Estimating Network Spillovers under Dense Measurement Error}
\thanks{Date: \today{}. }
\thanks{Author names are listed in alphabetical order.}
 \thanks{We thank Santiago Montoya-Blandon for comments and preliminary data analysis.
We thank Debopam Bhattacharya, Tom Boot, Guido Kuersteiner, Oliver Linton, Elena Manresa, Whitney Newey, Alexey Onatski, Bernard Salani\'{e}, Shuyang Sheng,  Xun Tang as well as participants at the Gothenburg Econometrics Workshop, the Oxford Econometrics Workshop, the VTSS Seminar, the (EC)$^2$ Conference, the 2025 Econometric Society World Congress and Sofie 2026 for their thoughtful comments. Áureo de Paula gratefully acknowledges financial support from the Economic and Social Research Council through the ESRC Institute for the Microeconomic Analysis of Public Policy (Grant Ref: ES/T014334/1), and from UK Research and Innovation (UKRI) under the UK Government’s Horizon Europe funding guarantee (Grant Ref: EP/X02931X/1). Any remaining errors are our own. 
 }
  \author[]{ Yingxing Li\textsuperscript{a} \and \'Aureo de Paula\textsuperscript{b} \and  Weining Wang\textsuperscript{c}}
\thanks{$^a$: School of Business, Sun Yat-sen University.}
\thanks{$^b$: Department of Economics, University College London, the Institute for Fiscal Studies and CeMMAP}
\thanks{$^c$: Department of Economics, University of Bristol. Corresponding author: Weining Wang, email: \texttt{weining.wang@bristol.ac.uk}.}

\begin{document}

\begin{abstract}
This paper analyzes spillover effects in spatial (network) models when the neighborhood (adjacency) matrix is contaminated by measurement error from reporting, imputation, or disclosure imperfections, leading to inconsistent estimation of network effects. We introduce a regularization framework for the latent network that allows for sparse and/or low-rank structure and accommodates potential correlation between measurement errors and outcomes. 
We propose two estimators: (i) a two-stage procedure that first denoises the adjacency matrix and then incorporates the denoised network into a regression analysis, and (ii) a Generalized Method of Moments (GMM) estimator that jointly estimates regression parameters and refines the network structure. We establish consistency, convergence rates and asymptotic normality for the proposed
spillover-effect estimators and show that their error bounds are
strictly smaller than those of naive estimation, which may be
inconsistent when measurement error does not vanish. Simulations demonstrate that, in the presence of noisy networks,
our approach could substantially reduce the estimation errors relative to conventional methods. %reduces the root mean squared error of spillover estimates relative to conventional methods by at least $40\%$. 
We apply our framework to examine international spillovers in economic growth, and tax competition across U.S. states.
\end{abstract}

\maketitle

\vspace{5mm}

\vspace{1mm} \noindent{Keywords:}  network analysis, measurement error, LASSO, nuclear norm penalty, penalized GMM

{
\section{Introduction}
\label{Sec:Introduction}
Network analysis provides powerful tools to quantify policy spillovers 
and identify pivotal agents for targeted interventions, particularly 
when policymakers face binding resource constraints 
\citep{Ballester2006, Zenou2016, Galeotti2020}. A large literature 
shows that economic outcomes often depend on interactions among 
interconnected units, giving rise to spillover effects that are central to policy design 
\citep{Acemoglu2012, Giorgi2020, Araujo2023}. These insights have 
motivated the widespread use of spatial and social interaction models 
to study peer effects, diffusion dynamics, and network-based policy 
interventions across a broad range of economic settings. 

Methodologically, the literature on spatial and social interaction models 
differs mainly in how the network structure is handled. A dominant strand 
assumes that the adjacency matrix is correctly specified and directly 
observed, and studies spillover effects conditional on this object 
\citep{Lee2007a, Bramoulle2009, Lee2010a, Banerjee2013, Elhorst2014, 
Lee2014, Yang2017, Zhu2020a, Kuersteiner2020,Zhu2019Network, Shin2021, 
Xu2024Dynamic,chernozhukov2021uniform}. An alternative and growing line of research 
relaxes this assumption by estimating the network jointly with model 
parameters, typically imposing dimension reduction constraints such as 
sparsity to address the curse of dimensionality in high-dimensional 
settings \citep{Haerdle2016,Manresa2016, Lam2020, Paula2023}.

Sparsity-based methods work well when the latent network is dominated 
by a few strong links, but they may fail in denser networks that arise 
in production, trade, and financial settings, where the observed 
adjacency matrix itself is contaminated by measurement error 
\citep{FismanWei2004, Upper2011, Anand2015}. In such settings, two 
distinct forms of latent structure may coexist beyond idiosyncratic noise. 
First, outcomes may be driven by a small number of pervasive (though 
individually weak) latent forces such as common shocks, shared 
technologies, or institutional linkages, naturally represented by a 
low-rank component \citep{Chudik2013, Kapetanios2021}. Second, 
networks may contain a small number of strong bilateral or 
idiosyncratic interactions, naturally represented by a sparse 
component \citep{Lam2020, Paula2023}. Imposing sparsity alone discards 
the pervasive signal, while ignoring measurement error produces 
non-vanishing bias as elementwise small errors accumulate along 
network paths \citep{Lewbeletal2024EJ}. The stakes are practical: in 
our cross-state tax-competition application, the standard spatial GMM 
estimator with the raw inverse squared distance weight matrix yields 
$\widehat\lambda > 1$ in the specification with contextual effects, violating the Leontief stability condition and 
implying explosive multipliers, while our  estimator returns 
$\widehat\lambda < 1$ within the admissible region. We therefore 
propose a regularization approach for the observed adjacency matrix 
that filters measurement error while preserving low-rank or sparse 
signal, developed formally in Section~\ref{Sec:Methodology}.
Throughout, by \emph{regularization} we mean imposing low-rank, sparse, or low-rank plus sparse restrictions directly on the 
latent adjacency matrix $W_0$ rather than on the regression 
coefficients, as identifying restrictions that make denoising of 
the observed network feasible.

We study spillover estimation when the econometrician observes
$W=W_0+E$, where $W_0$ is the latent adjacency matrix and $E$ is a
measurement-error matrix. As the number of nodes $n$ increases, even
small entrywise errors may accumulate across the $O(n^2)$ potential
links and along higher-order network paths, so their aggregate effect
need not vanish. We exploit low-rank, sparse, or low-rank-plus-sparse
structure in $W_0$ to attenuate the effect of this contamination on
spillover estimation. The problem is non-trivial because the
contaminated network enters both spatial regressors and instruments,
while the model's reduced form depends nonlinearly on the latent
network through the Leontief inverse  $(I_n-\lambda_0W_0)^{-1}$
\footnote{Here $I_n$ is the $n\times n$ identity matrix and
$\lambda_0$ is the spillover parameter.}. Standard corrections for additively
mismeasured regressors therefore do not directly apply.

Low-rank and sparse structures have deep foundations in statistics 
and econometrics, appearing in approximate factor models, 
latent-variable graphical models, and high-dimensional covariance 
estimation \citep{Bai2002, Chandrasekaran2010, Fan2011}. 
A broad statistical literature develops low-rank and/or sparse methods
for matrix recovery \citep[see, e.g.,][]{Cai2011, Candes2011, Agarwal2012,cai2021network}, covariance estimation \citep[see, e.g.,][]{Fan2018}, and network analysis
\citep[see, e.g.,][]{Alidaee2020, HayesLevin2026}. Unlike the latter prior work on network analysis that reconstructs latent positions under specified network-formation models, we study estimation and inference for the spillover parameter when the full adjacency matrix enters both the spatial regressors and the instruments and is measured with dense, potentially endogenous error.\footnote{\cite{Alidaee2020} consider a low-rank reconstruction of the network from Aggregate Relational Data (ARD).  \cite{HayesLevin2026} consider a spatial/interactions model where, using the notation we employ in \eqref{maineq}, $Y=\lambda_0 W_0 Y + X_{\texttt{unobs}} \alpha_0  + X_{\texttt{obs}} \beta_0 + \varepsilon$, where $X_{\texttt{unobs}}$ is an $n \times d$ matrix stacking $d$-dimensional \emph{unobserved} covariates (i.e., ``latent positions") and, similarly, $X_{\texttt{obs}}$ collects observed covariates. The remaining notation follows \eqref{maineq}, and the econometrician
additionally observes an adjacency matrix $W$. The paper considers two related models.  In their ``peer contagion'' model, one has $W_0=W$, whereas in the ``latent contagion'' model, $W_0=X_{\texttt{unobs}}X_{\texttt{unobs}}^\top$.  Assuming a ``dot-product'' network formation process in the first case and a  sub-gamma network-formation model in the second, $X_{\texttt{unobs}}$ can be estimated from a (truncated) singular-value decomposition of the observed adjacency matrix $W$. The models in both settings differ from the one considered here.  Whereas we do not consider ``latent positions'', we do allow for individual effects (which can be seen as a one-dimensional ``latent position''), do not impose particular network formation protocols, treat $W$ as an error-ridden version of $W_0$ and allow $W_0$ to be the sum of a low-rank matrix and a sparse matrix. \citet{cai2021network} estimate leading hub and authority
centralities and their regression effects using a purely low-rank
representation. In contrast, we study the full latent adjacency
matrix in a spatial autoregressive model, where measurement error
affects the endogenous spatial lag, network-based instruments,
and the implied higher-order propagation through the Leontief
inverse. Our framework also accommodates low-rank-plus-sparse
latent networks and dense, potentially endogenous measurement error.} We build on this literature by developing two 
complementary estimation approaches. The first is a two-stage 
procedure that denoises the observed adjacency matrix by separating 
low-rank or sparse signal from measurement error before estimating 
spillover effects. The second is a regularized GMM estimator that jointly estimates the network and the 
model parameters. % in a single step. 
After low-rank-plus-sparse regularization, the impact of network
measurement error is attenuated through averaging across the network,
resulting in an additional factor of $1/n$ in the estimation error bound. Consequently, the proposed estimator can remain
consistent in dense-error regimes in which GMM based on the raw
network may remain asymptotically biased. Simulation results show that the proposed estimators reduce the mean Frobenius-norm error in estimating the true direct network interaction matrix by at least 40\% relative to conventional GMM estimation based on noisy observed spatial weight matrices.

Our main theoretical contribution is to characterize how the bias 
of standard spillover estimators depends on the structure of the 
latent network and the measurement error, and to show that this 
bias need not vanish in the dense regime that characterizes economic 
networks. 
\citet{Lewbeletal2024EJ} show that link-misclassification bias in
linear social-interaction models may be negligible when
misclassification is not pervasive; we study settings in which the
latent network has many individually weak connections and network
measurement error need not be sparse. We further allow the 
measurement error to be correlated with unobservables in the outcome 
equation \citep{Candelaria2023, lewbel2024estimating}, which can 
introduce non-vanishing bias if ignored. We derive convergence rates 
for both the spillover-parameter estimator and our estimated
adjacency matrix. Beyond these formal results, the framework also 
delivers an interpretable decomposition of network connections (a 
low-rank component capturing pervasive relationships and a sparse 
component isolating localized interactions, with identification 
following standard strategies in the literature 
\citep{Candes2011, chandrasekaran2011rank}) and we use this decomposition to 
provide a novel expression of the network eigencentrality that 
quantifies the contributions of each component.

We illustrate the practical relevance of our framework through two 
empirical applications: international spillovers in GDP growth across 
23 economies and tax competition among 48 U.S.\ states. In both 
settings, our estimated network produces spillover estimates that 
differ markedly from those obtained using the raw adjacency matrix, 
highlighting the empirical importance of network denoising.  In the 
cross-economy GDP-growth application, regularization affects not only the scale of spillovers but also the inferred structure of interdependence. %denoising attenuates estimated spillovers by 5\%–21\% relative to the raw matrix, depending on the decomposition used. 
The network decomposition further reveals 
distinct regional and global components of cross-economy 
interdependence. These empirical findings demonstrate how our 
framework can uncover the structure of interdependence in economic 
networks and support more reliable inference in spatial econometric 
analysis. In the 
U.S.\ tax-competition application, the raw inverse-distance weight 
matrices could yield explosive spillover effect estimates $\widehat\lambda > 1$ that violate 
the Leontief stability condition, while our estimator 
returns stable estimates. 

The remainder of our paper is organized as follows. 
Section~\ref{Sec:Motivation} motivates the framework by explaining 
why measurement error matters in dense networks,
and demonstrating with the U.S.\ Input-Output Table that real-world 
dense networks exhibit exploitable regularization structure. 
Section~\ref{Sec:Methodology} presents the spatial interaction model 
and introduces the plug-in and regularized GMM estimators under 
low-rank, sparse, or low-rank-plus-sparse structure. 
Section~\ref{Sec:theory} provides assumptions and establishes 
consistency and asymptotic normality results for both estimators. 
Section~\ref{Sec:Simulation} presents simulation evidence. 
Section~\ref{Sec:Application} applies our methods to estimate 
spillover effects using noisy network data. 
Section~\ref{Sec:Conclusion} concludes. Proofs, algorithms, and 
additional technical details are provided in the appendices.
}

\section{Motivation} \label{Sec:Motivation}
This section sets up the notation and previews the regularization 
framework used throughout the paper. Let $\mathcal{N} \coloneqq 
\{1, \ldots, n\}$ denote the set of $n$ units or nodes and 
$\mathfrak{E}$ the set of edges between them, so that 
$(\mathcal{N}, \mathfrak{E})$ represents the graph. Let $W_0$ be the 
$n \times n$ adjacency matrix encoding the true connections, with 
$(i,j)$-th element $W_{0,ij}$; we write $W_{0,ij} \neq 0$ when unit 
$i$ is connected to unit $j$ and $W_{0,ij} = 0$ otherwise.\footnote{We 
allow the connection from $i$ to $j$ to differ from that from $j$ 
to $i$, so the graph may be directed and $W_0$ need not be symmetric.} 
We further impose the no-self-loop convention for the spatial and
social interaction model, so that $W_{0,ii}=0$ for all
$i=1,\ldots,n$.\footnote{This restriction is not imposed in the
input--output illustration, where diagonal entries record
within-sector intermediate-input flows.} Before formally motivating 
the regularization framework, we briefly preview the structural 
decomposition of $W_0$ that plays a central role throughout the paper.

To reduce estimation complexity for large $n$, it is natural to 
impose assumptions on $W_0$. The most common choice (sparsity) is plausible for social networks dominated by a few 
strong links, but fails to describe networks in which a large 
number of units share common interests or characteristics: a 
highly sparse network typically splits into several disconnected 
components. We therefore propose to model $W_0$ as the sum of a 
low-rank and a sparse component,
\begin{equation}
    \label{Eq:Std_Matrix_Structure1}
    W = \underbrace{L_0 + S_0}_{\equiv W_0} + E, 
\end{equation}
where $L_0$ captures pervasive, common connections and $S_0$ 
captures sporadic or idiosyncratic connections.

The components $L_0$ and $S_0$ need not individually satisfy the 
zero-diagonal constraint. Let $\operatorname{diag}(A)$ denote the diagonal matrix whose
diagonal entries coincide with those of $A$. When $W_0$ has a zero diagonal, we can 
define $L_0^* = L_0 - \diag(L_0)$ and $S_0^* = S_0 + \diag(L_0)$, 
which absorb the diagonal corrections so that both represent 
networks without self-loops\footnote{This ensures that $L_{0,ii}^* = 0$ and $S_{0,ii}^* = 0$ for all
$i=1,\ldots,n$ and rules out own-unit interactions from the network spillover channel.}:
\begin{equation}
    \label{Eq:Std_Matrix_Structure2}
    W = \underbrace{L_0^* + S_0^*}_{\equiv W_0} + E.
\end{equation}
Although $L_0^*$ need not be low-rank, its off-diagonal entries
coincide with those of the low-rank matrix $L_0$.
Online Appendix~\ref{set:Iden} illustrates the usefulness of this decomposition through 
concrete examples.\footnote{Appendices F and G are available in the arXiv version of this
paper \citep{li2026estimating}.} We further develop our theoretical analysis (Section \ref{Sec:theory}) in terms of $L_0$ and $S_0$  and 
note that the results extend directly to $L_0^*$ and $S_0^*$. Both theoretical and empirical evidence in 
the following sections show that this decomposition captures 
several meaningful network structures encountered in practice 
\citep{newman2010networks,Diebold2014, Barranca2015, Gamble2016}.

\subsection{Why dense measurement errors matter} 
\label{sec:densemeasurement}

Observed economic networks may be contaminated by reporting,
imputation, reconstruction, or link-classification errors.
For example, input--output tables are subject to reporting and revision;
bilateral trade records can differ between exporter and importer
reports; interbank exposures are often reconstructed from aggregate
balance-sheet information; and reported social links may contain
false positives and false negatives
\citep{Acemoglu2012,FismanWei2004,Upper2011,Anand2015,
Lewbeletal2024EJ}. 
We take the observational unit and aggregation level as given and
address measurement error at that level. We document four such sources of error in the Online Appendix~\ref{Sec:Sources}; 
here, we demonstrate why dense network measurement error can have a non-negligible impact on spillover estimation and potentially undermine consistency.

The existing literature, such as \citet{Lewbeletal2024EJ}, assumes 
that measurement error in the network is sufficiently sparse that 
it does not affect the consistency of the spillover estimator. In 
particular, consistency of the spillover parameter can be preserved 
when the aggregate magnitude of the error vanishes asymptotically, 
for example when
$\frac{1}{n}\sum_{i=1}^n\sum_{j=1}^n |E_{ij}|=o_p(1)$,
where $E_{ij}$ denotes the $(i,j)$th element of $E$.
This condition is naturally satisfied when the accumulated error 
grows slowly relative to the size of the network, as in settings 
where the underlying adjacency matrix itself is sparse. 
Intuitively, when only a small fraction of links are contaminated 
{and the magnitude of each contamination is small}, the 
resulting perturbation to the network propagation mechanism 
becomes negligible in large samples.

The situation is fundamentally different when measurement errors 
are \emph{dense}. Even when each entry of $E$ is elementwise small, 
the errors may accumulate across the $O(n^2)$ entries of the matrix 
and across the propagation paths generated by the spatial 
multiplier. The aggregate distortion may not vanish asymptotically, and a GMM estimator that does not acknowledge such measurement errors may fail to be consistent.

To recover $W_0$ from $W$ in the dense regime, some restriction on $W_0$ is needed. Depending on the application, 
the appropriate structure is:
\begin{itemize}
    \item[i)] \textbf{Sparse:} $W_0 = S_0^*$, where most links are 
    absent. Suitable for social networks.
    \item[ii)] \textbf{Low-rank:} $W_0 = L_0^*$, where connections 
    arise from a small number of common factors or shared 
    institutional features. A primary use case for dense networks is the analysis of financial exposure matrices. 
    \item[iii)] \textbf{Low-rank plus sparse:} $W_0 = L_0^* + S_0^*$, 
    the general case combining pervasive common connections with 
    idiosyncratic strong links.
\end{itemize}
Misspecifying this structure leaves omitted variable bias: assuming 
$W_0 = S_0^*$ when a non-negligible $L_0^*$ is present discards the 
pervasive systematic connections that drive spillovers; assuming 
$W_0 = L_0^*$ when strong idiosyncratic links are present absorbs 
those links into the residual. Moreover, the low-rank and sparse decomposition introduced above is the device we use to 
operationalize this regularization.
The decomposition need not be unique without additional
restrictions. We do not require the components
to have separate economic interpretations.
Nevertheless, Appendix~\ref{Appendix:decomp} discusses an 
interpretation for $L_0^*$, while Appendix
\ref{subsec:iden_examples} provides sufficient conditions for identification of its low-rank and sparse components.

The practical importance of regularizing $W_0$ rests on three 
observations. First, measurement error in dense networks is 
\emph{non-negligible}: it may not vanish with sample size, so 
standard asymptotic consistency arguments fail. Second, it is 
\emph{endogenous}: errors share common sources and are typically 
related to the true network entries, violating the classical 
measurement-error assumptions that justify, for example, standard 
instrumental-variable corrections. Third, it is \emph{economically 
consequential}. In our empirical illustration using the U.S.\ Input--Output Table (Section~\ref{empirical}), the Leontief inverse based on the low-rank plus sparse network estimate differs from that based on the raw network by 13\%, relative to the spectral norm of the latter. Imposing sparsity on the estimated network yields a 60\% difference relative to the same benchmark. These results indicate sensitivity to modeling choices and underscore the importance of carefully assessing measurement-error assumptions when drawing inferences about network spillovers.

{\subsection{Empirical illustration: the US Input-Output Table}
\label{empirical}

To demonstrate that measurement error in dense economic networks has economically
meaningful consequences, and that the low-rank plus sparse decomposition provides an effective device for
recovering the underlying true network structure, we revisit the 2002 US Input-Output
table from \cite{graham2020econometric}, compiled by the Bureau of Economic Analysis (US)
and studied by \cite{Acemoglu2012}. It contains an $n \times n$ network of $n=416$
sectors (after dropping rows and columns with zero direct input requirements), where
entry $(i,j)$ indicates the share of sector $i$'s output used in sector $j$'s production.
This network is dense by construction, making it a canonical example where sparsity cannot be assumed and measurement error is pervasive.
We treat the observed $W$ as a noisy version of the true network $W_0$ and minimize equation (\ref{Eq:LRSED}) without the diagonal constraint to recover $W_0$.

We compare the purely sparse estimate $\widehat{S}$ with
the low-rank-plus-sparse estimate $\widehat{W}=\widehat{L}+\widehat{S}$,
using the observed noisy $W$ as the reference. The purely sparse
estimate yields
$\|\widehat{S}-W\|_{\max}/\|W\|_{\max}=0.037$ but
$\|\widehat{S}-W\|_2/\|W\|_2=0.323$. The small entrywise discrepancy
alongside the much larger spectral-norm discrepancy indicates that
many tiny entrywise terms can aggregate into a non-negligible
spectral-norm difference. Adding the low-rank component leaves the max-norm
discrepancy unchanged,
$\|\widehat{W}-W\|_{\max}/\|W\|_{\max}=0.037$, while reducing the
spectral-norm discrepancy to
$\|\widehat{W}-W\|_2/\|W\|_2=0.145$. This indicates that the low-rank
component captures global structure discarded by the purely sparse
estimate.

To quantify the econometric consequences of measurement error, we compare
the Leontief inverse matrices \((I_n - \lambda W)^{-1}\),
\((I_n - \lambda \widehat{W})^{-1}\), and
\((I_n - \lambda \widehat{S})^{-1}\) following
\cite{acemoglu2016networks} with \(\lambda = 0.7\). %Using raw noisy \(W\) instead of 
Our estimate 
produces a \(13\%\) relative deviation in the spectral
norm:
\[
\frac{\|(I_n-0.7\widehat{W})^{-1}-(I_n-0.7W)^{-1}\|_2}
     {\|(I_n-0.7W)^{-1}\|_2}=13\%.
\]
Using a purely sparse estimate \(\widehat{S}\), which discards the
low-rank component, leads to an even larger \(60\%\) relative deviation:
\[
\frac{\|(I_n-0.7\widehat{S})^{-1}-(I_n-0.7W)^{-1}\|_2}
     {\|(I_n-0.7W)^{-1}\|_2}=60\%.
\]
These results demonstrate that measurement error in network data is not
merely a theoretical concern, as it can materially distort estimated
propagation channels. The choice of error structure therefore matters, and we recommend including both a sparse idiosyncratic component and a low-rank
systematic component for flexibility and reliable inference.
Figure~\ref{heatmap2} illustrates this point visually: the stripe pattern
present in the diffusion matrix based on the estimator
\((I_n - 0.7\widehat{W})^{-1}\) is substantially attenuated when the
low-rank component is omitted, indicating that misspecification of
network structure could bias spatial feedback in economically relevant
ways.

\begin{figure}[htbp]
\centering
\begin{subfigure}{0.3\textwidth}
    \includegraphics[width=\textwidth]{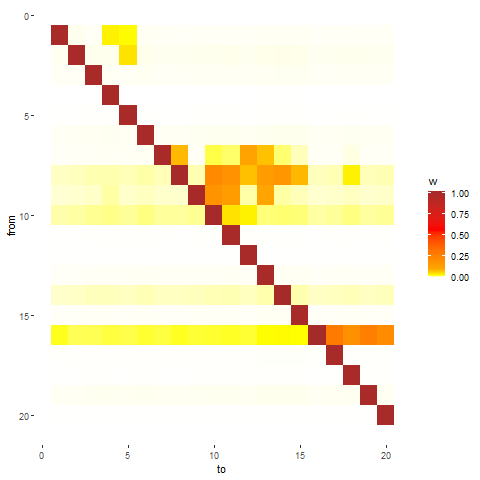}
\end{subfigure}
\hfill
\begin{subfigure}{0.3\textwidth}
    \includegraphics[width=\textwidth]{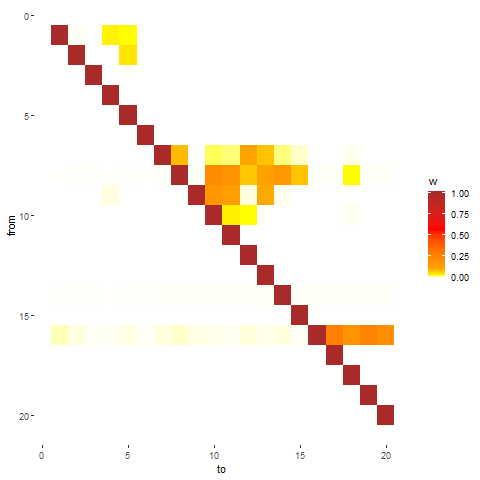}
\end{subfigure}
\hfill
\begin{subfigure}{0.3\textwidth}
    \includegraphics[width=\textwidth]{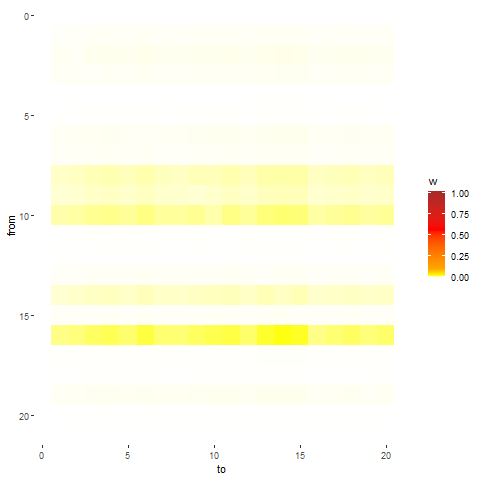}
\end{subfigure}
\caption{\footnotesize Heatmaps of 20 sectors from the diffusion matrix
$(I_n-0.7\widehat{W})^{-1}$ (left), $(I_n-0.7\widehat{S})^{-1}$ (centre), and
the difference (right). 
\label{heatmap2}}
\end{figure}
}

\section{Methodology} \label{Sec:Methodology}

We present the general model in which the latent network $W_0 = L_0 + S_0$
admits both a low-rank and a sparse component, with the purely low-rank and purely sparse settings arising as special cases.
We emphasize that this decomposition need not be unique or economically interpretable on its own; rather, the joint regularization imposed on $L_0$ and $S_0$ provides sufficient structure to separate $W_0$ from the noise $E$.  We provide the modeling definitions and estimation procedures we propose to tackle this problem. We offer additional assumptions and asymptotic results in Section~\ref{Sec:theory}.
    \subsection{Estimation}\label{Sec:Estimators}
	We formally describe here the social/spatial interaction model of interest. We consider a panel data scenario, but %will generally 
assume the true underlying network does not vary with time. Denote data points as $\{Y_t, X_t\}_{t=1, \cdots, T}$, where $Y_t$ is the $n \times 1$ outcome vector and $X_t$ is the $n \times K$ design matrix of exogenous covariates. We assume that the true model is generated by the unobserved $n \times n$ adjacency matrix $W_0$ encoding the true connections, and that we  have a noisy or contaminated observation, such that we can write
\begin{eqnarray}
Y_t &=& \lambda_0 W_0 Y_t + \alpha + \iota_t {\bf 1}_{n\times 1} + X_t\beta_0 + W_0 X_t \gamma_0 + \varepsilon_t, \qquad W   = W_0 + E,  \label{maineq}
\end{eqnarray}
where $\lambda_0$ represents the scalar (social/spatial) interaction or spillover effect, $\beta_0 \in \mathbb{R}^K$ are standard slope coefficients, $\gamma_0 \in \mathbb{R}^K$ are contextual effects (spillover effects of the observable characteristics $X$), $\alpha = (\alpha_1, \ldots, \alpha_n)^\top$ is the vector of individual fixed effects, $\iota_t$ is a time fixed effect at period $t$, ${\bf 1}_{n\times 1}$ is an $n \times 1$ vector of ones, and $\varepsilon_t$ is an $n$-dimensional vector of idiosyncratic disturbances, assumed to be independent and identically distributed over units $i$ and time $t$. We assume the covariates $X_t$ to be exogenous to the outcome unobservables and network formation process such that   $\exv[\varepsilon_{it}| X_t, W_0X_t, \cdots, W_0^dX_t] = 0$, for some integer $d>0$, $i = 1,\cdots, n$ and $t = 1, \cdots, T$. Crucially, we allow for the measurement errors to be endogenous in the sense that  $\exv[\varepsilon_{it}| E] \neq 0$, 
which introduces additional bias into conventional estimates when $W_0$ is replaced by $W$ \citep{boucher2025estimating, hardy2019estimating, Lewbeletal2024EJ}. \\
Collect all parameters of interest into $\theta_0 \coloneqq (\lambda_0, \beta^{\T}_0, \gamma^{\T}_0)^{\T} \in \mathbb{R}^{2K+1}$. When covariates $X_t$ are exogenous and the adjacency matrix $W_0$ is correctly observed,
the literature proposes to select $m$ linearly independent columns 
from $[X_t, W_0 X_t, \ldots, W_0^d X_t]$, for some positive 
integer $d$, as instruments for the regressors 
$\bar{X}_t(W_0) = [W_0 Y_t, X_t, W_0 X_t]$ in estimating $\theta_0$ 
\citep{Kelejian1998, Lee2003, Lee2007}.
The integer $d$ should be chosen such that the number of retained instruments %, denoted as $m$, 
is at least $2K+1$ to satisfy the order condition, i.e.  $2K+1 \leq m$ (see Assumption \ref{Assump:Instruments}). For an arbitrarily pre-specified $n \times n$ adjacency matrix $M$, let $Z_t(M)$ be the $n \times m$ matrix that collects the linearly independent columns from $[X_t, M X_t, \ldots, M^d X_t]$ and let ${X}_t(W)$ be the $n \times (2K+1)$ matrix that collects the regressors $[W Y_t, X_t, W X_t]$. Note that ${X}_t(W)$ uses the observed adjacency matrix $W$ as opposed to the true $W_0$ showing up in model \eqref{maineq}. \\
For the baseline theory in Section~\ref{Sec:theory}, we set
$\alpha=\mathbf 0_n$ and $\iota_t=0$ and use the original,
untransformed variables. Accordingly, throughout the theoretical
analysis, $Y_t$, $X_t(W)$, $Z_t(M)$, and $\varepsilon_t$ denote the
untransformed quantities. In empirical applications with individual
and time fixed effects, we apply the transformation described in
Remark~\ref{rem:twoway_demean} to the outcome and to each already
constructed regressor and instrument column before forming the moment
conditions. We discuss briefly how the theoretical results can be extended in that remark.\\
By defining $\varepsilon_t(\theta, W) \coloneqq Y_t - {X}_t(W) \theta$ and thus $\varepsilon_t =\varepsilon_t(\theta_0, W_0)$ , we obtain the population moment condition  $\exv[Z_t\T(M)\varepsilon_{t}] = \mathbf{0}_{m\times 1}$ for any exogenous $M$.
For the fixed-effects implementation, the analogous moment condition
is formed using the transformed variables under the strict-exogeneity
condition in Remark~\ref{rem:twoway_demean}.
Thus we have the empirical moment conditions:
\begin{align}
	\label{Eq:Moment_Conditions}
	\bar{g}_{nT}(\theta, W, M) & \coloneqq \frac{1}{nT} \sum_{t=1}^{T} Z_t(M) \T \varepsilon_t(\theta, W) \, .
\end{align}
The distinction made between the matrix used to construct the regressors ${X}_t(W)$ and the instruments $Z_t(M)$ is necessary to define our supervised estimator below. When $M=W$ such that no distinction is necessary, we simply write $\bar{g}_{nT}(\theta, W)  = \bar{g}_{nT}(\theta, W, W)$. For a given $m \times m$ moment-weighting matrix $\widehat{\Lambda}_{nT}^{-1}$, we define the sample GMM objective function as
\begin{equation*}
	J_{nT}(\theta, W, M; \widehat{\Lambda}_{nT}^{-1}) \coloneqq  (nT)\bar{g}_{nT}(\theta, W, M)\T \widehat{\Lambda}_{nT}^{-1} \, \bar{g}_{nT}(\theta, W, M) .
\end{equation*}
The GMM estimator is defined as
\begin{equation}
    \label{Eq:GMM_estimator}
    \widehat{\theta}_{\mathrm{GMM}}(W,M)
    \coloneqq
    \arg\min_{\theta\in\Theta}
    J_{nT}\!\left(
        \theta,W,M;\widehat{\Lambda}_{nT}^{-1}
    \right).
\end{equation}
When the same network matrix is used to construct both the regressors
and the instruments, we write
$\widehat{\theta}_{\mathrm{GMM}}
\coloneqq
\widehat{\theta}_{\mathrm{GMM}}(W,W)$.
In practice, we can obtain a two-step GMM estimator by using the weight matrix  $\widehat{\Lambda}_{nT}^{-1}$. \footnote{In practice, we obtain a feasible two-step GMM estimator by first computing
$\widetilde{\theta}$ using the identity weighting matrix and then using
$\widehat{\Lambda}_{nT}^{-1}$, where
\[
\widehat{\Lambda}_{nT}
=
\frac{1}{nT}\sum_{t=1}^T
Z_t(W)^\top
\operatorname{diag}\!\left(
\widehat{\varepsilon}_{1t}^{\,2},\ldots,
\widehat{\varepsilon}_{nt}^{\,2}
\right)
Z_t(W),
\qquad
\widehat{\varepsilon}_t
=
Y_t-X_t(W)\widetilde{\theta}.
\] } 
The rates of our estimators along with all formal assumptions and theoretical results are presented in Section \ref{Sec:theory}. When $W_0$ is observed without error, \cite{Kelejian1998} and \cite{Lee2007}, for example, guarantee that the standard GMM estimator $\widehat{\theta}_\text{GMM}$ is consistent for $\theta_0$ in \eqref{maineq}. However, when measurement error is present, small entrywise errors in $E$ can accumulate, leading to inconsistent estimation of the regression coefficients and spillover effects.\\
Therefore, when measurement error exists, we need a way to correct for this additional noise in our estimation. We make use of the information afforded by the low-rank plus sparse structure as in \eqref{Eq:Std_Matrix_Structure1} to obtain an estimate $\widehat W$. Theoretically, we quantify the estimation accuracy of $\widehat W$ %this estimate for $W_0$ 
and its asymptotic influence when used as an input for estimating $\theta_0$. We now describe the details leading to our two proposed estimators. Our first option is a ``plug-in estimator'', and the other is a ``supervised estimator'', both defined in the following subsection.
\vspace{-0.2cm}
\subsection{The plug-in estimator and the supervised estimator}
\label{Subsec:Plug-in}
Our proposed estimators make use of %take advantage of 
the structure for the adjacency $W$ in \eqref{maineq} by using penalized convex low-rank and sparse estimation. Our first proposal is simply a plug-in estimator that first produces a denoised estimate for $W_0$ and uses this estimate in place of the observed version when calculating the GMM estimator. Low-rank plus sparse regularization motivates the following
joint formulation for network denoising
\citep{Candes2011,Cao2017}:
\begin{equation}
\label{Eq:LRSED}
\min_{\substack{
L,S\in\mathbb R^{n\times n}\\
\operatorname{diag}(L+S)=\mathbf 0_{n\times n}}}
\left\{
\frac12\|W-L-S\|_F^2
+\nu_n\|L\|_*
+\tau_n|S|_{1,1}
\right\},
\end{equation}
where $\|\cdot\|_F$ and $\|\cdot\|_*$ denote the Frobenius
and nuclear norms, respectively, $|S|_{1,1} =\sum_{i=1}^n\sum_{j=1}^n |S_{ij}|$, and $\nu_n,\tau_n>0$
are tuning parameters. \footnote{It shall be noted that the constraint $\operatorname{diag}(L+S)=\mathbf{0}_{n\times n}$ is not necessary for some applications, e.g. the Input-Output matrix.}\\   
We use a sequential thresholding procedure as a computational
proxy for this joint formulation as in Algorithm~\ref{Algo:CD_LRSED}.
We denote the resulting component estimates by
$(\widehat L,\widehat S)$.
The theoretical results below concern this sequential estimator.
Defining $\widehat{W} \coloneqq \widehat{L} + \widehat{S}$, our plug-in estimator is given by replacing $W$ with its estimate in  \eqref{Eq:GMM_estimator}, 
\begin{eqnarray}
	\widehat{\theta}_{p}(\widehat W) \coloneqq \arg \min_{\theta \in \Theta} J_{nT}(\theta, \widehat{W},\widehat{W}; \widehat{\Lambda}_{nT}^{-1})\nonumber
= \left[{X}(\widehat{W})\T  Z(\widehat{W}) \widehat{\Lambda}_{nT}^{-1} Z(\widehat{W})\T  {X}(\widehat{W})\right]^{-1} {X}(\widehat{W})\T  Z(\widehat{W}) \widehat{\Lambda}_{nT}^{-1} Z(\widehat{W})\T  Y . 
     	\label{Eq:Plugin_Estimator}
\end{eqnarray}

The following joint optimization problem motivates our
supervised estimation procedure:

\begin{equation}
\label{Eq:Supervised_Estimator}
\begin{aligned}
\min_{\substack{
\theta\in\Theta,\;L,S\in\mathbb R^{n\times n}\\
\operatorname{diag}(L+S)=\mathbf 0_{n\times n}
%L_{ii}+S_{ii}=0,\;i=1,\ldots,n}
}}
\quad&
\frac{\xi_{nT}}{nT}
J_{nT}(\theta,L+S,M;\widehat{\Lambda}_{nT}^{-1})
+\frac12\|W-L-S\|_F^2
+\nu_{nT}\|L\|_*
+\tau_{nT}|S|_{1,1}.
\end{aligned}
\end{equation}
$\xi_{nT}\geq0$ controls the weight placed on the GMM criterion relative to the network-fitting
loss and regularization penalties. We use the computational
procedure in Algorithm~\ref{Algo:ADMM_GMM_FBS} in
Appendix~\ref{Appendix:Algorithms} as a proxy for the joint
formulation in \eqref{Eq:Supervised_Estimator} and denote its
output by $(\widehat{\theta}_s,\widehat L_s,\widehat S_s)$.
\citet{cai2021network} also study a supervised method based on
a purely low-rank structure for analyzing network centrality.
Their approach focuses on the leading singular vectors and
does not incorporate a spatial regression.

 \subsection{Bias of the GMM estimator}
Before establishing the properties of our estimator, we first examine the bias introduced by 
$E$ when using the classical GMM estimator.
We provide an intuitive argument illustrating how the measurement error matrix \(E\) biases the moment functions and the importance of its removal, particularly when the instrument matrix \(Z_t(M)\) also depends on \(W\), as in standard spatial regression.  For simplicity here, assume no contextual effects are present. When $Z_t(M)$ consists of exogenous instruments satisfying 
$\mathbb{E}[Z_t(M)^\top \varepsilon_t] = \mathbf{0}_{m \times 1}$, 
the moment function is given by
$\bar g_{nT}(\theta, W, M) = \frac{1}{nT} \sum_{t=1}^{T} 
Z_t(M)^\top (Y_t - \lambda W Y_t - X_t \beta)$.
Substituting $W = W_0 + E$,
\[
\bar g_{nT}(\theta, W, M) = \frac{1}{nT} \sum_{t=1}^{T} 
Z_t(M)^\top (Y_t - \lambda W_0 Y_t - X_t \beta) 
- \frac{\lambda}{nT} \sum_{t=1}^{T} Z_t(M)^\top E Y_t.
\]

Thus, if $E$ is independent of $(Z_t(M),Y_t)$ and has mean zero, then
$\mathbb{E}\!\left[Z_t(M)^\top E Y_t\right]
=\mathbf{0}_{m\times 1}$,
and hence
$\mathbb{E}\!\left[\bar g_{nT}(\theta_0,W,M)\right]
=\mathbf{0}_{m\times 1}$.
However, when $Z_t(M)$ is a plug-in version of the true instrument 
with $M = W$, i.e.,
$Z_t(W) = [X_t, W X_t] = [X_t, (W_0 + E) X_t]$,
the moment function becomes
$\bar g_{nT}(\theta, W) = \frac{1}{nT} \sum_{t=1}^{T} 
[X_t, (W_0 + E) X_t]^\top \big[(Y_t - \lambda W_0 Y_t - X_t \beta) 
- \lambda E Y_t\big]$.
Expanding the moment function $\Bar{g}_{nT}(\theta, W)$ with the 
instrument matrix $[X_t, (W_0+E)X_t]$, we find that the measurement-error-induced difference includes  
{\[
 \left[ \frac{1}{nT} \sum_{t=1}^{T}(-\lambda X_t^{\top} EY_t)\T,   \frac{1}{nT} \sum_{t=1}^{T}[X_t^\top E^\top (Y_t-\lambda W_0 Y_t-X_t\beta) - \lambda X_t^\top W_0^\top E Y_t - \lambda X_t^\top E^\top E Y_t]\T \right]\T.
\]} 
Even if \(E_{ij}\)'s are i.i.d. mean zero, with bounded second moment and independent of $(X_t,\varepsilon_t)$,  there will be moment condition misspecification. Specifically,  
$\mathbb{E}[\Bar{g}_{nT}(\theta_0, W)] \neq { \mathbf{0}_{m\times 1}}$,
since the expectation of the last term yields:  
$n^{-1}\mathbb{E}[X_t^\top E^\top E Y_t] 
\neq \mathbf{0}_{K\times 1},$
which might be non-negligible under our assumptions.  
Recall that \(\widehat{\theta}_{GMM}\) denotes the efficient two-step GMM estimator based on a given adjacency matrix \(W\), obtained in \eqref{Eq:GMM_estimator}. From  Proposition \ref{lemmabias} in the following section, it follows that \(\widehat{\theta}_{GMM}\) is likely to be inconsistent.

    \section{Asymptotic theory}
\label{Sec:theory}

In this section, we discuss the theoretical properties of our estimation procedures. Our main interest is to quantify the impact of using an estimated adjacency matrix on recovering both the true adjacency matrix and the spillover parameters of interest.
We shall use the following notation:  Let $A$ and $B$ be matrices of dimensions $n\times n$. For a symmetric matrix $A \in \mathbb{R}^{n \times n}$, let
$\lambda_i(A)$ denote its $i$th largest eigenvalue,
$i=1,\ldots,n$, and let $\lambda_{\min}(A)$ denote its smallest
eigenvalue. For any matrix $A \in \mathbb{R}^{n \times n}$,
let $\sigma_i(A)$ denote its $i$th largest singular value,
$i=1,\ldots,n$.  Let $\|.\|_2$ denote the Euclidean norm of a vector. We use $\|A\|_1$, $\|A\|_2$, $\|A\|_{\infty}$, $\|A\|_{\max}$ %, $\|A\|_F$ and $\|A\|_*$ 
to denote the $l_1$-norm, $l_2$-norm, $l_{\infty}$ norm, max norm,  
respectively. For a matrix $A \in \mathbb{R}^{n \times n}$, let:
%\begin{align*}
$m_r(A) := \max_{1 \leq i \leq n} \sum_{j=1}^n \mathbb{I}(A_{ij} \neq 0)$, $m_c(A) := \max_{1 \leq j \leq n} \sum_{i=1}^n \mathbb{I}(A_{ij} \neq 0)$, $m_s(A) := \sum_{i,j} \mathbb{I}(A_{ij} \neq 0)$, and $\deg_{\max}(A) := \max\{m_r(A), m_c(A)\}$. 
%\end{align*}

We use $e_i$ to denote the standard basis vector, i.e., the $i$th column of the identity matrix. For sequences $\{a_n\}_{n = 1}^{\infty}$ and $\{b_n\}_{n = 1}^{\infty}$,
we denote $a_n\lesssim b_n$ if there exists a positive constant $C$ such that $a_n/b_n\le C$ for all $n$, and {we write $a_n=o(b_n)$, or equivalently $a_n\ll b_n$,
if $a_n/b_n\to0$.
We write $a_n\asymp b_n$ if $a_n\lesssim b_n$
and $b_n\lesssim a_n$.} 
We use $\overset{d}{\to}$ and ${\to_p}$ to indicate convergence in distribution and in probability, respectively. We use $X_n = O_p(a_n)~(\lesssim_p a_n)$ to indicate $X_n/a_n$ is bounded in probability, i.e. for any $\e > 0$, there exists $M(\e) > 0$  %${\color{blue} M_{\e}}>0$ $M(\e) > 0$ 
such that $\P(|X_n/a_n| > M(\e) ) < \e$ for all $n \in \mathbb{N}$. We assume $\mbox{min}{(n,T)}\to \infty$ or $n \to \infty$.

\begin{assumption}[True weight]
\label{weight}
\begin{enumerate}[(i)]
    \item (Low-rank component) When $L_0\neq \mathbf{0}_{n\times n}$, $L_0$ is an $n\times n$ low-rank matrix 
    with $\mathrm{rank}(L_0) = r$, 
    $\min(m_r(L_0), m_c(L_0)) \to \infty$, and 
    $\|L_0\|_{\max} = O(1/\sqrt{m_r(L_0)m_c(L_0)})$.
 \item (Sparse component) $S_0^*=S_0-\diag(S_0)$ is an $n\times n$ sparse matrix 
with $\|S_0^*\|_{\max} = O(1)$. When $L_0=\mathbf{0}_{n\times n}$, set $r=0$ and assume
$\deg_{\max}(S_0^*)=o(\sqrt{n})$.

    \item (Relative sparsity) 
    When $L_0\neq \mathbf{0}_{n\times n}$, $\frac{m_s(S_0^*)}{m_r(L_0)m_c(L_0)} \to 0$ and 
    $\frac{\deg_{\max}(S_0^*)^2}{\min\{m_r(L_0),\,m_c(L_0)\}} \to 0$.
\end{enumerate}
\end{assumption}

Assumption~\ref{weight} collects conditions on $W_0$. Part~(i) 
requires the low-rank signal to be pervasive but individually 
weak: the condition $\|L_0\|_{\max} = O(1/\sqrt{m_r(L_0)m_c(L_0)})$ 
allows $L_0$ to be distinguished from  $S_0^*$; %a sparse perturbation; 
it fails when $L_0$ is spiked (as in the 
dominant-units example), in which case the low-rank part is better absorbed 
into $S_0^*$. The rank $r$ may diverge with $n$ provided the relevant rates still vanish. 
Part~(ii) requires $S_0^*$ to be bounded and Part (iii) allows $S_0^*$ to %have few but strong 
{have relatively few nonzero links compared with the dense part.} 

Hence the decomposition $W_0 = L_0 + S_0$ is uniquely determined when $n$ is large, making the 
rank parameters $r$ and $m_s(S_0^*)$ well-defined. We stress this is a 
\emph{mathematical} requirement only: $\widehat\theta_p$ uses 
$\widehat W = \widehat L + \widehat S$ as a single object, so we rely 
on no economic interpretation of the components, and inference about 
$\theta$ is robust to weak identification of the decomposition. 
Part~(iii) is imposed only when $L_0\neq\mathbf{0}_{n\times n}$ and is
automatically satisfied when $S_0=\mathbf{0}_{n\times n}$.
Together, the conditions in Assumption~\ref{weight}
imply $\|W_0\|_1=o(\sqrt n)$, which is used in the
rate analysis of $\widehat\theta_p(\widehat W)$. The following 
assumption imposes conditions on the measurement error $E$ in the 
observed matrix $W$, required for the large-sample analysis as 
$n \to \infty$.

\begin{assumption}[Noise]
    \label{Assump:E_noise}
    \begin{enumerate}[(i)]
 \item $E$ is an $n \times n$ noise matrix such that $E_{ii}=0$ for $i=1, \cdots, n$.         
        \item $E$ satisfies $\|E\|_{2}\leq w_{2n}$ with probability greater than $1-p_{n}$, where $p_{n}=o(1)$.
        \item $E$ satisfies $\|E\|_{\max}\leq w_{\max n}$ with probability greater than $1-p_{n}$, where $p_{n}=o(1)$.
    \end{enumerate}
\end{assumption}

Assumption~\ref{Assump:E_noise} is mild. It does not require the 
entries of $E$ to be independent or identically distributed, a 
useful feature in spatial and social network contexts, where 
measurement errors often exhibit dependence. It also does not 
require $w_{2n}$ or $w_{\max n}$ to vanish; these are sequences 
that bound the spectral and entrywise norms of $E$, and may grow 
with $n$. Whether the bounds delivered by 
Theorems~\ref{Thm:DK1}--\ref{Thm:DK22} are informative depends on 
joint conditions involving $(w_{2n}, w_{\max n}, r, m_s(S_0^*), n)$, 
discussed after each theorem.

The roles of $w_{2n}$ and $w_{\max n}$ are complementary. The 
spectral norm bound $w_{2n}$ controls the aggregate noise level and 
governs how well the low-rank component $L_0$ can be separated 
from $E$. The entrywise maximum $w_{\max n}$ controls the largest 
individual perturbation and governs the recovery of $S_0$.  We illustrate what $E$ looks like in the following examples.

\begin{example}[A sample-covariance error matrix]
\label{ex:covariance}
Suppose the error arises from estimating a covariance matrix: let 
$z_1, \ldots, z_T$ be independent mean-zero Gaussian vectors in 
$\mathbb{R}^n$ with covariance $W_z$, whose eigenvalues are positive and bounded away from 0 and $\infty$. Let 
$\widetilde E = T^{-1}\sum_{t=1}^T z_t z_t^\top - W_z$ be the 
sampling error. To respect the no-self-loop convention we take 
$E = \widetilde E - \mathrm{diag}(\widetilde E)$, so 
$E_{ii} = 0$ and Assumption~\ref{Assump:E_noise}(i) holds.
Let $n,T \to \infty$.
Since
$\|\operatorname{diag}(\widetilde E)\|_2
\leq \|\widetilde E\|_2$, we have
$\|E\|_2\leq 2\|\widetilde E\|_2$ and
$\|E\|_{\max}\leq\|\widetilde E\|_{\max}$.
Thus, removing the diagonal preserves the orders of the
upper bounds below.
By Theorem~6.5 of \citet{wainwright2019high}, taking
$\eta=C_0\{\sqrt{n/T}+n/T\}$ for a sufficiently large
constant $C_0$ gives
$\|E\|_2
\leq C\|W_z\|_2\{\sqrt{n/T}+n/T\}
=:w_{2n}$
with probability at least $1-c_1\exp(-c_2n)$.
Applying Bernstein's inequality to each centered covariance
entry and taking a union bound over all entries gives
$\|E\|_{\max}
\leq C\bigl(\max_i W_{z,ii}\bigr)
\{\sqrt{\log(n)/T}+\log(n)/T\}
=:w_{\max n}$
with probability at least $1-2n^{-c_3}$, for a sufficiently
large constant $C$ and positive constants $c_1,c_2,c_3$.
 Thus,
Assumption~\ref{Assump:E_noise} holds with
$w_{2n}\asymp\sqrt{n/T}+n/T$ and
$w_{\max n}\asymp\sqrt{\log(n)/T}+\log(n)/T$.
$\blacksquare$ 
\end{example}

\begin{example}[Dense Gaussian noise: comparison with \citet{Lewbeletal2024EJ}]
\label{lewbel:e}
Suppose $E_{ij}$ are i.i.d.\ $\mathcal N(0, \sigma^2_{n})$ for $i \neq j$ 
and $E_{ii} = 0$, with $\sigma_n$ possibly diminishing in $n$.
As in Example~\ref{ex:covariance}, applying the operator-norm bound 
of \citet{wainwright2019high} and a union bound over the entries 
gives, with probability approaching $1$, 
$\|E\|_2 \leq 4\sigma_n\sqrt n =: w_{2n}$ and 
$\|E\|_{\max} \leq 8\sigma_n\sqrt{\log n} =: w_{\max n}$.
Assumption~\ref{Assump:E_noise} holds with
$w_{2n}=4\sigma_n\sqrt n$ and
$w_{\max n}=8\sigma_n\sqrt{\log n}$.
In particular, $w_{\max n}\to0$ if
$\sigma_n=o(1/\sqrt{\log n})$.
By contrast, \citet{Lewbeletal2024EJ} impose the stronger condition
$\sum_{i,j}\mathbb E|E_{ij}|=O(n^s)$ for some $s<1$.
Since $\sum_{i,j}\mathbb E|E_{ij}|\asymp n^2\sigma_n$, their
condition requires $\sigma_n=O(n^{s-2})=o(1/n)$.
For $1/n\ll\sigma_n\ll1/\sqrt n$, we have
$w_{2n}\to0$, so Theorem~\ref{Thm:DK22}(b) implies consistency
of $\widehat\theta_{GMM}$, whereas the condition of
\citet{Lewbeletal2024EJ} fails.
Our estimator 
$\widehat\theta_p(\widehat W)$ further sharpens the rate via the rate discounted by $\frac{1}{n}$, see Theorem \ref{Thm:DK22}(a).
For $\sigma_n\geq c>0$, the available bound for the GMM estimator
does not guarantee consistency.
$\blacksquare$
\end{example}

We next study the properties of the estimators
$\widehat L$ and $\widehat S$ produced by
Algorithm~\ref{Algo:CD_LRSED}. Define
\[
a_n
\coloneqq
\begin{cases}
\{m_r(L_0)m_c(L_0)\}^{-1/2},
& L_0\neq\mathbf 0_{n\times n},\\
0,&L_0=\mathbf 0_{n\times n}.
\end{cases}
\]
Define the rates
\begin{align*}
&\mathcal R_F(W_0,E)
=
\bar r(w_{2n}+a_n)^2
+
m_s(S_0^*)w_{\max n}^2
+
m_s(S_0^*)a_n^2,
\mathcal R_2(W_0,E)=
w_{2n}^2
+m_s(S_0^*)w_{\max n}^2
+m_s(S_0^*)a_n^2,\\&
\mathcal R_s(W_0,E)
=
(w_{\max n}+a_n)\sqrt{m_s(S_0^*)}.
\end{align*}

\begin{theorem}
\label{Thm:DK1}
Let $\bar r\coloneqq r\vee1$ and
$\widehat S^*:=\widehat S-\operatorname{diag}(\widehat S)
=\widetilde S$.
Let the tuning parameters satisfy
$\nu_n
=
C_\nu\left(
w_{2n}+a_n
+
\frac{\mathcal R_s(W_0,E)}{\sqrt{\bar r}}
\right),
\tau_n=C_\tau(w_{\max n}+a_n)$,
for sufficiently large constants $C_\nu,C_\tau>0$.
Suppose Assumptions~\ref{weight}--\ref{Assump:E_noise} hold.

\noindent\textbf{(a) Frobenius bound.}
The two-step estimator satisfies
\begin{align}
\|\widehat L-L_0\|_F^2+\|\widehat S^*-S_0^*\|_F^2
&\lesssim_p\mathcal R_F(W_0,E),\\
|\widehat S^*-S_0^*|_{1,1}
&\lesssim_p
\mathcal R_s(W_0,E)\sqrt{m_s(S_0^*)}.
\end{align}

\noindent\textbf{(b) Sharpened spectral-norm bound.}
The two-step estimator satisfies
\[
\|\widehat L-L_0\|_2^2
+
\|\widehat S^*-S_0^*\|_2^2
\lesssim_p
\mathcal R_2(W_0,E).
\]

\noindent\textbf{(c) Nuclear-norm bound.}
The two-step estimator satisfies
\[
\|\widehat L-L_0\|_*
\lesssim_p
\bar r(w_{2n}+a_n)
+
\sqrt{\bar r}\,\mathcal R_s(W_0,E).
\]
\end{theorem}

The rates $\mathcal R_F$ and $\mathcal R_2$ contain three
contributions: a spectral-noise term associated with low-rank
recovery, an entrywise-noise term
$m_s(S_0^*)w_{\max n}^2$ associated with sparse recovery, and
a term $m_s(S_0^*)a_n^2$ reflecting the entrywise bound on
the low-rank component.
The factor $\max\{r,1\}$ multiplying the spectral-noise
contribution in $\mathcal R_F$ is replaced by unity in
$\mathcal R_2$. The Frobenius rate has the same structure
as the bounds in Corollary~1 and Proposition~1 of
\citet{Agarwal2012}.
The corresponding estimation-error bounds also hold,
up to constant factors, for the diagonal-adjusted
estimators of $L_0^*$ and $S_0^*$. Since
$L_0^*=L_0-\operatorname{diag}(L_0)$, $\widetilde S =\widehat S^* $and
$\widehat L^*=\widehat L-\operatorname{diag}(\widehat L)$,
$\|\widehat L^*-L_0^*\|_2 \leq 2\|\widehat L-L_0\|_2$.

By Theorem~\ref{Thm:DK1}(b), sufficient conditions for
operator-norm consistency are
\[
w_{2n}^2\to0,
\qquad
m_s(S_0^*)w_{\max n}^2\to0,
\qquad
m_s(S_0^*)a_n^2\to0,
\]
in keeping with the standard low-rank plus sparse literature. For Frobenius consistency, it suffices to replace
$w_{2n}^2\to0$ with $\max\{r,1\}w_{2n}^2\to0$,
while retaining the other two conditions. 
Under suitable rank and sparsity conditions, the rate
condition in Theorem~\ref{Thm:DK22} can establish consistency
of the spillover estimator without requiring operator-norm
consistency of the estimated components.
Next, we use $\widehat W$ as the working matrix for
estimating the model in~\eqref{maineq}. We now impose the conditions
used to analyze the spillover-parameter estimators.

\begin{assumption}[True parameters]
    \label{Assump:Theta}
  %  For simplicity, we set $\gamma_0 = \bm{0}_{K}$ for our theoretical results. 
  %  Additionally:
    \begin{enumerate}[(i)]
        \item The parameter space $\Theta = \mathcal{C} \times \mathcal{B}\times {\Gamma}$ 
        is a compact subset of $\mathbb{R} \times {\mathbb{R}^K} \times \mathbb{R}^K$ for some finite  dimension $K$. % < n$.
        \item The true value, denoted as $\theta_0 = (\lambda_0, \beta_0\T, \gamma_0 \T)\T$, lies in the interior of the parameter space $\Theta$.  Moreover, $\lambda_0\beta_0+\gamma_0\neq \mathbf{0}_{K\times 1}$. 
    \end{enumerate}
\end{assumption}
	
\begin{assumption}[Adjacency matrix]
    \label{Assump:True_Adjacency}
    $W_0$ is a non-stochastic spatial weight matrix satisfying the following conditions:  %  Additionally:
    \begin{enumerate}[(i)]
        \item The diagonal elements of $W_0$ are 0, such that $W_{0, ii} = 0$ for all $i = 1, \ldots, n$.
                \item  $I_n, W_0, W_0^2$ are linearly independent.

                \item Suppose that $J$ is fixed and that the absolute column sums of
$W_0$ are uniformly bounded in $n$, except possibly for those of its
first $J$ columns. Let $W_{0,u}$ contain the first $J$ columns of
$W_0$ and zeros elsewhere, and let
$W_{0,b}=W_0-W_{0,u}$. When $J=0$, set
$W_{0,u}=\mathbf 0_{n\times n}$ and $W_{0,b}=W_0$.
Assume that there exist constants $c_\infty,c_1>0$, independent of
$n$, such that
$
\sup_{\lambda\in\mathcal C}
|\lambda|\|W_0\|_\infty\leq1-c_\infty$, and $
\sup_{\lambda\in\mathcal C}
|\lambda|\|W_{0,b}\|_1\leq1-c_1.
$
\end{enumerate}
\end{assumption}

Assumption~\ref{Assump:Theta} imposes standard conditions on
the true parameters. The condition
$\lambda_0\beta_0+\gamma_0\neq\mathbf{0}_{K\times1}$
rules out complete cancellation of the endogenous and
contextual effects in the reduced form. Together with suitable
exogeneity and rank conditions, it supports identification
\citep{Bramoulle2009}.

Assumption \ref{Assump:True_Adjacency} (i) requires that the true adjacency matrix \(W_{0}\) have zero diagonal, which rules out self‐loops. 
\footnote{
Some motivating examples, such as input--output matrices, contain
nonzero diagonal entries representing within-sector transactions.
When these entries are retained, the decomposition is simpler because
no zero-diagonal constraint needs to be imposed. We therefore work
directly with $W_0=L_0+S_0$ and omit the diagonal adjustments used
to define $L_0^*$ and $S_0^*$. Accordingly, the zero-diagonal
restriction in Assumption~\ref{Assump:E_noise}(i) is not required for these applications.
}
Assumption~\ref{Assump:True_Adjacency}(ii) restricts the network
structure to avoid a source of non-identification. Instrument
exogeneity and relevance are imposed separately in
Assumption~\ref{Assump:Instruments}.
 Assumption~\ref{Assump:True_Adjacency}(iii) restricts the number of
dominant units to a fixed finite value and imposes uniform stability.
The first stability condition guarantees that
$I_n-\lambda W_0$ is invertible uniformly over
$\lambda\in\mathcal C$. 
For simplicity, we present the baseline theory for
model~\eqref{maineq} without individual and time fixed effects.
The assumptions refer to the original,
untransformed variables. The fixed-effects extension is discussed
separately in Remark~\ref{rem:twoway_demean} in the Appendix.

\begin{assumption}[Spatial errors]
\label{Assump:Spatial_Errors}
Let
$\mathcal Z_n=\sigma\{Z_t(W_0):t\in\mathbb Z\}$.
Conditional on $\mathcal Z_n$, the primitive disturbances
$\{\varepsilon_{it}\}_{i,t}$ are independent and satisfy
$\mathbb E[\varepsilon_{it}\mid\mathcal Z_n]=0,
\mathbb E[\varepsilon_{it}^2\mid\mathcal Z_n]=\sigma_0^2>0$.
Moreover, for some $\delta>2$ and a constant $C<\infty$,
$\mathbb E[|\varepsilon_{it}|^{2+\delta}\mid\mathcal Z_n]
\leq C$
almost surely, uniformly in $n,i,t$.
\end{assumption}

\begin{assumption}[Instrumental variables]
    \label{Assump:Instruments}
    \begin{enumerate}[(i)]
        \item The covariates $X_t$ %= (X_{1t}, \cdots, X_{Kt})$ 
        are stochastic $n \times K$ matrices with elements $x_{t,il}$ satisfying       
$\sup_{n,t,i,\,1\leq l\leq K}
\mathbb{E}|x_{t,il}|^{2+\delta}<\infty$
for some $\delta>2$. Let
$x_{t,i}=(x_{t,i1},\ldots,x_{t,iK})^\top$.
For each $i$, $\{x_{t,i}:t\in\mathbb Z\}$ is strictly stationary
and strongly mixing over $t$. The mixing coefficients are uniform
in $n$ and $i$: there exist constants $C<\infty$ and
$\rho\in(0,1)$ such that
$
\sup_{n,i}\alpha_{n,i}(h)\leq C\rho^h$, for $h\geq1$.
The cross-sectional processes
$\{x_{t,i}:t\in\mathbb Z\}$ are independent across $i$. Dependence between $E$ and $X_t$ is allowed, subject to the
conditional moment restrictions stated in
Lemma~\ref{rateG}(c) in the Appendix.
        \item $M$ is a pre-specified adjacency matrix.  The instrumental variable matrix $Z_t(M)$ 
        has columns chosen from $M^{\ell} X_{jt}$ for $\ell = 0, \ldots, d$ and $j = 1, \ldots, K$.
        We require that $
\sup_{n,t,i,\,1\leq l\leq m}
\mathbb E|z_{t,il}|^{2+\delta}<\infty$, $\delta>2$ and $\mathbb{E}[\varepsilon_t \mid Z_t(M)] = \mathbf{0}_{n \times 1}$.
        Moreover, $Z_t(M)$ has $m$ linearly independent columns, where $m$ is a fixed constant 
        satisfying $m \geq 2K+1$.

        \item $\lim_{\min(n,T) \to \infty} (nT)^{-1} \sum_{t=1}^T \mathbb{E}\bigl[Z_t(W_0)^\top Z_t(W_0)\bigr]$ 
        exists and is nonsingular, denoted as $\Sigma_{Z_0 Z_0}$. We assume that
$\Omega_0
\coloneqq
\lim_{\min(n,T)\to\infty}
\frac{1}{nT}
\operatorname{Var}\!\left(
\sum_{t=1}^T Z_t(W_0)^\top\varepsilon_t
\right)$
exists and is finite.

        \item Let $Q_{0t} \coloneqq \bigl[W_0(I_n - \lambda_0 W_0)^{-1}(X_t \beta_0 + W_0 X_t \gamma_0),\, 
        X_t,\, W_0 X_t\bigr]$. We assume the $m \times (2K+1)$ matrix
        $-\lim_{\min(n,T) \to \infty} (nT)^{-1} \sum_{t=1}^T \mathbb{E}\bigl[Z_t(W_0)^\top Q_{0t}\bigr]$
        exists, denoted as $G_0$, and is of full column rank.
    \end{enumerate}
\end{assumption}

\begin{assumption}[Moment condition and GMM weight matrix]
    \label{Assump:Moments_Weight}
The GMM weight matrix $\widehat{\Lambda}_{nT}^{-1} \in \mathbb{R}^{m \times m}$ converges in probability to some symmetric positive-definite matrix $\Lambda^{-1} \in \mathbb{R}^{m \times m}$, whose eigenvalues are bounded away from 0 and $\infty$. 
\end{assumption}

Assumptions~\ref{Assump:Spatial_Errors} and
\ref{Assump:Instruments} impose moment and dependence conditions.
The conditional independence requirement for the disturbances
in Assumption~\ref{Assump:Spatial_Errors} and the independence
of covariate processes across units in
Assumption~\ref{Assump:Instruments}(i) may be restrictive
in network settings.
The analysis can be extended to suitable forms of weak
cross-sectional dependence under additional regularity conditions. We maintain independence for expositional simplicity. Assumption \ref{Assump:Instruments} (ii) imposes regularity conditions on the instrumental variables.  Among other things, the requirement that $\mathbb{E}[\varepsilon_t |Z_t(M)] = \mathbf{0}_{n \times 1}$ precludes $M=W$  if the measurement error $E$ is endogenous.    The selection of linearly independent instrument columns
is imposed separately in part~(ii).
For $M=W_0$, consider the full candidate instrument matrix
$\check Z_t(W_0)=(X_t,W_0X_t,\ldots,W_0^dX_t)$.
The $(\ell_1,\ell_2)$ block of
$(nT)^{-1}\sum_{t=1}^T
\check Z_t(W_0)^\top\check Z_t(W_0)$
is
$(nT)^{-1}\sum_{t=1}^T
X_t^\top(W_0^{\ell_1})^\top W_0^{\ell_2}X_t$,
for $\ell_1,\ell_2=0,\ldots,d$.
The corresponding sample second-moment matrix for $Z_t(W_0)$
is obtained by retaining the rows and columns associated
with the selected instruments.
Assumption~\ref{Assump:Instruments}(iii) imposes non-singularity
of the limiting second-moment matrix $\Sigma_{Z_0Z_0}$,
while part~(iv) requires the limiting moment Jacobian matrix
$G_0$ to have full column rank. Assumption~\ref{Assump:Moments_Weight} imposes regularity
conditions on the GMM weight matrix.
For deterministic, time-invariant $M$ and a fixed selection
of instrument columns, $Z_t(M)$ is a measurable function of
$X_t$. Hence, for each $n$,
$\alpha_{Z(M),n}(h)\leq\alpha_{X,n}(h)$,
where $\alpha_{X,n}$ denotes the strong-mixing coefficient
of the full process $\{\operatorname{vec}(X_t)\}_t$.

Since \(W_0\) is not directly available for constructing instruments, one approach is to use \(Z_t(\widehat{W})\) as plug-in instruments, where \(\widehat{W}\) is obtained from  Algorithm \ref{Algo:CD_LRSED} %or \ref{Algo:ADMM_GMM_FBS}, 
and satisfies Theorem \ref{Thm:DK1}. When \(W_0\) is correctly observed, the minimizer \(\widehat{\theta}_p(W_0)\) satisfies the standard panel-data estimator convergence rate:
$\|\widehat{\theta}_p(W_0) - \theta_0\|_2 = O_p(1/\sqrt{nT})$.
However, when \(W_0\) is replaced by a denoised estimate \(\widehat{W}\), the associated minimizer \(\widehat{\theta}_p(\widehat{W})\) depends on the deviation \(\widehat{W} - W_0\).
Let
$\rho_n=w_{2n}+\mathcal R_s(W_0,E)+a_n, \ 
s_n=\mathcal R_s(W_0,E)\sqrt{m_s(S_0^*)}, \ 
\kappa_n=1+\|W_0\|_2+\rho_n$.
Define $B_0=(I_n-\lambda_0W_0)^{-1}$, and the dimension-discounted rate
\begin{equation}
\mathcal R_{nT}^*
\coloneqq
(1+\|B_0\|_2)\kappa_n^{d+1}
\frac{(r\vee1)\rho_n+s_n}{n}.
\label{eq:RstarnT}
\end{equation}

The corresponding convergence properties are detailed in the following results.

\begin{theorem}[Performance of the plug-in estimator]
\label{Thm:DK22}
Recall the definitions of
$B_0$, $\rho_n$, $s_n$, $\kappa_n$, and
$\mathcal R_{nT}^*$ given above.
Consider the model described in \eqref{maineq}. 
Suppose Assumptions \ref{weight}--\ref{Assump:Moments_Weight} hold. 
\medskip
\noindent\textbf{(a) The plug-in estimator.}
If $\mathcal R_{nT}^*=o(1)$, then
\[
\|\widehat\theta_p(\widehat W)-\theta_0\|_2
=
O_p(\mathcal R_{nT}^*)
+
O_p\!\left(\frac{1}{\sqrt{nT}}\right).
\]
\medskip
\noindent\textbf{(b) The GMM estimator.}
{Suppose additionally that $\|W_0\|_2=O(1)$ and $w_{2n}=O(1)$.}
\begin{align}
\|\widehat{\theta}_{GMM} - \theta_0\|_2 
= O_p(w_{2n} + w_{2n}^2) 
+ O_p\!\left(\frac{1}{\sqrt{nT}}\right).
\end{align} 
\medskip
\medskip
\noindent\textbf{(c) Asymptotic normality.}
If $\sqrt{nT}\,\mathcal R_{nT}^*=o(1)$, then
$\sqrt{nT}\,
\bigl(\widehat\theta_p(\widehat W)-\theta_0\bigr)
\overset{d}{\longrightarrow}
\mathcal N(\mathbf 0,\Sigma_{GMM})$, 
where
$\Sigma_{GMM}
=\left(G_0^\top \Lambda^{-1}G_0\right)^{-1}
G_0^\top \Lambda^{-1}\Omega_0\Lambda^{-1}G_0
\left(G_0^\top \Lambda^{-1}G_0\right)^{-1}$.
\end{theorem}

Under the assumptions of Theorem~\ref{Thm:DK22}(a), the plug-in
estimator $\widehat\theta_p(\widehat W)$ is consistent provided that
$\mathcal R_{nT}^*=o(1)$. The factor $1/n$ in $\mathcal R_{nT}^*$ discounts the noise contribution by the network 
size, so consistency holds even when $w_{2n}$ and $w_{\max n}$ do 
not vanish individually, for example, when $r$ and $m_s(S_0^*)$ 
remain bounded and $w_{2n}$ is of constant order. 
By contrast, the available bound for the GMM estimator based on the
noisy matrix does not guarantee consistency when $w_{2n}$ is bounded
away from zero; the moment-misspecification example above shows that
inconsistency can occur in this case. The noise-induced component of the plug-in estimator's rate is
strictly smaller whenever
$\mathcal R_{nT}^*
=o(w_{2n}+w_{2n}^2)$.
To show the rate advantage of our method, Remark~\ref{rem:explicit-rate-comparison} in the Appendix shows that, in the empirically relevant row-sparse regime with no dominant units, %bounded rank, and 
i.e. $m_s(S_0^*)\asymp n$, the network-error component improves from $O_p(w_{2n})$ for the raw-$W$ GMM estimator to $O_p(w_{2n}/\sqrt n)$ for our estimator.
The improvement reflects how the low-rank plus sparse decomposition enters the 
analysis. 
Denote
$\Delta L^* := \widehat L^*-L_0^*$ and
$\Delta S^* := \widetilde S-S_0^*$
as the recovery errors from Theorem~\ref{Thm:DK1}.
Thus $\widehat W-W_0
=
\Delta L^*+\Delta S^*$.
The network recovery error affects the spillover estimator
only through bilinear forms involving the instruments and
residuals, and these averages discount the error according
to its structure.
 The low-rank error enters through its nuclear norm, 
controlled by $\sqrt r\,\|\Delta L^*\|_F$, contributing $O(r)$ 
effective directions rather than all $n$; the sparse error enters 
through $|\Delta S^*|_{1,1}$, controlled by its $m_s(S_0^*)$ nonzero 
entries rather than all $n^2$. The measurement error can therefore 
be large in aggregate while the spillover estimator still converges.

Example~\ref{lewbel:e} constructs regimes in which
$|E|_{1,1}/n$ does not vanish, while the 
condition $\mathcal R_{nT}^*=o(1)$ may still hold. The former 
condition is sharp for sparse measurement error such as link 
misclassification in $0$--$1$ networks, but fails in the dense 
case that motivates our analysis: Example~\ref{lewbel:e} 
constructs the case in which $|E|_{1,1}$ is of order $n^2$ yet 
the plug-in estimator remains consistent.
The two analyses impose different restrictions on measurement
error. \citet{Lewbeletal2024EJ} control its aggregate magnitude
in the unnormalized adjacency matrix, whereas our framework
exploits structure in the latent network to accommodate
certain dense-error regimes. In the pure-sparse case
$L_0=\mathbf 0_{n\times n}$, the estimator obtained by fixing
$\widehat L=\mathbf 0_{n\times n}$ satisfies
$\|\widehat\theta_p(\widehat S)-\theta_0\|_2
=O_p(m_s(S_0^*)w_{\max n}/n)+O_p(1/\sqrt{nT})$,
under the conditions of Corollary~\ref{cor1}.

The following corollary specializes Theorem~\ref{Thm:DK22} to the
pure-sparse case $L_0=\mathbf 0_{n\times n}$. Within
Assumption~\ref{weight}, only part~(ii) on $S_0$ is needed, and the
coherence terms drop out. The rate is driven entirely by the sparse
recovery error $n^{-1}|\widehat S-S_0|_{1,1}$.
\begin{corollary}\label{cor1}
Suppose $L_0=\mathbf 0_{n\times n}$ and let $\widehat S$ denote the
pure-sparse estimator obtained by fixing
$\widehat L=\mathbf 0_{n\times n}$. Suppose
Assumption~\ref{weight}(ii),
Assumption~\ref{Assump:E_noise}(i),(iii), and
Assumptions~\ref{Assump:Theta}--\ref{Assump:Moments_Weight} hold.
Let
$\tau_n=C_\tau w_{\max n}, C_\tau\geq1$.
Assume additionally that
$\|S_0\|_1=O(1)$,
and
$\max_{1\leq j\leq K}
\frac{1}{T}\sum_{t=1}^T
\|X_{jt}\|_{\max}^2
=
O_p(1),
\frac{1}{T}\sum_{t=1}^T
\|\varepsilon_t\|_{\max}^2
=
O_p(1)$.

Then
\begin{align*}
\|\widehat\theta_p(\widehat S)-\theta_0\|_2
=
O_p\!\left(
\frac{|\widehat S-S_0|_{1,1}}{n}
\right)
+
O_p\!\left(\frac{1}{\sqrt{nT}}\right)
&=
O_p\!\left(
\frac{m_s(S_0^*)w_{\max n}}{n}
\right)
+
O_p\!\left(\frac{1}{\sqrt{nT}}\right).
\end{align*}
\end{corollary}

The rate in Corollary~\ref{cor1} follows from
Lemma~\ref{rateG}(b): the sparse recovery error enters the spillover
estimator only through
$n^{-1}|\widehat S-S_0|_{1,1}$. In the pure-sparse case, the
sparse-recovery argument used in the proof of
Theorem~\ref{Thm:DK1} gives
$\frac{1}{n}|\widehat S-S_0|_{1,1}
=
O_p\!\left(
\frac{m_s(S_0^*)w_{\max n}}{n}
\right)$.
Although Corollary~\ref{cor1} covers the pure-sparse case, the gain
from our general analysis is most pronounced when the low-rank
component $L_0$ is genuinely present and $r$ and $m_s(S_0^*)$ remain
controlled relative to $n$.

An additional distinction from \citet{Lewbeletal2024EJ} is that their
Assumption~3.2(ii) requires the row and column sums of both the true
and observed adjacency matrices to be uniformly bounded. Hence, by
the triangle inequality, this condition implies
$\|E\|_\infty=O_p(1)$ and
$\|E\|_1=O_p(1)$.
Our pure-sparse result does not impose either restriction on the raw
measurement error. We require only the high-probability entrywise
bound $\|E\|_{\max}\leq w_{\max n}$ and the
weaker condition
$\frac{m_s(S_0^*)w_{\max n}}{n}=o(1)$.
Thus, the row and column sums of $E$ may diverge, provided that its
individual entries are sufficiently small for the denoising step to
recover the sparse latent network.

{
\begin{remark}[Asymptotic collinearity]
It should be noted that  variation in peer connections plays an important role in identification and estimation as noted, for example, in \cite{Manski1993} (p.535) and \cite{Paula2017}.  In a \emph{single}, \emph{dense} and \emph{large} network, the leading eigenvector of a row-normalized and  positive adjacency matrix may become (asymptotically) proportional to a vector of ones.  In this case, peer averages become increasingly similar across individuals, complicating identification and estimation of spillover effects. A recent formulation of this issue in the context of peer effects is provided by \cite{hayes2024minimax}, who show that this collinearity arises when $W_0$ is row-normalized, positive, and exhibits unbounded minimum degrees. Our setting differs in two respects. 
For a nonnegative row-normalized matrix,
$W_0\mathbf 1_n=\mathbf 1_n$ holds exactly. When nodal covariates
are independent of the network and the minimum degree diverges, peer
averages may concentrate around a common value, producing asymptotic
collinearity; see \citet{hayes2024minimax}. Our assumptions do not
generally impose row normalization or diverging minimum degree, so
this conclusion does not follow automatically in our setting. The
presence of a sparse component $S_0$ does not by itself eliminate
the possibility of asymptotic collinearity. Without row normalization, $\mathbf 1_n$ need not
be an eigenvector of $W_0$, and the argument requires separate
verification.
$\blacksquare$
\end{remark}}

As mentioned previously, we emphasize that our results are also suitable 
for the case when $E$ is endogenous with respect to $\varepsilon_t$. As 
an example, in social networks this can arise when misclassification 
errors are not generated at random, instead being correlated with unobservables 
driving the network formation process, such as the framework considered 
by \citet{Candelaria2023}. To see this, we provide a rate analysis to compare the rates of our estimator and the GMM estimator under exogenous and endogenous measurement error. For simplicity, we adopt the following assumption on 
the relationship between these two sources of error, see similar 
assumptions as in \cite{Lee2014}: 

\begin{assumption}[Endogenous measurement error]
\label{Assump:E_noise_endogenous}
Suppose Model~\eqref{maineq} holds. For each $i$, let
$\varepsilon_i
=
(\varepsilon_{i1},\ldots,\varepsilon_{iT})^\top
\in\mathbb R^T,  \mathrm{and} \quad
E_{i,-i}
=
(E_{ij}:j\neq i)^\top
\in\mathbb R^{n-1}$
collect the off-diagonal elements of the $i$th row of $E$, with
$E_{ii}=0$.
Suppose that the vectors
$(\varepsilon_i^\top,E_{i,-i}^\top)^\top$ are independent across
$i$ and jointly Gaussian:
\[\begin{pmatrix}
\varepsilon_i\\
E_{i,-i}
\end{pmatrix}
\sim
\mathcal N\!\left(
\mathbf 0,\,
\Sigma_{\varepsilon E}
\right),
\]
where
\[
\Sigma_{\varepsilon E}
=
\begin{bmatrix}
\sigma_0^2 I_T
&
\sigma_0\sigma_E n^{-s}\rho_{\varepsilon E}\\
\sigma_0\sigma_E n^{-s}\rho_{\varepsilon E}^\top
&
\sigma_E^2n^{-2s}\Sigma_E
\end{bmatrix},
\qquad s>\frac12.
\]
Here,
$\rho_{\varepsilon E}\in\mathbb R^{T\times(n-1)}$ governs the
dependence between the outcome disturbances and the measurement
error, and
$\Sigma_E\in\mathbb R^{(n-1)\times(n-1)}$ is a correlation matrix.
Assume that there exist constants $c_E,C_E,C_\rho>0$, independent
of $n$ and $T$, such that
\[
c_E
\leq
\lambda_{\min}(\Sigma_E)
\leq
\lambda_{\max}(\Sigma_E)
\leq
C_E,
\qquad
\|\rho_{\varepsilon E}\|_2\leq C_\rho,
\]
and that $\Sigma_{\varepsilon E}$ is positive semidefinite.
\end{assumption}
Assumption~\ref{Assump:E_noise_endogenous} is imposed together with
Assumptions~\ref{Assump:E_noise} and~\ref{Assump:Spatial_Errors}.
Thus, the disturbances remain marginally i.i.d.\ across $i$ and $t$,
while dependence between $E$ and $\varepsilon_t$ is allowed. The
Gaussian specification is imposed solely to illustrate the rate in
the presence of endogenous measurement error. Denote $\widehat{\theta}_{GMM}(W,M)$ and $\widehat{\theta}_p(\widehat{W},M)$ 
the GMM and plug-in estimators with instrument $Z_t(M)$. In comparison to the 
exogenous measurement error case, the estimation errors of 
$\widehat{\theta}_{GMM}(W,M) = \widehat{\theta}_p(W,M)$ 
could be larger than that of 
our proposed estimators. This is because the deviation between sample moments 
$\Bar{g}_{nT}(\theta, W, M)$ and $\Bar{g}_{nT}(\theta, W_0, M)$ might be 
more severe. This is reflected in our simulation exercises, see 
Section~\ref{Sec:Simulation}. 
For simplicity, we retain the contextual-effects regressor $W_0X_t$
but assume that its true coefficient is zero,
$\gamma_0=\mathbf 0_K$. The matrix defined in
Assumption~\ref{Assump:Instruments}(iv) becomes
$Q_{0t}=[W_0B_0X_t\beta_0,\,X_t,\,W_0X_t]$.  Thus,
$\theta_0=(\lambda_0,\beta_0^\top,\mathbf 0_K^\top)^\top$.
Define
$G_0(M):=
-\lim_{\min(n,T)\to\infty}(nT)^{-1}
\sum_{t=1}^T\mathbb E[Z_t(M)^\top Q_{0t}]$.
 We suppress the argument $M$ when no ambiguity arises. The following proposition characterizes a rate bound 
for both the GMM estimator and our estimator. 
\begin{proposition}[Rates under endogenous measurement error]
\label{lemmabias}

Suppose that $M$ is fixed or exogenous, $G_0(M)$ has full column
rank, and the fixed-$M$ analogues of the conditional-moment
conditions in Lemma~\ref{rateG}(c) hold. Under
Assumptions~\ref{weight}--\ref{Assump:E_noise_endogenous}, with
$\|B_0\|_2(1+\|W_0\|_2)=O(1)$ and
$\mathcal R_{nT}^*=o(1)$, we have
\begin{align}
\|\widehat\theta_p(\widehat W,M)-\theta_0\|_2
&\lesssim_p
\mathcal R_{nT}^*
+\frac1{\sqrt{nT}},
\label{prop1_denoised}\\
\|\widehat\theta_{GMM}(W,M)-\theta_0\|_2
&\lesssim_p
n^{1/2-s}\|\rho_{\varepsilon E}\|_2
+w_{2n}
+\frac1{\sqrt{nT}}.
\label{prop1_gmm}
\end{align}

\end{proposition}

Three features of this comparison are worth highlighting.

\emph{First}, the endogeneity bias $n^{1/2-s}\|\rho_{\varepsilon E}\|_2$ 
appears only in the GMM rate~\eqref{prop1_gmm}. It originates from a 
bilinear form that is related to $E$: the correlation between $E$ and 
$\varepsilon_t$ contributes a non-zero  mean to 
$Z_t^\top E B_0\varepsilon_t$, whose magnitude is governed by 
$\rho_{\varepsilon E}$.
In the plug-in case, the analogous bilinear form involves
$(\Delta L^*+\Delta S^*)B_0\varepsilon_t$. Its low-rank component is
controlled by
$\|\Delta L^*\|_*
=
O_p\!\left((r\vee1)(w_{2n}+\mathcal R_s)\right)$,
whereas its off-diagonal sparse component is controlled by
$|\Delta S^*|_{1,1}
=
O_p\!\left(\mathcal R_s\sqrt{m_s(S_0^*)}\right)$.
Under the conditional-moment conditions, the plug-in bound therefore
contains no additional explicit term involving
$\rho_{\varepsilon E}$.
 \emph{Second}, even in the absence of endogeneity 
($\rho_{\varepsilon E} = 0$), the denoised rate improves upon the GMM 
rate: the noise-level term $w_{2n}$ 
in~\eqref{prop1_gmm} is replaced by the rate 
$\mathcal{R}^*_{nT}$ in~\eqref{prop1_denoised}, which is strictly 
smaller whenever 
$\mathcal R_{nT}^*=o(w_{2n}+w_{2n}^2)$. \emph{Third}, the denoised rate \emph{attenuates} the dependence on 
the noise magnitude $w_{2n}$ rather than eliminating it: $w_{2n}$ 
enters through the low-rank plus sparse coupling $\|\Delta L^*\|_2 = O_p(w_{2n}+\mathcal{R}_s)$, 
but its effect is multiplied by the factor $r/n$ rather 
than by 1. As a result, the rate scales not with the raw noise level 
$w_{2n}$  but with the ratio of the structural complexity ($r$, $m_s(S_0^*)$) 
to the sample size $n$. This reflects the benefit of our approach.
Finally, we consider the supervised estimator returned by
Algorithm~\ref{Algo:ADMM_GMM_FBS}. Denote its output by
$(\widehat\theta_s,\widehat L_s,\widehat S_s)$ and set
$\widehat W_s=\widehat L_s+\widehat S_s$.
Let $\mathcal E\coloneqq\sigma(E_{ij}:1\leq i,j\leq n)$.
\begin{theorem}[Performance of the supervised estimator]
\label{Thm:Supervised_Rate}
Suppose Assumptions~\ref{weight}--\ref{Assump:Moments_Weight}
hold. Let $M$ be fixed, exogenous and
satisfy the stated instrument and design conditions. Let $\bar r:=r\vee1$ and let
$\nu_{nT}
=C_\nu\left(w_{2n}+\mathcal R_s/\sqrt{\bar r}\right)$
and
$\tau_{nT}=C_\tau(w_{\max n}+a_n)$,
for sufficiently large constants $C_\nu,C_\tau>0$. Suppose
%\[
$\mathcal R^*_{nT}=o(1)$ and $
%\qquad\text{and}\qquad
\xi_{nT}=O(1).$
%\]
Then
\begin{equation}
\|\widehat\theta_s-\theta_0\|_2
=
O_p\left(
\frac1{\sqrt{nT}}+\mathcal R^*_{nT}
\right).
\label{superv8}
\end{equation}

If, in addition, 
$\|M-W_0\|_2=o_p(1)$,
%\[
$\sqrt{nT}\,\mathcal R^*_{nT}=o(1)$ and 
%\qquad\text{and}\qquad
$\xi_{nT}=o(1)$, 
%\]
then
\[
\sqrt{nT}
(\widehat\theta_s-\theta_0)
\overset{d}{\longrightarrow}
\mathcal N(\mathbf 0,\Sigma_{GMM}).
\]
\end{theorem}
The supervised estimator achieves the same rate as the plug-in 
estimator in Theorem~\ref{Thm:DK22} and it is strictly smaller
than that of the noisy-$W$ GMM estimator whenever
$\mathcal R_{nT}^*=o(w_{2n}+w_{2n}^2)$.
The rate depends on $w_{2n}$ through the dimension-discounted recovery rate $\mathcal R_{nT}^*$, rather than through the undiscounted noise term in the GMM rate. The condition on $\xi_{nT}$ 
ensures that the generated error from the first-iteration 
estimator $\widehat{\theta}^{[1]}$ does not inflate either the 
sparse recovery  in Step b) or the low-rank recovery in Step c) 
beyond the unsupervised noise level, so that the supervised 
decomposition matches the plug-in decomposition rates. \\
The following corollary specializes Theorem~\ref{Thm:Supervised_Rate} 
to the pure-sparse case $L_0 = \mathbf{0}_{n\times n}$, paralleling 
Corollary~\ref{cor1}: only Assumption~\ref{weight}(ii) on $S_0$ is 
needed, the coherence terms drop out, and the rate is driven entirely 
by the sparse recovery error $n^{-1}|\widehat S_s - S_0|_{1,1}$.
\begin{corollary}\label{cor2}
Suppose $L_0=\mathbf 0_{n\times n}$ and consider the pure-sparse
supervised estimator obtained by imposing
$\widehat{L}_s=\mathbf 0_{n\times n}$, so that
$\widehat W_s=\widehat S_s$.
Suppose Assumption~\ref{weight}(ii),
Assumption~\ref{Assump:E_noise}(i),(iii), and
Assumptions~\ref{Assump:Theta}--\ref{Assump:Moments_Weight} hold.
Let
$\tau_{nT}=C_\tau w_{\max n}, C_\tau>0$, and suppose that
$\|W_0\|_1=O(1),
\xi_{nT}=O(1)$.
Suppose also that the pre-specified instrument matrix $M$ satisfies
$\max_{0\leq\ell\leq d}\|M^\ell\|_\infty=O(1)$,
and that $\max_{1\leq j\leq K}
\frac{1}{T}\sum_{t=1}^T
\|X_{jt}\|_{\max}^2
=
O_p(1),
\frac{1}{T}\sum_{t=1}^T
\|\varepsilon_t\|_{\max}^2
=
O_p(1)$. Then
\begin{align}
\|\widehat\theta_s-\theta_0\|_2
&=
O_p\!\left(
\frac{|\widehat S_s-S_0|_{1,1}}{n}
\right)
+
O_p\!\left(
\frac{1}{\sqrt{nT}}
\right).
\label{eq:cor2-first-rate}
\end{align}
Moreover, the pure-sparse recovery bound for the supervised estimator
gives
\[
|\widehat S_s-S_0|_{1,1}
=
O_p\!\left(
m_s(S_0^*)w_{\max n}
\right).
\]
Consequently,
\begin{align}
\|\widehat\theta_s-\theta_0\|_2
&=
O_p\!\left(
\frac{m_s(S_0^*)w_{\max n}}{n}
\right)
+
O_p\!\left(
\frac{1}{\sqrt{nT}}
\right).
\label{eq:cor2-final-rate}
\end{align}
\end{corollary}
The theoretical performance of the supervised estimator is similar 
to that of the plug-in estimator, as confirmed by the simulation 
and application results in Sections~\ref{Sec:Simulation} 
and~\ref{Sec:Application}. We shall also note that neither
Theorem~\ref{Thm:DK22} nor ~\ref{Thm:Supervised_Rate} 
requires $E$ to be exogenous, as long as $Z_t(M)$ does not directly 
involve $W$. Furthermore, both theorems yield rates in which the 
noise level $w_{2n}$ enters only after dimension-discounting by 
the factor $r/n$ (rather than at the raw spectral rate 
$w_{2n}$ delivered by the noisy-$W$ baseline), reflecting the 
fundamental advantage of the low-rank plus sparse decomposition.

    \section{Simulations}
    \label{Sec:Simulation}
    In this section, we present simulation exercises to evaluate the performance of our proposed estimators. We first introduce four data-generating processes (DGPs) for the network adjacency matrix $W_0$. \footnote{ The simulation designs also explore settings beyond the
sufficient conditions of our asymptotic theory, allowing us
to assess the finite-sample robustness of the proposed
estimators.} We then describe how it is contaminated with noise, and discuss the Monte Carlo simulation results. 

\subsection{Generating True Weight Matrix}
\label{Subsec:NetworkDGPs}
The diagonal elements of $W_0$ are assumed to be 0. This could be achieved by setting all diagonal elements of $S_0$ to be the negatives of the corresponding diagonal entries of $L_0$. 
\begin{itemize}
\item[i)] \textbf{DGP 1} (Low rank)
In this setup, $L_0$ is generated by a low-rank matrix  
\begin{equation}
    \label{Eq:DGP_LowRank}
    L_0 = U DV\T, \quad U, V \in \mathbb{R}^{n \times r} \, , \text{ and} \quad U\T U = I_r, \quad V\T V = I_r,
\end{equation}
with rank $r=2$, $D$ is the diagonal matrix $D = \diag\{0.8^0, 0.8^1\}$, and the columns of $U$ and $V$, i.e. the left and right singular vectors, are the first 2 columns of two random $n \times n$ orthonormal matrices generated according to standard algorithms \citep{Stewart1980, Mezzadri2007}. We then set $W_0=L_0+S_0=L_0^*$, where $S_0$ is the diagonal matrix chosen such that $L_0^*$ has zero diagonal elements and its off-diagonal elements are the same as $L_0$. 

\item[ii)] \textbf{DGP 2} (Low-rank plus sparse)
Our second DGP uses the same $L_0$ as in DGP 1, but adds a sparse $n \times n$ component $S_0$, where in each row, we select $m_r=2$ off-diagonal positions uniformly at random without replacement, and assign independent \(\operatorname{Uniform}(0.5,1)\) values to these positions. All remaining off-diagonal entries are set to zero. The diagonal elements of $S_0$ are the negatives of the corresponding diagonal entries of $L_0$ and $W_0$ could also be written as 
\begin{equation}
    \label{Eq:DGP_LowRank_Sparse}
    W_0 =L_0^*+  S_0^*,
\end{equation}
where $L_0^*$ is defined as in DGP 1 using $U$, $V$ and $D$ generated as in \eqref{Eq:DGP_LowRank}, and $S_0^*$ has zero diagonal elements and its off-diagonal elements are the same as $S_0$. 

 \item[iii)] \textbf{DGP 3} (Dominant units) Our third DGP generates $W_0$ as a network with dominant units. Suppose the number of dominant units is $J=2$. The first $J$ units are dominant players in the network, whereas the remaining $n-J$ units are ordinary players. We partition $W_0$ as $[W_D, W_R]$, 
 where $W_D$ is the $n \times J$ dominant block of the adjacency matrix corresponding to the dominant units, and $W_{R}$ collects the remaining $n-J$ columns of the adjacency matrix.
For each column of $W_D$, the first $\lfloor n^{0.9} \rfloor$ entries, except the diagonal elements, are drawn independently from a $\text{Uniform}(0, 1)$ distribution, and the remaining entries are set equal to $0$. For $W_{R}$, we retain only predecessor and successor links
whose column indices correspond to ordinary units. Specifically, for $j>J$, we set $W_{0,ij}=0.5$ if $j=i-1$ or $j=i+1$. Note that $W_0$ could be viewed as a sparse matrix $S_0^*$ as each of its rows has at most $2+J$ nonzero elements, and we treat $L_0^*$ as the zero matrix. 

\item[iv)] \textbf{DGP 4} (Group) Our fourth DGP divides the individuals into two groups. The first $n/4$ members form group 1, and the remaining $3n/4$ form group 2. Only members within the same group have connections, and the connection strength between any two distinct individuals $i$ and $j$ within the same group is defined by $d_{g,g} u_{g,i}u_{g,j}/\sum_i u_{g,i}^2$, where $d_{1,1}=1$ and $d_{2,2}=0.9$, $u_{g,i}$ are latent individual-specific characteristics generated independently from a standard normal distribution. In DGP 4, we still have $W_0=L_0+S_0$, where $L_0=UDU^\top$, where $U=[U_1, U_2]$, $D$ is a $2\times 2$ diagonal matrix with $d_{1,1}=1$ and $d_{2,2}=0.9$, $U_1=[u_{1,1}, \cdots, u_{1,n/4}, 0, \cdots, 0]^\top/\sqrt{\sum_i u_{1,i}^2}$ and $U_2=[0, \cdots, 0, u_{2,1}, \cdots, u_{2,3n/4}]^\top/\sqrt{\sum_i u_{2,i}^2} $, and $S_0$ is a diagonal matrix to guarantee that $W_0$ has zero diagonal elements, or equivalently, $W_0=L_0^*$, where $L_0^*$ has zero diagonal entries and its off-diagonal elements are the same as $L_0$. 
\end{itemize}

\subsection{Monte Carlo setup}\label{dgp}

For our simulation exercises, we set $\lambda_0 = 0.25$, $\beta_0 = (-1, 2)^\top$. We generate two covariates for $X_t$ independently from $\mathcal N(0, 1)$ and $\mathcal N(5, 2^2)$ distributions, respectively. We consider all combinations of \(n\in \{40,80,120\}\) and \(T\in \{1,5,15,50\}\). We take as given a true network $W_0$ generated using one of the DGPs introduced in the previous subsection, and keep it constant across 1000 Monte Carlo replications.

Our exercises additionally consider the possibility that measurement errors are endogenous. That is, we allow for correlation between the outcome disturbances ($\varepsilon_t$) and the errors that contaminate the adjacency matrix $E$, as outlined in our Assumption \ref{Assump:E_noise_endogenous}. Let $\sigma_E $ and $\sigma_\varepsilon$ be two positive constants. Specifically, for $i,j = 1, \cdots, n$, we generate \(E_{ij}\overset{\mathrm{i.i.d.}}{\sim} \mathcal N(0, \sigma^2_E)\) for $i\neq j$, and set $E_{ii}=0$. We generate $\varepsilon_t$ such that its $i$th component $\varepsilon_{it}$ satisfies $\varepsilon_{it}=\rho \sigma_\varepsilon\sum_{j=1}^n E_{ij}/(\sqrt{n-1}\sigma_E) +v_{it}$, where $v_{it}$ are i.i.d. $ \mathcal N(0,  (1-\rho^2)\sigma_\varepsilon^2)$ that are independent of $E_{ij}$. Note that $\varepsilon_{t}|E$ is normally distributed and we could derive that 
\begin{eqnarray}
    \label{Eq:Correlated_Errors}
    \left ( \begin{matrix}
        \varepsilon_t \\
        \mbox{Vec}(E)
    \end{matrix} \right )
 \sim    \mathcal N \left( \mathbf{0}_{n^2 \times 1}, \begin{bmatrix}
        \sigma_{\varepsilon}^2 \cdot I_n & \Sigma_{\varepsilon E} \\
        \Sigma_{\varepsilon E}\T & \sigma_E^2 \cdot I_{n^2-n}
    \end{bmatrix} \right)  , \quad t = 1, \ldots, T.
\end{eqnarray}
where  
$\Sigma_{\varepsilon E}$ is the $n\times (n^2-n)$ cross-covariance matrix between $\varepsilon_{it}$ and $ \mbox{Vec}(E)$, and $ \mbox{Vec}(E)$ stacks  all $E_{ij}$ for $i\neq j$ into a long vector. This construction implies that %Note that it is implied from the previous sections that 
$\mbox{cov}(\varepsilon_{it}, E_{ij})=\rho\sigma_E\sigma_\varepsilon/\sqrt{n-1}$ for $i\neq j$ and $\mbox{cov}(\varepsilon_{it}, E_{kj})=0$ if $k\neq i$. In the simulation, we set $\sigma_\varepsilon=0.15$, $\sigma_E = 0.3/n^{0.7}$ and $\rho=0$ or  $0.7$. By setting $\rho = 0$, we recover the exogenous measurement error case.

Given the adjacency matrix \(W_0\), covariates $X_t$, parameters $\theta_0$, and disturbances $\varepsilon_t$, we use the reduced-form representation of model \eqref{maineq} to generate the outcome vector $Y_t$ in each period:
\begin{align}
    \label{Eq:SAR_reduced}
    Y_t = (I_n - \lambda_0 W_0)^{-1}(X_t \beta_0 + \varepsilon_t), \, t = 1, \ldots, T \, .
\end{align}
In each replication, we set $W = W_0 + E$, which is the observed noisy version of the adjacency matrix, to construct estimates of parameters $\theta_0$. Once the adjacency matrix is generated, we use the interquartile range of the observed off-diagonal entries as a robust proxy for the noise scale $\hat\sigma$. The
sparse and low-rank penalties are parameterized as $C_S \hat\sigma \log(n^2)$  and $C_L \hat\sigma \sqrt{n^2}$ respectively, with the constants chosen by the Bayesian Information Criterion (BIC). We also set $\xi_{nT} = 1$. 

\subsection{Simulation results}
Table \ref{Tab:RE} reports relative estimation errors for the direct interaction matrix \(\lambda_0W_0\), measured using the Frobenius norm. For each method, we first compute the Frobenius norm of the estimation error in each Monte Carlo replication and then average these norms across replications. We divide each proposed method’s mean error by the corresponding mean error of the GMM benchmark. Specifically, let $R$ denote the number of Monte Carlo replications and $r$ index a replication. The relative errors for the plug-in approach and the supervised approach are defined as 
\[
\mathrm{RE}_p =
\frac{
R^{-1}\sum_{r=1}^{R}
\left\|
\widehat\lambda_p^{(r)}\widehat W^{(r)}
-\lambda_0W_0
\right\|_F
}{
R^{-1}\sum_{r=1}^{R}
\left\|
\widehat\lambda_{\mathrm{GMM}}^{(r)}W^{(r)}
-\lambda_0W_0
\right\|_F
},\qquad 
\mathrm{RE}_s =
\frac{
R^{-1}\sum_{r=1}^{R}
\left\|
\widehat\lambda_s^{(r)}\widehat W_s^{(r)}
-\lambda_0W_0
\right\|_F
}{
R^{-1}\sum_{r=1}^{R}
\left\|
\widehat\lambda_{\mathrm{GMM}}^{(r)}W^{(r)}
-\lambda_0W_0
\right\|_F
},
\]
 where $\widehat W^{(r)}$ is obtained using Algorithm \ref{Algo:CD_LRSED}, $\widehat\lambda_p^{(r)}$ is the corresponding plug-in estimate, $\widehat\lambda_s^{(r)}$ and $\widehat W_s^{(r)}$ are obtained using Algorithm \ref{Algo:ADMM_GMM_FBS}. Ratios below one indicate smaller mean Frobenius-norm estimation errors than the GMM benchmark.

We highlight several key findings. First, the two proposed approaches perform similarly, with both yielding smaller Frobenius-norm estimation errors than conventional GMM across all DGPs and sample sizes considered. Second, both methods reduce estimation errors by at least $40\%$ when $n=40$ and $T=1$. These gains become more pronounced at larger sample sizes, reaching at least $65\%$ when $n=120$ and $T=50$. Overall, relative performance generally improves as the number of network nodes, $n$, and the number of time periods, $T$, increase, consistent with the theoretical results. Third, in many configurations, the relative gains are larger under endogenous network measurement error than under exogenous network measurement error. 

Additional results for the diffusion multiplier, $(I-\lambda_0 W_0)^{-1}$, are reported in Table \ref{Tab:Relative_SpilloverMultiplier_Exogenous} in the Appendix. This multiplier is policy-relevant because it captures the total propagation of an economic shock through the network, including its initial impact and subsequent feedback effects. The results exhibit similar patterns: both proposed methods yield smaller Frobenius-norm estimation errors for the diffusion multiplier across the DGPs and sample sizes considered. 

{\scriptsize
\begin{table}[htbp]
    \centering
    \captionsetup{width=0.975\linewidth}
    \caption{Relative Mean Frobenius Norm Estimation Errors for the Direct Interaction }
    \label{Tab:RE}
    % ===================== (a) Exogenous =====================
    \begin{adjustbox}{max width=0.95\textwidth}
    \begin{tabular}{lcc|cc|cc|cc}
        \toprule
   Exogenous &Plug-in & Supervised & Plug-in & Supervised                   & Plug-in & Supervised & Plug-in & Supervised \\
       \midrule
        L (Low-rank) & \multicolumn{2}{c}{$T = 1$} & \multicolumn{2}{c}{$T = 5$}
                     & \multicolumn{2}{c}{$T = 15$} & \multicolumn{2}{c}{$T = 50$} \\
        \midrule
  $n=40$ & 0.494 & 0.481 & 0.426 & 0.422 & 0.418 & 0.412 & 0.416 & 0.407 \\ 
    $n=80$ & 0.324 & 0.323 & 0.303 & 0.303 & 0.300 & 0.299 & 0.299 & 0.295 \\ 
    $n=120$ & 0.406 & 0.400 & 0.279 & 0.279 & 0.271 & 0.271 & 0.268 & 0.279 \\ 

    \midrule
        L+Sparse & \multicolumn{2}{c}{$T = 1$} & \multicolumn{2}{c}{$T = 5$}
                 & \multicolumn{2}{c}{$T = 15$} & \multicolumn{2}{c}{$T = 50$} \\
        \midrule

   $n=40$ & 0.410 & 0.408 & 0.394 & 0.394 & 0.391 & 0.390 & 0.390 & 0.389 \\ 
    $n=80$ & 0.295 & 0.295 & 0.281 & 0.280 & 0.277 & 0.277 & 0.276 & 0.276 \\ 
    $n=120$ & 0.242 & 0.242 & 0.228 & 0.228 & 0.225 & 0.225 & 0.224 & 0.224 \\ 
    \midrule
        Dominant & \multicolumn{2}{c}{$T = 1$} & \multicolumn{2}{c}{$T = 5$}
                 & \multicolumn{2}{c}{$T = 15$} & \multicolumn{2}{c}{$T = 50$} \\
        \midrule
  
  $n=40$ & 0.564 & 0.564 & 0.564 & 0.564 & 0.563 & 0.563 & 0.563 & 0.564 \\ 
  $n=80$ & 0.496 & 0.496 & 0.494 & 0.494 & 0.494 & 0.494 & 0.494 & 0.494 \\ 
  $n=120$ & 0.341 & 0.341 & 0.336 & 0.336 & 0.336 & 0.336 & 0.335 & 0.335 \\ 
    \midrule
        Group & \multicolumn{2}{c}{$T = 1$} & \multicolumn{2}{c}{$T = 5$}
                 & \multicolumn{2}{c}{$T = 15$} & \multicolumn{2}{c}{$T = 50$} \\
        \midrule
 $n=40$ & 0.483 & 0.506 & 0.444 & 0.471 & 0.440 & 0.466 & 0.439 & 0.463 \\ 
 $n=80$ & 0.359 & 0.354 & 0.315 & 0.313 & 0.312 & 0.310 & 0.311 & 0.309 \\ 
 $n=120$ & 0.273 & 0.276 & 0.257 & 0.260 & 0.254 & 0.257 & 0.254 & 0.255 \\ 
        \bottomrule
    \end{tabular}
    \end{adjustbox}

    \vspace{0.5em}

    % ===================== (b) Endogenous =====================
    \begin{adjustbox}{max width=0.95\textwidth}
    \begin{tabular}{lcc|cc|cc|cc}
        \toprule
  Endogenous &Plug-in & Supervised & Plug-in & Supervised                   & Plug-in & Supervised & Plug-in & Supervised \\
         \midrule
        L (Low-rank) & \multicolumn{2}{c}{$T = 1$} & \multicolumn{2}{c}{$T = 5$}
                     & \multicolumn{2}{c}{$T = 15$} & \multicolumn{2}{c}{$T = 50$} \\
        \midrule
$n=40$ & 0.448 & 0.435 & 0.412 & 0.406 & 0.407 & 0.400 & 0.407 & 0.401 \\ 
  $n=80$  & 0.338 & 0.338 & 0.333 & 0.332 & 0.332 & 0.329 & 0.332 & 0.323 \\ 
  $n=120$  & 0.307 & 0.302 & 0.216 & 0.217 & 0.212 & 0.216 & 0.211 & 0.237 \\ 
  \midrule
        L+Sparse & \multicolumn{2}{c}{$T = 1$} & \multicolumn{2}{c}{$T = 5$}
                     & \multicolumn{2}{c}{$T = 15$} & \multicolumn{2}{c}{$T = 50$} \\
        \midrule

$n=40$  & 0.402 & 0.400 & 0.389 & 0.388 & 0.386 & 0.384 & 0.385 & 0.384 \\ 
  $n=80$  & 0.292 & 0.291 & 0.279 & 0.279 & 0.276 & 0.276 & 0.276 & 0.275 \\ 
  $n=120$  & 0.236 & 0.236 & 0.223 & 0.223 & 0.221 & 0.221 & 0.220 & 0.220 \\ 
   \midrule
       Dominant & \multicolumn{2}{c}{$T = 1$} & \multicolumn{2}{c}{$T = 5$}
                     & \multicolumn{2}{c}{$T = 15$} & \multicolumn{2}{c}{$T = 50$} \\
        \midrule
  
  $n=40$  & 0.554 & 0.554 & 0.559 & 0.559 & 0.560 & 0.560 & 0.560 & 0.560 \\ 
  $n=80$  & 0.488 & 0.488 & 0.491 & 0.491 & 0.491 & 0.491 & 0.491 & 0.492 \\ 
  $n=120$ & 0.336 & 0.336 & 0.335 & 0.335 & 0.335 & 0.335 & 0.335 & 0.335 \\ 
  \midrule
       Group & \multicolumn{2}{c}{$T = 1$} & \multicolumn{2}{c}{$T = 5$}
                     & \multicolumn{2}{c}{$T = 15$} & \multicolumn{2}{c}{$T = 50$} \\
        \midrule  
  
  $n=40$  & 0.451 & 0.470 & 0.430 & 0.453 & 0.429 & 0.453 & 0.429 & 0.453 \\ 
  $n=80$  & 0.309 & 0.303 & 0.286 & 0.286 & 0.286 & 0.285 & 0.285 & 0.290 \\ 
  $n=120$  & 0.280 & 0.282 & 0.277 & 0.280 & 0.277 & 0.280 & 0.277 & 0.278 \\ 
   \hline
        \bottomrule
    \end{tabular}
    \end{adjustbox}

    \caption*{\footnotesize 
Note:  The Frobenius norm estimation errors for the direct interaction matrix under
the plug-in and supervised approaches are defined as
$\|\widehat\lambda_p\widehat W-\lambda_0W_0\|_F$
and
$\|\widehat\lambda_s\widehat W_s-\lambda_0W_0\|_F$,
respectively. Here, $\widehat\lambda_p$ is the plug-in estimate
based on the denoised network estimate $\widehat W$, whereas
$\widehat\lambda_s$ and $\widehat W_s$ are the supervised estimates.
Entries report the ratio of each method's mean Frobenius norm estimation
error across Monte Carlo replications to the corresponding mean
error of the conventional GMM estimator. The latter error is
computed as $\|\widehat\lambda_{\mathrm{GMM}}W-\lambda_0W_0\|_F$
in each replication, where $\widehat\lambda_{\mathrm{GMM}}$ is
the GMM estimate based on the observed noisy weight matrix $W$.
Endogenous measurement error is simulated with a correlation
parameter of $\rho=0.7$, as specified in
\eqref{Eq:Correlated_Errors}. The four DGPs are described in
Subsection~\ref{Subsec:NetworkDGPs}.
}     
\end{table}
}

    \section{Empirical Applications}
   \label{Sec:Application}
   This section illustrates the use of our estimators for social/spatial interactions. We consider two examples, applying our methods and examining the spillover effects. In each example, we first construct a raw spatial weight matrix and then apply Algorithms \ref{Algo:CD_LRSED} and \ref{Algo:ADMM_GMM_FBS}  described in the Appendix to obtain our estimates. To select the penalties, we use the interquartile range of the raw weight-matrix entries as a robust scale proxy, denoted by $\widehat\sigma$. The sparse and low-rank penalties are parameterized as $C_S\widehat\sigma\log(n^2)$ and $C_L\widehat\sigma\sqrt{n^2}$, respectively, with the constants chosen by the Bayesian Information Criterion (BIC). We estimate the spillover coefficient $\widehat\lambda$ and define the total influence of node $j$ as the sum of column $j$ of the Leontief inverse, which measures the aggregate response to a unit additive shock originating at that node. 

\subsection{Example 1}
In our first example, we set the annual gross domestic product (GDP) growth as our outcome of interest.  Motivated by \cite{MRW1992}, the explanatory variables include the annual growth rate of the working-age population defined as individuals aged 15 to 64, denoted by \(X_1\), and the log investment share, \(X_2\), which serves as a proxy for the saving rate. We obtain the dataset from the World Development Indicators (WDI) and the dataset covers a balanced panel of 23~economies from 1970 to~2023. \footnote{These economies are: Australia (AU), Canada (CA), Mainland China (CN), 
Denmark (DK), Finland (FI), France (FR), Germany (DE), Hong Kong SAR China (HK), Indonesia (ID), Ireland (IE), Italy (IT), Japan (JP), Korea Rep. (KR),
Malaysia (MY), Mexico (MX), Netherlands (NL), New Zealand (NZ), Singapore (SG), Spain (ES), Sweden (SE),
Thailand (TH), United Kingdom (GB), and United States (US).} We include economy and year fixed effects and estimate Model (\ref{maineq}) under two scenarios. Panel~A includes only \(X_1\) and \(X_2\) as covariates.  We assume that these variables are strictly exogenous and use the spatial lags \(WX_1\) and \(WX_2\) as instruments. Panel~B includes both the original covariates and the contextual variables \(WX_1\) and \(WX_2\) as regressors, using the second-order spatial lags \(W^2X_1\) and \(W^2X_2\) as instruments.

Economic activity is well known to exhibit substantial spatial dependence.
Neighboring economies tend to trade more, share regional demand and supply shocks, and compete for mobile factors of production. For this reason, we use a geographic‑proximity weight matrix as a natural and predetermined benchmark to study spillovers. The raw spatial weight matrix is constructed as the row-normalized inverse squared distance matrix based on Haversine distances and denoted by $W_2$. \footnote{ For economies \(i\neq j\), we compute the Haversine great-circle distance \(D_{ij}\) between capital cities (in millions of meters) to form the raw weight matrix by setting the $(i,j)$ entry to $1/D_{ij}^{2}$ for $i\neq j$, with the diagonal entries set to zero, and then apply row-normalization. } Note that row-normalization of the weight matrix is standard in the spatial econometrics literature, as it facilitates a clear interpretation of the spatial autoregressive coefficient as the strength of the average influence. All off-diagonal entries of $W_2$ are strictly positive, which contrasts sharply with the sparsity assumptions sometimes invoked in the literature. This dense structure makes $W_2$ a natural candidate for our estimation procedure since many of the small off-diagonal weights likely reflect weak, noise‑contaminated links rather than economically meaningful connections.

We first examine the differences between estimated adjacency matrices and the raw matrix $W_2$. Figure \ref{GR2} presents the graphical representations of the estimates $\widehat W_2$, $\widehat L_2$ and $\widehat S_2$ respectively. Table \ref{tab:app-gdp-denoising} reports the deviations of these estimates, as well as the supervised versions, from the raw matrix $W_2$. We find that the purely low-rank estimator $\widehat L_2$ exhibits a substantial departure from $W_2$, particularly in terms of the matrix maximum norm. Most of the elements in $\widehat L_2$ are small but not exactly zero, suggesting that a low-rank approximation alone may fail to preserve several important links in the original network. In contrast, the purely sparse estimator $\widehat S_2$ retains only 96 nonzero entries. Although this produces a more parsimonious network representation, it has the largest deviation from $W_2$ under the spectral norm. The low-rank plus sparse estimator provides a smaller adjustment to the raw matrix than the purely low-rank estimate, and 
both $\|\widehat{W}_2-W_2\|_{\max}$ and $\|\widehat{W}_2-W_2\|_{2}$ remain small relative to the corresponding deviations of the competitors. Heatmaps of the raw $W_{2}$, the estimated $\widehat{W}_{2}$, and $\widehat{S}_2$ are also shown in Figure~\ref{heatmapW2}. We observe that the low-rank plus sparse estimate $\widehat{W}_{2}$ appears to preserve the dominant network structure while allowing for pervasive weak link effects. In contrast, we find $\|\widehat{W}_2-\widehat W_s\|_{\max}/\|\widehat{W}_2\|_{\max}=0.007$ and $\|\widehat{W}_2-\widehat W_s\|_{2}/\|\widehat{W}_2\|_{2}=0.020$ for Panel A, and 
$\|\widehat{W}_2-\widehat W_s\|_{\max}/\|\widehat{W}_2\|_{\max}=0.005$ and $\|\widehat{W}_2-\widehat W_s\|_{2}/\|\widehat{W}_2\|_{2}=0.013$ for Panel B. Hence the supervised low‑rank plus sparse estimates differ from the unsupervised estimate to a much lesser extent. 

\begin{figure}[htbp]
\includegraphics[width=1\textwidth]{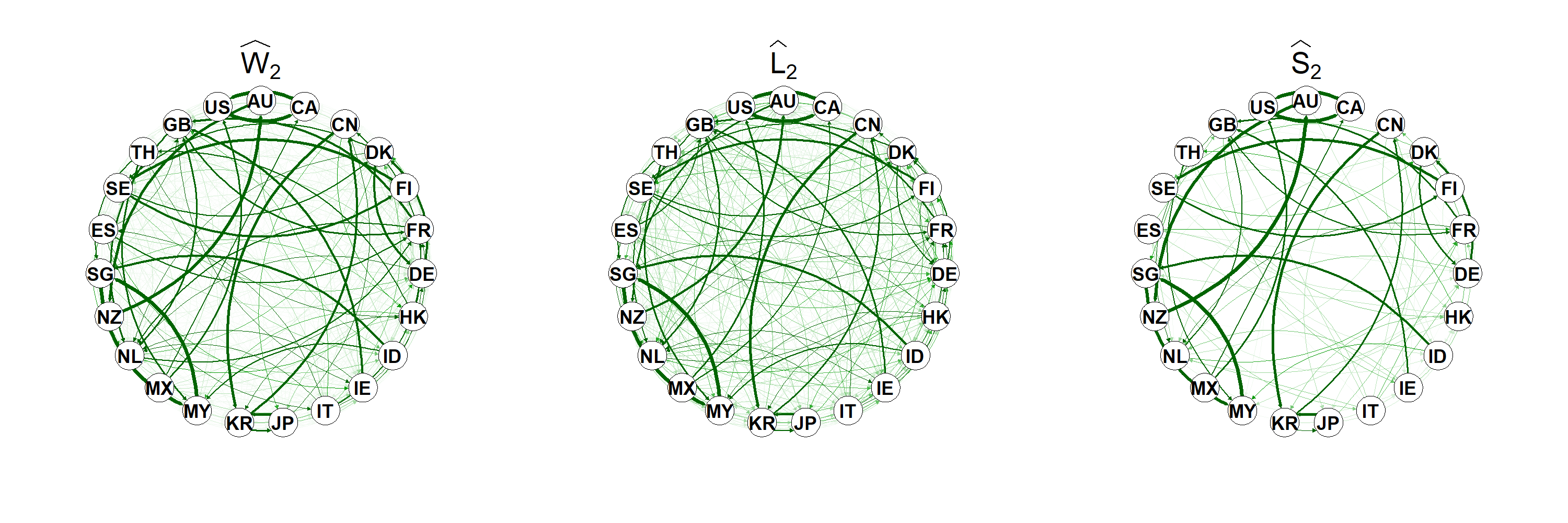}
\caption{\footnotesize Graphical representation of different spatial weight matrices associated with 23 economies, including the low-rank plus sparse estimate $\widehat W_2$ (left), the purely low-rank estimate $\widehat L_2$ (center) and the purely sparse estimate $\widehat S_2$ (right). \label{GR2}}
\end{figure}

\begin{table}[htbp]
	\centering
	\caption{Comparisons of Estimated Weight Matrices }
	\label{tab:app-gdp-denoising}
	\begin{tabular}{lrrrr}
		\toprule
		Estimator & $\|\widehat W-W_2\|_{\max}$ & $\|\widehat W-W_2\|_2$ & Rank($\widehat L$) & Nonzero Entries($\widehat S$)\\
		\midrule
		$\widehat S_2$  & 0.044 & 0.385 & -- & 96\\
		$\widehat L_2$  & 0.269 & 0.349 & 14 & --\\
		$\widehat W_2$  & 0.044 & 0.289 & 2 & 78\\
		$\widehat W_s$ (Panel A) & 0.047 & 0.301 & 2 & 115\\
		$\widehat W_s$ (Panel B) & 0.045 & 0.299 & 2 & 116\\
		\bottomrule
	\end{tabular}
    \caption*{\footnotesize Note: We calculate the differences between each estimated weight matrix and the raw weight matrix $W_2$. We consider the purely sparse matrix estimate $\widehat S_2$, the purely low-rank matrix estimate $\widehat L_2$, the low-rank plus sparse matrix estimate $\widehat W_2$ and its supervised estimates $\widehat W_s$ under model (\ref{maineq}) in Panel A (without contextual effects) or Panel B (with contextual effects). For the purely low-rank matrix $\widehat L_2$ and the low-rank and sparse matrix that admits $\widehat W_2=\widehat L_{\widehat W_2}+\widehat S_{\widehat W_2}$ and the supervised version that admits $\widehat W_s=\widehat L_{\widehat W_s}+\widehat S_{\widehat W_s}$, Rank ($\widehat L$) reports the rank of the low-rank component $\widehat L_2$, $\widehat L_{\widehat W_2}$ and $\widehat L_{\widehat W_s}$ respectively. For the purely sparse matrix $\widehat S_2$ and the low-rank and sparse matrix $\widehat W_2$ and the supervised version $\widehat W_s$, nonzero entries ($\widehat S$) reports the number of nonzero entries in the sparse component $\widehat S_2$, $\widehat S_{\widehat W_2}$ and $\widehat S_{\widehat W_s}$ respectively.}
\end{table}

\begin{figure}[htbp]
\includegraphics[width=1\textwidth]{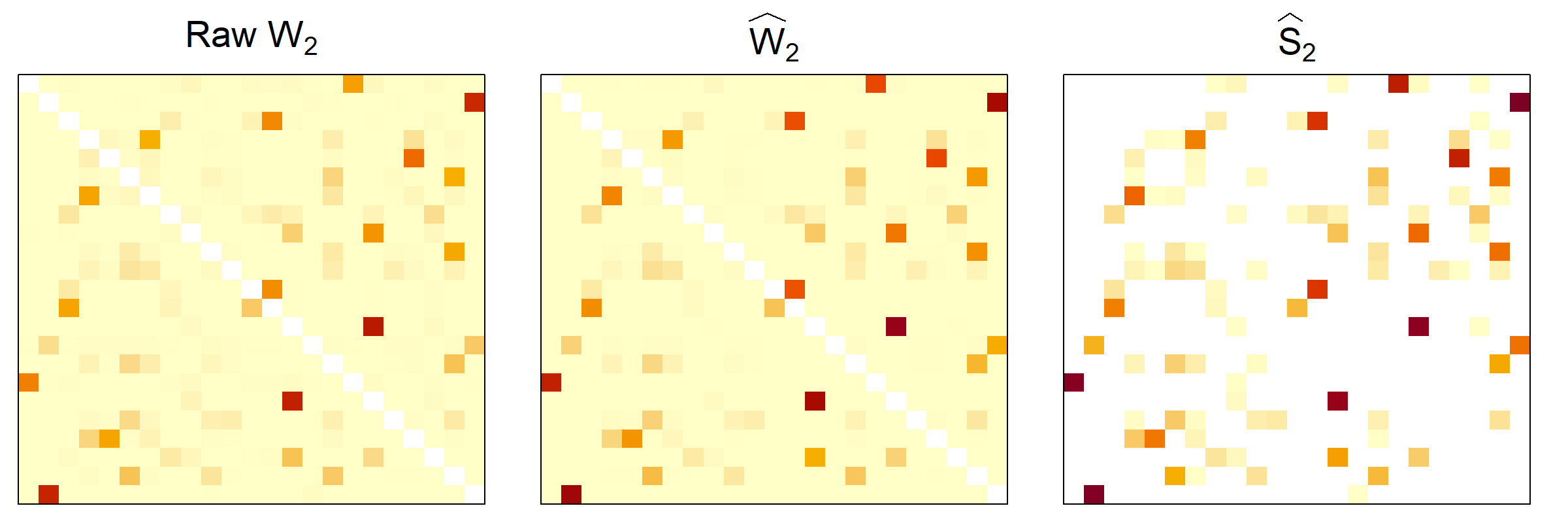}
\caption{\footnotesize  Heatmaps of different spatial weight matrices associated with 23 economies, including the raw weight matrix $W_2$ (left), the estimated low-rank and sparse matrix $\widehat W_2$ (center), and the estimated purely sparse matrix $\widehat S_2$ (right). \label{heatmapW2}}
\end{figure}

In the regression analysis, we evaluate the impact of the investment rate and working-age population growth, and report the estimated spillover coefficient $\widehat{\lambda}$ based on different weight matrices in Table~\ref{tab:app-gdp-spillovers}. For each candidate weight matrix $\widehat W$, we also compute the column sums of the Leontief inverse $(I_n-\widehat{\lambda} \widehat W)^{-1}$ 
and designate the economy with the largest column sum as the key player, i.e. the one with the largest estimated total influence in response to a unit additive shock. Relative changes compared to the raw $W_2$ benchmark are also reported.

In Panel A (without contextual effects), the purely sparse estimator $\widehat{S}_2$ only retains the strong connections and reduces $\widehat{\lambda}$ by 16.7\% and total influence (i.e., the corresponding column sum of the Leontief matrix, see above) by 11.3\% compared to those obtained from the raw $W_2$. Conversely, the purely low-rank estimator $\widehat{L}_2$ amplifies network connections by aggregating pervasive weak links and increases $\widehat{\lambda}$ by 11.3\% and total influence by 12.0\% compared to those obtained from the raw $W_2$. These results indicate that the sparse and low-rank components might capture distinct aspects of the network. The sparse component filters out scattered weak links, whereas the low-rank component summarizes pervasive dependence. Correspondingly, the low-rank-plus-sparse estimator $\widehat{W}_2$ produces more moderate deviations from the raw benchmark. In Panel B, which includes contextual covariates, the raw $W_2$ produces a negative estimate $\widehat{\lambda}$, while the sparse and the low-rank-plus-sparse estimates produce moderate positive values that might be more plausible. As higher-order spatial lags can magnify noise in the raw matrix, regularization has a greater impact on how the interaction structure is estimated.

\begin{table}[htbp]
	\centering
	\caption{GDP-Growth GMM Spillover Estimates and Influence Analysis}
	\label{tab:app-gdp-spillovers}
	\begin{tabular}{lccccc}
		\toprule
		Weight Matrix & $\widehat\lambda$  &  $\Delta \widehat\lambda$ & Key Player & Total Influence & $\Delta$Influence \\
		\midrule
		\multicolumn{6}{l}{\textit{Panel A: without contextual effects}}\\\\
		Raw $W_2$ & 0.377 & 0.0\% & SG & 2.211 & 0.0\%\\
		$\widehat S_2$ & 0.314& -16.7\% & SG & 1.961 & -11.3\%\\
		$\widehat L_2$ & 0.420 &  11.3\% & SG & 2.475 &12.0\%\\
		$\widehat W_2$ & 0.354 & -6.2\% & SG & 2.147 &-2.9\%\\
		$\widehat W_s$ & 0.356& -5.6\% & SG & 2.178 & -1.5\%\\
        \addlinespace
\midrule
        \multicolumn{6}{l}{\textit{Panel B: with contextual effects}}\\\\
		Raw $W_2$ & -0.109 & 0.0\% & MX & 0.985 & 0.0\%\\
		$\widehat S_2$ & 0.124 & 214.2\% & SG & 1.290& 31.0\%\\
		$\widehat L_2$ & -0.052 & 52.3\% & MX & 0.993 &0.8\%\\
		$\widehat W_2$ & 0.095 & 187.6\% & SG & 1.212 &23.0\%\\
		$\widehat W_s$ & 0.095 & 186.9\% & SG & 1.211  & 22.9\%\\
		\addlinespace
		\bottomrule
	\end{tabular}
       \caption*{\footnotesize Note: We report the estimated spillover coefficient $\widehat{\lambda}$ obtained using different weight matrices. The raw $W_2$
  matrix is the row-normalized inverse squared distance weight matrix. We consider its purely sparse estimator $\widehat S_2$, its purely low-rank estimator $\widehat L_2$, and its low-rank plus sparse estimator $\widehat W_2$ and its supervised version  $\widehat W_s$. %based on Algorithm \ref{Algo:CD_LRSED}.
  To make the results comparable, we row-normalize each estimated weight matrix and then calculate the estimated spillover $\widehat \lambda$. Panels A and B correspond to the scenarios without and with contextual effects respectively. We record the maximum column sum of the Leontief inverse and refer to it as the total influence of the key player. The relative changes of the spillover effect and the total influence are also computed using the raw matrix $W_2$ as the baseline. }
\end{table}

Moreover, the decomposition also yields interesting practical insights. Compared with $\widehat S_2$, we find that including the low-rank component in $\widehat W_2$ primarily rescales the overall magnitude of the spillover influence vector rather than altering the key player. This pattern indicates that, while the identity of the key player remains robust to the decomposition strategy, the estimated strength of network spillovers depends materially on whether the low-rank component is retained, underscoring the practical relevance of the full low-rank plus sparse decomposition over the purely sparse estimator. Recall that $\widehat W_2$ admits the decomposition $\widehat W_2=\widehat S_{\widehat W_2}+\widehat L_{\widehat W_2}$. By comparing the leading eigenvectors of $\widehat W_2$, $\widehat L_{\widehat W_2}$ and $\widehat S_{\widehat W_2}$, we find that ${\bf v}_{\widehat{W}_2}\approx 0.56\,{\bf v}_{\widehat{L}_{\widehat W_2}}+0.48\,{\bf v}_{\widehat{S}_{\widehat W_2}}$ from a regression perspective as discussed in equation (\ref{rhoLS}). Both the low-rank and the sparse components contribute substantially in this application, indicating that the identification of influential nodes depends not only on a few strong connections, but also on the cumulative effect of widespread weak links. Notably, as shown in Figure~\ref{fig:networksparseE1}, ${\bf v}_{\widehat{S}_{\widehat W_2}}$ has many zero entries, and its nonzero entries are concentrated among the Asia-Pacific economies, whereas ${\bf v}_{\widehat{L}_{\widehat W_2}}$ loads on all economies but assigns slightly smaller weights to Western economies. One possible interpretation is that the sparse eigenvector captures localized regional clustering centered on Asia-Pacific, while the low-rank eigenvector reflects a global common factor that is slightly attenuated for Western economies  in accounting for the broader geographic diffusion inherent in the low-rank structure.

\begin{figure}[htbp]
\includegraphics[width=1\textwidth]{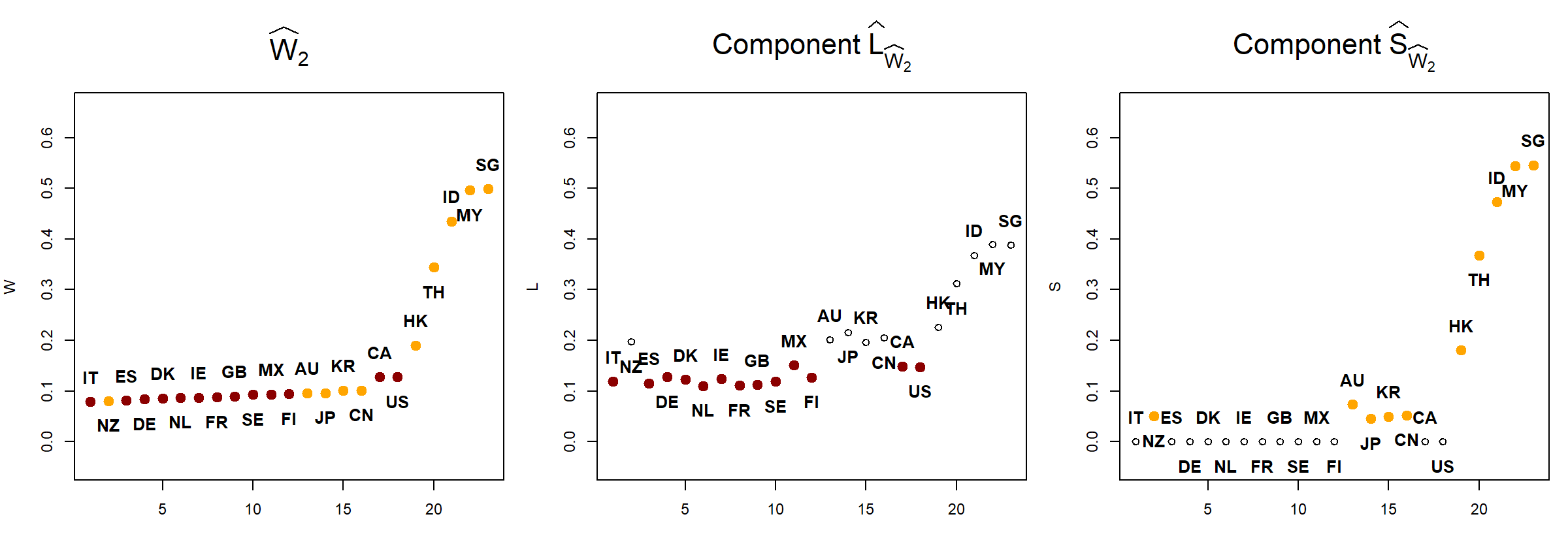}
\caption{\footnotesize Plots of eigenvector centrality associated with different weight matrices. Denote the low-rank component and the sparse component of our estimator $\widehat W_2$ as $\widehat L_{\widehat W_2}$ and $\widehat S_{\widehat W_2}$ respectively, i.e., $\widehat W_2=\widehat L_{\widehat W_2}+\widehat S_{\widehat W_2}$. From left to right, we provide the eigenvector centrality plots of $\widehat W_2$, $\widehat L_{\widehat W_2}$ and $\widehat S_{\widehat W_2}$ respectively. Orange dots mark economies whose entries in the leading eigenvector of $\widehat{S}_{\widehat W_2}$ are nonzero, while red dots mark economies whose entries in the leading eigenvector of $\widehat{L}_{\widehat W_2}$ are smaller than $0.2$. The economies corresponding to red dots are Canada, Denmark, Finland, France, Germany, Ireland, Italy, Mexico, Netherlands, Spain, Sweden, United Kingdom, and United States, and those corresponding to orange dots are Australia, Mainland China, Hong Kong SAR China,  Indonesia, Japan, Korea Rep., Malaysia, New Zealand, Singapore, Thailand.  
    }
 \label{fig:networksparseE1}
\end{figure}

%----------------------------------------------------------------------

\subsection{Example 2}
\label{sec:geo_application}
In our second example, we apply our methods to examine tax competition among the 48 contiguous U.S.\ states. This was originally studied by \cite{Besley1995} using data from 1962 to 1988, with a pre‑specified geographic weight matrix $W_g$ defined as the row‑normalized 0‑1 neighborhood adjacency matrix assigning 1 to state pairs that are geographically adjacent and zero otherwise. \cite{Besley1995} emphasize a political‑economy channel for tax‑setting spillovers, as policymakers respond to voters who benchmark their policies against those of peer states. Consequently, states may face both political and economic incentives to align their fiscal policies with those of their neighbors. This dataset has recently been extended to cover $1962$--$2015$ and analyzed by \cite{Paula2023}. We consider the spatial weight matrix $W_2$, which is a row-normalized inverse squared distance matrix based on Haversine distances between state capitals, as the observed weight matrix, and adopt our estimation procedures. We then apply Model (\ref{maineq}) with state and year fixed effects in our specifications. In addition to endogenous spatial effects captured by $WY$, the covariates include state-level characteristics such as per-capita income, the unemployment rate, and the proportions of young and elderly populations. Panel A excludes the contextual covariates $WX$, whereas Panel B includes them. Following \cite{Besley1995}, neighbors’ lagged per-capita state income and unemployment rates are used as instruments.

We consider three specifications, which yield a purely sparse matrix estimate $\widehat{S}_2$, a purely low‑rank matrix estimate $\widehat{L}_2$, and a low‑rank plus sparse estimate $\widehat{W}_2$. The purely sparse estimate $\widehat{S}_2$ captures the strong links and retains 90 nonzero entries. The estimate $\widehat{L}_2$ is a zero matrix, suggesting that 
the selected regularization procedure does not retain a low-rank component. The combined estimator $\widehat{W}_2$ coincides with the sparse‑only estimator $\widehat{S}_2$, and hence we only report results for $\widehat{W}_2$ in the subsequent analysis.

\begin{figure}[htbp]
\includegraphics[width=1\textwidth]{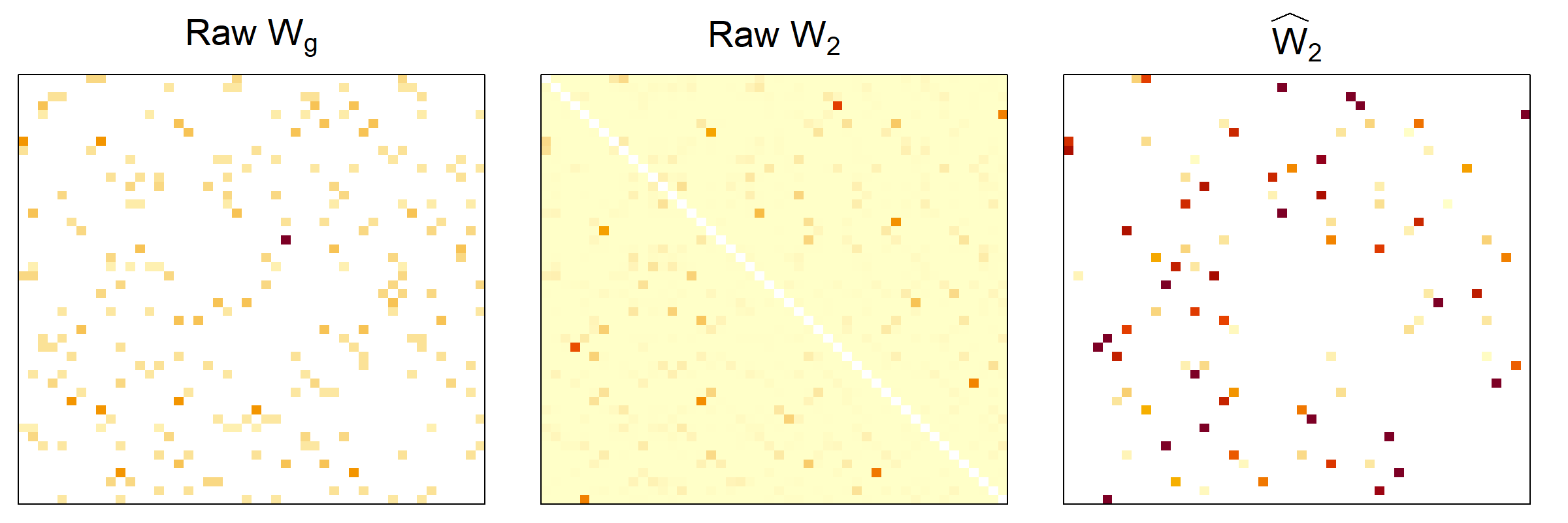}
   \caption{\footnotesize Heatmaps of adjacency matrices associated with 48 states. From left to right, this figure provides the heatmaps associated with the row-normalized 0-1 geographic adjacency matrix $W_g$, the row-normalized inverse squared distance matrix $W_2$, and the estimated low-rank plus sparse matrix $\widehat W_2$.\label{heatmap48}}
\end{figure}

\begin{figure}[htbp]
\includegraphics[width=1\textwidth]{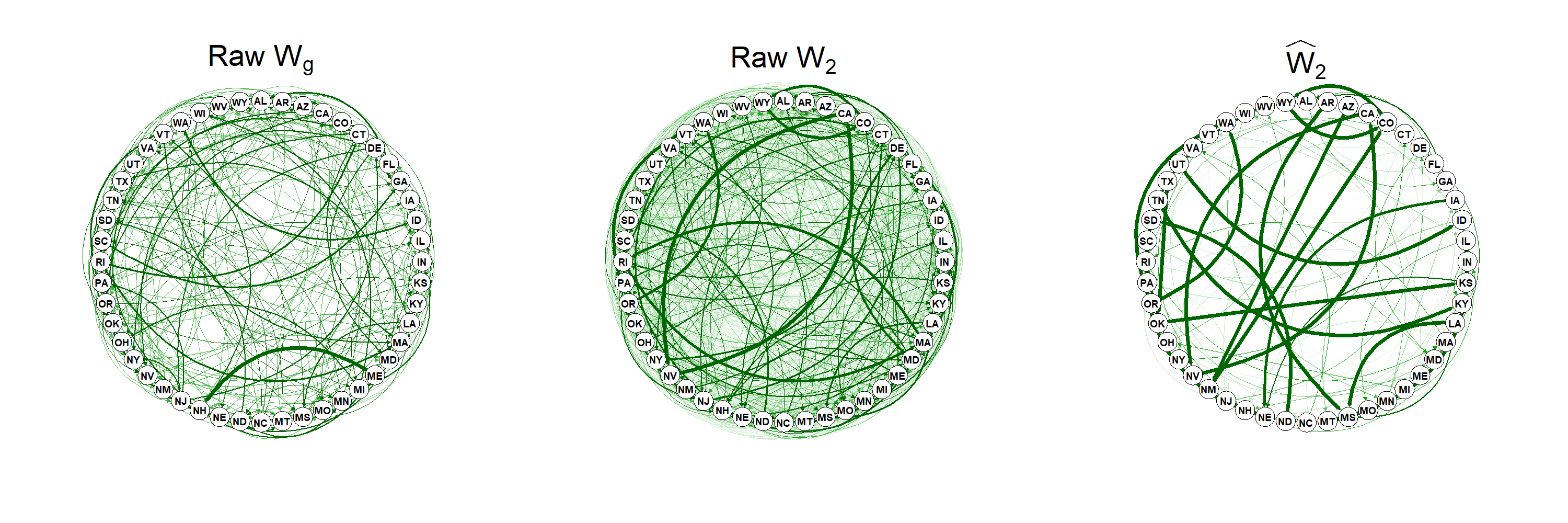}
 \caption{\footnotesize Graphical representations of different weight matrices. From left to right, this figure provides the network representation for the row-normalized  neighborhood  weight matrix $W_g$, the row-normalized inverse squared distance matrix $W_2$ and the estimated low-rank and sparse weight matrix $\widehat W_2$. 
\label{fig:three_graphs2}}
\end{figure}

Figures~\ref{heatmap48} and~\ref{fig:three_graphs2} provide heatmaps and graphical representations of the raw geographic weight matrices $W_g$, $W_2$ and the estimate $\widehat W_2$, respectively. We find that the raw inverse squared distance matrix $W_2$ is much denser than the raw geographic adjacency matrix $W_g$. It may be more vulnerable to measurement errors, and thus exaggerate the overall degree of network dependence and obscure the dominant transmission channels. By contrast, the estimated network discards many weak, noise‑dominated connections, thus providing a clearer visualization of the retained links and improving interpretability.

Table~\ref{tab:geo_ranktop} provides a more detailed characterization of the state-level network structure. States are sorted according to their column sums, which measure the aggregate strength of their outward connections. Massachusetts is identified as the most influential state under all three weight matrices. Its weighted out-degree is 1.70 under the row-normalized adjacency matrix $W_g$, 1.61 under the row-normalized inverse squared distance matrix $W_2$, and 0.72 under the estimate $\widehat W_2$. The estimated network therefore preserves the leading position of Massachusetts while substantially reducing the magnitude of its estimated outward influence.
This suggests that our estimate does not fundamentally alter the ranking of the major network hubs, but removes potentially noisy links that may inflate the overall strength of network connections rather than incorporating them into the core neighborhood topology. By contrast, the ranking based on $W_g$ differs more noticeably because the geographical adjacency matrix only captures direct neighboring relationships, whereas $W_2$ allows for distance-decaying connections between all states.

\begin{table}[ht]
	\centering
	\caption{Top Four States Ranked by Weighted Out-degree (WO)}
    \label{tab:geo_ranktop}
	\begin{tabular}{r|lr|lr|lr}
		\toprule
		& \multicolumn{2}{c|}{Raw $W_g$} & \multicolumn{2}{c|}{Raw $W_2$} & \multicolumn{2}{c}{$\widehat W_2$}\\
		Rank & State & WO & State & WO & State & WO\\
		\midrule
		1 & MA & 1.70 & MA & 1.61 & MA & 0.72\\
		2 & GA & 1.62 & MD & 1.57 & MD & 0.64\\
		3 & TN & 1.58 & RI & 1.53 & RI & 0.63\\
		4 & ID, NH & 1.53 & DE & 1.52 & DE & 0.63\\
		\bottomrule
	\end{tabular}
    \caption*{\footnotesize Notes: We consider three choices of the weight matrix, including the raw row-normalized adjacency matrix $W_g$, the raw row-normalized inverse squared distance matrix $W_2$, and the estimate $\widehat W_2$. For each matrix, we compute the weighted out-degree (column sums) and report the four states with the largest values in descending order. The abbreviations MA, GA, TN, ID, NH, MD, RI, DE stand for Massachusetts, Georgia, Tennessee, Idaho, New Hampshire, Maryland,  Rhode Island and Delaware, respectively.
}
\end{table}

\begin{table}[htbp]
	\centering
	\caption{Spillover Estimates and Influence Analysis}
	\label{tab:geo-spillovers}
	%\scriptsize
	\begin{tabular}{lrrlrr}
		\toprule
		Weight Matrix & $\widehat\lambda$ & $\Delta\widehat\lambda$ & Key Region & Influence Share &  $\Delta$Influence Share\\
		\midrule
		\multicolumn{6}{l}{\textit{Panel A: without contextual effects}}\\\\
		Raw $W_g$ & 0.658 & 30.3\% & Mountain & 0.181 & 3.9\%\\
		Raw $W_2$ & 0.505 & 0.0\% & SA & 0.174 & 0.0\% \\
	      $\widehat W_2$ & 0.220 & -56.5\% & SA & 0.168 & -3.3\%\\
		$\widehat W_s$ & 0.220 & -56.5\% & SA & 0.168 & -3.3\%\\
		\addlinespace
        \midrule
		\multicolumn{6}{l}{\textit{Panel B: with contextual effects}}\\\\
		Raw $W_g$ & 1.006 & -11.0\% & NA & NA & NA\\
		Raw $W_2$ & 1.130 & 0.0\% & NA & NA & NA\\
		$\widehat W_2$ & 0.811 & -28.3\% & SA & 0.177 & NA\\
		$\widehat W_s$ & 0.811 & -28.2\% & SA & 0.177 & NA\\
		\bottomrule
	\end{tabular}
    \caption*{\footnotesize Notes: 
    This table reports the estimated spillover coefficients obtained using different spatial weight matrices. \(W_g\) and \(W_2\) denote the row-normalized binary adjacency matrix and inverse-squared-distance weight matrix, respectively. \(\widehat{W}_2\) and \(\widehat{W}_s\) denote the estimated and supervised estimated weight matrices, respectively, constructed from \(W_2\). For comparability, each estimated weight matrix is also row-normalized before estimating the spillover coefficient. When \(|\widehat{\lambda}|<1\), regional influence is calculated by summing the columns of the Leontief inverse, \((I-\widehat{\lambda}W)^{-1}\), where \(W\) denotes the corresponding weight matrix, and aggregating these column sums within each of the nine geographic divisions to obtain \(RI_1,\ldots,RI_9\). The influence share of division \(i\) is defined as \(RI_i/\sum_{j=1}^{9} RI_j\). The division with the largest influence share is identified as the key region. SA denotes South Atlantic. Relative changes in the spillover coefficient and influence share are calculated using the corresponding results for \(W_2\) as the baseline, where available.} 
\end{table}

Table~\ref{tab:geo-spillovers} reports the estimated spillover coefficient $\widehat{\lambda}$ and the regional influence analysis using different weight matrices. We consider two scenarios, Panel A (without $WX$) and Panel B (with $WX$), respectively. In both panels, the raw weight matrices $W_g$ and $W_2$ produce quite similar results, and they both exhibit larger estimates outside the stability region in Panel B. Such large values are likely driven by measurement error or noise in the raw adjacency matrices, yielding  exaggerated spillover estimates that could be unreliable and potentially misleading. In contrast, the low-rank-plus-sparse estimator $\widehat{W}$, and its supervised variants, produce substantially smaller estimates within the stability region, and consistently identify the South Atlantic as the key region. Recall that each weight matrix is row-normalized prior to estimation.
By the Perron–Frobenius theorem, the spectral radius of a row-normalized non‑negative matrix is 1, and hence the standard Leontief-stability condition reduces to $|\widehat{\lambda}|<1$ \citep{LeSage2009}. Explosive spillover estimates can cause some elements of the Leontief inverse to become negative and hard to interpret. These results indicate that our approach might yield more interpretable spillover effects, avoiding the inflated influence suggested by raw network matrices.

\begin{table}[ht]
	\centering
	\caption{Top Four States Most Strongly Influenced by Massachusetts}
	\label{tab:n520-app-gmm-ma-influence}
	\begin{tabular}{lllll}
		\toprule
		Weight matrix & 1st & 2nd & 3rd & 4th\\
		\midrule
		\multicolumn{5}{l}{\textit{Panel A: Without contextual effects}}\\\\
		Raw $W_g$ & MA (1.30) & RI (0.59) & CT (0.48) & VT (0.45)\\
		Raw $W_2$ & MA (1.13) & RI (0.34) & NH (0.24) & CT (0.17)\\
		$\widehat W_2$ & MA (1.04) & RI (0.19) & NH (0.17) & ME (0.06)\\
		$\widehat W_s$ & MA (1.04) & RI (0.19) & NH (0.17) & ME (0.06)\\
		\midrule
		\multicolumn{5}{l}{\textit{Panel B: With contextual effects}}\\\\
		$\widehat W_2$ & MA (2.51) & RI (1.88) & NH (1.82) & ME (1.48)\\
		$\widehat W_s$  & MA (2.52) & RI (1.88) & NH (1.82) & ME (1.48)\\
		\bottomrule
	\end{tabular}
    \caption*{\footnotesize Notes: We evaluate the column corresponding to Massachusetts in the Leontief inverse %$(I_n-\widehat{\lambda}\widehat{W})^{-1}$ 
    and interpret its entries as the total influence that a shock originating from Massachusetts exerts. We report the four largest values in parentheses along with the corresponding states. Four choices of the weight matrix are considered: the raw geographical adjacency matrix $W_g$, the raw row-normalized inverse squared distance matrix $W_2$, the estimator $\widehat W_2$ and the supervised estimator $\widehat W_s$. The abbreviations MA, RI, CT, VT, NH, ME stand for Massachusetts, Rhode Island, Connecticut, Vermont, New Hampshire and Maine, respectively.
We do not include the results for the raw matrices  $W_{2}$ and $W_{g}$ in Panel B because the estimated spillover effects do not satisfy the Leontief stability  condition.
}
\end{table}

Table~\ref{tab:n520-app-gmm-ma-influence} further examines the propagation of shocks originating from Massachusetts. We calculate the Leontief inverse and examine the four largest elements in the column  corresponding to Massachusetts. The two states most strongly influenced by a shock originating from Massachusetts are Massachusetts itself and Rhode Island. Relative to the raw-network estimates, the magnitude of the estimated network influence is attenuated. When contextual effects $WX$ are incorporated, the Leontief influence values increase substantially. Under the estimate $\widehat W_2$, the influence measure for Massachusetts rises from 1.04 to 2.51, while the corresponding values for Rhode Island, New Hampshire, and Maine increase to 1.88, 1.82, and 1.48, respectively. Overall, the results consistently identify Massachusetts as an important regional hub, and including contextual covariates substantially changes the estimated response to a unit additive shock originating in Massachusetts. Our estimation procedure retains stronger local transmission channels while suppressing weaker long-range associations that may be subject to measurement error.

    \section{Conclusions}
    \label{Sec:Conclusion}
    This paper introduces a robust framework for studying interaction effects in social and spatial networks under noisy adjacency matrices. By leveraging the low-rank and sparse structure of real-world networks, our de-noising approach, combining LASSO and nuclear norm penalization, significantly improves estimation accuracy. We propose two estimation procedures—a two-step and a one-step supervised GMM estimator—that outperform standard GMM under noise and endogeneity. 
Applying our proposed methodology to examine the international spillover of economic growth and the tax competition across U.S. states, we find a pronounced discrepancy in estimating the spillover effects, which highlights the essential role of accurate network denoising. Failing to distinguish genuine connectivity from measurement noise can systematically distort spillover coefficients as well as the strength and direction of cross-unit interactions. Moreover, our decomposition framework might provide a new perspective on complex spatial dependencies by disentangling global and local components. This decomposition not only sharpens the economic interpretation but also improves econometric inference, thereby contributing to credible spatial econometric analysis.

\section{Declaration of generative AI and AI-assisted technologies in the writing process}

During the preparation of this work the authors used AI to suggest edits to some lengthy passages for concision and clarity. After using these tools, the authors reviewed and edited the content as needed and take full responsibility for the content of the published article.

    % Bibliography and references
\begingroup
\singlespacing
\setlength{\bibsep}{0pt}
\setlength{\parskip}{0pt}
    \bibliographystyle{apalike}
    \bibliography{LowRank_new2}
    \endgroup

    % Appendices divided into sections
    \begin{appendices}
    \numberwithin{figure}{section}
    \numberwithin{table}{section}
    \numberwithin{equation}{section}
  \section{Algorithms}
    \label{Appendix:Algorithms}
    \subsection{Algorithms}\label{algo}

For a singular value decomposition $A=UDV^\top$, let
$\SVT_{\nu}(A)=U(D-\nu I)_+V^\top$, and let
${\bf S}_{\tau}(\cdot)$ denote entrywise soft thresholding.

\begin{algorithm}[H]
\caption{Sequential low-rank-plus-sparse network denoising}
\label{Algo:CD_LRSED}
\begin{algorithmic}[1]
\Require Observed matrix $W$ and penalty parameters $\nu_n,\tau_n>0$.

\State Obtain the preliminary sparse estimate
$\widetilde S_{ij}=\operatorname{sign}(W_{ij})(|W_{ij}|-\tau_n)_+$ for
$i\neq j$, and set $\widetilde S_{ii}=0$.

\State Obtain the low-rank estimate
$\widehat L=\SVT_{\nu_n}(W-\widetilde S)$.

\State Set $\widehat S_{ij}=\widetilde S_{ij}$ for $i\neq j$ and
$\widehat S_{ii}=-\widehat L_{ii}$, so that
$\operatorname{diag}(\widehat L+\widehat S)=\mathbf 0_{n\times n}$.

\State Set $\widehat W=\widehat L+\widehat S$.

\State \Return $\widehat L$, $\widehat S$, and $\widehat W$.
\end{algorithmic}
\end{algorithm}
\noindent{\footnotesize\textit{Note.}
The algorithm is stated for networks with zero diagonal.
In applications where diagonal entries are economically
meaningful, such as input--output matrices, soft thresholding
may also be applied to diagonal entries in Step 1.
For input--output applications, Step 3 may instead set
$\widehat S_{ii}=W_{ii}-\widehat L_{ii}$, so that
$\operatorname{diag}(\widehat W)=\operatorname{diag}(W)$.
\par}

\begin{algorithm}[H]
\caption{Supervised estimation of the spatial autoregressive model}
\label{Algo:ADMM_GMM_FBS}
\begin{algorithmic}[1]
\Require $W$, $\{Y_t,X_t\}_{t=1}^T$, $M$,
$\widehat\Lambda_{nT}^{-1}$, and
$\nu_{nT},\tau_{nT},\xi_{nT}$.
\State Let $Z_t$ contain the $m$ linearly independent columns of
$(X_t,MX_t,\ldots,M^dX_t)$.
\State Starting from any feasible computational initialization, obtain
$(\widehat L^{[0]},\widehat S^{[0]})$ as a minimizer of
$\frac12\|W-L-S\|_F^2
+\nu_{nT}\|L\|_*
+\tau_{nT}|S|_{1,1}$
subject to $\operatorname{diag}(L+S)=0_{n\times n}$, and set
$\widehat W^{[0]}=\widehat L^{[0]}+\widehat S^{[0]}$.
\State Obtain
$\widehat\theta^{[1]}
=\arg\min_{\theta\in\Theta}
J_{nT}(\theta,\widehat W^{[0]},M;
\widehat\Lambda_{nT}^{-1})$.
\State Given $\widehat\theta^{[1]}$ and $\widehat L^{[0]}$, obtain
$\widehat S^{[1]}$ by minimizing
$\frac{\xi_{nT}}{nT}
J_{nT}(\widehat\theta^{[1]},\widehat L^{[0]}+S,M;
\widehat\Lambda_{nT}^{-1})
+\frac12\|W-\widehat L^{[0]}-S\|_F^2
+\tau_{nT}\sum_{i,j}|S_{ij}|$
subject to $\operatorname{diag}(\widehat L^{[0]}+S)=0_{n\times n}$.
\State Given $\widehat\theta^{[1]}$ and $\widehat S^{[1]}$, obtain
$\widehat L^{[1]}$ by minimizing
$\frac{\xi_{nT}}{nT}
J_{nT}(\widehat\theta^{[1]},L+\widehat S^{[1]},M;
\widehat\Lambda_{nT}^{-1})
+\frac12\|W-L-\widehat S^{[1]}\|_F^2
+\nu_{nT}\|L\|_*$
subject to $\operatorname{diag}(L+\widehat S^{[1]})=0_{n\times n}$.
\State Set
$\widehat W_s=\widehat L^{[1]}+\widehat S^{[1]}$ and obtain
$\widehat\theta_s
=\arg\min_{\theta\in\Theta}
J_{nT}(\theta,\widehat W_s,M;
\widehat\Lambda_{nT}^{-1})$.
\State \Return
$\widehat\theta_s$,
$\widehat L_s=\widehat L^{[1]}$,
$\widehat S_s=\widehat S^{[1]}$, and
$\widehat W_s$.
\end{algorithmic}
\end{algorithm}

\noindent{\footnotesize\textit{Note.}
The algorithm is stated for networks with zero diagonal.
In applications where diagonal entries are economically
meaningful, such as input--output matrices, the constraint
$\operatorname{diag}(L+S)=0_{n\times n}$, together with the corresponding
elementwise restrictions imposed in subsequent steps,
can be omitted.
\par}

\section{Additional Tables}
\begin{table}[htbp]
    \centering
    \captionsetup{width=0.975\linewidth}
    \caption{Relative Mean Frobenius Norm Estimation Errors of the Leontief Inverse}  %$(I_n-\lambda_0 W_0)^{-1}$ }
\label{Tab:Relative_SpilloverMultiplier_Exogenous}
    % ===================== (a) Exogenous =====================
    \begin{adjustbox}{max width=0.95\textwidth}
    \begin{tabular}{lcc|cc|cc|cc}
        \toprule 
 Exogenous &Plug-in & Supervised & Plug-in & Supervised                   & Plug-in & Supervised & Plug-in & Supervised \\
        \midrule
        L (Low-rank) & \multicolumn{2}{c}{$T = 1$} & \multicolumn{2}{c}{$T = 5$}
                     & \multicolumn{2}{c}{$T = 15$} & \multicolumn{2}{c}{$T = 50$} \\
        \midrule
 $n=40$  & 0.493 & 0.480 & 0.425 & 0.421 & 0.418 & 0.412 & 0.416 & 0.406 \\ 
   $n=80$  & 0.322 & 0.321 & 0.301 & 0.301 & 0.298 & 0.297 & 0.297 & 0.293 \\ 
   $n=120$  & 0.411 & 0.405 & 0.280 & 0.279 & 0.271 & 0.271 & 0.268 & 0.280 \\ 
  \midrule
        L+Sparse & \multicolumn{2}{c}{$T = 1$} & \multicolumn{2}{c}{$T = 5$}
                     & \multicolumn{2}{c}{$T = 15$} & \multicolumn{2}{c}{$T = 50$} \\
        \midrule
  
   $n=40$  & 0.414 & 0.413 & 0.396 & 0.394 & 0.391 & 0.390 & 0.390 & 0.389 \\ 
   $n=80$  & 0.296 & 0.296 & 0.278 & 0.278 & 0.274 & 0.274 & 0.273 & 0.272 \\ 
   $n=120$  & 0.245 & 0.245 & 0.228 & 0.228 & 0.224 & 0.224 & 0.223 & 0.223 \\ 
    \midrule
        Dominant & \multicolumn{2}{c}{$T = 1$} & \multicolumn{2}{c}{$T = 5$}
                     & \multicolumn{2}{c}{$T = 15$} & \multicolumn{2}{c}{$T = 50$} \\
        \midrule
  
   $n=40$  & 0.552 & 0.552 & 0.550 & 0.550 & 0.549 & 0.549 & 0.549 & 0.550 \\ 
   $n=80$  & 0.511 & 0.511 & 0.507 & 0.507 & 0.507 & 0.507 & 0.506 & 0.507 \\ 
   $n=120$  & 0.337 & 0.337 & 0.328 & 0.328 & 0.327 & 0.327 & 0.326 & 0.326 \\ 
    \midrule
        Group & \multicolumn{2}{c}{$T = 1$} & \multicolumn{2}{c}{$T = 5$}
                     & \multicolumn{2}{c}{$T = 15$} & \multicolumn{2}{c}{$T = 50$} \\
        \midrule
  
   $n=40$  & 0.533 & 0.552 & 0.473 & 0.493 & 0.466 & 0.484 & 0.464 & 0.477 \\ 
   $n=80$  & 0.415 & 0.403 & 0.337 & 0.334 & 0.332 & 0.328 & 0.329 & 0.325 \\ 
   $n=120$  & 0.323 & 0.326 & 0.299 & 0.302 & 0.296 & 0.298 & 0.295 & 0.294 \\ 
        \bottomrule
    \end{tabular}
    \end{adjustbox}

    \vspace{0.5em}

    % ===================== (b) Endogenous =====================
    \begin{adjustbox}{max width=0.95\textwidth}
    \begin{tabular}{lcc|cc|cc|cc}
        \toprule
       Endogenous & Plug-in & Supervised & Plug-in & Supervised                    & Plug-in & Supervised & Plug-in & Supervised \\
        \midrule
        L (Low-rank) & \multicolumn{2}{c}{$T = 1$} & \multicolumn{2}{c}{$T = 5$}
                     & \multicolumn{2}{c}{$T = 15$} & \multicolumn{2}{c}{$T = 50$} \\
        \midrule
$n=40$ & 0.448 & 0.435 & 0.412 & 0.406 & 0.407 & 0.400 & 0.407 & 0.401 \\ 
 $n=80$ & 0.336 & 0.336 & 0.332 & 0.330 & 0.331 & 0.327 & 0.330 & 0.321 \\ 
 $n=120$ & 0.313 & 0.308 & 0.219 & 0.219 & 0.214 & 0.218 & 0.213 & 0.239 \\ 
   \midrule
        L+Sparse & \multicolumn{2}{c}{$T = 1$} & \multicolumn{2}{c}{$T = 5$}
                 & \multicolumn{2}{c}{$T = 15$} & \multicolumn{2}{c}{$T = 50$} \\
        \midrule
  
  $n=40$ & 0.405 & 0.403 & 0.389 & 0.387 & 0.385 & 0.383 & 0.384 & 0.382 \\ 
 $n=80$ & 0.292 & 0.292 & 0.277 & 0.276 & 0.273 & 0.273 & 0.272 & 0.272 \\ 
  $n=120$  & 0.237 & 0.237 & 0.223 & 0.222 & 0.219 & 0.219 & 0.219 & 0.218 \\ 
    \midrule
        Dominant & \multicolumn{2}{c}{$T = 1$} & \multicolumn{2}{c}{$T = 5$}
                 & \multicolumn{2}{c}{$T = 15$} & \multicolumn{2}{c}{$T = 50$} \\
        \midrule
  
  $n=40$ & 0.536 & 0.536 & 0.542 & 0.542 & 0.542 & 0.543 & 0.542 & 0.543 \\ 
  $n=80$ & 0.492 & 0.492 & 0.497 & 0.497 & 0.498 & 0.498 & 0.498 & 0.499 \\ 
 $n=120$ & 0.327 & 0.327 & 0.324 & 0.324 & 0.324 & 0.324 & 0.324 & 0.324 \\ 
    \midrule
        Group & \multicolumn{2}{c}{$T = 1$} & \multicolumn{2}{c}{$T = 5$}
                 & \multicolumn{2}{c}{$T = 15$} & \multicolumn{2}{c}{$T = 50$} \\
        \midrule  
  $n=40$ & 0.486 & 0.501 & 0.441 & 0.456 & 0.437 & 0.449 & 0.436 & 0.445 \\ 
  $n=80$ & 0.351 & 0.342 & 0.301 & 0.299 & 0.299 & 0.295 & 0.297 & 0.297 \\ 
  $n=120$ & 0.308 & 0.311 & 0.296 & 0.298 & 0.295 & 0.296 & 0.295 & 0.291 \\ 
 
         \bottomrule
    \end{tabular}
    \end{adjustbox}
    \caption*{\footnotesize
Note: The Frobenius norm estimation errors for the Leontief inverse
under the plug-in and supervised approaches are $\|(I_n-\widehat\lambda_p\widehat W)^{-1}
-(I_n-\lambda_0W_0)^{-1}\|_F$
and
$\|(I_n-\widehat\lambda_s\widehat W_s)^{-1}
-(I_n-\lambda_0W_0)^{-1}\|_F$,
respectively. Here, $\widehat\lambda_p$ is the plug-in estimate
based on the denoised network estimate $\widehat W$, whereas
$\widehat\lambda_s$ and $\widehat W_s$ are the supervised estimates.
Entries report the ratio of each method's mean estimation error
across Monte Carlo replications to the corresponding mean error
of the conventional GMM estimator. The latter error is computed as
$\|(I_n-\widehat\lambda_{\mathrm{GMM}}W)^{-1}
-(I_n-\lambda_0W_0)^{-1}\|_F$
in each replication, where $\widehat\lambda_{\mathrm{GMM}}$ is
the GMM estimate based on the observed noisy weight matrix $W$.
Endogenous measurement error is simulated with a correlation
parameter of $\rho=0.7$, as specified in
\eqref{Eq:Correlated_Errors}. The four DGPs are described in
Subsection~\ref{Subsec:NetworkDGPs}.
}
\end{table}
\newpage
    \section{Motivation for Low-rank plus sparse decomposition}
    \label{Appendix:decomp}
    
%\subsection{Motivation for Low-Rank Plus Sparse Structure} 

The low-rank plus sparse decomposition is useful in network analysis because it can improve interpretability and computational efficiency across various applications. Additionally, it is closely related to {eigenvector centrality} and the {Leontief inverse}, both of which play crucial roles in policy targeting by identifying key structural dependencies and intervention points \citep{jackson2008social, newman2010networks}.

\subsection{Application to eigenvector centrality} \label{sub:ec}
Our decomposition can be useful in interpreting and understanding centrality as such a measure can be shown to be dominated by either the sparse or low-rank components in certain potentially relevant circumstances. The (eigenvector) centrality of each node corresponds to the eigenvector entries attached to the highest eigenvalue of the adjacency matrix representing the network.{\footnote{  For a connected network $W_0$ whose entries are all non-negative, the Perron-Frobenius theorem implies that its maximum eigenvalue (in terms of radius) is real and non-negative.}}  The low-rank plus sparse decomposition separates the adjacency matrix \( W_0\) into a low-rank structure \( L_0^* \) (or $L_0$, before diagonal adjustment), representing the core economic interactions, and a sparse component \( S_0^* \) (or $S_0$), capturing highly influential entities or individual heterogeneity \citep{gabaix2011granular, carvalho2014}. The eigenvector centrality under this decomposition expands to:
\[
%(L_0 + S_0) \mathbf{v}_{W_0} =
(L_0^* + S_0^*) \mathbf{v}_{W_0}= \lambda_{W_0}\mathbf{v}_{W_0},
\]
 where $\lambda_{W_0}$ and $\mathbf{v}_{W_0}$ are the leading eigenvalue and eigenvector associated with $W_0$.
The entries of $\mathbf{v}_{W_0}$ are related to the leading eigenvectors of $\bL_0$ and $\bS_0$, denoted by $\mathbf{v}_{\bL_0}$ and $\mathbf{v}_{\bS_0}$, respectively, which capture different aspects of the underlying structure.  
 If %the leading eigenvector 
 $\mathbf{v}_{W_0}$ is closer to $\mathbf{v}_{\bL_0}$ or $\mathbf{v}_{\bS_0}$, this indicates that the low-rank or the sparse structure, respectively, is more dominant.
 
We now formalize how $\mathbf{v}_{W_0}$ is informed by $\mathbf{v}_{\bL_0}$ and $\mathbf{v}_{\bS_0}$. %Denote by $\mathbf{v}_{\bL_0}$ and $\mathbf{v}_{\bS_0}$ the corresponding (maximal eigenvalue) eigenvectors for $\bL_0$ and $\bS_0$, respectively.  In addition, let 
Let $\rho_{\bL_0\bS_0}$ denote the inner product between $\bvL$ and $\bvS$, i.e., $\rho_{\bL_0\bS_0}=\bvL ^\top \bvS$. Similarly define $\rho_{W_0\bL_0}=\bvW^\top \bvL$ and $\rho_{W_0\bS_0}=\bvW^\top \bvS$ and let $ k_{\bL_0}\bvL+k_{\bS_0}\bvS$ denote the best linear approximation of $\bvW$ using $\bvL$ and $\bvS$. %{\color{red}PLEASE ADD THE DERIVATION IN the APPENDIX}
Then from a regression perspective, it satisfies: %$\hat \bvW=k_{L_0}\bvL+k_{S_0}\bvS$, where 
\begin{eqnarray}\label{rhoLS}
k_{\bL_0}=\frac{\rho_{W_0\bL_0}-\rho_{\bL_0\bS_0}\rho_{W_0\bS_0}}{1-\rho_{\bL_0\bS_0}^2}, \qquad  k_{\bS_0}=\frac{\rho_{W_0\bS_0}-\rho_{\bL_0\bS_0}\rho_{W_0\bL_0}}{1-\rho_{\bL_0\bS_0}^2} 
\end{eqnarray}
(see the Online Appendix \ref{perb}).
Further calculations show that the coefficients $k_{\bL_0}$ and $k_{\bS_0}$ should be related to the maximum eigenvalues of $\bL_0$ and $\bS_0$ respectively, i.e. $\lambda_{\bL_0}$ and $\lambda_{\bS_0}$. If $\bL_0$ and $\bS_0$ are symmetric, we have 
\begin{eqnarray}
\frac{k_{\bL_0}}{k_{\bS_0}} &\approx& \frac{\rho_{\bL_0\bS_0}(\lambda_{\bL_0}-\bvS^\top \bL_0\bvS )}{\lambda_{S_0^*}+\bvS^\top \bL_0\bvS-\lambda_{L_0^*}-\bvL^\top \bS_0\bvL}, \  \text{if} \  \lambda_{\bL_0}+\bvL^\top \bS_0\bvL< \lambda_{\bS_0}+\bvS^\top \bL_0\bvS  \label{eq:kSL1},\\
%\frac{k_{S_0}}{k_{L_0}} &\approx & \frac{\rho_{L_0S_0}(\bvS^\top L_0\bvS -\lambda_{L_0})}{\lambda_{L_0}+\bvL^\top S_0\bvL-\lambda_{S_0}-\bvS^\top L_0\bvS}   \qquad \text{if} \qquad \lambda_{L_0}+\bvL^\top S_0\bvL> \lambda_{S_0}+\bvS^\top L_0\bvS   \label{eq:kSL2}\\
\frac{k_{\bS_0}}{k_{\bL_0}} &\approx& \frac{\rho_{\bL_0\bS_0}(\lambda_{\bS_0}-\bvL^\top \bS_0\bvL )}{\lambda_{\bL_0}+\bvL^\top \bS_0\bvL-\lambda_{\bS_0}-\bvS^\top \bL_0\bvS}, \   \text{if} \ \lambda_{\bL_0}+\bvL^\top \bS_0\bvL> \lambda_{\bS_0}+\bvS^\top \bL_0\bvS , \label{eq:kSL2}
\end{eqnarray}
{where the approximation error diminishes
when $\rho_{\bL_0\bS_0}\to 0$.} (These approximations are established under the symmetry and
spectral conditions stated in the Online Appendix~\ref{perb}.)
In practical scenarios involving large networks, $\bvL$ is typically dense, with few zero entries, while $\bvS$ tends to be sparse, featuring many zero entries, thus implying that the correlation $\rho_{\bL_0 \bS_0}$ is small.
Consequently, if $\rho_{\bL_0\bS_0}$ is close to 0, then $k_{\bL_0}/k_{\bS_0}\approx 0$ if (\ref{eq:kSL1}) holds or $k_{\bS_0}/k_{\bL_0}\approx 0$ if (\ref{eq:kSL2}) holds.
The inequality
$\lambda_{L_0^*} + \bvL^\top S_0^* \bvL
\quad \lessgtr \quad
\lambda_{S_0^*} + \bvS^\top L_0^* \bvS$
determines which component $L_0^*$ or $S_0^*$ plays a more dominant role in shaping the eigenvector centrality of $W_0$.
Specifically, the left-hand side approximates the effective contribution of the low-rank component $L_0^*$ (and the part of $S_0^*$ aligned with it), whereas the right-hand side captures the influence of the sparse component $S_0^*$.
In contrast, the inner product $\bvL^\top \bvS$ measures the degree of alignment between the leading eigenvectors of $L_0^*$ and $S_0^*$, indicating whether the two structures reinforce or counteract each other in determining the dominant direction of $W_0$.
For the detailed calculations above, please refer to the Online Appendix \ref{perb}.  \\
Ultimately, this means that when $\rho_{\bL_0\bS_0}\to 0$ we have:
\[ \mathbf{v}_{W_0} \approx k_{\bL_0} \mathbf{v}_{\bL_0} \quad \mbox{or} \quad \mathbf{v}_{W_0} \approx k_{\bS_0} \mathbf{v}_{\bS_0}. \]
Therefore, we obtain a measure that either emphasizes the systemic structure of the network \( \bL_0 \) or the eigenvector centrality of \( \bS_0 \) which can be attributed to the individual shocks or dominant units \citep{battiston2012}. For example, in interbank networks, systemically important banks often dominate centrality rankings when using the full adjacency matrix \citep{degryse2007}. Applying the low-rank plus sparse decomposition allows us to disentangle the influence of connections with some important institutions  (captured in \( \bS_0 \))  from the broader systemic structure (captured in \( \bL_0 \)). This approach thus potentially provides a more nuanced understanding of systemic risk, as it distinguishes between the centrality of highly connected hubs and the underlying economic network \citep{Acemoglu2015, greenwood2015}.

To illustrate this, we now present two numerical examples to examine the underlying structure of the dominant eigenvector for a network with $n=100$. In both examples, the network $W_0=\bL_0+\bS_0$, where  $L_0=\mathbf{v}_{L_0} \mathbf{v}_{L_0}^{\top}$ has the leading eigenvector $ \mathbf{v}_{L_0}$ whose elements are all $1/10$, and $L_0^*=L_0-1/100 I_{100}$ with $\mathbf{v}_{L_0^*}{=} \mathbf{v}_{L_0}$.%\mathbf{1}_{100}^{\top} /10$. 
 The sparse matrix has the form $\bS_0=kS_A^*$ for some pre-specified sparse matrix $S_A^*$ with zero diagonals.  
The constant $k$ affects the relative strength between the low-rank component $\bL_0$ and the sparse component $\bS_0$, and we consider $k=0.4$ and $k=4$. In Case 1, $S_A^*$ is chosen such that it only has two nonzero elements $S^*_{A,12}=S^*_{A,21}=1$. 
In this case, the maximum eigenvalues of $\bL_0$ and $\bS_0$ are 0.99 and $k$ respectively. If $k=0.4$, the calculated leading eigenvectors show that $\mathbf{v}_{W_0}\approx 0.983\mathbf{v}_{\bL_0}+0.091\mathbf{v}_{\bS_0}$, signifying a predominantly uniform structure with minor local adjustments. However, when $k=4$, the eigenvalue of $\bS_0$ dominates that of $\bL_0$, and we find that $\mathbf{v}_{W_0} \approx 0.046 \mathbf{v}_{\bL_0}+ 0.992\mathbf{v}_{\bS_0} $, reflecting a reduced influence of $ \bL_0 $ due to the increased eigenvalue of $ \bS_0 $.

In Case 2, we set $S_A^*$ to be an asymmetric sparse matrix with 
$S^*_{A,12}=1$ and $S^*_{A,i1}=1$ for $2\leq i\leq 30$. This 
represents the case when the first individual is a key player that 
affects $29$ individuals. Because $S_A^*$ is asymmetric, its 
nonzero eigenvalues are $\pm 1$; we take $\bvS$ to be the 
eigenvector associated with the eigenvalue $+1$. This $\bvS$ has 
more nonzero entries than in Case~1, so $\rho_{\bL_0\bS_0} = 
\bvL^\top\bvS$ is larger. When $W_0=\bL_0+kS_A^*$ with $k=0.4$, 
$\mathbf{v}_{W_0}\approx 0.845\,\bvL+0.244\,\bvS$; when $k=4$, 
$\mathbf{v}_{W_0}\approx 0.148\,\bvL+0.911\,\bvS$, implying a 
negligible correction from $\bL_0$.

\subsection{Application to the Leontief inverse}\label{inverse}
We also highlight how our decomposition is reflected in the Leontief 
inverse for $W_0$. Defining $I_n$ as the $n\times n$ identity matrix, 
the Leontief inverse is given by $(I_n - \lambda_0 W_0)^{-1}$. For 
notational simplicity, throughout this subsection we set $\lambda_0 = 1$ 
and write the Leontief inverse as $(I_n - W_0)^{-1}$; the general case 
follows by replacing $W_0$ with $\lambda_0 W_0$ throughout.
Assume that the low-rank component $L_0$ admits the singular value decomposition as % follows a factorized structure:
%\begin{equation}
$L_0 = U D V^\top,$  
%\end{equation}
where  $D$ is a square diagonal matrix of dimension $r\times r$, and %whose rank is significantly smaller than $n$, $U$ and $V$ are matrices with a small number of columns, indicating that $L_0$ is a low-rank matrix, $D$ is a square diagonal matrix whose rank is significantly smaller than $n$, 
$U$ and $V$ are the matrices that collect the left and right singular vectors respectively. When $r<n$,
this allows dimensionality reduction while retaining key structural information. %By applying the \textbf{Woodbury matrix identity} \citep{Woodbury}, the diffusion matrix $(I_n - W_0)^{-1}$ has a simpler form: 
We seek to calculate the Leontief matrix:
%\begin{equation}
$(I_n - W_0)^{-1} = (I_n - (UDV^\top + S_0))^{-1}.$
%\end{equation}
Applying the Woodbury matrix identity \citep{Woodbury}: % \cite{hager1989updating}:
\begin{equation*}
(A + UDV^\top)^{-1} = A^{-1} - A^{-1} U (D^{-1} + V^\top A^{-1} U)^{-1} V^\top A^{-1},
\end{equation*}
where $A = I_n - S_0$ is assumed to be invertible, we have %, we obtain (assume the invertibility of $(I_n - S_0)$):
\begin{equation}\label{Eq:Woodbury21}
(I_n - UDV^\top - S_0)^{-1} = (I_n - S_0)^{-1} + (I_n - S_0)^{-1} U (D^{-1} - V^\top (I_n - S_0)^{-1} U)^{-1} V^\top (I_n - S_0)^{-1}.
\end{equation}
Denote $\|.\|_{\max}$ and $\|.\|_2$ as the maximum norm and two-norm of a matrix respectively.
Suppose $\|\mbox{diag}(L_0)\|_{\max}\to 0$. By matrix algebra\footnote{Take $A = I_n - S_0^*$ and $B=I_n - S_0$ and note that $ A^{-1}-B^{-1}= A^{-1}(B-A)B^{-1}$.} and $\bS_0-S_0=\mbox{diag}(L_0)$, we have 
\begin{equation*}
\|(I_n - S_0^*)^{-1}-
(I_n - S_0)^{-1}\|_2= \|(I_n - S_0^*)^{-1}\mbox{diag}(L_0)
(I_n - S_0)^{-1}\|_2 \to 0. 
\end{equation*}
If we substitute $(I_n - S_0)^{-1}$ with $(I_n - S_0^*)^{-1}$ in Equation (\ref{Eq:Woodbury21}) and use Woodbury's formula, we have  
\begin{equation}\label{Eq:Woodbury22}
(I_n - W_0)^{-1} \approx (I_n - S_0^*)^{-1} + (I_n - S_0^*)^{-1} U (D^{-1} - V^\top (I_n - S_0^*)^{-1} U)^{-1} V^\top (I_n - S_0^*)^{-1}.
\end{equation}
This decomposition reveals two key effects in the network. The first term, $(I_n - \bS_0)^{-1}$, represents the Leontief matrix for the sparse component, capturing localized effects such as idiosyncratic shocks or sporadic interactions. For example, in {economic networks}, this can model firm-specific disruptions or localized financial shocks \cite{gabaix2011granular, Acemoglu2012}. In {supply chains}, it reflects disruptions in specific suppliers without affecting the entire network \cite{barrot2016input, boehm2019input}.

The second term in the Woodbury expansion of \( (I_n - W_0)^{-1} \) 
captures the influence of dense network structures, exploiting the 
low-rank form of \( U \) and \( V \). When \( L_0 = k\mathbf{v}_{L_0} 
\mathbf{v}_{L_0}^\top \), the second term in Equation~(\ref{Eq:Woodbury22}) yields, for symmetric \( S_0^* \), a rank-one matrix \( \tilde k\, \tilde v \tilde v^\top \),  where 
\( \tilde v = (I_n-S_0^*)^{-1} \mathbf{v}_{L_0} \) and 
\( \tilde k = \bigl(1/k - \mathbf{v}_{L_0}^\top (I_n-S_0^*)^{-1} 
\mathbf{v}_{L_0}\bigr)^{-1} \), which reflects the difference between 
\( (I_n-W_0)^{-1} \) and \( (I_n-S_0^*)^{-1} \) due to the presence 
of \( L_0 \). Note that $\tilde{k}\tilde{v}\tilde{v}^\top$ lies along the direction of $\tilde{v}$, with both $\mathbf{v}_{L_0}$ and $S_0^*$ contributing to it. Its effect can be substantial, particularly when $L_0$ is dominant. In economic networks, this balances uniform information spread, as in \cite{jackson2008social} where dense ties drive equilibrium outcomes in coordination games, against localized dynamics. Thus, the correction reconciles global connectivity, akin to \cite{Acemoglu2012}’s analysis of aggregate economic fluctuations via network linkages.

When studying the network effect, it is often assumed $W_0$ has bounded eigenvalues %less than 1 
even when the dimension of the network grows. If the network is highly connected, columns of $U$ and $V$ are dense vectors with many non-zero entries whose magnitude must converge to $0$ due to the bounded row-sum constraint. %the unitary constraint. 
Unlike factor models, the singular values in $D$ may not diverge as $n$ increases, which makes the elements of $L_0$ corresponding to a dense network being very small. However, our calculation above implies that in a dense network, naive thresholding methods are insufficient, as this might result in missing the second term in Equation (\ref{Eq:Woodbury22}) and inconsistent evaluation of network effects that overlook systemic dependencies. We empirically illustrate this point in our empirical application Section \ref{Sec:Application}.

Moreover, utilizing the low-rank and sparse structure, we only need $O(nr + s)$ parameters to store the matrix information,where $r=\operatorname{rank}(L_0)$ and $s=m_s(S_0)$, which highlights the computational advantages we refer to earlier in the section. Note that it costs less time to compute the inverse of a sparse matrix compared to a dense matrix. With the help of Equation (\ref{Eq:Woodbury21}) or (\ref{Eq:Woodbury22}), we could simplify the calculation of the Leontief matrix.

     \section{Identification of the low-rank and sparse components}
    \label{sec:iden}
    In this section, we discuss the identification of low-rank and sparse decomposition.  In other words, given a weight matrix $W_0$ we study discuss regularity conditions for it could be uniquely written as $W_0=L_0+S_0$.  As discussed earlier, we focus on $L_0$ and $S_0$ in the theoretical derivations, since the properties of $L_0^*$ and $S_0^*$ can be derived from them.
We present in this section conditions regarding the identification of 
$L_0$ and $S_0$, and Online Appendix \ref{exampleid} discusses the conditions using examples.

\label{subsec:iden_examples}
The separation of low-rank and sparse components is studied
directly by \citet{Candes2011} and
\citet{chandrasekaran2011rank}. We adopt a deterministic
sufficient condition based on the latter approach.
\footnote{The recovery guarantee of \citet{Candes2011}
allows a positive fraction of nonzero sparse entries under
incoherence and a uniformly random support model.
Our Chandrasekaran-based condition does not require random
support and permits structured network support satisfying
the stated incoherence and degree restrictions, although
it can be more restrictive in terms of the maximum number
of nonzero entries in any row or column. These are different
sets of sufficient conditions, rather than one being uniformly
more general than the other. \citet{Zhou2010stable} extend
principal component pursuit to accommodate an additional
dense perturbation bounded in Frobenius norm.}

Before proceeding, we recall from Section~\ref{Sec:theory} the 
sparsity measures $m_r(A)$, $m_c(A)$, $m_s(A)$, and 
$\deg_{\max}(A) = \max\{m_r(A), m_c(A)\}$, and additionally define 
the \emph{incoherence index} of a nonzero matrix $A$ as 
$\mbox{inc}(A)= \sqrt{\max(\mu_u(A), \mu_{v}(A))}$, where
\begin{equation*}
\mu_u(A) \coloneqq \max_{1\leq i\leq n} \sum_{j=1}^r u_{ij}^2(A), 
\qquad 
\mu_v(A) \coloneqq \max_{1\leq i\leq n} \sum_{j=1}^r v_{ij}^2(A),
\end{equation*}
$r$ is the rank of $A$, and $u_{ij}(A)$, $v_{ij}(A)$ are the $j$-th 
elements of the $i$-th left and right singular vectors. The measure 
is scale-invariant and reflects whether $A$ concentrates its mass in 
few entries: for a diagonal matrix with distinct diagonal entries 
$\mathrm{inc}(A) = 1$, whereas for a fully connected matrix with 
equal entries $\mathrm{inc}(A) = n^{-1/2}$.

The incoherence index controls the maximum entry via 
$\|L_0\|_{\max} \le \sigma_1(L_0)\,\mathrm{inc}(L_0)^2$,\footnote{For 
$L_0 = UDV^\top$, $L_{0,ij} = (U^\top e_i)^\top D (V^\top e_j)$, so 
by Cauchy--Schwarz 
$|L_{0,ij}| \le \sigma_1(L_0)\|U^\top e_i\|_2\|V^\top e_j\|_2 \le 
\sigma_1(L_0)\,\mathrm{inc}(L_0)^2$, using 
$\|U^\top e_i\|_2^2 = \sum_{k=1}^r u_{ik}^2 \le \mu_u(L_0)$.} so when 
$\sigma_1(L_0) = O(1)$, the quantity $\mathrm{inc}(L_0)^2$ plays the 
same role as $1/\sqrt{m_r(L_0)m_c(L_0)}$ in the coherence bound of 
Assumption~\ref{weight}(i): both equal $1/n$ for a fully dense 
$L_0$ and are of constant order for a spiked $L_0$. We use 
$\mathrm{inc}(L_0)$ here because it yields the sharp identification 
condition $\mathrm{inc}(L_0)\,\deg_{\max}(S_0)\le 1/16$.

%\subsection{Identification conditions}\label{sec:iden}
In the following, we provide a theoretical result that establishes the conditions under which $L_0$ and $S_0$ can be disentangled. 

\begin{assumption}[Low rank plus sparse components]
    \label{Assump:W_structure}
Suppose $W_0$ is an $n\times n$ matrix %with constant entries. It 
        which can be decomposed as $W_0=L_0+S_0$, where $L_0$ is a low-rank matrix of a  rank $0\leq r< n$ and $S_0\neq \mathbf{0}_{n\times n}$. %and $r = o(n)$ .
Moreover, $$\mbox{inc}(L_0){\operatorname{deg}_{\max }(S_0)} \leq 1/16.$$
\end{assumption}

\begin{lemma}\label{identification}
Under Assumption \ref{Assump:W_structure}, for any given decomposition 
\(W_0 = L_0 + S_0\) with %\(L_0 \neq \mathbf{0}_{n\times n}\) and 
\(S_0 \neq \mathbf{0}_{n\times n}\), the pair \((L_0, S_0)\) 
can be  uniquely recovered.\footnote{Uniqueness here means that there exists an 
optimization scheme that pins down the pair \((L_0, S_0)\) uniquely; 
see \cite{Chan2009}. Note that the identification is local, in the sense that what we can identify is a discrete set. Local perturbations of the solution can also be identified; however, at the global level, multiple identified sets may exist.}
\end{lemma}

{The proof is provided in Section F.3 of the arXiv version
of this paper, see \cite{li2026estimating}.}

Assumption \ref{Assump:W_structure} is a direct consequence of Corollary~3 in \cite{Chan2009}. %,  \cite{chandrasekaran2009sparse}. 
This assumption imposes that the sparse matrix \( S_0 \) is not excessively dense, while the low-rank matrix \( L_0 \) must not be overly sparse. The incoherence measure provides a quantitative assessment of the dense level of \( L_0 \). Specifically, if \( L_0 \) is sparse, its incoherence measure is of constant order, which might preclude the above inequality from holding when \( S_0 \neq {\bf 0}_{n\times n}\). Conversely, if \( L_0 \) is dense and contains infinitely many small entries, its incoherence measure tends toward zero. For instance, if \( L_0 \) possesses a singular vector with entries of order \( 1/\sqrt{n} \), the incoherence measure attains an order of {\( 1/\sqrt{n} \)}, which is the smallest order achievable. In such a case,  it could be separated from \( S_0 \) if $\operatorname{deg}_{\max }(S_0)=o(\sqrt{n})$, which includes the sparse network with bounded degree, or the dominant network with finite keyplayers. It also implies that $m_s(S_0) = o(n^{\frac{3}{2}})$ as $m_r(S_0)\vee m_c(S_0) =o(\sqrt{n})$.

 \section{Proofs of main results}
    \label{Appendix:Proofs}
    
\subsection{Notation}
Recall that $\|A\|_1 = \max_{1\leq j\leq n} \sum_{i=1}^n |A_{ij}|$, $\|A\|_2 =  \sigma_1(A)$, $\|A\|_{\infty} = \max_{1\leq i\leq n}\sum_{j=1}^n|A_{ij}|$, $\|A\|_{\max} = \max_{1\leq i,j\leq n}|A_{i,j}|$, $\|A\|_F = \sqrt{\sum A_{ij}^2} $, and $\|A\|_*=\sum_{i=1}^n \sigma_i(A)$.
Define $\langle A, B \rangle_F = \tr(A\T B)$ as the Frobenius inner product between $A$ and $B$.
For any square matrix $A$, define
$\operatorname{off}(A)
\coloneqq
A-\operatorname{diag}(A)$,
where $\operatorname{diag}(A)$ denotes the diagonal matrix containing
the diagonal entries of $A$.
Throughout the appendix, write
$\Sigma_{Z_0Q_0}:=-G_0$ and
$\Sigma_{Z_0Q_0}(M):=-G_0(M)$.
\subsection{Proof of Main Theorems}

\begin{remark}[Within transformation and two-way fixed effects]
\label{rem:twoway_demean}

Throughout the spillover-estimation proofs, we consider
model~\eqref{maineq} without individual or time fixed effects.
Thus, $Y_t$, $X_t$, $Z_t(M)$, and $\varepsilon_t$ denote
the original, untransformed quantities. For a balanced panel
with individual effects $\alpha$ and time effects
$\iota_t\mathbf 1_n$, the two-way within-transformed
disturbance is
$\ddot\varepsilon_{it}
=\varepsilon_{it}-\bar\varepsilon_{i\cdot}
-\bar\varepsilon_{\cdot t}+\bar\varepsilon_{\cdot\cdot}$.
Under Assumption~\ref{Assump:Spatial_Errors}, the
transformation induces weak cross-sectional and temporal
dependence, with
$\operatorname{Cov}(\ddot\varepsilon_{it},
\ddot\varepsilon_{js})
=\sigma_0^2(\mathbb I\{i=j\}-1/n)
(\mathbb I\{t=s\}-1/T)$.
In particular, for $i\neq j$,
$\operatorname{Cov}(\ddot\varepsilon_{it},
\ddot\varepsilon_{jt})
=-\sigma_0^2(1-1/T)/n$;
for $t\neq s$,
$\operatorname{Cov}(\ddot\varepsilon_{it},
\ddot\varepsilon_{is})
=-\sigma_0^2(1-1/n)/T$;
and for $i\neq j$ and $t\neq s$, the covariance
equals $\sigma_0^2/(nT)$.
Thus, although the induced dependence is weak at
the individual level, its aggregate contribution
may affect the asymptotic variance without
necessarily changing the convergence rate.

Applying the same transformation to the regressors
and instruments, including spatial-lag terms, and
writing $\ddot Z_t(M)$ for the transformed instruments,
symmetry and idempotence of the within projection imply
$\sum_{t=1}^T\ddot Z_t(M)^\top\ddot\varepsilon_t
=\sum_{t=1}^T\ddot Z_t(M)^\top\varepsilon_t$.
Under strict exogeneity over the full time dimension,
$\mathbb E[\varepsilon_t\mid Z_1(M),\ldots,Z_T(M)]
=\mathbf 0_{n\times 1}$, the transformed moments remain mean zero.
Under suitable moment, identification, and dependence
conditions for the transformed design, the convergence
rate may remain unchanged, while the asymptotic
variance generally differs from that of the baseline
model. In particular, the limiting GMM covariance
retains the sandwich form with the Jacobian, moment
covariance, and weighting matrix evaluated for the
transformed variables.

\end{remark}

\begin{proof}[Proof of Theorem \ref{Thm:DK1}]
We prove the result for the two-step estimator in
Algorithm~\ref{Algo:CD_LRSED}. Let $\widetilde S$ denote the
preliminary sparse estimator obtained in Step~1, before the diagonal
adjustment, and let $\widehat L$ and $\widehat S$ denote the final
component estimates.
Consider the event
$\mathcal E_n=
\{\|E\|_2\leq w_{2n},\;
  \|E\|_{\max}\leq w_{\max n}\}$,
whose probability approaches one by
Assumption~\ref{Assump:E_noise}. By Assumption~\ref{weight}(i) and the definition of $a_n$,
$\|L_0\|_{\max}\lesssim a_n$.
Hence, for sufficiently large $C_\tau$ and all sufficiently large $n$,
$\tau_n=C_\tau(w_{\max n}+a_n)
\geq2\{\|E\|_{\max}+\|L_0\|_{\max}\}$
on $\mathcal E_n$.

We verify the following conditions for Step~1, using the equivalent
full entrywise-$\ell_1$ formulation with
$\Phi(S)=|S|_{1,1}$, $\theta^*=S_0^*$, and $\lambda_n=\tau_n$.
Step~2 is handled by the singular value decomposition argument below.
Define $|S|_1 = \sum_{i,j\in \{1,\cdots, n\}:i\neq j}|S_{i,j}| $.
Namely:
\begin{itemize}
\item[Step 1] For a given $L=\mathbf{0}_{n\times n}$, we optimize:
\begin{align}
\label{Eq:LRSEDsub1}
\widetilde S
\coloneqq
\arg\min_{S\in\mathbb R^{n\times n}}
\frac12\|W-L-S\|_F^2+\tau_n|S|_1.
\end{align}
Since $\operatorname{diag}(W)=0$, this gives
$\widetilde S_{ii}=0$.
The same minimizer is obtained by replacing $|S|_1$
with $|S|_{1,1}$, since the diagonal minimizer remains zero.
\item[Step 2]
For the estimator in the previous step $\widetilde S$, we optimize:
\begin{align}
\widehat L
\coloneqq
\arg\min_{L\in\mathbb R^{n\times n}}
\frac12\|W-L-\widetilde S\|_F^2+\nu_n\|L\|_*.
\end{align}
Then set $\widehat S_{ij}=\widetilde S_{ij}$ for $i\neq j$ and
$\widehat S_{ii}=-\widehat L_{ii}$, so that
$\widehat W=\widehat L+\widehat S$ has zero diagonal.

{
Then we can apply Lemma~\ref{theorem919}
(Theorem~9.19 in \cite{wainwright2019high}) to Step~1. For Step~2,
we use the singular-value-thresholding representation. The relevant conditions
are verified as follows:
\begin{itemize}

\item[i)]
The loss is quadratic. Therefore,
$
\mathcal E_n(\Delta)
=
\frac12\|\Delta\|_F^2,
$
so the curvature condition holds globally with $\kappa=1$ and zero
tolerance.

\item[ii)]
For Step~1, the entrywise $\ell_1$ norm is decomposable with respect
to the support space of $S_0^*$, and
$
\Psi(\overline{\mathbb M})
\leq
\sqrt{m_s(S_0^*)}.
$

\item[iii)]

The gradient of the Step~1 loss evaluated at $S_0^*$ is
$
\nabla\mathcal L_n(S_0^*)
=
S_0^*-W
=
-\bigl(L_0-\operatorname{diag}(L_0)+E\bigr).
$
Since the dual of the entrywise $\ell_1$ norm is the maximum norm,
$\Phi^*\!\left(\nabla\mathcal L_n(S_0^*)\right)
\leq
\|L_0\|_{\max}+\|E\|_{\max}
\leq \tau_n/2
\quad\text{on }\mathcal E_n$.
Thus, the gradient condition holds for the stated choice of
$\tau_n$.

\end{itemize}
According to {Lemma~\ref{theorem919}  in the Online Appendix in  \citep{li2026estimating}  and the Step~1} Frobenius bound,
we have
$|\widetilde S-S_0^*|_{1,1}
\lesssim_p
(w_{\max n}+a_n)m_s(S_0^*)
=
\mathcal R_s\sqrt{m_s(S_0^*)}$.
The same argument also gives
\begin{equation}
\label{Eq:S_F}
\|\widetilde S-S_0^*\|_F^2
\lesssim_p
m_s(S_0^*)(w_{\max n}+a_n)^2
\lesssim
m_s(S_0^*)w_{\max n}^2+m_s(S_0^*)a_n^2.
\end{equation}
In particular,
$\|\widetilde S-S_0^*\|_2
\leq\|\widetilde S-S_0^*\|_F
\lesssim_p\mathcal R_s$.
Denote $
\operatorname{SVT}_{\nu_n}(A_n)
\coloneqq
U\operatorname{diag}
\left\{
\bigl(\sigma_j(A_n)-\nu_n\bigr)_+
\right\}V^\top$, where
$(x)_+\coloneqq\max\{x,0\}.
$
For Step~2, let
\[
D\coloneqq\widetilde S-S_0^*,
\qquad
\Delta\coloneqq\widehat L-L_0,
\qquad
A_n\coloneqq W-\widetilde S
=L_0+E-D-\operatorname{diag}(L_0).
\]
Write the compact singular-value decomposition of $L_0$ as
$L_0=U\Sigma V^\top$, and define the tangent space
$T
\coloneqq
\left\{
UX^\top+YV^\top:
X,Y\in\mathbb R^{n\times r}
\right\}$.
When $r=0$, set $T=\{0\}$. Let $\mathcal P_T$ and
$\mathcal P_{T^\perp}$ denote the corresponding orthogonal
projections. By optimality of $\widehat L$,
$\frac12\|\Delta\|_F^2
\leq
\langle E-D-\operatorname{diag}(L_0),\Delta\rangle
+
\nu_n
\left(
\|L_0\|_*-\|L_0+\Delta\|_*
\right)$.
On $\mathcal E_n$, the choice of $\nu_n$ satisfying
$\nu_n\geq
2(w_{2n}+\|\operatorname{diag}(L_0)\|_2)$
gives
$\|E-\operatorname{diag}(L_0)\|_2
\leq\nu_n/2$.
Hence, by nuclear--operator norm
duality and Frobenius duality,
$\langle E-D-\operatorname{diag}(L_0),\Delta\rangle
\leq
\frac{\nu_n}{2}\|\Delta\|_*
+
\|D\|_F\|\Delta\|_F$.
Moreover, decomposability of the nuclear norm gives
\[
\|L_0+\Delta\|_*
\geq
\|L_0\|_*
+
\|\mathcal P_{T^\perp}\Delta\|_*
-
\|\mathcal P_T\Delta\|_*.
\]
Thus
\begin{equation}
\|\Delta\|_F^2
+
\nu_n\|\mathcal P_{T^\perp}\Delta\|_*
\leq
3\nu_n\|\mathcal P_T\Delta\|_*
+
2\|D\|_F\|\Delta\|_F.
\label{Eq:Step2Cone}
\end{equation}
}
Since
$\operatorname{rank}(\mathcal P_T\Delta)\leq2r$,
$\|\mathcal P_T\Delta\|_*
\leq
\sqrt{2r}\,\|\Delta\|_F$.
Dropping the nonnegative second term on the left-hand side of
\eqref{Eq:Step2Cone} therefore gives
$\|\Delta\|_F
\leq
3\sqrt{2r}\,\nu_n+2\|D\|_F$.
By \eqref{Eq:S_F},
$\|D\|_F=O_p\{\mathcal R_s(W_0,E)\}$. Consequently,
\begin{align}
\|\widehat L-L_0\|_F^2
&\lesssim_p
\bar r(w_{2n}+a_n)^2+\mathcal R_s(W_0,E)^2 \nonumber\\
&\lesssim_p
\bar r(w_{2n}+a_n)^2
+m_s(S_0^*)w_{\max n}^2
+m_s(S_0^*)a_n^2.
\label{Eq:L_F}
\end{align}

To obtain the nuclear-norm bound, apply
$2\|D\|_F\|\Delta\|_F
\leq
\frac12\|\Delta\|_F^2+2\|D\|_F^2$
in \eqref{Eq:Step2Cone}. This gives the approximate cone condition
$\|\mathcal P_{T^\perp}\Delta\|_*
\leq
3\|\mathcal P_T\Delta\|_*
+
\frac{2\|D\|_F^2}{\nu_n}$.
It follows that
\begin{align*}
\|\Delta\|_*
\leq
4\|\mathcal P_T\Delta\|_*
+
\frac{2\|D\|_F^2}{\nu_n}
\leq
24r\nu_n
+
8\sqrt{2r}\,\|D\|_F
+
\frac{2\|D\|_F^2}{\nu_n}.
\end{align*}
Because
$\nu_n
\geq
C_\nu
\frac{\mathcal R_s(W_0,E)}{\sqrt{\bar r}}$,
we conclude that
\begin{equation}
\|\widehat L-L_0\|_*
\lesssim_p
\bar r(w_{2n}+a_n)
+
\sqrt{\bar r}\,\mathcal R_s(W_0,E)
\lesssim_p
\bar r\rho_n.
\label{Eq:L_nuclear}
\end{equation}

The solution is obtained by singular-value thresholding:
$
\widehat L
=
\operatorname{SVT}_{\nu_n}(A_n)$, $
\|A_n-\widehat L\|_2\leq\nu_n.
$
Moreover,
$
A_n-L_0
=
E-(\widetilde S-S_0^*)-\operatorname{diag}(L_0).
$
Since $W_{0,ii}=0$, we have $S_{0,ii}=-L_{0,ii}$.
After the diagonal adjustment,
$\widehat S-S_0
=
(\widetilde S-S_0^*)
-
\operatorname{diag}(\widehat L-L_0)$.
Hence,
$\|\widehat S-S_0\|_F^2
\leq
\|\widetilde S-S_0^*\|_F^2
+
\|\widehat L-L_0\|_F^2$.
Combining this with \eqref{Eq:S_F} and \eqref{Eq:L_F} yields
$\|\widehat L-L_0\|_F^2
+
\|\widetilde S-S_0^*\|_F^2
\lesssim_p
\mathcal R_F(W_0,E)$,
which proves the Frobenius bound in part~(a).
Therefore,
\begin{align}
\|\widehat L-L_0\|_2
&\leq
\|\widehat L-A_n\|_2+\|A_n-L_0\|_2 \\
&\leq
\nu_n+\|E\|_2+\|\widetilde S-S_0^*\|_2
+\|\operatorname{diag}(L_0)\|_2
\lesssim_p
w_{2n}+a_n+\mathcal R_s.
\end{align}
By the diagonal adjustment,
$\widehat S-S_0
=
(\widetilde S-S_0^*)
-
\operatorname{diag}(\widehat L-L_0)$.
Since
$\|\operatorname{off}(B)\|_2\leq2\|B\|_2$ and
$\|\operatorname{diag}(B)\|_2\leq\|B\|_2$, we obtain
$\|\widehat S-S_0\|_2
\leq
\|\widetilde S-S_0^*\|_2+\|\widehat L-L_0\|_2
\lesssim_p
w_{2n}+a_n+\mathcal R_s$.
Consequently,
$\|\widehat L-L_0\|_2^2+\|\widetilde S-S_0^*\|_2^2
\lesssim_p
(w_{2n}+a_n)^2+\mathcal R_s^2
\lesssim
\mathcal R_2(W_0,E)$.
Finally,
\begin{align*}
|\widehat S-S_0|_{1,1}
&\leq
|\widetilde S-S_0^*|_{1,1}
+
|\operatorname{diag}(\widehat L-L_0)|_{1,1}\\
&\lesssim_p
\mathcal R_s\sqrt{m_s(S_0^*)}
+n(w_{2n}+a_n+\mathcal R_s).
\end{align*}
\end{itemize}
\end{proof}
\vspace{-0.7cm}

\begin{lemma}
\label{rateG}
Define the gradient matrices
\begin{eqnarray*}
\widehat{G}(\widehat W) &:=& \frac{\partial \bar{g}_{nT}(\theta, \widehat W)}{\partial \theta^{\top}}
= -\frac{1}{nT} \sum_{t=1}^T Z^{\top}_t(\widehat W)
\begin{bmatrix} \widehat W Y_t & X_t & \widehat{W} X_t \end{bmatrix},\\
\widehat{G}(W_0) &:=& \frac{\partial \bar{g}_{nT}(\theta, W_0)}{\partial \theta^{\top}}
= -\frac{1}{nT} \sum_{t=1}^T Z^{\top}_t(W_0)
\begin{bmatrix} W_0 Y_t & X_t & W_0 X_t \end{bmatrix},
\end{eqnarray*}
which do not depend on the evaluation point $\theta$.
Suppose Assumptions \ref{weight}--\ref{Assump:Moments_Weight} hold. 
Let $\Delta L := \widehat{L}-L_0$, $\Delta S := \widehat{S}-S_0$, 
$B_0 := (I_n - \lambda_0 W_0)^{-1}$, and 
$\mathcal R_s :=
(w_{\max n}+a_n)\sqrt{m_s(S_0^*)}$.\\
Throughout, $Y_t$, $Z_t$, $X_t$, and $\varepsilon_t$ denote the
original, untransformed quantities, as specified in
Remark~\ref{rem:twoway_demean}.\\
\noindent\textbf{(a) Baseline rate using the noisy matrix $W$.} 
Assume
$
\|W_0\|_2=O(1)$, and $
w_{2n}=O(1).$
Using $W = W_0 + E$ directly,
$\|\widehat{G}(W)-\widehat{G}(W_0)\|_2 
= O_p(w_{2n} + w_{2n}^2)$, 
$\|\bar{g}_{nT}(\theta_0, W)
-\bar{g}_{nT}(\theta_0, W_0)\|_2 
= O_p(w_{2n} + w_{2n}^2)$.

\noindent\textbf{(b) Alternative bound via the $|\cdot|_{1,1}$ norm.}
Suppose %, in addition, that
$
\|W_0\|_1=O(1)$, $
\|E\|_\infty=O_p(1),
$
and
\[
\max_{1\leq j\leq K}
\frac{1}{T}\sum_{t=1}^T
\|X_{jt}\|_{\max}^2
=O_p(1),
\qquad
\frac{1}{T}\sum_{t=1}^T
\|\varepsilon_t\|_{\max}^2
=O_p(1),
\]
where $X_{jt}$ is the $j$th column of $X_t$,
$\|X_{jt}\|_{\max}=\max_{1\leq i\leq n}|x_{t,ij}|$, and
$\|\varepsilon_t\|_{\max}
=\max_{1\leq i\leq n}|\varepsilon_{it}|$.
Then
$\|\widehat G(W)-\widehat G(W_0)\|_2
=
O_p\!\left(n^{-1}|E|_{1,1}\right)$,$ 
\|\bar g_{nT}(\theta_0,W)
-\bar g_{nT}(\theta_0,W_0)\|_2
=
O_p\!\left(n^{-1}|E|_{1,1}\right)$.\\
\noindent\textbf{(c) Improved rate using the denoised matrix
$\widehat W$.}
Let
$\mathcal E:=\sigma(E)$,
$U_t:=\lambda_0Y_t+X_t\gamma_0,
$
and consider the finite collection of primitive pairs
\[
\mathcal P
=
\left\{
(X_{jt},X_{kt}),
(X_{jt},Y_t),
(X_{jt},U_t),
(X_{jt},\varepsilon_t):
1\leq j,k\leq K
\right\}.
\]
Suppose that
\begin{equation}
\sup_{(a_t,b_t)\in\mathcal P}
\sup_{1\leq t\leq T}
\left\|
\mathbb E[a_tb_t^\top\mid\mathcal E]
\right\|_2
=
O_p(1),
\label{eq:conditional-moment-op}
\end{equation}
and
\begin{equation}
\sup_{(a_t,b_t)\in\mathcal P}
\sup_{1\leq t\leq T}
\left\|
\mathbb E[a_tb_t^\top\mid\mathcal E]
\right\|_{\max}
=
O_p(1).
\label{eq:conditional-moment-max}
\end{equation}
Assume also that, uniformly over the pairs in $\mathcal P$ and over
all $\mathcal E$-measurable matrices $A$,
\begin{equation}
\operatorname{Var}\left(
\sum_{t=1}^T a_t^\top A b_t
\,\middle|\,\mathcal E
\right)
\lesssim_p
T\|A\|_F^2.
\label{eq:conditional-bilinear}
\end{equation}

Independence between $E$ and the regressors or disturbances is not
required; the conditional-moment conditions above are the operative
restriction. Under independence, the two conditional-moment bounds
reduce to their corresponding unconditional counterparts in the main
assumptions, whereas under dependence they restrict the strength of
that dependence.

Recall
$
\rho_n
:=
w_{2n}+\mathcal R_s+ a_n$, $
s_n
:=
\mathcal R_s\sqrt{m_s(S_0^*)}$, $
\kappa_n
:=
1+\|W_0\|_2+\rho_n,
$
and
\begin{equation}
\mathcal R_{nT}^*
\coloneqq
(1+\|B_0\|_2)\kappa_n^{d+1}
\frac{(r\vee1)\rho_n+s_n}{n}.
%\label{eq:RstarnT}
\end{equation}
If $\mathcal R_{nT}^*=o(1)$, then
\begin{align*}
\|\widehat G(\widehat W)-\widehat G(W_0)\|_2
&=
O_p(\mathcal R_{nT}^*),\qquad
\|\bar g_{nT}(\theta_0,\widehat W)
-\bar g_{nT}(\theta_0,W_0)\|_2
=
O_p(\mathcal R_{nT}^*).
\end{align*}
\end{lemma}
Part~(b) controls network-induced perturbations through the
aggregate error magnitude $|E|_{1,1}/n$, reflecting the type
of restriction used by \citet{Lewbeletal2024EJ}. For this
bound to vanish, the total absolute measurement error must
be $o_p(n)$.
Part~(c) provides an alternative route to vanishing
perturbations by exploiting the low-rank plus sparse
structure of the latent network. Under its stated
assumptions, it requires $\mathcal R_{nT}^*=o(1)$,
without requiring $|E|_{1,1}=o_p(n)$. Consequently,
it accommodates dense-error regimes in which
$\mathcal R_{nT}^*\to0$ while $|E|_{1,1}/n$ does not
converge to zero. In these regimes, our recovery-based
bound vanishes even though the aggregate-error bound
in part~(b) does not. This is the principal advantage
of exploiting structure in the latent adjacency matrix.
There can also be a rate improvement over the
spectral-norm baseline: when the assumptions of both
parts~(a) and~(c) hold, part~(c) provides a smaller-order
upper bound whenever
$\mathcal R_{nT}^*=o(w_{2n}+w_{2n}^2)$.
\begin{proof}[Proof of Lemma~\ref{rateG}]
\noindent\textbf{Part (a): $\|\cdot\|_2$ based bound.}
Throughout this part, $d=2$, so the columns of $Z_t(W)$ are chosen
from $W^lX_{jt}$, where $l=0,1,2$ and $j=1,\ldots,K$. Since $K$
and the number of instruments are fixed, it suffices to bound each
entry of the relevant matrices.

Recall that
$
B_0=(I_n-\lambda_0W_0)^{-1}.
$
The stability condition on $W_0$ implies
\[
\|W_0\|_\infty=O(1),\qquad
\|B_0\|_\infty=O(1),\qquad
\|B_0W_0\|_\infty=O(1).
\]
Together with the uniform moment bounds in
Assumptions~\ref{Assump:Spatial_Errors} and \ref{Assump:Instruments}, %, together with these row-sum bounds,
it implies
\begin{align}
& \max_{1\leq j\leq K}
\frac{1}{nT}\sum_{t=1}^T\|X_{jt}\|_2^2
=O_p(1), \label{eq:rateG-energy-X}\\
& \frac{1}{nT}\sum_{t=1}^T\|\varepsilon_t\|_2^2
=O_p(1), \label{eq:rateG-energy-eps}\\
& \max_{\substack{1\leq j\leq K\\H\in\{I_n,W_0\}}}
\frac{1}{nT}\sum_{t=1}^T\|B_0H X_{jt}\|_2^2
+\frac{1}{nT}\sum_{t=1}^T\|B_0\varepsilon_t\|_2^2
=O_p(1). \label{eq:rateG-energy-B}
\end{align}
Indeed, uniformly in $t$, $j$, and $H\in\{I_n,W_0\}$,
the uniform componentwise second-moment bounds and the bounded row
sums of $B_0H$ imply
$
\mathbb E\|B_0H X_{jt}\|_2^2=O(n).
$
Moreover, Assumption~\ref{Assump:Spatial_Errors} implies
$\mathbb E[\varepsilon_t\varepsilon_t^\top]
=\sigma_0^2 I_n$,
and hence
$\lambda_{\max}\!\left(
\mathbb E[\varepsilon_t\varepsilon_t^\top]
\right)=\sigma_0^2$.
Consequently,
$
\mathbb E\|B_0\varepsilon_t\|_2^2
=
\operatorname{tr}\!\left(
B_0^\top B_0
\mathbb E[\varepsilon_t\varepsilon_t^\top]
\right)
\leq
\sigma_0^2\|B_0\|_F^2
\leq
\sigma_0^2n\|B_0\|_\infty^2
=O(n),
$
where
$
\|B_0\|_F^2
=\sum_{i=1}^n\|B_{0,i\cdot}\|_2^2
\leq
\sum_{i=1}^n\|B_{0,i\cdot}\|_1^2
\leq
n\|B_0\|_\infty^2.
$
Therefore, the expectations of the normalized time averages are
uniformly bounded, and Markov's inequality gives
\eqref{eq:rateG-energy-X}--\eqref{eq:rateG-energy-B}.
Since
$
Y_t=B_0(X_t\beta_0+W_0X_t\gamma_0+\varepsilon_t),
$
the compactness of the parameter space, the fact that $K$ is fixed,
and \eqref{eq:rateG-energy-B} implies
\begin{equation}
\frac{1}{nT}\sum_{t=1}^T\|Y_t\|_2^2=O_p(1).
\label{eq:rateG-energy-Y}
\end{equation}
Consequently, defining
$
U_t:=\lambda_0Y_t+X_t\gamma_0,
$
we also have
\begin{equation}
\frac{1}{nT}\sum_{t=1}^T\|U_t\|_2^2=O_p(1).
\label{eq:rateG-energy-U}
\end{equation}
Define
$
A_l:=W^l-W_0^l, l=0,1,2.
$
Since $W=W_0+E$,
$A_0=0, A_1=E,
$
and
$
A_2=W^2-W_0^2=W_0E+EW_0+E^2.
$
Therefore, using $\|W_0\|_2=O(1)$ and
$\|E\|_2=O_p(w_{2n})$,
\begin{equation}
\max_{0\leq l\leq2}\|A_l\|_2
=O_p(w_{2n}+w_{2n}^2).
\label{eq:rateG-power-difference}
\end{equation}
Moreover,
$
\|W\|_2
\leq\|W_0\|_2+\|E\|_2
=O_p(1),
$ where the final equality uses $w_{2n}=O(1)$.
We first consider the gradient. By its definition,
\begin{align*}
\|\widehat G(W)-\widehat G(W_0)\|_2
&\lesssim
\left\|
\frac{1}{nT}\sum_{t=1}^T
\big[Z_t^\top(W)-Z_t^\top(W_0)\big]X_t
\right\|_2+
\left\|
\frac{1}{nT}\sum_{t=1}^T
\big[Z_t^\top(W)W-Z_t^\top(W_0)W_0\big]Y_t
\right\|_2\\
&\quad+
\left\|
\frac{1}{nT}\sum_{t=1}^T
\big[Z_t^\top(W)W-Z_t^\top(W_0)W_0\big]X_t
\right\|_2\\
&=:\|\mathcal R_{n1}\|_2
+\|\mathcal R_{n2}\|_2
+\|\mathcal R_{n3}\|_2.
\end{align*}

For $\mathcal R_{n1}$, every $(j,k,l)$ entry satisfies
\begin{align*}
&\left|
\frac{1}{nT}\sum_{t=1}^T
X_{jt}^\top A_l^\top X_{kt}
\right|\leq
\|A_l\|_2
\left(
\frac{1}{nT}\sum_{t=1}^T\|X_{jt}\|_2^2
\right)^{1/2}
\left(
\frac{1}{nT}\sum_{t=1}^T\|X_{kt}\|_2^2
\right)^{1/2}.
\end{align*}
It follows from \eqref{eq:rateG-energy-X} and
\eqref{eq:rateG-power-difference} that
\begin{equation}
\|\mathcal R_{n1}\|_2
=O_p(w_{2n}+w_{2n}^2).
\label{eq:rateG-Rn1}
\end{equation}

For $\mathcal R_{n2}$ and $\mathcal R_{n3}$, define
$D_l:=(W^l)^\top W-(W_0^l)^\top W_0.
$
Because
$
D_l=A_l^\top W+(W_0^l)^\top E,
$
we have
$
\|D_l\|_2
\leq
\|A_l\|_2\|W\|_2+\|W_0\|_2^l\|E\|_2.
$
Thus,
\begin{equation}
\max_{0\leq l\leq2}\|D_l\|_2
=O_p(w_{2n}+w_{2n}^2).
\label{eq:rateG-Dl}
\end{equation}

Each $(j,l)$ entry of $\mathcal R_{n2}$ satisfies
\begin{align*}
\left|
\frac{1}{nT}\sum_{t=1}^T
X_{jt}^\top D_lY_t
\right|
&\leq
\|D_l\|_2
\left(
\frac{1}{nT}\sum_{t=1}^T\|X_{jt}\|_2^2
\right)^{1/2}
\times
\left(
\frac{1}{nT}\sum_{t=1}^T\|Y_t\|_2^2
\right)^{1/2}.
\end{align*}
Therefore, by \eqref{eq:rateG-energy-X},
\eqref{eq:rateG-energy-Y}, and \eqref{eq:rateG-Dl},
\begin{equation}
\|\mathcal R_{n2}\|_2
=O_p(w_{2n}+w_{2n}^2).
\label{eq:rateG-Rn2}
\end{equation}

Similarly, every $(j,k,l)$ entry of $\mathcal R_{n3}$ satisfies
\begin{align*}
\left|
\frac{1}{nT}\sum_{t=1}^T
X_{jt}^\top D_lX_{kt}
\right|
&\leq
\|D_l\|_2
\left(
\frac{1}{nT}\sum_{t=1}^T\|X_{jt}\|_2^2
\right)^{1/2}\times
\left(
\frac{1}{nT}\sum_{t=1}^T\|X_{kt}\|_2^2
\right)^{1/2}.
\end{align*}
Hence
\begin{equation}
\|\mathcal R_{n3}\|_2
=O_p(w_{2n}+w_{2n}^2).
\label{eq:rateG-Rn3}
\end{equation}

Combining \eqref{eq:rateG-Rn1}--\eqref{eq:rateG-Rn3}, we obtain
$
\|\widehat G(W)-\widehat G(W_0)\|_2
=O_p(w_{2n}+w_{2n}^2).
$

We next consider the sample moment. Define
$
\varepsilon_t(\theta_0,M)
:=Y_t-\lambda_0MY_t-X_t\beta_0-MX_t\gamma_0.
$
Then
$
\varepsilon_t(\theta_0,W_0)=\varepsilon_t
$
and
$
\delta_t
:=\varepsilon_t(\theta_0,W)-\varepsilon_t
=-E(\lambda_0Y_t+X_t\gamma_0)
=-EU_t.
$
Consequently,
\begin{align*}
&\bar g_{nT}(\theta_0,W)
-\bar g_{nT}(\theta_0,W_0)=
\frac{1}{nT}\sum_{t=1}^T
Z_t^\top(W_0)\delta_t
+
\frac{1}{nT}\sum_{t=1}^T
\big[Z_t^\top(W)-Z_t^\top(W_0)\big]
\varepsilon_t(\theta_0,W)\\
&\quad=:\mathcal R_1+\mathcal R_2.
\end{align*}

For each instrument component of $\mathcal R_1$,
\begin{align*}
\left|
\frac{1}{nT}\sum_{t=1}^T
X_{jt}^\top(W_0^l)^\top EU_t
\right|
&\leq
\|E\|_2
\left(
\frac{1}{nT}\sum_{t=1}^T
\|W_0^lX_{jt}\|_2^2
\right)^{1/2}\times
\left(
\frac{1}{nT}\sum_{t=1}^T\|U_t\|_2^2
\right)^{1/2}.
\end{align*}
Because $\|W_0\|_2=O(1)$, \eqref{eq:rateG-energy-X} and
\eqref{eq:rateG-energy-U} imply
$
\|\mathcal R_1\|_2=O_p(w_{2n}).
$

For $\mathcal R_2$, use
$
\varepsilon_t(\theta_0,W)=\varepsilon_t+\delta_t
$
to write
\begin{align*}
\mathcal R_2
&=
\underbrace{
\frac{1}{nT}\sum_{t=1}^T
\big[Z_t^\top(W)-Z_t^\top(W_0)\big]\varepsilon_t
}_{\mathcal R_{21}}+
\underbrace{
\frac{1}{nT}\sum_{t=1}^T
\big[Z_t^\top(W)-Z_t^\top(W_0)\big]\delta_t
}_{\mathcal R_{22}}.
\end{align*}

For every $(j,l)$ component of $\mathcal R_{21}$,
\begin{align*}
\left|
\frac{1}{nT}\sum_{t=1}^T
X_{jt}^\top A_l^\top\varepsilon_t
\right|
&\leq
\|A_l\|_2
\left(
\frac{1}{nT}\sum_{t=1}^T\|X_{jt}\|_2^2
\right)^{1/2}\times
\left(
\frac{1}{nT}\sum_{t=1}^T\|\varepsilon_t\|_2^2
\right)^{1/2}.
\end{align*}
Thus,
$
\|\mathcal R_{21}\|_2
=O_p(w_{2n}+w_{2n}^2).
$ Finally, since $\delta_t=-EU_t$, every component of
$\mathcal R_{22}$ satisfies
\begin{align*}
\left|
\frac{1}{nT}\sum_{t=1}^T
X_{jt}^\top A_l^\top EU_t
\right|
&\leq
\|A_l\|_2\|E\|_2
\left(
\frac{1}{nT}\sum_{t=1}^T\|X_{jt}\|_2^2
\right)^{1/2}\times
\left(
\frac{1}{nT}\sum_{t=1}^T\|U_t\|_2^2
\right)^{1/2}.
\end{align*}
It follows that
$
\|\mathcal R_{22}\|_2
=O_p(w_{2n}^2+w_{2n}^3)
=O_p(w_{2n}^2),
$
where the final equality uses $w_{2n}=O(1)$.
Combining the preceding bounds gives
$
\|\bar g_{nT}(\theta_0,W)
-\bar g_{nT}(\theta_0,W_0)\|_2
=O_p(w_{2n}+w_{2n}^2),$
which completes the proof of part~(a). 

%\medskip
\noindent\textbf{Part (b): $|\cdot|_{1,1}$-based bound.}\\
Throughout this part, $d=2$. We use the inequalities
\begin{equation}
|AB|_{1,1}
\leq
\|A\|_1|B|_{1,1},
\qquad
|AB|_{1,1}
\leq
\|B\|_\infty|A|_{1,1}.
\label{eq:l11-product}
\end{equation}
For vectors $x,y\in\mathbb R^n$, we also use
\begin{equation}
|x^\top Ay|
\leq
\|x\|_{\max}|A|_{1,1}\|y\|_{\max}.
\label{eq:l11-bilinear}
\end{equation}
Recall the definitions
$
A_l:=W^l-W_0^l,
D_l:=(W^l)^\top W-(W_0^l)^\top W_0$, for 
$l=0,1,2.$ Note that $W=W_0+E$,
$
A_0=0$, $
A_1=E,
$
and
$
A_2=W_0E+EW_0+E^2.
$
Using~\eqref{eq:l11-product},
\begin{align*}
|A_2|_{1,1}
&\leq
|W_0E|_{1,1}
+|EW_0|_{1,1}
+|E^2|_{1,1}\leq
\left(
\|W_0\|_1
+\|W_0\|_\infty
+\|E\|_\infty
\right)|E|_{1,1}.
\end{align*}
By the additional hypothesis of part~(b),
$\|W_0\|_1=O(1)$. Moreover,
Assumption~\ref{Assump:True_Adjacency}(iii) gives
$\|W_0\|_\infty=O(1)$, while
$\|E\|_\infty=O_p(1)$ by assumption. Therefore,
\begin{equation}
\max_{0\leq l\leq2}|A_l|_{1,1}
=
O_p(|E|_{1,1}).
\label{eq:Al-l11}
\end{equation}
As $
D_l=A_l^\top W+(W_0^l)^\top E,
$ 
%\begin{align*}
$|D_l|_{1,1}
%&\leq
%|A_l^\top W|_{1,1}
%+|(W_0^l)^\top E|_{1,1}
\leq
\|W\|_\infty|A_l^\top|_{1,1}
+\|(W_0^l)^\top\|_1|E|_{1,1}=
\|W\|_\infty|A_l|_{1,1}
+\|W_0^l\|_\infty|E|_{1,1}.$
%\end{align*}
Since
$
\|W\|_\infty
\leq
\|W_0\|_\infty+\|E\|_\infty
=
O_p(1)
$
and
$
\|W_0^l\|_\infty
\leq
\|W_0\|_\infty^l
=
O(1),
$
it follows that
\begin{equation}
\max_{0\leq l\leq2}|D_l|_{1,1}
=
O_p(|E|_{1,1}).
\label{eq:Dl-l11}
\end{equation}
We next record bounds for the data factors. Since
$
Y_t
=
B_0
\left(
X_t\beta_0+W_0X_t\gamma_0+\varepsilon_t
\right),
$
Assumption~\ref{Assump:True_Adjacency}(iii) gives
$
\|B_0\|_\infty=O(1)$, and $
\|W_0\|_\infty=O(1).
$
Together with the boundedness of $\beta_0$ and $\gamma_0$ and the
fact that $K$ is fixed, this gives
$
\|Y_t\|_{\max}
\lesssim
\sum_{k=1}^K\|X_{kt}\|_{\max}
+\|\varepsilon_t\|_{\max}.
$ By the additional maximum-norm assumptions in part~(b),
\begin{equation}
\frac1T\sum_{t=1}^T\|Y_t\|_{\max}^2
=
O_p(1).
\label{eq:Y-max-energy}
\end{equation}
Similarly, defining
$
U_t:=\lambda_0Y_t+X_t\gamma_0,
$
we have
\begin{equation}
\frac1T\sum_{t=1}^T\|U_t\|_{\max}^2
=
O_p(1).
\label{eq:U-max-energy}
\end{equation}
With $\mathcal R_{n1}$, $\mathcal R_{n2}$, and
$\mathcal R_{n3}$ defined as in part~(a), we first bound the
gradient difference. For $\mathcal R_{n1}$, every $(j,k,l)$ entry
satisfies
\begin{align*}
&\left|
\frac1{nT}\sum_{t=1}^T
X_{jt}^\top A_l^\top X_{kt}
\right|\leq
\frac{|A_l|_{1,1}}{n}
\left(
\frac1T\sum_{t=1}^T\|X_{jt}\|_{\max}^2
\right)^{1/2}
\left(
\frac1T\sum_{t=1}^T\|X_{kt}\|_{\max}^2
\right)^{1/2}.
\end{align*}
Thus, by~\eqref{eq:Al-l11} and the additional maximum-norm
assumptions,
$
\|\mathcal R_{n1}\|_2
=
O_p\left(n^{-1}|E|_{1,1}\right).
$
For $\mathcal R_{n2}$, every $(j,l)$ entry satisfies
%\begin{align*}
$$\left|
\frac1{nT}\sum_{t=1}^T
X_{jt}^\top D_lY_t
\right|
\leq
\frac{|D_l|_{1,1}}{n}
\left(
\frac1T\sum_{t=1}^T\|X_{jt}\|_{\max}^2
\right)^{1/2}\times
\left(
\frac1T\sum_{t=1}^T\|Y_t\|_{\max}^2
\right)^{1/2}.$$
%\end{align*}
Therefore, by~\eqref{eq:Dl-l11},
\eqref{eq:Y-max-energy}, and the additional maximum-norm
assumptions,
$
\|\mathcal R_{n2}\|_2
=
O_p\left(n^{-1}|E|_{1,1}\right).
$
Similarly, every $(j,k,l)$ entry of $\mathcal R_{n3}$ satisfies
\begin{align*}
\left|
\frac1{nT}\sum_{t=1}^T
X_{jt}^\top D_lX_{kt}
\right|
&\leq
\frac{|D_l|_{1,1}}{n}
\left(
\frac1T\sum_{t=1}^T\|X_{jt}\|_{\max}^2
\right)^{1/2}\times
\left(
\frac1T\sum_{t=1}^T\|X_{kt}\|_{\max}^2
\right)^{1/2}.
\end{align*}
Hence,
$\|\mathcal R_{n3}\|_2
=
O_p\left(n^{-1}|E|_{1,1}\right).
$
Since $K$, $d$, and the number of instruments are fixed, the
preceding entrywise bounds also hold for the corresponding matrix
$2$-norms. Combining the three terms gives
\[
\|\widehat G(W)-\widehat G(W_0)\|_2
=
O_p\left(n^{-1}|E|_{1,1}\right).
\]
We next consider the sample-moment difference. Recall that
$
\delta_t
:=
\varepsilon_t(\theta_0,W)-\varepsilon_t
=
-EU_t.
$
As in part~(a), decompose
\begin{align*}
&\bar g_{nT}(\theta_0,W)
-\bar g_{nT}(\theta_0,W_0)=
\underbrace{
\frac1{nT}\sum_{t=1}^T
Z_t^\top(W_0)\delta_t
}_{\mathcal R_1}
+
\underbrace{
\frac1{nT}\sum_{t=1}^T
\left[
Z_t^\top(W)-Z_t^\top(W_0)
\right]
\varepsilon_t(\theta_0,W)
}_{\mathcal R_2}.
\end{align*}
For each instrument component of $\mathcal R_1$,
\begin{align*}
&\left|
\frac1{nT}\sum_{t=1}^T
X_{jt}^\top(W_0^l)^\top EU_t
\right|\leq
\frac{|(W_0^l)^\top E|_{1,1}}{n}
\left(
\frac1T\sum_{t=1}^T\|X_{jt}\|_{\max}^2
\right)^{1/2}
\left(
\frac1T\sum_{t=1}^T\|U_t\|_{\max}^2
\right)^{1/2}.
\end{align*}
Furthermore,
$
|(W_0^l)^\top E|_{1,1}
\leq
\|(W_0^l)^\top\|_1|E|_{1,1}
=
\|W_0^l\|_\infty|E|_{1,1}
=
O(|E|_{1,1}).
$
Hence,
$
\|\mathcal R_1\|_2
=
O_p\left(n^{-1}|E|_{1,1}\right).
$
For $\mathcal R_2$, write
$
\mathcal R_2=\mathcal R_{21}+\mathcal R_{22},
$
where
%\begin{align*}
$\mathcal R_{21}
:=
\frac1{nT}\sum_{t=1}^T
\left[
Z_t^\top(W)-Z_t^\top(W_0)
\right]\varepsilon_t,$ 
$\mathcal R_{22}
:=
\frac1{nT}\sum_{t=1}^T
\left[
Z_t^\top(W)-Z_t^\top(W_0)
\right]\delta_t.$ 
%\end{align*}
For each component of $\mathcal R_{21}$,
\begin{align*}
\left|
\frac1{nT}\sum_{t=1}^T
X_{jt}^\top A_l^\top\varepsilon_t
\right|
&\leq
\frac{|A_l|_{1,1}}{n}
\left(
\frac1T\sum_{t=1}^T\|X_{jt}\|_{\max}^2
\right)^{1/2}\times
\left(
\frac1T\sum_{t=1}^T\|\varepsilon_t\|_{\max}^2
\right)^{1/2}.
\end{align*}
Thus,
$
\|\mathcal R_{21}\|_2
=
O_p\left(n^{-1}|E|_{1,1}\right).
$
Finally, because $\delta_t=-EU_t$, every component of
$\mathcal R_{22}$ satisfies
\begin{align*}
&\left|
\frac1{nT}\sum_{t=1}^T
X_{jt}^\top A_l^\top EU_t
\right|\leq
\frac{|A_l^\top E|_{1,1}}{n}
\left(
\frac1T\sum_{t=1}^T\|X_{jt}\|_{\max}^2
\right)^{1/2}
\left(
\frac1T\sum_{t=1}^T\|U_t\|_{\max}^2
\right)^{1/2}.
\end{align*}
Using~\eqref{eq:l11-product} and~\eqref{eq:Al-l11},
$
|A_l^\top E|_{1,1}
\leq
\|E\|_\infty|A_l^\top|_{1,1}
=
\|E\|_\infty|A_l|_{1,1}
=
O_p(|E|_{1,1}).
$
Therefore,
$
\|\mathcal R_{22}\|_2
=
O_p\left(n^{-1}|E|_{1,1}\right).
$ Combining the preceding bounds yields
\[
\|\bar g_{nT}(\theta_0,W)
-\bar g_{nT}(\theta_0,W_0)\|_2
=
O_p\left(n^{-1}|E|_{1,1}\right),
\]
which completes the proof of part~(b).\\
\noindent\textbf{Proof of part~(c).}
Define
$\widehat W
:=\operatorname{off}(\widehat L+\widehat S)$
and $\Delta_W:=\widehat W-W_0$, where
$\operatorname{off}(A)=A-\operatorname{diag}(A)$.
Let
$\Delta L:=\widehat L-L_0$ and
$\Delta S:=\widehat S-S_0$.
The off-diagonal operator is redundant under the
zero-diagonal constraint, but we retain it to make
the diagonal bookkeeping explicit. Since
$\operatorname{diag}(W_0)=0$,
$\Delta_W
=\operatorname{off}(\Delta L)
+\operatorname{off}(\Delta S)$.
Recall that
$\rho_n:=w_{2n}+a_n+\mathcal R_s$,
$s_n:=\mathcal R_s\sqrt{m_s(S_0^*)}$,
and
$\kappa_n:=1+\|W_0\|_2+\rho_n$.
Theorem~\ref{Thm:DK1} and its cone condition imply
\begin{align}
\|\Delta L\|_*
&=
O_p\bigl((r\vee1)\rho_n\bigr),
\label{eq:DL-nuclear}\\
\|\Delta L\|_F+\|\Delta S\|_F
&=
O_p\bigl(\sqrt{r\vee1}\,\rho_n\bigr),
\label{eq:DW-frobenius}\\
|\operatorname{off}(\Delta S)|_{1,1}
&=
O_p(s_n),
\label{eq:DS-off-l1}\\
\|\Delta_W\|_2
&=
O_p(\rho_n).
\label{eq:DW-operator}
\end{align}
Consequently,
$\|\widehat W\|_2
\leq\|W_0\|_2+\|\Delta_W\|_2
=O_p(\kappa_n)$.
Let $\mathcal E=\sigma(E)$. Since $\widehat W$ is constructed from
$W=W_0+E$, the matrices $\widehat W$, $\Delta_W$, $\Delta L$, and
$\Delta S$ are $\mathcal E$-measurable. We first record the bilinear-form bound used below. Consider
$
\mathcal B_{nT}(A)
=
\frac1{nT}\sum_{t=1}^T a_t^\top A b_t,
$
where $(a_t,b_t)$ is one of
$(X_{jt},X_{kt})$, $(X_{jt},Y_t)$,
$(X_{jt},U_t)$, $
(X_{jt},\varepsilon_t)$, 
and $A$ is $\mathcal E$-measurable.
Suppose that
$
A
=
P_E\operatorname{off}(\Delta L)Q_E,
$
where
$
\|P_E\|_2\|Q_E\|_2
=
O_p\left(
(1+\|B_0\|_2)\kappa_n^{d+1}
\right).
$
Define
$
M_E
=
\frac1T\sum_{t=1}^T
\mathbb E[a_tb_t^\top\mid\mathcal E].
$
The conditional-moment condition gives
$
\|M_E\|_2=O_p(1).
$
Since the off-diagonal operator is self-adjoint under the Frobenius
inner product,
\begin{align*}
\left|
\mathbb E[\mathcal B_{nT}(A)\mid\mathcal E]
\right|
&=
\frac1n
\left|
\left\langle
\operatorname{off}(\Delta L),
P_E^\top M_EQ_E^\top
\right\rangle_F
\right|=
\frac1n
\left|
\left\langle
\Delta L,
\operatorname{off}
\left(P_E^\top M_EQ_E^\top\right)
\right\rangle_F
\right|.
\end{align*}
For any matrix $N$,
$
\|\operatorname{off}(N)\|_2
\leq
\|N\|_2+\|\operatorname{diag}(N)\|_2
\leq
2\|N\|_2.
$
Therefore, nuclear--operator norm duality and
\eqref{eq:DL-nuclear} give
\begin{align}
\left|
\mathbb E[\mathcal B_{nT}(A)\mid\mathcal E]
\right|
&\leq
\frac{2}{n}
\|\Delta L\|_*
\|P_E\|_2\|M_E\|_2\|Q_E\|_2 
=O_p\left(
(1+\|B_0\|_2)\kappa_n^{d+1}
\frac{(r\vee1)\rho_n}{n}
\right).
\label{eq:bilinear-mean-L}
\end{align}
Moreover,
\begin{align*}
\|A\|_F
&\leq
\|P_E\|_2
\|\operatorname{off}(\Delta L)\|_F
\|Q_E\|_2\leq
\|P_E\|_2\|\Delta L\|_F\|Q_E\|_2=
O_p\left(
(1+\|B_0\|_2)\kappa_n^{d+1}
\sqrt{r\vee1}\,\rho_n
\right).
\end{align*}
The conditional variance condition and conditional Chebyshev's
inequality yield
\begin{align}
\mathcal B_{nT}(A)
-
\mathbb E[\mathcal B_{nT}(A)\mid\mathcal E]
&=
O_p\left(
\frac{\|A\|_F}{n\sqrt T}
\right)=
O_p\left(
(1+\|B_0\|_2)\kappa_n^{d+1}
\frac{\sqrt{r\vee1}\,\rho_n}{n\sqrt T}
\right).
\label{eq:bilinear-fluctuation-L}
\end{align}
If instead
$
A
=
P_E\operatorname{off}(\Delta S)Q_E,
$
then entrywise $\ell_1$--max norm duality gives
\begin{align*}
\left|
\mathbb E[\mathcal B_{nT}(A)\mid\mathcal E]
\right|
&=
\frac1n
\left|
\left\langle
\operatorname{off}(\Delta S),
P_E^\top M_EQ_E^\top
\right\rangle_F
\right|\leq
\frac1n
|\operatorname{off}(\Delta S)|_{1,1}
\|P_E^\top M_EQ_E^\top\|_{\max}.
\end{align*}
Since
$
\|P_E^\top M_EQ_E^\top\|_{\max}
\leq
\|P_E^\top M_EQ_E^\top\|_2
\leq
\|P_E\|_2\|M_E\|_2\|Q_E\|_2,
$
we obtain
\begin{equation}
\left|
\mathbb E[\mathcal B_{nT}(A)\mid\mathcal E]
\right|
=
O_p\left(
(1+\|B_0\|_2)\kappa_n^{d+1}
\frac{s_n}{n}
\right).
\label{eq:bilinear-mean-S}
\end{equation}
Similarly,
%\begin{align*}
$\|A\|_F
\leq
\|P_E\|_2
\|\operatorname{off}(\Delta S)\|_F
\|Q_E\|_2\leq
\|P_E\|_2\|\Delta S\|_F\|Q_E\|_2,$
%\end{align*}
and hence
\begin{equation}
\mathcal B_{nT}(A)
-
\mathbb E[\mathcal B_{nT}(A)\mid\mathcal E]
=
O_p\left(
(1+\|B_0\|_2)\kappa_n^{d+1}
\frac{\sqrt{r\vee1}\,\rho_n}{n\sqrt T}
\right).
\label{eq:bilinear-fluctuation-S}
\end{equation}
Because $r\vee1\geq1$ and $T\geq1$,
we have $
\frac{\sqrt{r\vee1}\,\rho_n}{n\sqrt T}
\leq
\frac{(r\vee1)\rho_n}{n}.
$
Thus every bilinear form above is
\begin{equation}
O_p\left(
(1+\|B_0\|_2)\kappa_n^{d+1}
\frac{(r\vee1)\rho_n+s_n}{n}
\right)
=
O_p(\mathcal R_{nT}^*).
\label{eq:generic-denoised-bilinear}
\end{equation}
Left factors such as $W_0^lX_{jt}$ or
$\widehat W^aX_{jt}$ are handled by absorbing the corresponding
$\mathcal E$-measurable network matrices into $P_E$ or $Q_E$.
This explains why the conditional-moment conditions are stated using
the primitive vectors $X_{jt}$.\\
As in part~(a), write
\begin{align*}
&\quad \bar g_{nT}(\theta_0,\widehat W)
-\bar g_{nT}(\theta_0,W_0)\\
&=
\frac1{nT}\sum_{t=1}^T
Z_t(W_0)^\top
\left[
\varepsilon_t(\theta_0,\widehat W)
-\varepsilon_t(\theta_0,W_0)
\right]+
\frac1{nT}\sum_{t=1}^T
\left[
Z_t(\widehat W)-Z_t(W_0)
\right]^\top
\varepsilon_t(\theta_0,\widehat W)=:
\mathcal R_1+\mathcal R_2.
\end{align*}Let
$
F_t
:=
B_0(X_t\beta_0+W_0X_t\gamma_0).
$
Since
$
Y_t=F_t+B_0\varepsilon_t
$
and
$
U_t=\lambda_0Y_t+X_t\gamma_0,
$
we have
\begin{align*}
\mathcal R_1
&=
-\lambda_0\frac1{nT}\sum_{t=1}^T
Z_t(W_0)^\top\Delta_WB_0\varepsilon_t%\\
%&\quad
-\lambda_0\frac1{nT}\sum_{t=1}^T
Z_t(W_0)^\top\Delta_WF_t%\\
%&\quad
-\frac1{nT}\sum_{t=1}^T
Z_t(W_0)^\top\Delta_WX_t\gamma_0\\
&=:
A_{1nT}+A_{2nT}+A_{3nT}.
\end{align*}
Every coordinate of $A_{1nT}$ is a finite sum of terms involving the
pair $(X_{jt},\varepsilon_t)$, with $B_0$ and the powers of $W_0$
absorbed into $P_E$ and $Q_E$. Hence,
$
\|A_{1nT}\|_2
=
O_p(\mathcal R_{nT}^*).
$ After writing
$
F_t
=
B_0X_t\beta_0+B_0W_0X_t\gamma_0,
$
every coordinate of $A_{2nT}$ involves a pair
$(X_{jt},X_{kt})$. Therefore,
$
\|A_{2nT}\|_2
=
O_p(\mathcal R_{nT}^*).
$
The same argument using $(X_{jt},X_{kt})$ gives
$
\|A_{3nT}\|_2
=
O_p(\mathcal R_{nT}^*).
$
Consequently,
\begin{equation}
\|\mathcal R_1\|_2
=
O_p(\mathcal R_{nT}^*).
\label{eq:R1-denoised}
\end{equation}
For $\mathcal R_2$, write
\begin{align*}
\mathcal R_2
&=
\frac1{nT}\sum_{t=1}^T
\left[
Z_t(\widehat W)-Z_t(W_0)
\right]^\top\varepsilon_t+
\frac1{nT}\sum_{t=1}^T
\left[
Z_t(\widehat W)-Z_t(W_0)
\right]^\top
\left[
\varepsilon_t(\theta_0,\widehat W)-\varepsilon_t
\right]\\
&=:
\mathcal R_{21}+\mathcal R_{22}.
\end{align*}
For every $l\leq d$,
\begin{equation}
\widehat W^l-W_0^l
=
\sum_{a=0}^{l-1}
\widehat W^a\Delta_WW_0^{l-1-a}.
\label{eq:power-telescope}
\end{equation}
Thus each coordinate of $\mathcal R_{21}$ is a finite sum of
bilinear forms involving $(X_{jt},\varepsilon_t)$ and at least one
factor $\Delta_W$. It follows that
$
\|\mathcal R_{21}\|_2
=
O_p(\mathcal R_{nT}^*).
$
Moreover,
$
\varepsilon_t(\theta_0,\widehat W)-\varepsilon_t
=
-\Delta_WU_t.
$
Thus each coordinate of $\mathcal R_{22}$ involves the pair
$(X_{jt},U_t)$ and at least two recovery-error factors. One factor
is paired with the data bilinear form, while every remaining factor
is bounded using
$
\|\Delta_W\|_2=O_p(\rho_n).
$
These additional factors are absorbed into $\kappa_n^{d+1}$.
Therefore,
$
\|\mathcal R_{22}\|_2
=
O_p(\mathcal R_{nT}^*).
$
Combining these bounds yields
\begin{equation}
\|\mathcal R_2\|_2
=
O_p(\mathcal R_{nT}^*).
\label{eq:R2-denoised}
\end{equation}
Equations~\eqref{eq:R1-denoised} and
\eqref{eq:R2-denoised} give
$
\|\bar g_{nT}(\theta_0,\widehat W)
-\bar g_{nT}(\theta_0,W_0)\|_2
=
O_p(\mathcal R_{nT}^*).
$\\
\noindent\emph{Gradient difference.}
As in part~(a), decompose
\begin{align*}
\|\widehat G(\widehat W)-\widehat G(W_0)\|_2
&\lesssim
\left\|
\frac1{nT}\sum_{t=1}^T
\left[
Z_t(\widehat W)-Z_t(W_0)
\right]^\top X_t
\right\|_2+
\left\|
\frac1{nT}\sum_{t=1}^T
\left[
Z_t(\widehat W)^\top\widehat W
-Z_t(W_0)^\top W_0
\right]Y_t
\right\|_2\\
&\quad+
\left\|
\frac1{nT}\sum_{t=1}^T
\left[
Z_t(\widehat W)^\top\widehat W
-Z_t(W_0)^\top W_0
\right]X_t
\right\|_2\\
&=:
\|\mathcal R_{n1}\|_2
+\|\mathcal R_{n2}\|_2
+\|\mathcal R_{n3}\|_2.
\end{align*}
For an instrument column of the form $W^lX_{jt}$, define
$
A_l
:=
\widehat W^l-W_0^l$, and $
D_l
:=
(\widehat W^l)^\top\widehat W
-(W_0^l)^\top W_0.
$
By \eqref{eq:power-telescope}, $A_l$ is a finite sum of terms
containing at least one factor $\Delta_W$. Moreover,
$
D_l
=
A_l^\top\widehat W
+
(W_0^l)^\top\Delta_W.
$
Therefore, every term in $D_l$ also contains at least one factor
$\Delta_W$.\\
Each coordinate of $\mathcal R_{n1}$ has the form
$
\frac1{nT}\sum_{t=1}^T
X_{jt}^\top A_l^\top X_{kt},
$
and hence involves the pair $(X_{jt},X_{kt})$. Therefore,
$
\|\mathcal R_{n1}\|_2
=
O_p(\mathcal R_{nT}^*).
$ Each coordinate of $\mathcal R_{n2}$ has the form
$
\frac1{nT}\sum_{t=1}^T
X_{jt}^\top D_lY_t,
$
and hence involves the pair $(X_{jt},Y_t)$. Thus,
$
\|\mathcal R_{n2}\|_2
=
O_p(\mathcal R_{nT}^*).
$Finally, each coordinate of $\mathcal R_{n3}$ has the form
$
\frac1{nT}\sum_{t=1}^T
X_{jt}^\top D_lX_{kt},
$
which again involves the pair $(X_{jt},X_{kt})$. Hence,
$
\|\mathcal R_{n3}\|_2
=
O_p(\mathcal R_{nT}^*).
$Combining the three bounds gives
$
\|\widehat G(\widehat W)-\widehat G(W_0)\|_2
=
O_p(\mathcal R_{nT}^*).
$
This completes the proof of part~(c).
\end{proof}
\vspace{-1em}
\begin{proof}[Proof of Theorem \ref{Thm:DK22}]
\textbf{Step 1}
    The proof consists of two steps. First, we show that when $W_0$ is correctly specified, we establish the large sample performance of $\widehat \theta_{p}(W_0)$. 
	In the first step, when $W_0$ is correctly specified, $\e_t(\theta_0, W_0)=\e_t$ and hence
   $\bar{g}_{nT}(\theta_0, W_0)=\sum_{t=1}^T Z^{\top}_t(W_0)\e_t/(nT).$ 
	Note that the gradient $$\frac{\partial \bar{g}_{nT}(\theta, W)}{\partial \theta\T} =  \frac{\partial\frac{1}{nT}\sum_{t=1}^T Z_t(W)\T(Y_t- \lambda W Y_t - X_t \beta  -WX_t \gamma)}{\partial \theta\T}
				=-\frac{1}{nT} \sum_{t=1}^T Z_t(W)\T (WY_t, X_t,WX_t),$$ with $Y_t=(I_n-\lambda_0 W_0)^{-1}( X_t \beta_0  +W_0X_t \gamma_0+\e_t)$.
                Define $\E[\widehat{G}(W)] = G(W)$.
                Recall that  $\widehat{G}(W) =\frac{\partial \bar{g}_{nT}(\theta, W)}{\partial \theta^{\top}} ,$ and $Q_{0t}=(W_0(I_n-\lambda_0 W_0)^{-1}(X_t\beta_0+W_0X_t \gamma_0), X_t,W_0X_t)$. We have   
{\begin{eqnarray*}
					&&\frac{\partial \bar{g}_{nT}(\theta, W_0)}{\partial \theta\T}=-\frac{ \sum_{t=1}^T Z^{\top}_t(W_0)Q_{0t}}{nT}-\frac{ \sum_{t=1}^T Z^{\top}_t(W_0) (W_0(I_n-\lambda_0 W_0)^{-1}\e_t, {\bf 0}_{n\times (2K)})}{nT} \coloneqq \mathcal{R}_1+\mathcal{R}_2\label{trueW}
				\end{eqnarray*} }
        			To prove $\mathcal{R}_1+\mathcal{R}_2\to -\Sigma_{Z_0Q_0}$ holds, 
			 Assumption \ref{Assump:Instruments} yields the first term $\mathcal{R}_1$ converges to  $-\Sigma_{Z_0Q_0}$.
                For the second term $\mathcal{R}_2$ we have $\frac{1}{nT}\sum_{t=1}^T Z^{\top}_t(W_0) (W_0(I_n-\lambda_0 W_0)^{-1}\e_t)$, 
                let us compute the variance,
                $$A_z = \frac{ \sigma_0^2\sum_{t=1}^T Z^{\top}_t(W_0) (W_0(I_n-\lambda_0 W_0)^{-1}(I_n-\lambda_0 W_0)^{-1\top}W_0^{\top}Z_t(W_0)}{n^2T^2}.$$
                Assume that the column sums \(\sum_{i=1}^n |W_{0,im}| \leq c_n.\)
                From Assumption \ref{weight} and its remark we have $\|S_0\|_1=O( \sqrt{n})$ , and moreover, from Assumption \ref{weight}, $\|L_0\|_2\leq \sqrt{m_r(L_0)m_c(L_0)}\|L_0\|_{\max}$ is bounded and  we have $\|L_0\|_1=O(\sqrt{n}\|L_0\|_2).$ So we can deduce that $\|W_0\|_1 = O(\sqrt{n})$ .
                Since we have 
                \begin{eqnarray*}
                 \|A_z\|_{\infty}&\leq& \frac{ \sigma_0^2\sum_{t=1}^T \|Z^{\top}_t(W_0) W_0(I_n-\lambda_0 W_0)^{-1}\|_{\infty}\|(I_n-\lambda_0 W_0)^{-1\top}W_0^{\top}Z_t(W_0)\|_{\infty}}{n^2T^2}   \\
                 &\leq &\frac{ \sigma_0^2\sum_{t=1}^T \|Z^{\top}_t(W_0)\|_{\infty} \|W_0(I_n-\lambda_0 W_0)^{-1}\|_{\infty}\|(I_n-\lambda_0 W_0)^{-1\top}W_0^{\top}Z_t(W_0)\|_{\infty}}{n^2T^2},
                \end{eqnarray*}
            for $d = 0,1,2$.
            {To fill in details since there are finite number of columns of $W_0$ to have unbounded column sums of order bounded by $c_n$, we have by {Lemma \ref{propbound} in the Online Appendix in  \citep{li2026estimating}  } with $m$ and $M$ being constants, $\|X^{\top}_tW_0^{l\top}\|_{\infty} \lesssim  \|{ X_t\T}\|_{\max}c_n+ \|{ X_t\T}\|_{\infty}\lesssim n$,
                 and $$\|(I_n-\lambda_0 W_0)^{-1\top}W_0^{\top}Z_t(W_0)\|_{\infty}\lesssim\|(I_n-\lambda_0 W_0)^{-1\top}W_0^{\top}\|_{\infty}\|Z_t(W_0)\|_{\infty}\lesssim c_n,$$
  $$\|A_z\|_{\infty}\leq \frac{ \sigma_0^2\sum_{t=1}^T \|Z^{\top}_t(W_0)\|_{\infty}\|(I_n-\lambda_0 W_0)^{-1\top}W_0^{\top}Z_t(W_0)\|_{\infty}}{n^2T^2}\lesssim \frac{n Tc_n}{n^2T^2} .$$
 Thus to ensure that $\|A_z\|_{\infty}\to 0$ we need to have $\frac{ c_n}{nT}\to 0$ to ensure $\mathcal{R}_2 \to_p 0$.
  It shall be noted that this is already ensured by our Assumptions \ref{weight}.}
 %Assumption \ref{Assump:True_Adjacency} also assumes $(I -\lambda W_{0,BB})^{-1}$ exists.
      It is negligible since its $2$-norm is of order 
$O_p(\sqrt{c_n/(nT)}) = O_p(1/(n^{1/4}\sqrt T)) = o_p(1)$. \\
				By the linearity of the moment condition, we have  
				$$\bar{g}_{nT}(\widehat\theta_p(W_0), W_0)=\bar{g}_{nT}(\theta_{0}, W_0)+ \widehat{G}(W_0) ( \widehat\theta_p(W_0)-\theta_0).$$
                Hence the first order condition yields that 
                \begin{eqnarray*}
               && {\bf 0}=\widehat{G}(W_0)^{\top}
               \widehat{\Lambda}_{nT}^{-1}  \bar{g}_{nT}(\widehat\theta_p(W_0), W_0)
           =\widehat{G}(W_0)^{\top} \widehat{\Lambda}_{nT}^{-1} \bar{g}_{nT}(\theta_0, W_0)
           +\widehat{G}(W_0)^{\top} \widehat{\Lambda}_{nT}^{-1}  \widehat{G}(W_0)( \widehat\theta_p(W_0)-\theta_0),
                \end{eqnarray*}
or equivalently, \begin{eqnarray*}
        &&  \widehat{G}(W_0)^{\top} \widehat{\Lambda}_{nT}^{-1}   \widehat{G}(W_0)( \widehat\theta_p(W_0)-\theta_0)= -\widehat{G}(W_0)^{\top} \widehat{\Lambda}_{nT}^{-1}  \bar{g}_{nT}(\theta_0, W_0).
                \end{eqnarray*}
                By Assumption \ref{Assump:Instruments} and \ref{Assump:Moments_Weight}, we conclude that $\widehat{G}(W_0)^{\top}\widehat{\Lambda}_{nT}^{-1} \to_{p} -\Sigma^{\top}_{Z_0Q_0} \Lambda^{-1},$ and 
				$$\widehat{G}(W_0)^{\top}\widehat{\Lambda}_{nT}^{-1} \widehat{G}(W_0)\to_{p} \Sigma^{\top}_{Z_0Q_0} \Lambda^{-1} \Sigma_{Z_0Q_0},$$ 
                which is nonsingular as Assumption \ref{Assump:Instruments} assumes $\Sigma_{Z_0Q_0}$ is of full column rank and Assumption \ref{Assump:Moments_Weight} assumes $\Lambda$ is a positive definite matrix. 
We shall first prove the consistency. For convenience we introduce the vector $O_p(.)$ notation. Let $a_n$ be a positive sequence. A finite dimension vector $A_n = O_p(a_n)$ means that $\|A_n/a_n\| = O_p(1)$ for any norm $\|.\|$.
Thus       %\begin{eqnarray*}
        $\widehat{G}(W_0)^{\top} \widehat{\Lambda}_{nT}^{-1}   \widehat{G}(W_0)( \widehat\theta_p(W_0)-\theta_0)= -\widehat{G}(W_0)^{\top} \widehat{\Lambda}_{nT}^{-1}  \bar{g}_{nT}(\theta_0, W_0),$
%                \end{eqnarray*}
{  \begin{eqnarray*}
            && [\Sigma^{\top}_{Z_0Q_0}\Lambda^{-1} \Sigma_{Z_0Q_0}]( \widehat\theta_p(W_0)-\theta_0)+ o_p(\|( \widehat\theta_p(W_0)-\theta_0)\|_2)= -\widehat{G}(W_0)^{\top} \widehat{\Lambda}_{nT}^{-1}  \bar{g}_{nT}(\theta_0, W_0).
                \end{eqnarray*}
                By Assumptions \ref{Assump:Instruments} and \ref{Assump:Moments_Weight}, we have that 
    \begin{eqnarray*}
     -\widehat{G}(W_0)^{\top} \widehat{\Lambda}_{nT}^{-1}  \bar{g}_{nT}(\theta_0, W_0) =  -{G}(W_0)^{\top} {\Lambda}^{-1} \bar{g}_{nT}(\theta_0, W_0)+O_p(\|\bar{g}_{nT}(\theta_0, W_0)\|_2).
                \end{eqnarray*}
By Lemma \ref{rateG} and the above derivation,} we have that, 
   $
     \|-\widehat{G}(W_0)^{\top} \widehat{\Lambda}_{nT}^{-1}  \bar{g}_{nT}(\theta_0, W_0)\|_2 =  O_p(\frac{1}{\sqrt{nT}}).
    $
Thus the consistency follows, $\| \widehat\theta_p(W_0)-\theta_0 \|_2= o_p(1)$.
Let $N=nT$ and $z_{ti}=Z_t(W_0)^\top e_i$.
Conditional on $\mathcal Z_n$, the vectors
$z_{ti}\varepsilon_{it}$ are independent and centered.
Their normalized total conditional covariance satisfies
\[
\frac{\sigma_0^2}{N}
\sum_{t=1}^T Z_t(W_0)^\top Z_t(W_0)
\to_p
\sigma_0^2\Sigma_{Z_0Z_0}
=\Omega_0.
\]
Take a common moment exponent $\delta>2$ for the instruments
and disturbances, and set $\eta=(\delta-2)/2>0$.
For every fixed $a\in\mathbb R^m$, Holder's inequality gives
$\sup_{n,t,i}
\mathbb E|a^\top z_{ti}\varepsilon_{it}|^{2+\eta}<\infty$.
Hence the sum of the conditional $(2+\eta)$ moments of the
normalized summands converges to zero in probability by
Markov's inequality. The conditional Lindeberg--Feller
argument and the Cramer--Wold device therefore yield
\begin{equation}
\label{thm2clt}
\Omega_0^{-1/2}\sqrt{nT}\,
\bar g_{nT}(\theta_0,W_0)
\overset{d}{\longrightarrow}
\mathcal N(\mathbf0_{m\times 1},I_m).
\end{equation}
Define the variance covariance matrix $$\Sigma_{GMM}= (\Sigma_{Z_0Q_0}^{\top}\Lambda^{-1} \Sigma_{Z_0Q_0})^{-1} \Sigma_{Z_0Q_0}^{\top}\Lambda^{-1}(\sigma_0^2\Sigma_{Z_0Z_0}) \Lambda^{-1} \Sigma_{Z_0Q_0}(\Sigma_{Z_0Q_0}^{\top}\Lambda^{-1} \Sigma_{Z_0Q_0})^{-1}$$
                Together with equation (\ref{thm2clt}), we have
   	\begin{eqnarray*}
\sqrt{nT}\,\Sigma_{GMM}^{-1/2}(\widehat\theta_p(W_0)-\theta_0)
&=& -\Sigma_{GMM}^{-1/2}\big(G(W_0)^\top\Lambda^{-1}G(W_0)\big)^{-1}
G(W_0)^\top\Lambda^{-1}\sqrt{nT}\,\bar g_{nT}(\theta_0,W_0) + o_p(1)\\
&\overset{d}{\to}& \mathcal N(\mathbf 0_{(2K+1)\times 1}, I_{2K+1}).
\end{eqnarray*}
\textbf{Step 2}
Second, we study the deviation between
$\widehat\theta_p(\widehat W)$ and $\widehat\theta_p(W_0)$ to
establish the large-sample performance of
$\widehat\theta_p(\widehat W)$. By the linearity of the sample
moments,
\[
\widehat\theta_p(\widehat W)-\theta_0
=
-\left(
\widehat G(\widehat W)^\top
\widehat\Lambda_{nT}^{-1}
\widehat G(\widehat W)
\right)^{-1}
\widehat G(\widehat W)^\top
\widehat\Lambda_{nT}^{-1}
\bar g_{nT}(\theta_0,\widehat W).
\]
Decompose
$\bigl(\widehat\theta_p(\widehat W)-\theta_0\bigr)
-
\bigl(\widehat\theta_p(W_0)-\theta_0\bigr)
=
\mathcal R_{1,nT}
+
\mathcal R_{2,nT}
+
\mathcal R_{3,nT}$,
where
\begin{align*}
\mathcal R_{1,nT}
&=
\Big[
-\left(
\widehat G(\widehat W)^\top
\widehat\Lambda_{nT}^{-1}
\widehat G(\widehat W)
\right)^{-1}
+
\left(
\widehat G(W_0)^\top
\widehat\Lambda_{nT}^{-1}
\widehat G(W_0)
\right)^{-1}
\Big]\times
\widehat G(\widehat W)^\top
\widehat\Lambda_{nT}^{-1}
\bar g_{nT}(\theta_0,\widehat W),
\\[1ex]
\mathcal R_{2,nT}
&=
\left(
\widehat G(W_0)^\top
\widehat\Lambda_{nT}^{-1}
\widehat G(W_0)
\right)^{-1}\times
\Big[
-\widehat G(\widehat W)^\top
+\widehat G(W_0)^\top
\Big]
\widehat\Lambda_{nT}^{-1}
\bar g_{nT}(\theta_0,\widehat W),
\\[1ex]
\mathcal R_{3,nT}
&=
\left(
\widehat G(W_0)^\top
\widehat\Lambda_{nT}^{-1}
\widehat G(W_0)
\right)^{-1}
\widehat G(W_0)^\top
\widehat\Lambda_{nT}^{-1}\times
\Big[
\bar g_{nT}(\theta_0,W_0)
-
\bar g_{nT}(\theta_0,\widehat W)
\Big].
\end{align*}
We apply Lemma~\ref{rateG}(c) to bound each remainder. Recall
$\mathcal R_{nT}^*$ from~\eqref{eq:RstarnT}. Lemma~\ref{rateG}(c)
gives
\begin{align}
& \|\widehat G(\widehat W)-\widehat G(W_0)\|_2
=
O_p(\mathcal R_{nT}^*),
\label{eq:DK22-G-rate}\\
& \|\bar g_{nT}(\theta_0,\widehat W)
-\bar g_{nT}(\theta_0,W_0)\|_2
=
O_p(\mathcal R_{nT}^*).
\label{eq:DK22-g-rate} 
\end{align}
For $\mathcal R_{3,nT}$, since
$
\left\|
\left(
\widehat G(W_0)^\top
\widehat\Lambda_{nT}^{-1}
\widehat G(W_0)
\right)^{-1}
\right\|_2
=
O_p(1)
$
and
$
\left\|
\widehat G(W_0)^\top
\widehat\Lambda_{nT}^{-1}
\right\|_2
=
O_p(1),
$
{ Equation~\eqref{eq:DK22-g-rate}} gives
$
\|\mathcal R_{3,nT}\|_2
=
O_p(\mathcal R_{nT}^*).
$ For $\mathcal R_{2,nT}$, { Step 1 and
Equation~\eqref{eq:DK22-g-rate}} give
\[
\|\bar g_{nT}(\theta_0,\widehat W)\|_2
\leq
\|\bar g_{nT}(\theta_0,W_0)\|_2
+
\|\bar g_{nT}(\theta_0,\widehat W)
-\bar g_{nT}(\theta_0,W_0)\|_2
=
O_p\left(
\frac1{\sqrt{nT}}+\mathcal R_{nT}^*
\right).
\]
Together with { Equation~\eqref{eq:DK22-G-rate}}, this yields
$
\|\mathcal R_{2,nT}\|_2
=
O_p\left[
\mathcal R_{nT}^*
\left(
\frac1{\sqrt{nT}}+\mathcal R_{nT}^*
\right)
\right]
=
o_p(\mathcal R_{nT}^*),
$
provided $\mathcal R_{nT}^*=o(1)$.
For $\mathcal R_{1,nT}$, define
$
\widehat A_{nT}
=
\widehat G(\widehat W)^\top
\widehat\Lambda_{nT}^{-1}
\widehat G(\widehat W)$, and 
%\qquad
$A_{0,nT}
=
\widehat G(W_0)^\top
\widehat\Lambda_{nT}^{-1}
\widehat G(W_0).
$
{ Equation \eqref{eq:DK22-G-rate}} implies
$
\|\widehat A_{nT}-A_{0,nT}\|_2
=
O_p(\mathcal R_{nT}^*).
$
Using the inverse identity
$
\widehat A_{nT}^{-1}-A_{0,nT}^{-1}
=
\widehat A_{nT}^{-1}
(A_{0,nT}-\widehat A_{nT})
A_{0,nT}^{-1},
$
together with
$
\|\widehat A_{nT}^{-1}\|_2
+
\|A_{0,nT}^{-1}\|_2
=
O_p(1),
$
we obtain
$
\|\widehat A_{nT}^{-1}-A_{0,nT}^{-1}\|_2
=
O_p(\mathcal R_{nT}^*).
$
Therefore,
$
\|\mathcal R_{1,nT}\|_2
=
O_p\left[
\mathcal R_{nT}^*
\left(
\frac1{\sqrt{nT}}+\mathcal R_{nT}^*
\right)
\right]
=
o_p(\mathcal R_{nT}^*),
$
provided $\mathcal R_{nT}^*=o(1)$.

Combining the three terms and recalling from Step 1 that
$
\|\widehat\theta_p(W_0)-\theta_0\|_2
=
O_p\left(\frac1{\sqrt{nT}}\right),
$
we obtain
\begin{align*}
\|\widehat\theta_p(\widehat W)-\theta_0\|_2
&\leq
\|\widehat\theta_p(W_0)-\theta_0\|_2
+
\|\mathcal R_{1,nT}\|_2
+
\|\mathcal R_{2,nT}\|_2
+
\|\mathcal R_{3,nT}\|_2
=
O_p\left(
\frac1{\sqrt{nT}}+\mathcal R_{nT}^*
\right).
\end{align*}
By the definition in~\eqref{eq:RstarnT},
$
\mathcal R_{nT}^*
=
(1+\|B_0\|_2)\kappa_n^{d+1}
\frac{(r\vee1)\rho_n+s_n}{n},
$
where
$\rho_n=w_{2n}+a_n+\mathcal R_s$,
$s_n=\mathcal R_s\sqrt{m_s(S_0^*)}$,
and
$\kappa_n=1+\|W_0\|_2+\rho_n$.
Thus the preceding display retains all spectral factors without
requiring them to be uniformly bounded.
Under the boundedness conditions in
Theorem~\ref{Thm:DK22}, and when $\rho_n=O(1)$, we have
$\kappa_n=O(1)$ and hence
$
\mathcal R_{nT}^*
=
O\left(
\frac{
(r\vee1)(w_{2n}+a_n+\mathcal R_s)
+
\mathcal R_s\sqrt{m_s(S_0^*)}
}{n}
\right).
$

\noindent\textbf{Rate comparison.}
In contrast, using the noisy matrix $W=W_0+E$ directly without
denoising, Lemma~\ref{rateG}(a) gives
$
\|\widehat\theta_p(W)-\theta_0\|_2
=
O_p\left(
\frac1{\sqrt{nT}}+w_{2n}+w_{2n}^2
\right).
$
Therefore, the denoised estimator has a strictly smaller plug-in
error whenever
$
\mathcal R_{nT}^*
=
o(w_{2n}+w_{2n}^2).
$
Equivalently,
$
(1+\|B_0\|_2)\kappa_n^{d+1}
\bigl\{
(r\vee1)\rho_n+s_n
\bigr\}
=
o\left(
n(w_{2n}+w_{2n}^2)
\right).
$ Finally, if
$
\mathcal R_{nT}^*
=
o\left(\frac1{\sqrt{nT}}\right),
$
then
$
\sqrt{nT}
\left(
\widehat\theta_p(\widehat W)
-
\widehat\theta_p(W_0)
\right)
=
o_p(1).
$
Combining this result with Step 1 and Slutsky's theorem gives
\[
\sqrt{nT}\,
\Sigma_{GMM}^{-1/2}
\left(
\widehat\theta_p(\widehat W)-\theta_0
\right)
\overset{d}{\longrightarrow}
\mathcal N(\mathbf 0_{(2K+1)\times 1},I_{2K+1}).
\]
\end{proof}
\begin{remark}[Explicit comparison under row sparsity]
\label{rem:explicit-rate-comparison}
Suppose there are no dominant units. Then
$\|W_0\|_2=O(1)$, and $\|B_0\|_2
=
\|(I_n-\lambda_0W_0)^{-1}\|_2
=O(1).
$
If $\rho_n=O(1)$, this also gives $\kappa_n=O(1)$.
Suppose that the sparse component has a bounded number of nonzero
entries in each row and column. Then
$
m_s(S_0^*)\asymp n.
$
Let
$
b_n\coloneqq w_{\max n}+a_n.
$
Because
$
\mathcal R_s(W_0,E)
=
b_n\sqrt{m_s(S_0^*)}
\asymp b_n\sqrt n,
$
we have
$
s_n
=
\mathcal R_s(W_0,E)\sqrt{m_s(S_0^*)}
\asymp b_n n.
$
Consequently, for bounded rank $r$,
\begin{align}
\mathcal R_{nT}^*
&=
O\left(
\frac{
w_{2n}+b_n\sqrt n+b_nn
}{n}
\right) %\nonumber\\
=
O\left(
\frac{w_{2n}}{n}+b_n
\right).
\label{eq:rowsparse-plugin-rate}
\end{align}

In the common regime in which
$
b_n
=
w_{\max n}+a_n
=
O\left(\frac{w_{2n}}{\sqrt n}\right),
$
it follows that
$
\mathcal R_{nT}^*
=
O\left(\frac{w_{2n}}{\sqrt n}\right).
$
By contrast, when $w_{2n}=o(1)$, the network-error component of the
estimator based on the raw adjacency matrix is
$
w_{2n}+w_{2n}^2
=
O(w_{2n}).
$
Thus, under row sparsity, the plug-in estimator improves the
noise-induced rate from
$
O(w_{2n})$ to 
%\quad\text{to}\quad
$O\left(\frac{w_{2n}}{\sqrt n}\right).
$
The improvement is therefore of order $\sqrt n$, rather than $n$.
Including sampling variation, the corresponding rates are
\[
\|\widehat\theta_p(\widehat W)-\theta_0\|_2
=
O_p\left(
\frac{w_{2n}}{\sqrt n}
+
\frac{1}{\sqrt{nT}}
\right),
%\]
%and
%\[
\qquad 
\|\widehat\theta_{\mathrm{GMM}}-\theta_0\|_2
=
O_p\left(
w_{2n}
+
\frac{1}{\sqrt{nT}}
\right).
\]
\end{remark}
\subsection{Proof of Proposition \ref{lemmabias}.}
\begin{proof}
For simplicity, we retain the contextual-effects regressor
$W_0X_t$ but assume that its true coefficient is zero,
$\gamma_0=\mathbf 0_{K\times 1}$. 
Recall that
$Q_{0t}=[W_0B_0X_t\beta_0,\,X_t,\,W_0X_t]$.
Define
$G_0(M):=
-\lim_{\min(n,T)\to\infty}
\frac{1}{nT}\sum_{t=1}^T
\mathbb E[Z_t(M)^\top Q_{0t}]$.

Recall that 
$
\bar g_{nT}(\theta,W,M)
=
\frac{1}{nT}\sum_{t=1}^T Z_t(M)^\top
\left(
Y_t-\lambda WY_t-X_t\beta-WX_t\gamma
\right).
$
and $Z_t(M)$ is of the form $M^\ell X_t$ for $\ell = 0, 1, 2 $. Moreover, $\mathbb{E}\{\widehat G(M)\} = G(M)$. Define 
$\widehat\theta_p(\widehat W, M)$ as the plug-in estimator using 
$Z_t(M)$:
\begin{equation}
\widehat\theta_p(\widehat W, M) = 
\Big[\widetilde X(\widehat W)^\top Z(M)\widehat\Lambda_{nT}^{-1}
Z(M)^\top \widetilde X(\widehat W)\Big]^{-1}
\widetilde X(\widehat W)^\top Z(M)\widehat\Lambda_{nT}^{-1}
Z(M)^\top Y.
\end{equation}
We compare with the estimator using $W$:
\[
\widehat\theta_p(W, M) = \widehat\theta_{GMM} = 
\Big[\widetilde X(W)^\top Z(M)\widehat\Lambda_{nT}^{-1}
Z(M)^\top \widetilde X(W)\Big]^{-1}
\widetilde X(W)^\top Z(M)\widehat\Lambda_{nT}^{-1}
Z(M)^\top Y.
\]
By the first-order condition for $\widehat\theta_p(\widehat W, M)$ and 
linearity of the moment function in $\theta$, we have the exact 
expansion
\[
\widehat\theta_p(\widehat W, M) - \theta_0 
= -\big[\widehat G(\widehat W, M)^\top \widehat\Lambda_{nT}^{-1} 
\widehat G(\widehat W, M)\big]^{-1} 
\widehat G(\widehat W, M)^\top \widehat\Lambda_{nT}^{-1}\,
\bar g_{nT}(\theta_0, \widehat W, M).
\]
Under Assumptions~\ref{weight}--\ref{Assump:Moments_Weight}, 
$\widehat G(\widehat W, M) \to_p -\Sigma_{Z_0Q_0}(M)$ and 
$\widehat\Lambda_{nT} \to_p \Lambda$, so that
\[
\widehat\theta_p(\widehat W, M) - \theta_0 
= [\Sigma_{Z_0Q_0}(M)^\top \Lambda^{-1} \Sigma_{Z_0Q_0}(M)]^{-1} 
\Sigma_{Z_0Q_0}(M)^\top \Lambda^{-1}\,
\bar g_{nT}(\theta_0, \widehat W, M)\,(1 + o_p(1)).
\]
The deviation from the limiting form is $o_p$ relative to 
$\|\bar g_{nT}(\theta_0, \widehat W, M)\|_2$, so the rate of 
$\|\widehat\theta_p(\widehat W, M) - \theta_0\|_2$ matches that of 
$\|\bar g_{nT}(\theta_0, \widehat W, M)\|_2$ up to multiplicative 
constants. Let $B_0 := (I_n - \lambda_0 W_0)^{-1}$ 
and 
$\mathcal{R}_s := (w_{\max n} + a_n)\sqrt{m_s(S_0^*)}$.\\
\textbf{Rate of the GMM estimator (using $W$).}
Let $Z_{jt}$ be the $j$th column of the $n \times m$ matrix $Z_t(M)$. 
The residual at $\theta_0$ when using $W$ is
$\varepsilon_t(\theta_0, W) = \varepsilon_t - \lambda_0 E Y_t - E X_t \gamma_0$,
where $E = W - W_0$. Substituting 
$Y_t = B_0(X_t\beta_0 + W_0 X_t\gamma_0 + \varepsilon_t)$, the $j$th 
element of $\bar g_{nT}(\theta_0, W, M)$ is
\begin{align*}
& \frac{1}{nT}\sum_{t=1}^T Z^\top_{jt}(M)
[\varepsilon_t - \lambda_0 E Y_t - E X_t\gamma_0]= \frac{1}{nT}\sum_{t=1}^T Z^\top_{jt}(M)\varepsilon_t 
- \frac{\lambda_0}{nT}\sum_{t=1}^T Z^\top_{jt}(M) E B_0 (X_t\beta_0 + W_0 X_t\gamma_0) \\
&\qquad - \frac{1}{nT}\sum_{t=1}^T Z^\top_{jt}(M) E X_t\gamma_0
- \frac{\lambda_0}{nT}\sum_{t=1}^T Z^\top_{jt}(M) E B_0 \varepsilon_t=: O_p(1/\sqrt{nT}) + \mathcal{R}_{1,j} + \mathcal{R}_{3,j} + \mathcal{R}_{2,j}.
\end{align*}
For each $i$, let
$
E_{i,-i}
=
(E_{i1},\ldots,E_{i,i-1},E_{i,i+1},\ldots,E_{in})^\top
\in\mathbb R^{n-1}
$
denote the vector of off-diagonal entries in row $i$ of $E$.
{
For $t=1,\ldots,T$, let
$\rho_{\varepsilon E,t\cdot}$ denote the $t$th row of
$\rho_{\varepsilon E}$ and define
$\delta_{nt}
:=
\left(
\sigma_E^2n^{-2s}\Sigma_E
\right)^{-1}
\sigma_0\sigma_E n^{-s}
\rho_{\varepsilon E,t\cdot}^{\top}$.
By joint Gaussianity,
$\mathbb E[\varepsilon_{it}\mid E]
=
E_{i,-i}^{\top}\delta_{nt}$.
Define
$\mu_{E,t}
:=
\left(
E_{1,-1}^{\top}\delta_{nt},
\ldots,
E_{n,-n}^{\top}\delta_{nt}
\right)^{\top}$.
Then
$\varepsilon_t=\mu_{E,t}+\xi_t$,
where the collection $\{\xi_t\}_{t=1}^T$ is independent of $E$ and
has conditional mean zero.The eigenvalue bounds on $\Sigma_E$ imply
$\mathbb E\|\mu_{E,t}\|_2^2
\lesssim
n\|\rho_{\varepsilon E,t\cdot}\|_2^2$.
Consequently,
\[
\frac1T\sum_{t=1}^T\|\mu_{E,t}\|_2^2
=
O_p\!\left(
\frac{n}{T}\|\rho_{\varepsilon E}\|_F^2
\right)
=
O_p\!\left(
n\|\rho_{\varepsilon E}\|_2^2
\right),
\]
where the last equality uses
$\|\rho_{\varepsilon E}\|_F^2
\leq T\|\rho_{\varepsilon E}\|_2^2$.
Now decompose
$\mathcal R_{2,j}
=
-\frac{\lambda_0}{nT}
\sum_{t=1}^T
Z_{jt}(M)^\top E B_0\mu_{E,t}
-
\frac{\lambda_0}{nT}
\sum_{t=1}^T
Z_{jt}(M)^\top E B_0\xi_t$.
For the first term, Cauchy--Schwarz gives
$\max_j\left|
\frac1{nT}\sum_{t=1}^T
Z_{jt}(M)^\top EB_0\mu_{E,t}
\right|
\leq
\frac{\|E\|_2\|B_0\|_2}{n}
\max_j\left(\frac1T\sum_{t=1}^T\|Z_{jt}(M)\|_2^2\right)^{1/2}
\left(\frac1T\sum_{t=1}^T\|\mu_{E,t}\|_2^2\right)^{1/2}
=
O_p\!\left(
\|B_0\|_2n^{1/2-s}\|\rho_{\varepsilon E}\|_2
\right)$.
Moreover, since $\xi_t=\varepsilon_t-\mu_{E,t}$, the disturbance-energy
bound and the preceding bound imply
$\frac1T\sum_{t=1}^T\|\xi_t\|_2^2=O_p(n)$.
Therefore, by Cauchy--Schwarz,
$\max_j\left|
\frac1{nT}\sum_{t=1}^T
Z_{jt}(M)^\top EB_0\xi_t
\right|
\leq
\frac{\|E\|_2\|B_0\|_2}{n}
\max_j\left(\frac1T\sum_{t=1}^T\|Z_{jt}(M)\|_2^2\right)^{1/2}
\left(\frac1T\sum_{t=1}^T\|\xi_t\|_2^2\right)^{1/2}
=O_p(\|B_0\|_2w_{2n})$.
Therefore,
$\max_j|\mathcal R_{2,j}|
\lesssim_p
\|B_0\|_2
\left\{
n^{1/2-s}\|\rho_{\varepsilon E}\|_2+w_{2n}
\right\}$.
}
Then
\begin{align*}
\max_j |\mathcal{R}_{1,j}| 
&\leq |\lambda_0|\frac{1}{nT}\sum_t \|Z_{jt}(M)\|_2 \|B_0\|_2
\|X_t\beta_0 + W_0 X_t\gamma_0\|_2 \|E\|_2 \lesssim_p \|B_0\|_2(1+\|W_0\|_2)\,w_{2n},\\
\max_j |\mathcal{R}_{3,j}| 
&\leq \frac{1}{nT}\sum_t \|Z_{jt}(M)\|_2 \|E\|_2 \|X_t\gamma_0\|_2 
\lesssim_p w_{2n}.
\end{align*}
Hence 
$
\|\bar g_{nT}(\theta_0, W, M)\|_2 = O_p\big(n^{1/2-s}\|\rho_{\varepsilon E}\|_2 
+ w_{2n}\|B_0\|_2(1+\|W_0\|_2)\, + 1/\sqrt{nT}\big),
$
yielding
\[
\|\widehat\theta_{GMM}(W, M) - \theta_0\|_2 
\lesssim_p n^{1/2-s}\|\rho_{\varepsilon E}\|_2 
+ \|B_0\|_2(1+\|W_0\|_2)\,w_{2n} + 1/\sqrt{nT},
\]
which under the boundedness assumption $\|B_0\|_2(1+\|W_0\|_2) = O(1)$ 
gives the rate~\eqref{prop1_gmm}.\\
\noindent\textbf{Step 2: Effect of replacing $W_0$ by $\widehat W$.}
Throughout Step~2, we fix the same $M$ and suppress it from the notation.
We now study the deviation between
$\widehat\theta_p(\widehat W)$ and $\widehat\theta_p(W_0)$.
Because the sample moment is linear in $\theta$, the first-order
condition gives the exact representation
\[
\widehat\theta_p(\widehat W)-\theta_0
=
-\left[
\widehat G(\widehat W)^\top
\widehat\Lambda_{nT}^{-1}
\widehat G(\widehat W)
\right]^{-1}
\widehat G(\widehat W)^\top
\widehat\Lambda_{nT}^{-1}
\bar g_{nT}(\theta_0,\widehat W).
\]
Similarly,
$
\widehat\theta_p(W_0)-\theta_0
=
-\left[
\widehat G(W_0)^\top
\widehat\Lambda_{nT}^{-1}
\widehat G(W_0)
\right]^{-1}
\widehat G(W_0)^\top
\widehat\Lambda_{nT}^{-1}
\bar g_{nT}(\theta_0,W_0).
$
Define
$
\widehat A_{nT}
:=
\widehat G(\widehat W)^\top
\widehat\Lambda_{nT}^{-1}
\widehat G(\widehat W),$ and 
$A_{0,nT}
:=
\widehat G(W_0)^\top
\widehat\Lambda_{nT}^{-1}
\widehat G(W_0).
$
Then
\[
\bigl\{\widehat\theta_p(\widehat W)-\theta_0\bigr\}
-
\bigl\{\widehat\theta_p(W_0)-\theta_0\bigr\}
=
\mathcal R_{1,nT}
+
\mathcal R_{2,nT}
+
\mathcal R_{3,nT},
\]
where
%\begin{align*}
$\mathcal R_{1,nT}
=
\left(
-\widehat A_{nT}^{-1}
+
A_{0,nT}^{-1}
\right)
\widehat G(\widehat W)^\top
\widehat\Lambda_{nT}^{-1}
\bar g_{nT}(\theta_0,\widehat W)$,
%\\[1ex]
$\mathcal R_{2,nT}
=
A_{0,nT}^{-1}
\left[
-\widehat G(\widehat W)^\top
+\widehat G(W_0)^\top
\right]\\ 
\widehat\Lambda_{nT}^{-1}
\bar g_{nT}(\theta_0,\widehat W)$,
%\\[1ex]
$\mathcal R_{3,nT}
=
A_{0,nT}^{-1}
\widehat G(W_0)^\top
\widehat\Lambda_{nT}^{-1}
\left[
\bar g_{nT}(\theta_0,W_0)
-\bar g_{nT}(\theta_0,\widehat W)
\right].$
%\end{align*}

Recall the rate defined in~\eqref{eq:RstarnT}:
\[
\mathcal R_{nT}^*
=
(1+\|B_0\|_2)\kappa_n^{d+1}
\frac{(r\vee1)\rho_n+s_n}{n},
\]
where
\[
\rho_n = w_{2n}+a_n+\mathcal R_s,
\qquad
s_n = \mathcal R_s\sqrt{m_s(S_0^*)},
\qquad
\kappa_n = 1+\|W_0\|_2+\rho_n.
\]
Lemma~\ref{rateG}(c) gives
\begin{align}
\|\widehat G(\widehat W)-\widehat G(W_0)\|_2
&=
O_p(\mathcal R_{nT}^*),
\|\bar g_{nT}(\theta_0,\widehat W)
-\bar g_{nT}(\theta_0,W_0)\|_2
=
O_p(\mathcal R_{nT}^*).\label{eq:DK22-G-rate3}
\end{align}

By Step~1,
$
\|\bar g_{nT}(\theta_0,W_0)\|_2
=
O_p\left(\frac1{\sqrt{nT}}\right)$, $
\|\widehat G(W_0)\|_2
=O_p(1)$, 
and
$A_{0,nT}
\to_p
G_0^\top\Lambda^{-1}G_0.
$
Since $G_0$ has full column rank and $\Lambda^{-1}$ is positive
definite,
$
\|A_{0,nT}^{-1}\|_2
=
O_p(1).
$

Moreover, { Equation %~\eqref{eq:DK22-G-rate} 
\eqref{eq:DK22-G-rate3} gives} 
$
\|\widehat A_{nT}-A_{0,nT}\|_2
=
O_p(\mathcal R_{nT}^*).
$
If $\mathcal R_{nT}^*=o(1)$, Weyl's inequality implies
$
\lambda_{\min}(\widehat A_{nT})
\geq
\lambda_{\min}(A_{0,nT})
-
\|\widehat A_{nT}-A_{0,nT}\|_2,
$
and hence
$
\|\widehat A_{nT}^{-1}\|_2
+
\|A_{0,nT}^{-1}\|_2
=
O_p(1).
$
Using the inverse identity
$
\widehat A_{nT}^{-1}-A_{0,nT}^{-1}
=
\widehat A_{nT}^{-1}
\left(
A_{0,nT}-\widehat A_{nT}
\right)
A_{0,nT}^{-1},
$
we obtain
\begin{equation}
\|\widehat A_{nT}^{-1}-A_{0,nT}^{-1}\|_2
=
O_p(\mathcal R_{nT}^*).
\label{eq:DK22-inverse-rate}
\end{equation}

Next, Equation %~\eqref{eq:DK22-g-rate} 
\eqref{eq:DK22-G-rate3} 
and Step~1 imply
\begin{align}
\|\bar g_{nT}(\theta_0,\widehat W)\|_2
\leq
\|\bar g_{nT}(\theta_0,W_0)\|_2
+
\|\bar g_{nT}(\theta_0,\widehat W)
-\bar g_{nT}(\theta_0,W_0)\|_2 =
O_p\left(
\frac1{\sqrt{nT}}+\mathcal R_{nT}^*
\right).
\label{eq:DK22-ghat-rate}
\end{align}

%For $\mathcal R_{3,nT}$,
 Equation %~\eqref{eq:DK22-g-rate} or 
\eqref{eq:DK22-G-rate3} and the boundedness of the
multiplying matrices give
$
\|\mathcal R_{3,nT}\|_2
=
O_p(\mathcal R_{nT}^*).
$
 Equations~\eqref{eq:DK22-G-rate3} and
\eqref{eq:DK22-ghat-rate} give
%\begin{align*}
$\|\mathcal R_{2,nT}\|_2
=
O_p\left[
\mathcal R_{nT}^*
\left(
\frac1{\sqrt{nT}}+\mathcal R_{nT}^*
\right)
\right]=
o_p(\mathcal R_{nT}^*),$ 
%\end{align*}
provided $\mathcal R_{nT}^*=o(1)$. Similarly, Equations~\eqref{eq:DK22-inverse-rate} and
\eqref{eq:DK22-ghat-rate} yield
%\begin{align*}
$\|\mathcal R_{1,nT}\|_2
=
O_p\left[
\mathcal R_{nT}^*
\left(
\frac1{\sqrt{nT}}+\mathcal R_{nT}^*
\right)
\right]=
o_p(\mathcal R_{nT}^*).$
Combining the three remainder bounds gives
$
\left\|
\widehat\theta_p(\widehat W)
-
\widehat\theta_p(W_0)
\right\|_2
=
O_p(\mathcal R_{nT}^*).
$
Since Step~1 establishes
$
\|\widehat\theta_p(W_0)-\theta_0\|_2
=
O_p\left(\frac1{\sqrt{nT}}\right),
$
we conclude that
\begin{equation}
\|\widehat\theta_p(\widehat W)-\theta_0\|_2
=
O_p\left(
\frac1{\sqrt{nT}}+\mathcal R_{nT}^*
\right).
\label{eq:DK22-plugin-rate}
\end{equation}

\end{proof}

\subsection{Proof of Theorem \ref{Thm:Supervised_Rate}}

\begin{proof}
This step estimates the weight matrix $W_0=L_0+S_0$ jointly with
the regression coefficient $\theta$. Specifically, we minimize the
objective in~\eqref{Eq:Supervised_Estimator} using a
block-coordinate procedure. When $L$ and $S$ are given, $\theta$
is estimated by standard GMM using the plug-in weight matrix
$L+S$. When $\theta$ is given, minimization over $(L,S)$ becomes
a penalized regression problem involving a nuclear-norm penalty on
$L$ and an elementwise $\ell_1$ penalty on the entries
of $S$. To streamline the exposition, we establish the result using the
instrumental variable $Z_t(M)$ for an arbitrary fixed matrix $M$.
The same proof applies when $M$ is exogenous or
$\sigma(E)$-measurable and satisfies the corresponding norm and
moment conditions. Additional arguments are required when $M$ is
estimated using the same outcome data as the supervised estimator.

Throughout the proof, we impose
$
\operatorname{diag}(L+S)=0,
$
so that all estimated adjacency matrices have zero diagonal.
Recall that
$
\bar g_{nT}(\theta,W,M)
=
\frac{1}{nT}\sum_{t=1}^T
Z_t^\top(M)
\left(
Y_t-\lambda WY_t-X_t\beta-WX_t\gamma
\right).
$
For
$
\theta=(\lambda,\beta^\top,\gamma^\top)^\top,
$
define the $n^2\times m$ matrix
$
\widetilde X_t(\theta,M)
:=
\frac{1}{n}
(\lambda Y_t+X_t\gamma)\otimes Z_t(M),
$
and let
$
\overline{\widetilde X}_{nT}(\theta,M)
:=
\frac{1}{T}\sum_{t=1}^T
\widetilde X_t(\theta,M).$ 
At the true parameter, write
$
\widetilde X_t(M)
:=
\widetilde X_t(\theta_0,M)
=
\frac{1}{n}
(\lambda_0Y_t+X_t\gamma_0)\otimes Z_t(M).
$       The vectorization identity\\
$
\operatorname{vec}
\left(
Z_t^\top(M)W(\lambda Y_t+X_t\gamma)
\right)
=
\left[
(\lambda Y_t+X_t\gamma)\otimes Z_t(M)
\right]^\top
\operatorname{vec}(W)
$
implies that, for fixed $\theta$ and $M$,
\begin{equation}
\bar g_{nT}(\theta,W_1,M)
-
\bar g_{nT}(\theta,W_2,M)
=
-\overline{\widetilde X}_{nT}(\theta,M)^\top
\operatorname{vec}(W_1-W_2).
\label{eq:supervised_moment_affine}
\end{equation}

Define
$
q_{nT}(\theta,W,M)
:=
\bar g_{nT}(\theta,W,M)^\top
\widehat\Lambda_{nT}^{-1}
\bar g_{nT}(\theta,W,M).
$Write $\mathcal R_s:=\mathcal R_s(W_0,E)$ and let
$\bar r:=r\vee1$. 
Define
$\mathcal C_n\coloneqq
\{(L,S)\in\mathbb R^{n\times n}\times\mathbb R^{n\times n}:
\operatorname{diag}(L+S)=0\}$.
Define
$Q_n^0(L,S)\coloneqq
\frac12\|W-L-S\|_F^2
+\nu_{nT}\|L\|_*
+\tau_{nT}|S|_{1,1}$,
where $|S|_{1,1}\coloneqq\sum_{i,j}|S_{ij}|$, and define
$Q_{nT}(\theta,L,S;M)\coloneqq
\xi_{nT}q_{nT}(\theta,L+S,M)+Q_n^0(L,S)$
for $(L,S)\in\mathcal C_n$.
Let
$(\widehat L^{[0]},\widehat S^{[0]})
\in\arg\min_{(L,S)\in\mathcal C_n}Q_n^0(L,S)$
be a measurable global minimizer and set
$\widehat W^{[0]}:=\widehat L^{[0]}+\widehat S^{[0]}$.
Take
$\nu_{nT}
=C_\nu(w_{2n}+\mathcal R_s/\sqrt{\bar r})$
and
$\tau_{nT}=C_\tau(w_{\max n}+a_n)$
for sufficiently large constants $C_\nu,C_\tau>0$.

Recall that
$\rho_n:=w_{2n}+a_n+\mathcal R_s$,
$s_n:=\mathcal R_s\sqrt{m_s(S_0^*)}$,
$\kappa_n:=1+\|W_0\|_2+\rho_n$, and
$\mathcal R_{nT}^*:=
(1+\|B_0\|_2)\kappa_n^{d+1}
(\bar r\rho_n+s_n)/n$.
Here $\mathcal R_s=(w_{\max n}+a_n)
\sqrt{m_s(S_0^*)}$.
For the joint initializer, suppose that
\begin{align}
\|\widehat L^{[0]}-L_0\|_F^2
+\|\widehat S^{[0]}-S_0\|_F^2
&=
O_p\!\left(
\bar r(w_{2n}+a_n)^2+\mathcal R_s^2
\right), \nonumber\\
\|\widehat L^{[0]}-L_0\|_*
&=
O_p\!\left(
\bar r(w_{2n}+a_n)
+\sqrt{\bar r}\,\mathcal R_s
\right), \nonumber\\
|\operatorname{off}(\widehat S^{[0]}-S_0)|_{1,1}
&=
O_p(s_n).
\label{eq:supervised_init_rates}
\end{align}

By the additional rank-free aggregate operator-norm condition imposed
in the theorem,
$\|\widehat W^{[0]}-W_0\|_2=O_p(\rho_n)$.
Moreover,
$\|\widehat L^{[0]}-L_0\|_*
+
|\operatorname{off}(\widehat S^{[0]}-S_0)|_{1,1}
=
O_p(\bar r\rho_n+s_n)$.
Therefore, repeating the proof of Lemma~\ref{rateG}(c) with these
joint-initializer bounds gives
\begin{align}
\|\widehat G(\widehat W^{[0]},M)
-\widehat G(W_0,M)\|_2
&=O_p(\mathcal R_{nT}^*), \nonumber\\
\|\bar g_{nT}(\theta_0,\widehat W^{[0]},M)
-\bar g_{nT}(\theta_0,W_0,M)\|_2
&=O_p(\mathcal R_{nT}^*).
\label{eq:supervised_joint_perturbation_rates}
\end{align}
The supervised estimator is obtained by the following
block-coordinate procedure, initialized at
$(\widehat L^{[0]},\widehat S^{[0]})$.
\begin{itemize}
\item[a)] Given $(\widehat L^{[0]},\widehat S^{[0]})$, define
$\widehat\theta^{[1]}
=
\arg\min_{\theta\in\Theta}
q_{nT}(\theta,\widehat W^{[0]},M)$.

\item[b)] Given $\widehat\theta^{[1]}$ and
$\widehat L^{[0]}$, define
\begin{align*}
\widehat S^{[1]}
=
\arg\min_{\substack{S\in\mathbb R^{n\times n}:\\
\operatorname{diag}(\widehat L^{[0]}+S)=0}}
\Bigg\{&
\xi_{nT}
q_{nT}
\big(
\widehat\theta^{[1]},
\widehat L^{[0]}+S,M
\big)
+
\frac12
\|W-\widehat L^{[0]}-S\|_F^2
+
\tau_{nT}|S|_{1,1}
\Bigg\}.
\end{align*}

\item[c)] Given $\widehat\theta^{[1]}$ and
$\widehat S^{[1]}$, define
\begin{align*}
\widehat L^{[1]}
=
\arg\min_{\substack{L\in\mathbb R^{n\times n}:\\
\operatorname{diag}(L+\widehat S^{[1]})=0}}
\Bigg\{&
\xi_{nT}
q_{nT}
\big(
\widehat\theta^{[1]},
L+\widehat S^{[1]},M
\big)
+
\frac12
\|W-L-\widehat S^{[1]}\|_F^2
+
\nu_{nT}\|L\|_*
\Bigg\}.
\end{align*}

\item[d)] Set
$\widehat W^{[1]}
:=
\widehat L^{[1]}+\widehat S^{[1]}$
and define
$\widehat\theta_s
=
\arg\min_{\theta\in\Theta}
q_{nT}(\theta,\widehat W^{[1]},M)$.
\end{itemize}
We now establish the rate of the resulting estimator.

\noindent\textbf{Step a).}
The estimator $\widehat\theta^{[1]}$ is the plug-in GMM estimator
based on the adjacency-only estimator $\widehat W^{[0]}$.
The argument used in the proof of Theorem~\ref{Thm:DK22},
together with \eqref{eq:supervised_joint_perturbation_rates}, gives
\begin{equation}
\|\widehat\theta^{[1]}-\theta_0\|_2
=
O_p\left(
\frac{1}{\sqrt{nT}}+\mathcal R^*_{nT}
\right).
\label{eq:supervised_theta_first}
\end{equation}
For conciseness, define
$
a_{nT}
:=
\frac{1}{\sqrt{nT}}+\mathcal R^*_{nT}.
$
Then
$
\|\widehat\theta^{[1]}-\theta_0\|_2
=
O_p(a_{nT}).
$ Because $\widehat\theta^{[1]}$ minimizes
$q_{nT}(\theta,\widehat W^{[0]},M)$,
\begin{equation}
q_{nT}(\widehat\theta^{[1]},\widehat W^{[0]},M)
\leq
q_{nT}(\theta_0,\widehat W^{[0]},M).
\label{eq:first_theta_minimization}
\end{equation}
The oracle sampling bound and Lemma~\ref{rateG}(c) imply
$
\|
\bar g_{nT}(\theta_0,\widehat W^{[0]},M)
\|_2
=
O_p(a_{nT}).
$Under Assumption~\ref{Assump:Moments_Weight}, there exist constants
$0<c<C<\infty$ such that, with probability approaching one,
$
c\|v\|_2^2
\leq
v^\top\widehat\Lambda_{nT}^{-1}v
\leq
C\|v\|_2^2
$
for every conformable vector $v$. Consequently,
\begin{align*}
c\|
\bar g_{nT}
(\widehat\theta^{[1]},\widehat W^{[0]},M)
\|_2^2
&\leq
q_{nT}
(\widehat\theta^{[1]},\widehat W^{[0]},M)
% \\
% &\leq
% q_{nT}
% (\theta_0,\widehat W^{[0]},M)
% \\
\leq
C\|
\bar g_{nT}
(\theta_0,\widehat W^{[0]},M)
\|_2^2.
\end{align*} It follows that
\begin{equation}
\|
\bar g_{nT}
(\widehat\theta^{[1]},\widehat W^{[0]},M)
\|_2
=
O_p(a_{nT}).
\label{eq:first_theta_moment}
\end{equation}
 We also verify the operator-norm bound used below. Let
$
U_t(\theta):=\lambda Y_t+X_t\gamma.$
For each of the finitely many columns $Z_{t,\cdot j}(M)$ of
$Z_t(M)$,
\[
\left\|
\frac1n
U_t(\theta)\otimes Z_{t,\cdot j}(M)
\right\|_2
=
\frac1n
\|U_t(\theta)\|_2
\|Z_{t,\cdot j}(M)\|_2.
\]
By Lemma~\ref{lem9} in the Online Appendix in  \citep{li2026estimating} and the moment conditions on the regressors and
instruments, both norms on the right-hand side are
$O_p(\sqrt n)$, uniformly for $\theta$ in a neighborhood of
$\theta_0$. Since the number of instruments $m$ is fixed,
$
\|
\overline{\widetilde X}_{nT}(\theta_0,M)
\|_2
=
O_p(1).
$
Moreover, \eqref{eq:supervised_theta_first} gives
$\widehat\theta^{[1]}-\theta_0=o_p(1)$ under $a_{nT}=o(1)$.
Therefore,
\begin{equation}
\|
\overline{\widetilde X}_{nT}
(\widehat\theta^{[1]},M)
\|_2
=
O_p(1).
\label{eq:Xtilde_operator_bound}
\end{equation}
\noindent\textbf{Steps b) and c).}
These two updates are performed with
$\widehat\theta^{[1]}$ fixed. By the definition of
$\widehat S^{[1]}$,
%\begin{align*}
$Q_{nT}
\big(
\widehat\theta^{[1]},
\widehat L^{[0]},
\widehat S^{[1]};M
\big)
\leq
Q_{nT}
\big(
\widehat\theta^{[1]},
\widehat L^{[0]},
\widehat S^{[0]};M
\big).$
%\end{align*}
Similarly, by the definition of $\widehat L^{[1]}$,
%\begin{align*}
$Q_{nT}
\big(
\widehat\theta^{[1]},
\widehat L^{[1]},
\widehat S^{[1]};M
\big)
\leq
Q_{nT}
\big(
\widehat\theta^{[1]},
\widehat L^{[0]},
\widehat S^{[1]};M
\big).$
%\end{align*}
Combining these inequalities,
\begin{equation}
Q_{nT}
\big(
\widehat\theta^{[1]},
\widehat L^{[1]},
\widehat S^{[1]};M
\big)
\leq
Q_{nT}
\big(
\widehat\theta^{[1]},
\widehat L^{[0]},
\widehat S^{[0]};M
\big).
\label{eq:supervised_descent}
\end{equation}
Let
$
\Delta_L
:=
\widehat L^{[1]}-\widehat L^{[0]}$,
$\Delta_S
:=
\widehat S^{[1]}-\widehat S^{[0]}$,
and define
$
H_{nT}
:=
\widehat W^{[1]}-\widehat W^{[0]}
=
\Delta_L+\Delta_S.
$
We first establish a lower bound for the change in the
adjacency-only objective. 
The function
\[
Q_n^0(L,S)
=
\frac12\|W-L-S\|_F^2
+\nu_{nT}\|L\|_*
+\tau_{nT}|S|_{1,1}
\]
is convex in $(L,S)$ and is $1$-strongly convex in the direction
$L+S$. More precisely, for any two feasible pairs
$(L_1,S_1)$ and $(L_2,S_2)$, the strong-convexity remainder in
the quadratic part is
$
\frac12
\|
(L_1+S_1)-(L_2+S_2)
\|_F^2.
$
Because $(\widehat L^{[0]},\widehat S^{[0]})$ minimizes
$Q_n^0$ over the convex zero-diagonal feasible set, its
first-order optimality condition and the strong convexity of the
quadratic part imply
\begin{align}
&Q_n^0(\widehat L^{[1]},\widehat S^{[1]})
-
Q_n^0(\widehat L^{[0]},\widehat S^{[0]})
\geq
\frac12
\left\|
(\widehat L^{[1]}+\widehat S^{[1]})
-
(\widehat L^{[0]}+\widehat S^{[0]})
\right\|_F^2
=
\frac12\|H_{nT}\|_F^2.
\label{eq:supervised_strong_curvature}
\end{align}
For fixed $\theta$ and $M$, equation
\eqref{eq:supervised_moment_affine} shows that
$q_{nT}(\theta,W,M)$ is a convex quadratic function of
$\operatorname{vec}(W)$. Its gradient is
\begin{equation}
\nabla_{\operatorname{vec}(W)}
q_{nT}(\theta,W,M)
=
-2\overline{\widetilde X}_{nT}(\theta,M)
\widehat\Lambda_{nT}^{-1}
\bar g_{nT}(\theta,W,M).
\label{eq:q_W_gradient}
\end{equation}
The exact quadratic expansion around $\widehat W^{[0]}$ gives
\begin{align*}
&\quad q_{nT}
(\widehat\theta^{[1]},\widehat W^{[0]},M)
-
q_{nT}
(\widehat\theta^{[1]},\widehat W^{[1]},M)
\\
&=
-
\left\langle
\nabla_{\operatorname{vec}(W)}
q_{nT}
(\widehat\theta^{[1]},\widehat W^{[0]},M),
\operatorname{vec}(H_{nT})
\right\rangle
-
\operatorname{vec}(H_{nT})^\top
\overline{\widetilde X}_{nT}
(\widehat\theta^{[1]},M)
\widehat\Lambda_{nT}^{-1}
\overline{\widetilde X}_{nT}
(\widehat\theta^{[1]},M)^\top
\operatorname{vec}(H_{nT}).
\end{align*}
The final term is nonpositive. Therefore,
\begin{align}
&\quad q_{nT}
(\widehat\theta^{[1]},\widehat W^{[0]},M)
-
q_{nT}
(\widehat\theta^{[1]},\widehat W^{[1]},M)
\leq
\left|
\left\langle
\nabla_{\operatorname{vec}(W)}
q_{nT}
(\widehat\theta^{[1]},\widehat W^{[0]},M),
\operatorname{vec}(H_{nT})
\right\rangle
\right|
\nonumber\\
&\leq
\left\|
\nabla_{\operatorname{vec}(W)}
q_{nT}
(\widehat\theta^{[1]},\widehat W^{[0]},M)
\right\|_2
\|H_{nT}\|_F.
\label{eq:q_decrease_bound}
\end{align}

On the other hand, the descent inequality
\eqref{eq:supervised_descent} implies
\begin{align}
&Q_n^0(\widehat L^{[1]},\widehat S^{[1]})
-
Q_n^0(\widehat L^{[0]},\widehat S^{[0]})
\leq
\xi_{nT}
\left[
q_{nT}
(\widehat\theta^{[1]},\widehat W^{[0]},M)
-
q_{nT}
(\widehat\theta^{[1]},\widehat W^{[1]},M)
\right].
\label{eq:Q0_descent_connection}
\end{align}
Combining
\eqref{eq:supervised_strong_curvature},
\eqref{eq:q_decrease_bound}, and
\eqref{eq:Q0_descent_connection}, %we obtain
$
\frac12\|H_{nT}\|_F^2
\leq
\xi_{nT}
\left\|
\nabla_{\operatorname{vec}(W)}
q_{nT}
(\widehat\theta^{[1]},\widehat W^{[0]},M)
\right\|_2
\|H_{nT}\|_F.
$
Hence,
\begin{equation}
\|H_{nT}\|_F
\leq
2\xi_{nT}
\left\|
\nabla_{\operatorname{vec}(W)}
q_{nT}
(\widehat\theta^{[1]},\widehat W^{[0]},M)
\right\|_2.
\label{eq:H_gradient_bound}
\end{equation}

By \eqref{eq:q_W_gradient},
\eqref{eq:first_theta_moment},
\eqref{eq:Xtilde_operator_bound}, and
$\|\widehat\Lambda_{nT}^{-1}\|_2=O_p(1)$,
\begin{align*}
&\left\|
\nabla_{\operatorname{vec}(W)}
q_{nT}
(\widehat\theta^{[1]},\widehat W^{[0]},M)
\right\|_2
\leq
2
\|
\overline{\widetilde X}_{nT}
(\widehat\theta^{[1]},M)
\|_2
\|
\widehat\Lambda_{nT}^{-1}
\|_2
\|
\bar g_{nT}
(\widehat\theta^{[1]},\widehat W^{[0]},M)
\|_2
=
O_p(a_{nT}).
\end{align*}
It follows from \eqref{eq:H_gradient_bound} that
\begin{equation}
\|
\widehat W^{[1]}-\widehat W^{[0]}
\|_F
=
O_p(\xi_{nT}a_{nT}).
\label{eq:supervised_W_stability}
\end{equation}
This stability bound is important because
$\widehat W^{[1]}$ depends on the outcome data and therefore is
not $\sigma(E)$-measurable. We do not apply
Lemma~\ref{rateG}(c) directly to $\widehat W^{[1]}$. Instead, we
apply that lemma only to the initial estimator
$\widehat W^{[0]}$ and control the additional supervised update
through \eqref{eq:supervised_W_stability}.The bound also implies the aggregate recovery inequality
$
\|
\widehat W^{[1]}-W_0
\|_F
\leq
\|
\widehat W^{[1]}-\widehat W^{[0]}
\|_F
+
\|
\widehat W^{[0]}-W_0
\|_F.
$
This inequality controls the aggregate adjacency estimator, but
does not by itself give separate recovery rates for
$\widehat L^{[1]}$ and $\widehat S^{[1]}$.
\noindent\textbf{Step d).}
Define 
$
\widehat G_s
:=
\frac{\partial
\bar g_{nT}(\theta,\widehat W^{[1]},M)}
{\partial\theta^\top}
$
and
$
\widehat G^{[0]}
:=
\frac{\partial
\bar g_{nT}(\theta,\widehat W^{[0]},M)}
{\partial\theta^\top}.
$
We first verify that the GMM Jacobian evaluated at
$\widehat W^{[1]}$ remains well conditioned. Because the
Jacobian is affine in $W$, the same design bounds used above give
\begin{equation}
\|
\widehat G_s-\widehat G^{[0]}
\|_2
\lesssim_p
\|
\widehat W^{[1]}-\widehat W^{[0]}
\|_F
=
O_p(\xi_{nT}a_{nT}).
\label{eq:supervised_Jacobian_Lipschitz}
\end{equation}
When $\xi_{nT}=O(1)$ and $a_{nT}=o(1)$, the right-hand side is
$o_p(1)$.
Moreover, the fixed-$M$ version of Lemma~\ref{rateG}(c), applied
only to the $\sigma(E)$-measurable estimator
$\widehat W^{[0]}$, and the oracle law of large numbers give
\begin{align}
\|
\widehat G^{[0]}-G_0(M)
\|_2
&\leq
\|
\widehat G^{[0]}-\widehat G(W_0,M)
\|_2
+
\|
\widehat G(W_0,M)-G_0(M)
\|_2
\nonumber\\
&=
O_p(\mathcal R^*_{nT})+o_p(1)
=
o_p(1).
\label{eq:initial_Jacobian_convergence}
\end{align}
Combining
\eqref{eq:supervised_Jacobian_Lipschitz} and
\eqref{eq:initial_Jacobian_convergence}, we obtain
$
\widehat G_s
=
G_0(M)+o_p(1).
$
By the identification condition,
$
\lambda_{\min}
\left(
G_0(M)^\top\Lambda^{-1}G_0(M)
\right)
>0.
$
Since
$\widehat\Lambda_{nT}^{-1}\to_p\Lambda^{-1}$, Weyl's inequality
implies that, for some constant $c>0$,
$
\lambda_{\min}
\left(
\widehat G_s^\top
\widehat\Lambda_{nT}^{-1}
\widehat G_s
\right)
\geq c+o_p(1).
$
Consequently,
\begin{equation}
\left\|
\left(
\widehat G_s^\top
\widehat\Lambda_{nT}^{-1}
\widehat G_s
\right)^{-1}
\widehat G_s^\top
\widehat\Lambda_{nT}^{-1}
\right\|_2
=
O_p(1).
\label{eq:supervised_GMM_multiplier}
\end{equation}
As the moment function is affine in $\theta$,
$
\bar g_{nT}
(\widehat\theta_s,\widehat W^{[1]},M)
=
\bar g_{nT}
(\theta_0,\widehat W^{[1]},M)
+
\widehat G_s
(\widehat\theta_s-\theta_0).
$
The first-order condition for the final GMM step is
$
\widehat G_s^\top
\widehat\Lambda_{nT}^{-1}
\bar g_{nT}
(\widehat\theta_s,\widehat W^{[1]},M)
=
0.
$
Therefore,
\begin{align}
\widehat\theta_s-\theta_0
=
-&
\left(
\widehat G_s^\top
\widehat\Lambda_{nT}^{-1}
\widehat G_s
\right)^{-1}
\widehat G_s^\top
\widehat\Lambda_{nT}^{-1}
\bar g_{nT}
(\theta_0,\widehat W^{[1]},M).
\label{eq:supervised_GMM_expansion}
\end{align}
Using \eqref{eq:supervised_moment_affine},
\begin{align}
&\bar g_{nT}
(\theta_0,\widehat W^{[1]},M)
-
\bar g_{nT}
(\theta_0,\widehat W^{[0]},M)
%\nonumber\\
%&\qquad=
= -\overline{\widetilde X}_{nT}(\theta_0,M)^\top
\operatorname{vec}
(\widehat W^{[1]}-\widehat W^{[0]}).
\label{eq:final_moment_difference}
\end{align}
It follows that
\begin{align*}
\|
\bar g_{nT}
(\theta_0,\widehat W^{[1]},M)
\|_2
&\leq
\|
\bar g_{nT}
(\theta_0,\widehat W^{[0]},M)
\|_2
%\\
%&\quad+
+\|
\overline{\widetilde X}_{nT}(\theta_0,M)
\|_2
\|
\widehat W^{[1]}-\widehat W^{[0]}
\|_F
%\\
%&=
=O_p(a_{nT})
+
O_p(\xi_{nT}a_{nT}).
\end{align*}
Combining this bound with
\eqref{eq:supervised_GMM_expansion} and
\eqref{eq:supervised_GMM_multiplier}, we obtain
$
\|
\widehat\theta_s-\theta_0
\|_2
=
O_p(a_{nT})
+
O_p(\xi_{nT}a_{nT}).
$
If $\xi_{nT}=O(1)$, it follows that
\begin{equation}
\|
\widehat\theta_s-\theta_0
\|_2
=
O_p\left(
\frac{1}{\sqrt{nT}}+\mathcal R^*_{nT}
\right).
\label{eq:supervised_final_rate}
\end{equation}
This proves the rate in~\eqref{superv8}.
Finally, we establish the asymptotic distribution. Define
$
\widehat\Psi(W)
:=
\left[
\widehat G(W)^\top
\widehat\Lambda_{nT}^{-1}
\widehat G(W)
\right]^{-1}
\widehat G(W)^\top
\widehat\Lambda_{nT}^{-1}.
$
The GMM expansions for $\widehat\theta_s$ and
$\widehat\theta^{[1]}$ give
\begin{align*}
\widehat\theta_s-\widehat\theta^{[1]}
={}&
-
\left[
\widehat\Psi(\widehat W^{[1]})
-
\widehat\Psi(\widehat W^{[0]})
\right]
\bar g_{nT}
(\theta_0,\widehat W^{[0]},M)
%\\
%&-
-\widehat\Psi(\widehat W^{[1]})
\left[
\bar g_{nT}
(\theta_0,\widehat W^{[1]},M)
-
\bar g_{nT}
(\theta_0,\widehat W^{[0]},M)
\right].
\end{align*}
The Jacobian perturbation bound and the inverse identity imply
\[
\|
\widehat\Psi(\widehat W^{[1]})
-
\widehat\Psi(\widehat W^{[0]})
\|_2
=
O_p\left(
\|
\widehat W^{[1]}-\widehat W^{[0]}
\|_F
\right).
\]
Together with
\eqref{eq:first_theta_moment},
\eqref{eq:final_moment_difference}, and
\eqref{eq:supervised_W_stability}, this yields
\begin{align*}
\|
\widehat\theta_s-\widehat\theta^{[1]}
\|_2
&=
O_p\left(
a_{nT}
\|
\widehat W^{[1]}-\widehat W^{[0]}
\|_F
\right)
+
O_p\left(
\|
\widehat W^{[1]}-\widehat W^{[0]}
\|_F
\right)
=
O_p(\xi_{nT}a_{nT}).
\end{align*}

Suppose, in addition, that
$
\sqrt{nT}\,\mathcal R^*_{nT}=o(1)$, and 
%\qquad\text{and}\qquad
$\xi_{nT}=o(1).
$
Then
\begin{align*}
\sqrt{nT}\,\xi_{nT}a_{nT}
&=
\xi_{nT}
\left(
1+\sqrt{nT}\,\mathcal R^*_{nT}
\right)
=o(1).
\end{align*}
Consequently,
$
\sqrt{nT}
\left(
\widehat\theta_s-\widehat\theta^{[1]}
\right)
=o_p(1).
$

The asymptotic distribution of $\widehat\theta_s$ is therefore the
same as that of the initial plug-in estimator
$\widehat\theta^{[1]}$ established in
Theorem~\ref{Thm:DK22}. Thus, supervision may affect finite-sample
performance, but under $\xi_{nT}=o(1)$ it is asymptotically
negligible at the $\sqrt{nT}$ scale and does not alter the
first-order limiting distribution. This completes the proof.
\end{proof}

\newpage
    \section{Technical Notes}
    
\subsection{Proof of Perturbation Lemma}\label{perb}
Justification of Equation (\ref{rhoLS}):\\
\noindent\textbf{Assumption
(two-dimensional spectral concentration and separation).}
For each $n$, let $\bL_0$ and $\bS_0$ be real symmetric matrices
with normalized leading eigenvectors $\bvL$ and $\bvS$ and
corresponding simple leading eigenvalues
$\lambda_{\bL_0}$ and $\lambda_{\bS_0}$. Let
\[
\rho_n \coloneqq \bvL^\top\bvS,
\qquad |\rho_n|<1,
\qquad \rho_n\to0.
\]
Let $\bvW$ be the normalized leading eigenvector of
$W_0=\bL_0+\bS_0$ and write
\[
\bvW=k_1\bvL+k_2\bvS+k_3e,
\qquad
e\perp\operatorname{span}\{\bvL,\bvS\},
\qquad \|e\|_2=1.
\]
Assume
$
\|\bL_0\|_2+\|\bS_0\|_2\le C$,
$|k_1|+|k_2|\asymp1$,
$ |k_3|=o(|\rho_n|).
$
Define
$
\widetilde a_n
  \coloneqq
  \lambda_{\bL_0}-\bvS^\top\bL_0\bvS$,
and $
\widetilde c_n
  \coloneqq
  \lambda_{\bS_0}-\bvL^\top\bS_0\bvL.
$
Assume that, for some constants $g_0,\kappa>0$, one of the
following two cases holds eventually:
\begin{enumerate}[(i)]
\item $\widetilde a_n-\widetilde c_n\ge g_0$ and
      $|k_1|\ge\kappa$;
\item $\widetilde c_n-\widetilde a_n\ge g_0$ and
      $|k_2|\ge\kappa$.
\end{enumerate}
\begin{lemma}[Perturbation of the leading eigenvector]\label{lem:perturb}
   Under the two-dimensional spectral concentration and separation 
   conditions stated above, Equations~\eqref{rhoLS},~\eqref{eq:kSL1} and \ref{eq:kSL2} hold; in case~(i), 
   $|\bvW^\top\bvL| > |\bvW^\top\bvS|$ for all sufficiently 
   large $n$, and in case~(ii) the reverse inequality holds.
   \end{lemma}

\begin{proof}
Define
$\rho_{W_0\bL_0}
\coloneqq
\bvW^\top\bvL
=
k_1+k_2\rho_n,
\qquad
\rho_{W_0\bS_0}
\coloneqq
\bvW^\top\bvS
=
k_1\rho_n+k_2.
$
Note that
$\rho_n=\rho_{\bL_0\bS_0}=\bvL^\top\bvS.
$
Moreover, $\bvL$ and $\bvS$ are normalized eigenvectors. Hence
\begin{equation*}
\begin{pmatrix}
k_1\\
k_2
\end{pmatrix}
=
\begin{pmatrix}
1 & \rho_n\\
\rho_n & 1
\end{pmatrix}^{-1}
\begin{pmatrix}
\rho_{W_0\bL_0}\\
\rho_{W_0\bS_0}
\end{pmatrix}
=
\begin{pmatrix}
\dfrac{\rho_{W_0\bL_0}-\rho_n\rho_{W_0\bS_0}}
      {1-\rho_n^2}\\[3mm]
\dfrac{\rho_{W_0\bS_0}-\rho_n\rho_{W_0\bL_0}}
      {1-\rho_n^2}
\end{pmatrix},
\end{equation*}
and Equation~(\ref{rhoLS}) holds.

We now prove Equation~(\ref{eq:kSL1}) in the main text.

Recall the decomposition
\[
\bvW=k_1\bvL+k_2\bvS+k_3e,
\]
where $\|e\|_2=1$ and $e$ is orthogonal to both $\bvL$ and $\bvS$.
Then
\begin{align}
(\bL_0+\bS_0)\bvW
={}&
k_1\lambda_{\bL_0}\bvL
+k_2\bL_0\bvS
+k_1\bS_0\bvL
+k_2\lambda_{\bS_0}\bvS
+r_n,
\label{eqn:L+S}
\end{align}
where
$r_n\coloneqq k_3(\bL_0+\bS_0)e.
$
Because $\|\bL_0\|_2+\|\bS_0\|_2\leq C$, we have 
$\|r_n\|_2\leq C|k_3|.$

Let $
r_{Ln}\coloneqq\bvL^\top r_n$ and $
r_{Sn}\coloneqq\bvS^\top r_n.
$
Then
$
|r_{Ln}|+|r_{Sn}|=O(|k_3|).
$
Multiplying Equation~(\ref{eqn:L+S}) by $\bvL^\top$ and using the
symmetry of $\bL_0$ gives
\begin{align*}
&k_1\lambda_{\bL_0}
+k_2\lambda_{\bL_0}\rho_n
+k_1\bvL^\top\bS_0\bvL
+k_2\lambda_{\bS_0}\rho_n
+r_{Ln}%\\
%&\hspace{3cm}
=\lambda_{W_0}(k_1+k_2\rho_n).
\end{align*}
Similarly, multiplying by $\bvS^\top$ and using the symmetry of
$\bS_0$ gives
\begin{align*}
&k_1\lambda_{\bL_0}\rho_n
+k_2\bvS^\top\bL_0\bvS
+k_1\lambda_{\bS_0}\rho_n
+k_2\lambda_{\bS_0}
+r_{Sn}%\\
%&\hspace{3cm}
=\lambda_{W_0}(k_1\rho_n+k_2).
\end{align*}

Recall that
$
\widetilde a_n
=
\lambda_{\bL_0}-\bvS^\top\bL_0\bvS,
\qquad
\widetilde c_n
=
\lambda_{\bS_0}-\bvL^\top\bS_0\bvL.
$ 
Since $\lambda_{\bL_0}$ and $\lambda_{\bS_0}$ are the largest
eigenvalues of the symmetric matrices $\bL_0$ and $\bS_0$, 
respectively, we have 
$
\widetilde a_n\geq0$, and $
\widetilde c_n\geq0.
$

Consider case~(i), so that
$
\widetilde a_n-\widetilde c_n\geq g_0>0$, and $|k_1|\geq\kappa>0.
$
Set
$
q_n\coloneqq\frac{k_2}{k_1}.
$
Cross-multiplying the two projected eigenvalue equations and retaining
the remainder gives
\[
\widetilde a_n\rho_nq_n^2
+
(\widetilde a_n-\widetilde c_n)q_n
-
\widetilde c_n\rho_n
=
O(|k_3|).
\]
Let
$
g_n\coloneqq\widetilde a_n-\widetilde c_n.
$
The sequence $q_n$ is bounded because $|k_1|\geq\kappa$ and
$|k_1|+|k_2|\asymp1$. Moreover,
\[
2\widetilde a_n\rho_nq_n+g_n
=
g_n+o(1)
\geq \frac{g_0}{2}
\]
for all sufficiently large $n$. Hence the $O(|k_3|)$ perturbation of
the quadratic produces an $O(|k_3|)$ perturbation of its bounded root.
Therefore,
\begin{align*}
\frac{k_2}{k_1}
&=
\frac{2\widetilde c_n\rho_n}
{g_n+\sqrt{g_n^2+
4\widetilde a_n\widetilde c_n\rho_n^2}}
+O(|k_3|)% \\
%&=
=
\frac{\widetilde c_n\rho_n}{g_n}
+O(\rho_n^3)+O(|k_3|)\\
&=
\frac{
(\lambda_{\bS_0}-\bvL^\top\bS_0\bvL)\rho_n}
{
\lambda_{\bL_0}-\bvS^\top\bL_0\bvS
-
(\lambda_{\bS_0}-\bvL^\top\bS_0\bvL)}
+o(|\rho_n|),
\end{align*}
where the last equality uses $\rho_n\to0$ and
$|k_3|=o(|\rho_n|)$.

Since $q_n=O(\rho_n)$,
$
|\bvW^\top\bvL|
=
|k_1+k_2\rho_n|
=
|k_1|\{1+o(1)\},
$
whereas
$
|\bvW^\top\bvS|
=
|k_1\rho_n+k_2|
=
O(|\rho_n|).
$
Thus, for all sufficiently large $n$,
$
|\bvW^\top\bvL|>|\bvW^\top\bvS|.
$

Now consider case~(ii), so that
$\widetilde c_n-\widetilde a_n\geq g_0>0$,
$|k_2|\geq\kappa>0.$ 
Set
$
p_n\coloneqq\frac{k_1}{k_2}.
$
Cross-multiplying the projected eigenvalue equations gives
$
\widetilde c_n\rho_np_n^2
+
(\widetilde c_n-\widetilde a_n)p_n
-
\widetilde a_n\rho_n
=
O(|k_3|).
$
Let
$
h_n\coloneqq\widetilde c_n-\widetilde a_n.
$ As in case~(i), $p_n$ is bounded and
$
2\widetilde c_n\rho_np_n+h_n
=
h_n+o(1)
\geq\frac{g_0}{2}
$
for all sufficiently large $n$. Therefore,
\begin{align*}
\frac{k_1}{k_2}
&=
\frac{2\widetilde a_n\rho_n}
{h_n+\sqrt{h_n^2+
4\widetilde a_n\widetilde c_n\rho_n^2}}
+O(|k_3|)=
\frac{\widetilde a_n\rho_n}
{\widetilde c_n-\widetilde a_n}
+O(\rho_n^3)+O(|k_3|)%\\
%&=
=\frac{\widetilde a_n\rho_n}
{\widetilde c_n-\widetilde a_n}
+o(|\rho_n|).
\end{align*}

Since $p_n=O(\rho_n)$,
$
|\bvW^\top\bvS|
=
|k_2+k_1\rho_n|
=
|k_2|\{1+o(1)\},
$ 
whereas
$
|\bvW^\top\bvL|
=
|k_2\rho_n+k_1|
=
O(|\rho_n|).
$
Consequently, for all sufficiently large $n$,
$
|\bvW^\top\bvS|>|\bvW^\top\bvL|.
$
\end{proof}
The perturbation result above assumes that both $\bL_0$ and
$\bS_0$ are symmetric; the asymmetric case is outside its scope.

\subsection{Some Useful Lemmas}

We first present some useful lemmas for the proof of the theorem.
The notation within this section does not necessarily align the main text and the proofs of the theorem since they are independent results. %Assume that the rank of the matrix \( A \) is \( r \).
%{Denote that} $ \mu_u(A) \coloneqq \max_{1\leq i\leq n} \left\{\frac{1}{r} \sum_{j=1}^r u_{ij}^2(A) \right\} \leq \max_{1\leq i\leq n}\|u_i(A)\|_{\max}$
%and 
%$\mu_v(A)\coloneqq \max_{1\leq i\leq n} \left\{\frac{1}{r} \sum_{j=1}^r v_{ij}^2(A) \right\} \leq \max_{1\leq i\leq n}\|v_i(A)\|_{\max}$. Then 
%Also since have:

To summarise the steps we shall first give a few definitions.
Let $L_0 \in \mathbb{R}^{n \times n}$ be a low rank matrix with rank $r$.
Let $\mathcal{T}(L_0)$ denote the span of all matrices whose row 
space is contained in the row space of $L_0$ or whose column space 
is contained in the column space of $L_0$. Suppose  the singular 
value decomposition is $L_0 = U\Sigma V^\top$ with 
$U, V \in \mathbb{R}^{n\times r}$. % and $\operatorname{rank}(L_0) = r$, 
Then this space is
$$
\mathcal{T}(L_0)=\left\{U X^{\top}+Y V^{\top} \mid X, Y \in 
\mathbb{R}^{n \times r}\right\}.
$$
%Analogously, % to $\deg_{\max}(A)$, 
Define 
$\deg_{\min}(A) \coloneqq \min\Big\{\min_{1\leq i\leq n}\sum_{j=1}^n 
\I(A_{ij}\neq 0),\ \min_{1\leq j\leq n}\sum_{i=1}^n \I(A_{ij}\neq 0)\Big\},
$
the minimum number of nonzero entries in any row or column of $A$.
 For any matrix $S_0 \in \mathbb{R}^{n \times n}$, the  space $\Omega(S_0)$ is defined by,
$
\Omega(S_0)=\left\{N \in  \mathbb{R}^{n \times n} \mid \operatorname{support}(N) \subseteq \operatorname{support}(S_0)\right\}.
$

Given $\mathcal{T}(L_0)$ and $\Omega(S_0)$, we define two measures that quantify the denseness 
and sparsity levels of a matrix.
$\xi(L_0) = \max_{N \in \mathcal{T}(L_0), \|N\|_2 \leq 1}\|N\|_{\max}$, 
$\mu(S_0) =\max _{N \in \Omega(S_0),\|N\|_{\max} \leq 1}\|N\|_2$ .

\begin{lemma}(\cite{Chan2009})\label{ineq}
 Let $A \in \mathbb{R}^{n \times n}$ be any matrix. We have that
$$\operatorname{inc}(A) \leq \xi(A) \leq 2 \operatorname{inc}(A), \qquad 
\operatorname{deg}_{\min }(A) \leq \mu(A) \leq \operatorname{deg}_{\max }(A), \qquad 
%\begin{equation*}
\xi(A) \mu(A) \geq 1.$$
%\end{equation*}
\end{lemma}

Let $\Omega$ denote the parameter space for $\theta$, and let
$\mathcal L$ denote the population cost function. Define the target
parameter by
$
\theta^*
\in
\arg\min_{\theta\in\Omega}\mathcal L(\theta).
$ Based on the observed sample
$
Z_1^n=\{Z_1,\ldots,Z_n\},
$
we estimate $\theta^*$ by solving
$
\widehat{\theta}
\in
\arg\min_{\theta\in\Omega}
\left\{
\mathcal L_n(\theta;Z_1^n)
+\lambda_n\Phi(\theta)
\right\},
$
where $\mathcal L_n$ is the empirical cost function,
$\lambda_n>0$ is a regularization parameter, and
$\Phi:\Omega\rightarrow\mathbb R_+$ is a norm. Let $\Phi^*$ denote
the dual norm of $\Phi$.

Let $\mathbb M$ and $\overline{\mathbb M}$ be subspaces satisfying
$
\mathbb M\subseteq\overline{\mathbb M},
$
and denote their orthogonal complements by $\mathbb M^\perp$ and
$\overline{\mathbb M}^{\perp}$, respectively. For any subspace $\mathbb S$, let $\theta_{\mathbb S}$ denote the
orthogonal projection of $\theta$ onto $\mathbb S$. In particular,
$\theta_{\mathbb M^\perp}^*$ denotes the component of $\theta^*$ in
$\mathbb M^\perp$.

Our bound involves the quantity
\[
\begin{aligned}
\varepsilon_n^2
\left(\overline{\mathbb M},\mathbb M^\perp\right)
\coloneqq{}&
\underbrace{
9\frac{\lambda_n^2}{\kappa^2}
\Psi^2\left(\overline{\mathbb M}\right)
}_{\text{estimation error}}
%\\
%&+
+\underbrace{
\frac{8}{\kappa}
\left\{
\lambda_n\Phi\left(\theta_{\mathbb M^\perp}^*\right)
+
16\tau_n^2
\Phi^2\left(\theta_{\mathbb M^\perp}^*\right)
\right\}
}_{\text{approximation error}},
\end{aligned}
\]
which depends on the choice of the subspace pair
$\left(\overline{\mathbb M},\mathbb M^\perp\right)$.
\begin{assumption}\label{conditionwainright}
\begin{itemize}
\item[i)]
The cost function $\mathcal L_n$ is convex and satisfies the local
restricted strong convexity condition with curvature $\kappa>0$,
radius $R>0$, and tolerance $\tau_n^2\geq0$. That is,
$
\mathcal E_n(\Delta)
\coloneqq
\mathcal L_n(\theta^*+\Delta)
-\mathcal L_n(\theta^*)
-\left\langle
\nabla\mathcal L_n(\theta^*),\Delta
\right\rangle
\geq
\frac{\kappa}{2}\|\Delta\|^2
-\tau_n^2\Phi^2(\Delta)
$
for every
$
\Delta\in\mathbb B(R)
\coloneqq
\left\{
\Delta:
\theta^*+\Delta\in\Omega,\ 
\|\Delta\|\leq R
\right\}.
$

\item[ii)]
There exists a pair of subspaces
$
\left(\mathbb M,\overline{\mathbb M}^{\perp}\right)$,
$
\mathbb M\subseteq\overline{\mathbb M},
$
such that the regularizer $\Phi$ is decomposable with respect to
$\left(\mathbb M,\overline{\mathbb M}^{\perp}\right)$; that is,
$
\Phi(u+v)=\Phi(u)+\Phi(v)
$ for every $u\in\mathbb M$ and
$v\in\overline{\mathbb M}^{\perp}$.

\item[iii)]
Define the good event
$
\mathbb G(\lambda_n)
\coloneqq
\left\{
\Phi^*\left(\nabla\mathcal L_n(\theta^*)\right)
\leq\frac{\lambda_n}{2}
\right\},
\qquad \lambda_n\geq0.
$

\item[iv)]
For a given norm $\|\cdot\|$ and subspace $\mathbb S$, define the
subspace Lipschitz constant by
\[
\Psi(\mathbb S)
\coloneqq
\sup_{u\in\mathbb S\setminus\{\mathbf 0\}}
\frac{\Phi(u)}{\|u\|}.
\]

\end{itemize}
\end{assumption}

The following are two important lemmas that we use to prove
Theorem~\ref{Thm:DK1}.

\begin{lemma}\label{theorem919}
{\rm (Theorem 9.19 in \cite{wainwright2019high}).}
Suppose Assumption~\ref{conditionwainright} holds. Conditional on
the event $\mathbb{G}(\lambda_n)$, the following statements hold:

\begin{enumerate}[(a)]
\item Any optimal solution satisfies
$
\Phi\left(\widehat{\theta}-\theta^*\right)
\leq
4\left\{
\Psi\left(\overline{\mathbb{M}}\right)
\left\|\widehat{\theta}-\theta^*\right\|_2
+
\Phi\left(\theta_{\mathbb{M}^{\perp}}^*\right)
\right\}.
$

\item For any subspace pair
$\left(\overline{\mathbb{M}},\mathbb{M}^{\perp}\right)$ satisfying
$\tau_n^2\Psi^2\left(\overline{\mathbb{M}}\right)
\leq \frac{\kappa}{64}$ and $
\varepsilon_n\left(
\overline{\mathbb{M}},\mathbb{M}^{\perp}
\right)
\leq R,
$
we have
\[
\left\|\widehat{\theta}-\theta^*\right\|_2^2
\leq
\varepsilon_n^2\left(
\overline{\mathbb{M}},\mathbb{M}^{\perp}
\right).
\]
\end{enumerate}
\end{lemma}

\begin{assumption}\label{conditionwainrightPhi}
The cost function satisfies the $\Phi^*$-norm curvature condition
with curvature $\kappa_\Phi>0$, tolerance $\tau_{n,\Phi}\geq0$,
and radius $R_\Phi>0$:
\[
\Phi^*\left(
\nabla\mathcal L_n(\theta^*+\Delta)
-\nabla\mathcal L_n(\theta^*)
\right)
\geq
\kappa_\Phi\Phi^*(\Delta)
-\tau_{n,\Phi}\Phi(\Delta)
\]
for every
$
\Delta\in\mathbb B_{\Phi^*}(R_\Phi)
\coloneqq
\left\{
\Delta\in\Omega:
\Phi^*(\Delta)\leq R_\Phi
\right\}.
$
\end{assumption}

\begin{lemma}[Theorem 9.24 in \cite{wainwright2019high}]
\label{thm924}
Suppose $\theta^*\in\mathbb M$, the decomposability condition in
Assumption~\ref{conditionwainright}(ii) holds, and
Assumption~\ref{conditionwainrightPhi} holds. Suppose also that
$
\tau_{n,\Phi}\Psi^2(\overline{\mathbb M})
<\frac{\kappa_\Phi}{32}.
$
Conditional on
$
\mathbb G(\lambda_n)
\cap
\left\{
\Phi^*(\widehat\theta-\theta^*)\leq R_\Phi
\right\},
$
any optimal solution satisfies
\[
\Phi^*(\widehat\theta-\theta^*)
\leq
\frac{3\lambda_n}{\kappa_\Phi}.
\]
\end{lemma}

\begin{lemma}[Conditional Lindeberg--Feller central limit theorem]
\label{feller}
For each $n$, let $\mathcal F_n$ be a sigma-field and let
$Y_{n,1},\ldots,Y_{n,k_n}\in\mathbb R^m$, with $m$ fixed,
be conditionally independent given $\mathcal F_n$.
Assume finite second moments and
$\mathbb E[Y_{n,i}\mid\mathcal F_n]=0$.
Suppose that, for every $\varepsilon>0$,
$
\sum_{i=1}^{k_n}
\mathbb E\left[
\|Y_{n,i}\|_2^2
\mathbf 1\{\|Y_{n,i}\|_2>\varepsilon\}
\,\middle|\,\mathcal F_n
\right]
\to_p0
$
and
$
\sum_{i=1}^{k_n}
\operatorname{Cov}(Y_{n,i}\mid\mathcal F_n)
\to_p\Sigma,
$
where $\Sigma$ is a deterministic positive semidefinite matrix.
Then
\[
\sum_{i=1}^{k_n}Y_{n,i}
\overset{d}{\longrightarrow}\mathcal N(0,\Sigma).
\]
\end{lemma}
\begin{lemma}\label{propbound}
Let $A,B\in\mathbb R^{n\times n}$. Let
$\mathcal C_1\subseteq\{1,\ldots,n\}$ satisfy
$|\mathcal C_1|\leq K$, where $K$ is fixed, and define
$
\mathcal C_2
=
\{1,\ldots,n\}\setminus\mathcal C_1.
$
Let
$
c_n
=
\max_{k\in\mathcal C_1}
\sum_{j=1}^n|B_{kj}|$, and $
c
=
\max_{k\in\mathcal C_2}
\sum_{j=1}^n|B_{kj}|.
$
Then
\[
\|AB\|_\infty
\leq
K\|A\|_{\max}c_n+\|A\|_\infty c.
\]
Here, the maximum over an empty set is defined to be zero.
\end{lemma}

% \begin{lemma}\label{lem7}
% Suppose $\lambda_{\max} (\lambda W_0)<1$. Prove that $\|(I - \lambda W_0)^{-1}\|_2$ is uniformly bounded.
% \end{lemma}
% \begin{proof}
% Since $\lambda_{\max} (\lambda W_0)<1$, $(I - \lambda W_0)^{-1}=I_n+\lambda W_0+\lambda^2 W_0^2+\cdots $. Utilizing the fact that $\|AB\|_2\leq \|A\|_2\|B\|_2$ and $\|A+B\|_2\leq \|A\|_2+\|B\|_2$,  we have 
% $\|(I - \lambda W_0)^{-1}\|_2=\|I_n+\lambda W_0+\lambda^2 W_0^2+\cdots\|\leq 1/(1-\lambda \|W_0\|_2)$. 
% \end{proof}

\begin{lemma}\label{lem8}
Suppose Assumptions~\ref{weight} and
\ref{Assump:True_Adjacency}(iii) hold, $J$ is fixed, and
$
\|W_0\|_1=O(\sqrt{n}).
$
Assume that there exist constants $c_\infty,c_1>0$, independent
of $n$, such that
\[
\sup_{\lambda\in\mathcal C}
|\lambda|\|W_0\|_\infty\leq 1-c_\infty,
\qquad
\sup_{\lambda\in\mathcal C}
|\lambda|\|W_{0,b}\|_1\leq 1-c_1.
\]
Then, uniformly over $\lambda\in\mathcal C$,
\[
\big\|(I_n-\lambda W_0)^{-1}\big\|_2=O(n^{1/4}).
\]
\end{lemma}

\begin{lemma}\label{lem9}
Suppose Assumptions~\ref{Assump:True_Adjacency}(iii) and
\ref{Assump:Instruments}(i) hold.
Let $X_{jt}$ be the $j$th column of $X_t$. Then, for
$\ell=0,1,2$,
\[
\|W_0^\ell X_{jt}\|_2=O_p(\sqrt n),
\qquad
\big\|W_0^\ell(I_n-\lambda_0W_0)^{-1}X_{jt}\big\|_2
=O_p(\sqrt n).
\]
\end{lemma}

\subsection{Proof of Lemma \ref{identification}, \ref{propbound}-\ref{lem9}.}

\begin{proof}

% To identify $L_0$ and $S_0$, we compare their element sizes and sparsity levels directly. According to Assumption \ref{Assump:W_structure}, each element of $S_0$ has a constant magnitude, while the maximum value of $L_0$ diminishes to zero. Additionally, $L_0$ is dense in both rows and columns because $m_c(L_0) Asymp m_r(L_0) \to \infty$. In contrast, $S_0$ consists of elements of constant order, with sparsity constrained by the finite constant $m_r(S_0)$.
%If we consider two possible decompositions, $W_0 = L_0 + S_0$ and $W_0 = \tilde{L}_0 + \tilde{S}_0$, it becomes impossible for $\operatorname{rank}(L_0) \neq \operatorname{rank}(\tilde{L}_0)$. Additionally, we have $| S_0 - \tilde{S}_0 |_{\max} = o(1)$.
It follows from Corollary 3 in  \cite{Chan2009}. 

%The Theorem ensures the identification of $L_0$ and $S_0$ which means that sparse matrix $S_0$ need to be not too dense in eigenvalues and low rank matrix $L_0$ can not be too sparse.

%The incoherence measure can be defined as,
%$\operatorname{inc}(M) = \max \left[\max _i\left\|P_U e_i\right\|, \max _i\left\|P_V e_i\right\|\right] .$
%We can define $P_U  = U(U^{\top}U)^{-1}U^{\top}$ and same as $P_V$.

$\mu(S_0)$ can be bounded by $\operatorname{deg}_{\max }(S_0)$
and $\xi(L_0)$ can be bounded by its incoherence measure $2\operatorname{inc}(L_0)$.
The identification condition is $ \xi(L_0) \mu(S_0)<1/8$.
A stronger condition can be  $\operatorname{deg}_{\max }(S_0)\operatorname{inc}(L_0)\leq 1/16$, due to Lemma \ref{ineq}.
%{\color{red} Thus we have if $\mbox{inc}(L_0)\to 0 $,  $\operatorname{deg}_{\max }(L_0)\geq \mbox{inc}(L_0)^{-1}\to \infty$. }
\end{proof}

\textbf{Proof of Lemma \ref{propbound}.}
\begin{proof}
Recall that
$
\norminf{AB}
=
\max_{1\leq i\leq n}\sum_{j=1}^n|(AB)_{ij}|$,
$\norminf{A}
=
\max_{1\leq i\leq n}\sum_{k=1}^n|A_{ik}|,
$
and
$
\normmax{A}
=
\max_{1\leq i,k\leq n}|A_{ik}|.
$
For every $i$,
\begin{align*}
\sum_j |(AB)_{ij}|
&\leq
\sum_j\sum_{k\in\mathcal C_1}|A_{ik}B_{kj}|
+
\sum_j\sum_{k\in\mathcal C_2}|A_{ik}B_{kj}|\\
&\leq
\max_{k\in\mathcal C_1}|A_{ik}|
\sum_{k\in\mathcal C_1}\sum_j|B_{kj}|
+
\sum_{k\in\mathcal C_2}|A_{ik}|
\max_{k\in\mathcal C_2}\sum_j|B_{kj}|\leq
K\normmax{A}c_n+\norminf{A}c.
\end{align*}
Taking the maximum over $i$ gives the stated result
$
\norminf{AB}
\leq
K\normmax{A}c_n+\norminf{A}c.
$
%which is the stated result.
\end{proof}

\textbf{Proof of Lemma~\ref{lem8}.}
\begin{proof}
Let
$
B_n(\lambda)=(I_n-\lambda W_{0,b})^{-1}.
$
Since zeroing columns cannot increase the maximum absolute row sum,
$
\|W_{0,b}\|_\infty\leq\|W_0\|_\infty.
$
Therefore, the two stability conditions and the Neumann-series
argument give
$
\sup_{\lambda\in\mathcal C}\|B_n(\lambda)\|_1
\leq c_1^{-1}$,$ 
%\qquad
\sup_{\lambda\in\mathcal C}\|B_n(\lambda)\|_\infty
\leq c_\infty^{-1}.
$
Consequently,
\[
\sup_{\lambda\in\mathcal C}\|B_n(\lambda)\|_2
\leq
\sup_{\lambda\in\mathcal C}
\sqrt{\|B_n(\lambda)\|_1\|B_n(\lambda)\|_\infty}
=O(1).
\]

If $J=0$, then $W_{0,u}=0$ and the result follows immediately.
Hence, suppose $J\geq1$. Let
$
\Pi_J=(e_1,\ldots,e_J)\in\mathbb R^{n\times J}
$
be the matrix selecting the first $J$ coordinates, and define
$
A_J:=W_{0,u}\Pi_J\in\mathbb R^{n\times J}.
$
Since $W_{0,u}$ vanishes outside its first $J$ columns,
$
W_{0,u}=A_J\Pi_J^\top.
$
Set
$
U=B_n(\lambda)A_J$, 
$V=\Pi_J$, and $
Q_n(\lambda)=I_n-\lambda UV^\top.
$
Then
$
I_n-\lambda W_0
=
(I_n-\lambda W_{0,b})Q_n(\lambda).
$

Let
$
R_n(\lambda)=(I_n-\lambda W_0)^{-1}.
$
The first stability condition and the Neumann-series argument give
$
\sup_{\lambda\in\mathcal C}
\|R_n(\lambda)\|_\infty
\leq c_\infty^{-1}.
$
Both $I_n-\lambda W_0$ and $I_n-\lambda W_{0,b}$ are invertible.
Hence, $Q_n(\lambda)$ is invertible and
$
Q_n(\lambda)^{-1}
%= (I_n-\lambda W_{0,b})R_n(\lambda).
=R_n(\lambda)(I_n-\lambda W_{0,b}).
$
Therefore,
\begin{align*}
\sup_{\lambda\in\mathcal C}
\|Q_n(\lambda)^{-1}\|_\infty
&\leq
\sup_{\lambda\in\mathcal C}
\|R_n(\lambda)\|_\infty
\bigl(1+|\lambda|\|W_{0,b}\|_\infty\bigr)\leq
\frac{2-c_\infty}{c_\infty}
=O(1).
\end{align*}

Define
$
M_J(\lambda)=I_J-\lambda V^\top U.
$
Since
$
V^\top Q_n(\lambda)=M_J(\lambda)V^\top,
$
right-multiplication by $Q_n(\lambda)^{-1}V$ gives
$
M_J(\lambda)
\bigl[V^\top Q_n(\lambda)^{-1}V\bigr]
=I_J.
$
Thus, $M_J(\lambda)$ is invertible and
$
M_J(\lambda)^{-1}
=
V^\top Q_n(\lambda)^{-1}V.
$
Because $J$ is fixed,
$
\sup_{\lambda\in\mathcal C}
\|M_J(\lambda)^{-1}\|_2
\leq
J\sup_{\lambda\in\mathcal C}
\|Q_n(\lambda)^{-1}\|_\infty
=O(1).
$
The Woodbury formula may therefore be applied:
$
Q_n(\lambda)^{-1}
=
I_n+\lambda U M_J(\lambda)^{-1}V^\top.
$

Since $\mathcal C$ is compact,
$\sup_{\lambda\in\mathcal C}|\lambda|<\infty$. Moreover,
\begin{align*}
\|U\|_2\|V^\top\|_2
&\leq
\|B_n(\lambda)\|_2\|W_{0,u}\|_2\leq
O(1)
\sqrt{\|W_{0,u}\|_1\|W_{0,u}\|_\infty}=O(n^{1/4}),
\end{align*}
because
$\|W_{0,u}\|_1\leq\|W_0\|_1=O(\sqrt n)$, and $
\|W_{0,u}\|_\infty\leq\|W_0\|_\infty=O(1).
$
Therefore, the Woodbury representation gives
$
\sup_{\lambda\in\mathcal C}
\|Q_n(\lambda)^{-1}\|_2=O(n^{1/4}).
$
Finally,
$
R_n(\lambda)
=
Q_n(\lambda)^{-1}B_n(\lambda),
$
and hence
$
\sup_{\lambda\in\mathcal C}
\|(I_n-\lambda W_0)^{-1}\|_2
=O(n^{1/4}),
$
as required.
\end{proof}
\textbf{Proof of Lemma~\ref{lem9}.}
\begin{proof}
Fix $j$ and  write $X=X_{jt}$. For $\ell=0,1,2$,
Assumption~\ref{Assump:True_Adjacency}(iii) gives
$
\|W_0^\ell\|_\infty
\leq\|W_0\|_\infty^\ell
=O(1).
$
Let $\mu=\mathbb E X$. By
Assumption~\ref{Assump:Instruments}(i), the components of $X$
are independent across $i$ and have means and variances that are
uniformly bounded in $n,t,i$, and $j$. Hence,
%\begin{align*}
$$\mathbb E\|W_0^\ell(X-\mu)\|_2^2
%&\leq C\|W_0^\ell\|_F^2\\
\leq Cn\|W_0^\ell\|_\infty^2
=O(n).$$
%\end{align*}
Moreover,
$
\|W_0^\ell\mu\|_2^2
\leq
n\|W_0^\ell\|_\infty^2\|\mu\|_{\max}^2
=O(n).
$
It follows that
$
\mathbb E\|W_0^\ell X\|_2^2=O(n),
$
and therefore, by Markov's inequality,
$
\|W_0^\ell X\|_2=O_p(\sqrt n).
$

Let
$
R_{0n}=(I_n-\lambda_0W_0)^{-1}.
$
By the Neumann-series bound under
Assumption~\ref{Assump:True_Adjacency}(iii),
$
\|R_{0n}\|_\infty\leq c_\infty^{-1}.
$
Consequently,
$
\|W_0^\ell R_{0n}\|_\infty=O(1).
$
Applying the same second-moment argument with
$A_n=W_0^\ell R_{0n}$ gives
$
\mathbb E\|W_0^\ell R_{0n}X\|_2^2=O(n).
$
Hence, by Markov's inequality,
\[
\|W_0^\ell(I_n-\lambda_0W_0)^{-1}X\|_2
=O_p(\sqrt n), \qquad \ell=0,1,2.
\]
%for $\ell=0,1,2$.
\end{proof}

\subsection{Proof of Corollary~\ref{cor1}}
\begin{proof}
When $L_0=\mathbf 0_{n\times n}$ and the estimator is restricted to
$\widehat L=\mathbf 0_{n\times n}$, we have
$
W_0=S_0$, 
%\qquad
$\widehat W=\widehat S.
$
The pure-sparse estimator is obtained by entrywise soft thresholding:
\[
\widehat S_{ij}
=
\operatorname{sign}(W_{ij})
\bigl(|W_{ij}|-\tau_n\bigr)_+,
\qquad i\neq j,
\]
with $\widehat S_{ii}=0$. Consider the event
$
\mathcal E_n
=
\left\{
\|E\|_{\max}\leq w_{\max n}
\right\},
$
whose probability approaches one by
Assumption~\ref{Assump:E_noise}. Since
$\tau_n=C_\tau w_{\max n}$ with $C_\tau\geq1$, on $\mathcal E_n$,
\[
S_{0,ij}=0
\quad\Longrightarrow\quad
|W_{ij}|=|E_{ij}|\leq\tau_n
\quad\Longrightarrow\quad
\widehat S_{ij}=0.
\]
Hence,
$
\operatorname{supp}(\widehat S)
\subseteq
\operatorname{supp}(S_0).
$
Moreover, the soft-thresholding formula implies
$
|\widehat S_{ij}-S_{0,ij}|
\leq
\min\{|S_{0,ij}|,\,2\tau_n\}.
$
Writing
$
\Delta_S\coloneqq\widehat S-S_0,
$
we consequently obtain
$
|\Delta_S|_{1,1}
\leq
2\tau_n m_s(S_0)
=
O_p\!\left(
m_s(S_0)w_{\max n}
\right).
$
In addition,
$
\|\Delta_S\|_\infty
\leq
\|S_0\|_\infty
=
O(1),
$
where the last equality follows from
Assumption~\ref{Assump:True_Adjacency}(iii).

Apply Lemma~\ref{rateG}(b) with the working matrix
$\widehat S=S_0+\Delta_S$ and with $\Delta_S$ in place of the raw
measurement error. This gives
$
\|\widehat G(\widehat S)-\widehat G(S_0)\|_2
+
\|\bar g_{nT}(\theta_0,\widehat S)
-\bar g_{nT}(\theta_0,S_0)\|_2
=
O_p\!\left(
\frac{|\Delta_S|_{1,1}}{n}
\right).
$
The same GMM expansion as in the proof of
Theorem~\ref{Thm:DK22}, together with the nonsingularity condition in
Assumption~\ref{Assump:Instruments}, therefore yields
$
\|\widehat\theta_p(\widehat S)-\theta_0\|_2
=
O_p\!\left(
\frac{|\widehat S-S_0|_{1,1}}{n}
\right)
+
O_p\!\left(
\frac{1}{\sqrt{nT}}
\right).
$
Finally, substituting the sparse-recovery bound gives
\[
\|\widehat\theta_p(\widehat S)-\theta_0\|_2
=
O_p\!\left(
\frac{m_s(S_0)w_{\max n}}{n}
\right)
+
O_p\!\left(
\frac{1}{\sqrt{nT}}
\right),
\]
which proves the corollary.
\end{proof}

\subsection{Proof of Corollary~\ref{cor2}}
\begin{proof}
Let
$
\Delta_{S,s}
\coloneqq
\widehat S_s-S_0$, and 
$U_t
\coloneqq
\lambda_0Y_t+X_t\gamma_0.
$
Because $Z_t(M)$ is constructed from the pre-specified matrix $M$,
$
\bar g_{nT}(\theta_0,\widehat S_s,M)
-
\bar g_{nT}(\theta_0,S_0,M)
=
-\frac{1}{nT}
\sum_{t=1}^T
Z_t(M)^\top\Delta_{S,s}U_t.
$
Using the bilinear-form inequality
$
|x^\top Ay|
\leq
\|x\|_{\max}|A|_{1,1}\|y\|_{\max},
$
together with the stated instrument and data bounds, gives
$
\left\|
\bar g_{nT}(\theta_0,\widehat S_s,M)
-
\bar g_{nT}(\theta_0,S_0,M)
\right\|_2
=
O_p\!\left(
\frac{|\Delta_{S,s}|_{1,1}}{n}
\right).
$

For a pre-specified $M$, define
$
\widehat G(S,M)
\coloneqq
-\frac{1}{nT}
\sum_{t=1}^T
Z_t(M)^\top
\begin{bmatrix}
SY_t & X_t & SX_t
\end{bmatrix}.
$
The same bilinear-form argument yields
$
\|\widehat G(\widehat S_s,M)-\widehat G(S_0,M)\|_2
=
O_p\!\left(
\frac{|\Delta_{S,s}|_{1,1}}{n}
\right).
$
Provided that
$
\frac{|\Delta_{S,s}|_{1,1}}{n}
=o_p(1),
$
the preceding bound and the rank condition in
Assumption~\ref{Assump:Instruments} imply that the GMM Jacobian
evaluated at $\widehat S_s$ remains nonsingular with probability
approaching one.

Because $\xi_{nT}>0$, minimization of the supervised objective with
respect to $\theta$ gives the usual GMM first-order condition.
Therefore,
\begin{align*}
\widehat\theta_s-\theta_0
={}&
-
\left[
\widehat G(\widehat S_s,M)^\top
\widehat\Lambda_{nT}^{-1}
\widehat G(\widehat S_s,M)
\right]^{-1} %\\
%&\quad\times
\widehat G(\widehat S_s,M)^\top
\widehat\Lambda_{nT}^{-1}
\bar g_{nT}(\theta_0,\widehat S_s,M).
\end{align*}
Since
$
\|\bar g_{nT}(\theta_0,S_0,M)\|_2
=
O_p\!\left(\frac{1}{\sqrt{nT}}\right),
$
the preceding moment and Jacobian bounds imply
\[
\|\widehat\theta_s-\theta_0\|_2
=
O_p\!\left(
\frac{|\widehat S_s-S_0|_{1,1}}{n}
\right)
+
O_p\!\left(
\frac{1}{\sqrt{nT}}
\right),
\]
which proves \eqref{eq:cor2-first-rate}.
Finally, the pure-sparse recovery bound %for the supervised estimator,
$
|\widehat S_s-S_0|_{1,1}
=
O_p\!\left(
m_s(S_0)w_{\max n}
\right),
$
gives
\[
\|\widehat\theta_s-\theta_0\|_2
=
O_p\!\left(
\frac{m_s(S_0)w_{\max n}}{n}
\right)
+
O_p\!\left(
\frac{1}{\sqrt{nT}}
\right),
\]
which proves \eqref{eq:cor2-final-rate}.
\end{proof}

\section{Additional evidence on measurement error and network structure}

\subsection{Additional Figures}\label{app:additional_figures}
\begin{figure}[htbp]
\centering
%\caption{Heatmap of different weight matrices}
\begin{subfigure}{0.4\textwidth}
\includegraphics[width = \textwidth]{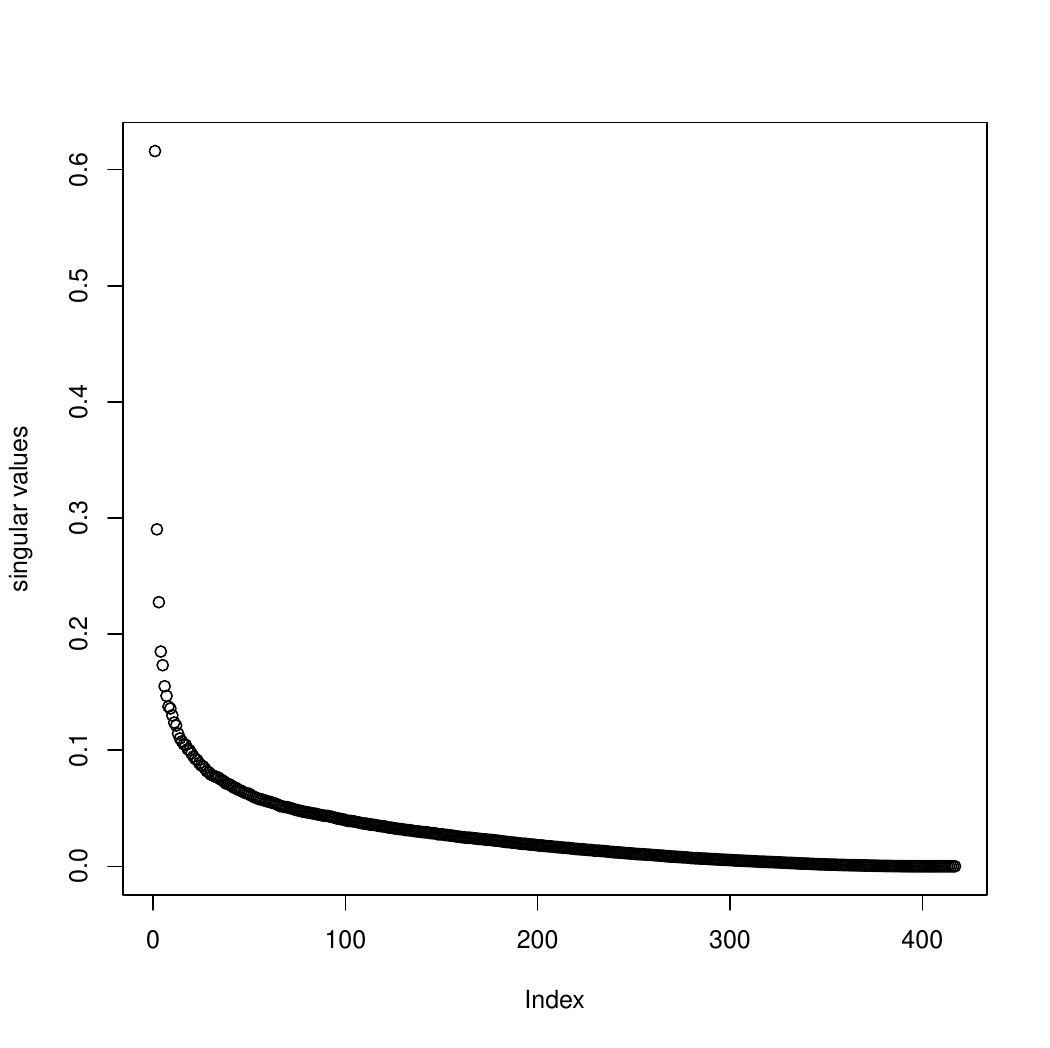}
%	\caption{Standard adjacency}
%	\label{fig:block_adjacency_standard}
\end{subfigure}
\caption{Plots of singular values associated with the tiny weight matrix which is calculated by subtracting the sparse network using a threshold 0.05 from the input-output weight table 2002. \label{fig:singular_values}
}
\end{figure}

\begin{figure}[htbp]
\centering
\begin{subfigure}{0.3\textwidth}
    \includegraphics[width=\textwidth]{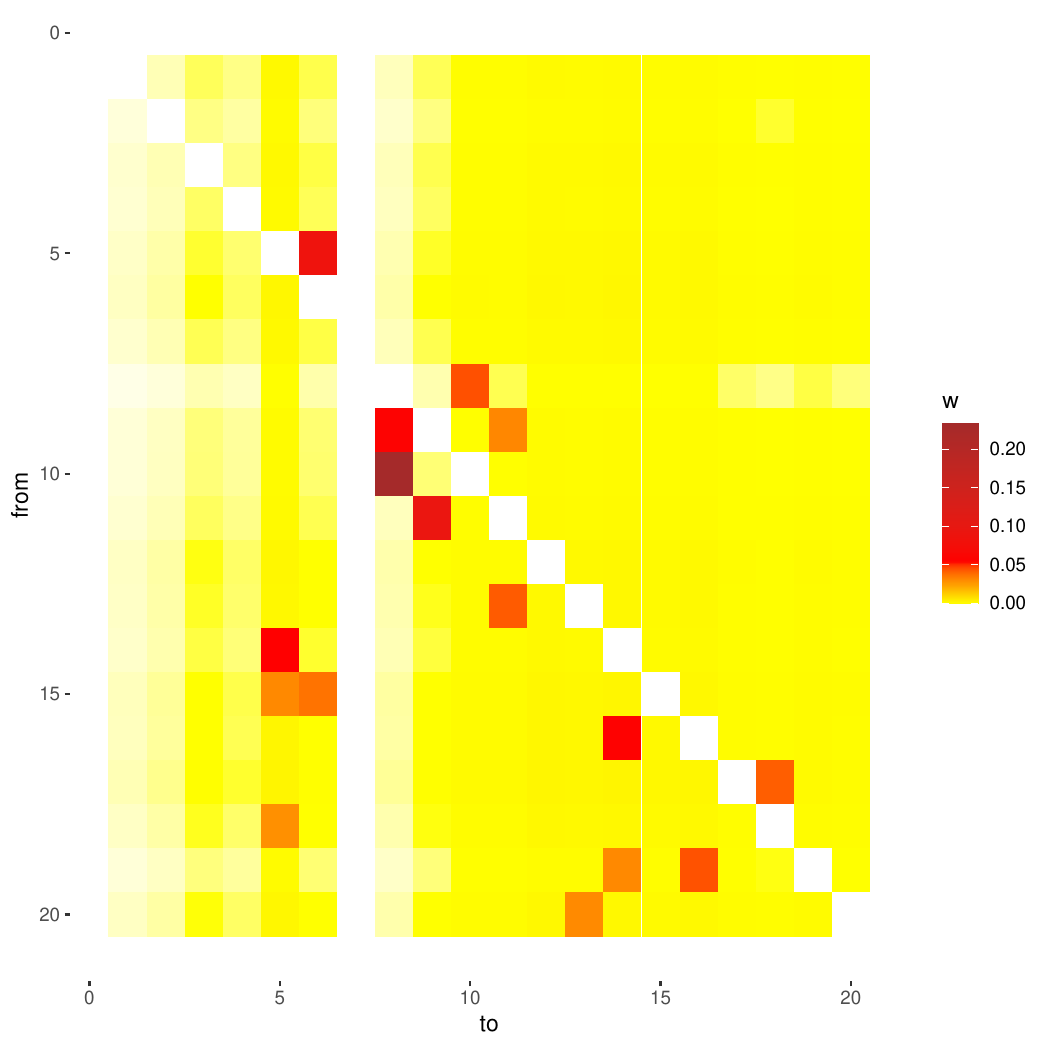}
\end{subfigure}
\hfill
\begin{subfigure}{0.3\textwidth}
    \includegraphics[width=\textwidth]{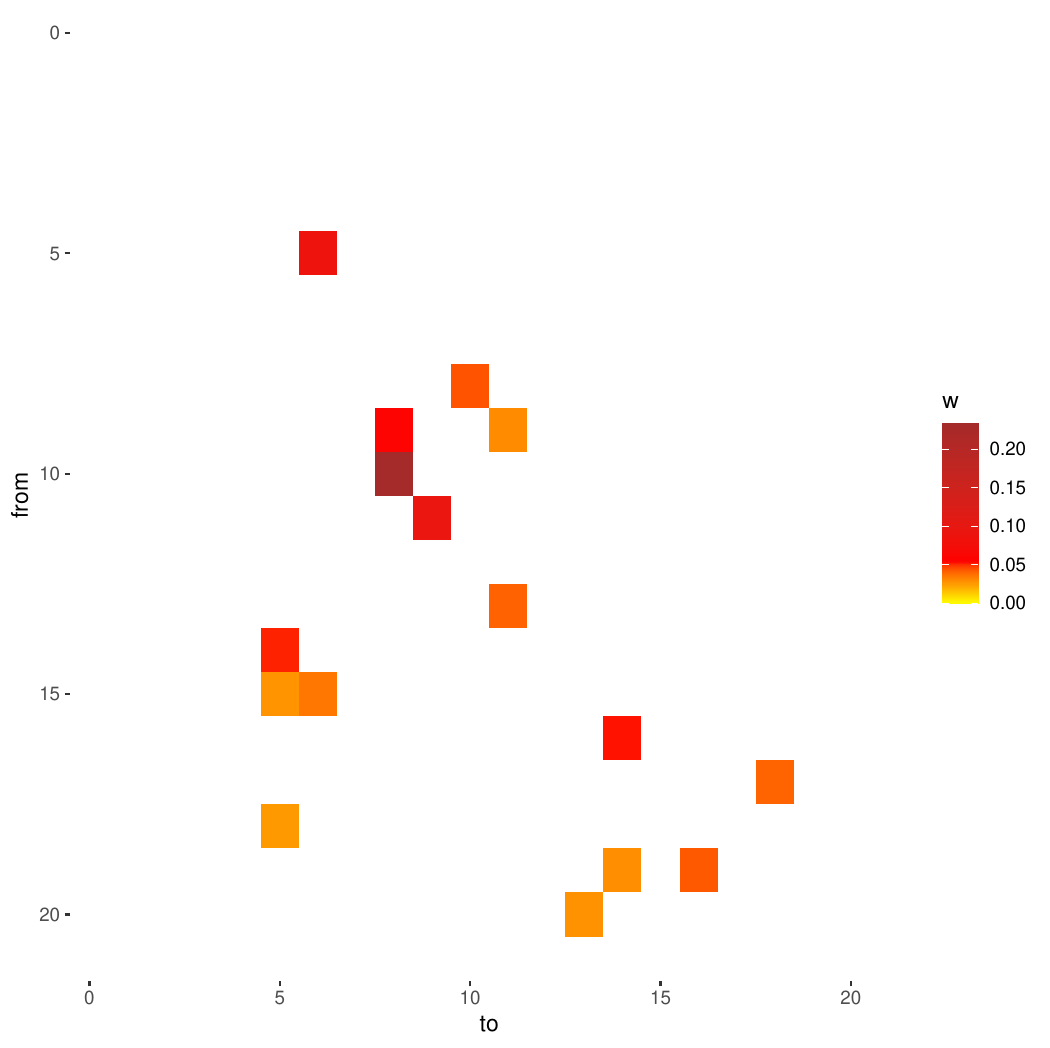}
\end{subfigure}
\hfill
\begin{subfigure}{0.3\textwidth}    \includegraphics[width=\textwidth]{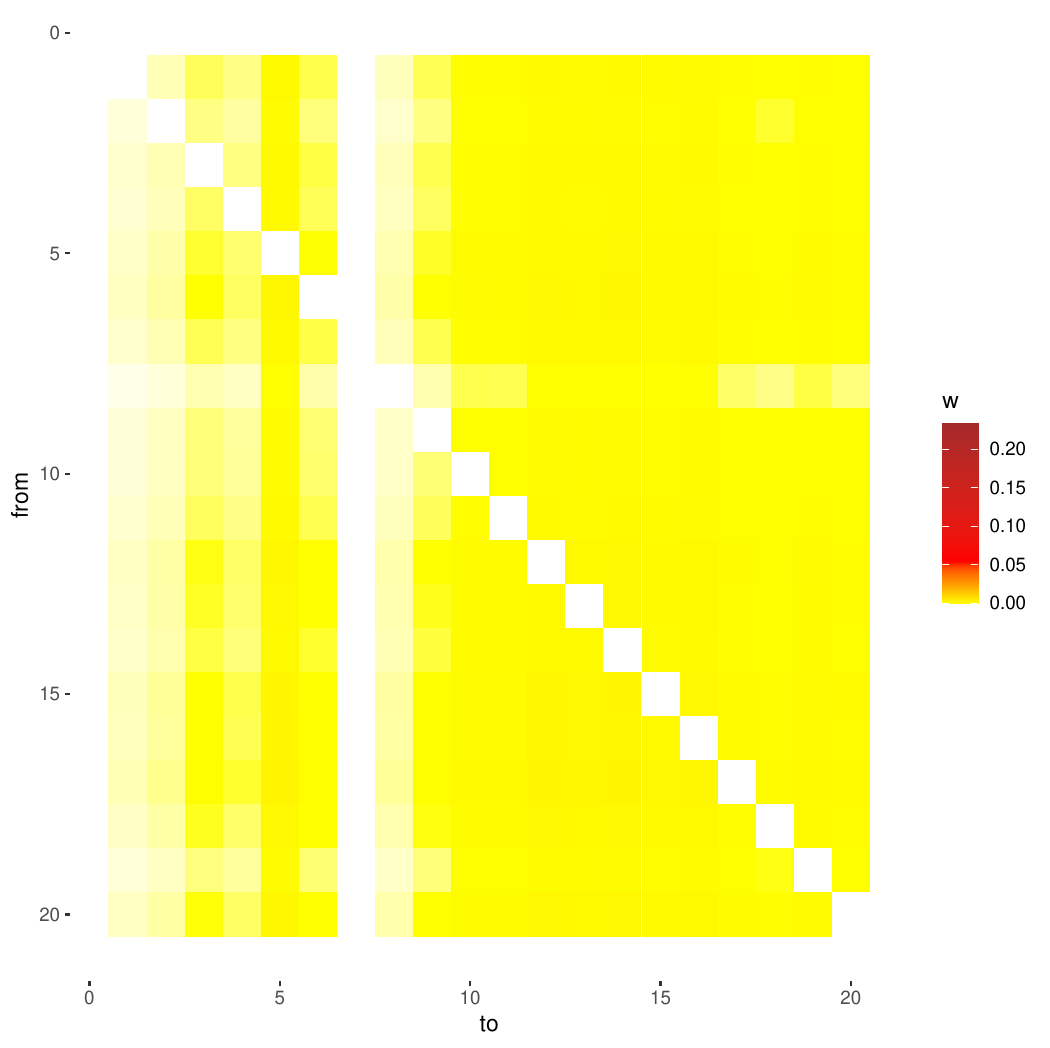}
\end{subfigure}
\caption{\footnotesize Heatmaps of 20 sectors from the 2002 Input-Output Table $W$ (left),
estimated sparse component $\widehat{S}$ (centre), and estimated low-rank component
$\widehat{L}$ (right). The stripe pattern in $\widehat{L}$ reveals pervasive
systematic structure missed by a purely sparse
representation.\label{fig:component_heatmaps}}
\end{figure}

\subsection{Sources of measurement error in economic networks}\label{Sec:Sources}

The economic literature documents several mechanisms through which
the observed adjacency matrix $W$ may depart from the latent network
$W_0$ mediating outcomes of interest. Many policy-relevant production, trade, and reconstructed
financial networks contain numerous weak links and can be dense
relative to conventional social networks. Consequently, measurement
error may affect a large fraction of their entries. A related source
of potentially dense measurement error arises when the adjacency
matrix is constructed from geographical or economic distance
\citep{LeSage2009, Getis2004}: errors in the underlying distance
measurement (arising, for example, from centroid approximations, omitted
transport links, or bandwidth choices) can propagate across entire
rows or columns of $W$. We discuss several additional sources of
measurement error below.

\subsubsection*{(i) Aggregation and reporting error in input--output networks}

A widely studied dense economic network is the input--output table, where 
the $(i,j)$ entry records the share of sector $i$'s output used as an intermediate 
input in the production for sector $j$. \citet{Acemoglu2012} establish that 
propagation multipliers %(the Leontief inverses that govern how sectoral shocks aggregate to macroeconomic fluctuations) 
are highly sensitive to 
the precise pattern of bilateral input shares across all $n^2$ entries of 
this matrix. The observed Bureau of Economic Analysis (US) table, for example, aggregates 
transactions across hundreds of heterogeneous firms and products into 
sector-level flows, and is further subject to reporting inaccuracies, 
imputation of missing entries, periodic revisions, and reconciliation 
between the supply and use tables. Given these, the analyst observes a sector-level $W$ that departs from any 
meaningful sector-level $W_0$, and the resulting errors propagate along 
multi-step Leontief chains.\footnote{Firm-level studies show that heterogeneous buyer--supplier
links matter for shock propagation; see, e.g.
\citet{barrot2016input} and \citet{boehm2019input}. We do not address
the loss of information caused by aggregating such links to the sector
level. Instead, we take the sector-level network as primitive and study
measurement error in its reported entries.} Our method takes the sector-level network as primitive and denoises
around a sector-level truth. Section~\ref{empirical}
illustrates how denoising and retaining the low-rank component affect
the implied cross-sector propagation multipliers.

\subsubsection*{(ii) Mirror-statistic discrepancies in bilateral trade networks}

International trade flows provide a second canonical example of a dense network where measurement error affects every bilateral entry. Trade flows are recorded twice: by the exporting country as exports and by the importing country as imports. These two independent records of the same physical shipment should agree, but in practice they differ substantially and systematically. \citet{FismanWei2004} exploit this bilateral mirror-statistic structure to measure the \emph{evasion gap}, i.e., the log difference between Hong Kong's reported exports to China and China's reported imports from Hong Kong at the six-digit product level and show that a one-percentage-point increase in the tariff rate is associated with a three-percent increase in the evasion gap.
This tariff-related underreporting may affect many product categories
and can be amplified by product misclassification, as importers
reclassify high-tariff goods into lower-tariff categories. The
resulting measurement error need not be confined to a small number
of trade links and may affect a substantial fraction of the bilateral
trade network in a manner correlated with tariff rates.
Furthermore, because tariff schedules are systematically related to trade volumes
and revealed comparative advantage, the measurement error $E_{ij}$
tends to relate with the true $W_{0,ij}$, violating classical
errors-in-variables assumptions and further biasing estimates of
trade-network spillovers.

\subsubsection*{(iii) Entropy-based reconstruction error in interbank exposure networks}

Bilateral interbank exposures form a naturally dense network for studying systemic risk and financial contagion. Each bank's balance sheet records its total interbank lending and borrowing, but the bilateral breakdown, who lent to whom and in what amount, is typically confidential. Regulatory data under the Basel II framework require reporting of exposures only above a threshold of 10\% of regulatory capital, so the majority of bilateral links are unobserved \citep{Anand2015}. Researchers and regulators therefore reconstruct the full $n \times n$ bilateral exposure matrix from the observable marginals (total lending and borrowing of each bank) using, for example, maximum-entropy algorithms. As \citet{Upper2011} and the systematic evaluation of \citet{Anand2015} document, the maximum-entropy solution fills in the matrix as uniformly as possible, implicitly treating the network as if every bank lent to every other bank in proportion to its total activity. This assumption yields a fully dense reconstructed matrix that may systematically underestimate concentration and overestimate diversification: in reality, interbank lending is typically relationship-based and concentrated among a small set of counterparts. The reconstruction error $E = W - W_0$ may therefore be dense. It affects every entry of the matrix, and is potentially correlated with true bilateral exposures, since banks with large total lending are disproportionately misrepresented. Stress-testing exercises that rely on such reconstructed network may underestimate contagion risk precisely because the entropy assumption smooths away the concentrated exposures that drive cascading failure.

\subsubsection*{(iv) Recall error and endogeneity in reported network links}

Even when network data are collected via direct reports rather than algorithmic reconstruction, the observed adjacency matrix is subject to recall error, misunderstanding, and measurement timing mismatches. \citet{Lewbeletal2024EJ} study economic network models in which some reported binary links are misclassified: true connections are recorded as absent (false negatives) and absent connections are recorded as present (false positives). They establish that this misclassification introduces new sources of endogeneity beyond the standard simultaneity problem in spatial models for outcome ($Y$): the mismeasured adjacency matrix contaminates the spatial lag $WY$, creating correlation between regressors and structural errors, and invalidating standard instrumental variables based on powers of $W$. The result is that conventional Two-Stage Least Squares (2SLS) estimators ignoring misclassification are inconsistent even when the fraction of misclassified links is small relative to $n$, because in a dense network a small misclassification rate still corresponds to a large absolute number of erroneous entries. Panel data settings compound this problem: if the network is observed at lower frequency than the outcome variable, links created or dissolved between survey waves appear as persistent misclassifications correlated with individual fixed effects that cannot be differenced away \citep{lewbel2024estimating}.

The four mechanisms above are distinct but share the 
three features highlighted at the end of Section~\ref{sec:densemeasurement}: 
the measurement error is dense, correlated across entries through 
shared sources (aggregation, tariff incentives, reconstruction 
algorithm, or recall bias), and asymptotically non-negligible. These 
properties together explain why classical errors-in-variables 
corrections may be insufficient: they presume sparse or vanishing 
noise that can be handled by instrumental variable (IV) or 
measurement-error bias corrections. Imposing a 
regularization on $W_0$ (as a low-rank matrix, a sparse matrix, 
or their combination) provides the identifying restrictions needed 
to disentangle the systematic signal from the noise, motivating the 
framework developed in the remainder of the paper.
Furthermore, we show in Appendix~\ref{set:Iden} how the network can be constructed mathematically.

\subsection{Examples of network structures}\label{set:Iden}

To build intuition, we describe four common economic network structures and show how each fits into the low-rank, sparse, or low-rank plus sparse framework.
For simplicity, throughout this subsection, we assume 
the network is observed without measurement error, so $W = W_0$. 
\footnote{Note that the link matrix or adjacency matrix could be used to describe the connection strength between individuals. If it does not satisfy the bounded row sum assumption, it could be further scaled by its $\ell_1$ norm.}

\subsubsection{Complete network}
We begin by considering one of the simplest examples: a fully connected network. In this setup, the link matrix $W_0$ is an $n \times n$ matrix of nonzero values, where it is common to set $W_{0, ii} = 0$ to avoid self-loops. The link matrix corresponding to a complete network generated by nonzero constants $\bm{c}=(c_1, \cdots, c_n)\T$ is defined as 
\begin{equation*}
    W_0 = \begin{bmatrix}
        0 & c_{1}c_2 & c_{1}c_3 &\cdots & c_{1}c_n \\
        c_{2}c_1 & 0 & c_{2}c_3 &\cdots & c_2c_n \\
       % c & 1 & 0 &\cdots & 1 \\
        \vdots & \vdots & \vdots & \ddots & \vdots \\
        c_nc_1 & c_nc_2 & c_nc_3 & \cdots & 0
    \end{bmatrix} ,
\end{equation*}
{where $c_i$'s are positive constants within the bounded support $[c_{\min}, c_{\max}]$, and they could be viewed as individual characteristics that vary across $i$.  Let $L_0= \bm{c} \bm{c}\T$, which is a symmetric low-rank matrix with rank $1$. Its singular value is $\|\bm{c}\|_2^2$, and the associated singular vector  is  %proportional to 
$\bm{c} /\|\bm{c}\|_2$. We could then also express $W_0=L_0+S_0=L^*_0$, where the sparse component $S_0$ is a diagonal matrix whose $(i,i)$th element is $- c_{i}^2$ to guarantee no self-loops in $W_0$. Such a network could be used to capture the common interaction structure driven by individuals or agents who interact through a common market  or platform. }

\subsubsection{Low-rank network}

Motivated by the structure of the latent factor, we consider the  network in which the strength of the connection is determined by $W_{0,ij}=a_i+z_iz_j$ if $i\neq j$, { where $a_{i}$'s are individual characteristics and $z_i$'s are some latent variables. \footnote{In a more general term, the low rank component can be understood as network effects related to a principal agent, see \citet{Galeotti2020} . The principal agent is formed by connections with respect to many nodes within the network, and the effects are dense.}
Denote $\bm a$ as a vector consisting of $a_i$. In this example, $W_0=L_0+S_0=L_0^*$, where $L_0={\bm a} {\bf 1}_n^\top + {\bm z} {\bm z}^\top$ is a matrix of rank no more than 2 with ${\bm a}=(a_1, \cdots, a_n)$ and ${\bm z}=(z_1, \cdots, z_n)$, and $S_0$ is a diagonal matrix whose $(i,i)$th element is $-L_{0,ii}$ to remove self-loops. Networks of this kind arise in production networks and financial systems, where common shocks or institutional linkages generate pervasive, dense interactions.}

\subsubsection{Group network}
Another example of interest is the pure group network, where there are two or more communities, potentially of distinct sizes to avoid  eigenvalue multiplicity in the link matrix. Moreover, individuals within the same community are assumed to be connected with each other, while individuals do not connect to others outside their community. For example, suppose $n=20$, and $6$ individuals belong to community $1$ and the rest belong to community $2$. We define  
\begin{equation*}
    W_0 = \begin{bmatrix}
       L_1 & \bm{0}_{6 \times 14} \\
        \bm{0}_{14 \times 6} & L_2 %\bm{1}_{9} \bm{1}_{9}\T - I_9
    \end{bmatrix}+S_0=L_0^*\, ,
\end{equation*}
where $L_1=(c_1, c_2, \cdots, c_6)^{\top}(c_1, c_2, \cdots c_6)$ and $L_2=(c_7, \cdots, c_{20})^{\top}(c_7, \cdots, c_{20})$, and $S_0$ is a diagonal matrix such that $S_{0,ii}=-c_i^2$. Note that $W_0$ still admits the low-rank and sparse decomposition. The low-rank component admits a block structure and its rank equals the number of communities. One of its two singular vectors is proportional to $(0, \cdots, 0, c_7, \cdots, c_{20})$, while the other is 
proportional to $(c_1, \cdots, c_6, 0, \cdots, 0)$. See Figure \ref{fig:group} for the network graph and the heatmap of $W_0$. {This setting captures trade blocs, regional banking clusters, and sectoral linkages, and the rank of $L_0$ equals the number of communities.  We could also extend it to allow sporadic between-group connections, which are described by few nonzero off-diagonal elements in $S_0$. Then the general model could be written as $W_0=L_0+S_0=L_0^*+S_0^*$, where the low-rank component captures the within-group connections and the sparse component represents the between-group connections.} 

\subsubsection{Dominant units}
\label{Subsec:Dominant_Units_Simulation}
An interesting example comes from considering networks containing a set of dominant units also known as key players, see \citet{CalvoArmengol2009, Zenou2016}. One definition of dominant units used in the network literature identifies  them as the set of units whose actions have large and persistent  effects on the other agents.\footnote{We follow the definition of dominant units as in \citet{Pesaran2020, Pesaran2021, Lee2022}. } 
One way to measure  the effect of the unit $j$ on others is to use its weighted out-degree $d_j$, 
defined simply by the $j$-th column sum of $W_0$:
 $d_j = \sum_{i = 1}^n W_{0,ij}$.
We then call the unit $j$ a dominant unit if its weighted out-degree is large or grows as the sample size does. For example, suppose $n=20$ and there is only one dominant unit, namely the first individual, who is connected to $13$ individuals. We assume that each nondominant unit is also connected to its
immediate neighbors. Then the link matrix $W_0$ is specified as below: 
\begin{equation*}
    W_0 = \begin{bmatrix}
            0 & w_{1,2} & 0 & \cdots & \cdots & \cdots &\cdots&\cdots& 0 \\
            w_{2,1} & 0 & w_{2,3} & \cdots & \cdots &\cdots&\cdots& \cdots & 0 \\
            \vdots & \vdots & \vdots & \vdots & \ddots & \vdots & \vdots &\vdots &\vdots\\
            w_{14,1} & 0 & \cdots  & w_{14,13} & 0 & w_{14,15} & \cdots & 0 & 0 \\
            0 & \cdots & \cdots & 0 & w_{15,14} & 0 & w_{15,16} & \cdots  & 0 \\
            \vdots & \vdots & \vdots & \vdots & \vdots & \vdots & \vdots &\vdots &\vdots\\
            0 & 0 & 0 & 0 & \cdots &\cdots&\cdots& w_{20,19} & 0 \\
        \end{bmatrix} \, .
\end{equation*}
See Figure \ref{fig:dominant} for the network graph and the heat map of $W_0$. 

This network is naturally sparse. The number of nonzero entries in
each row is bounded, although the column corresponding to the dominant
unit may contain a growing number of nonzero entries. While this
column can be represented as a sparse rank-one matrix, the full
adjacency matrix need not be low-rank because it also contains links
among neighboring units. We therefore treat $W_0$ as a sparse matrix.
Appendix~\ref{exampleid} provides further discussion of alternative
low-rank-plus-sparse representations and their identification.

\begin{figure}[htbp]
\centering
\caption{Graphical representations of simulated adjacency matrices associated with different groups}
\begin{subfigure}{0.4\textwidth}
	\includegraphics[width = \textwidth]{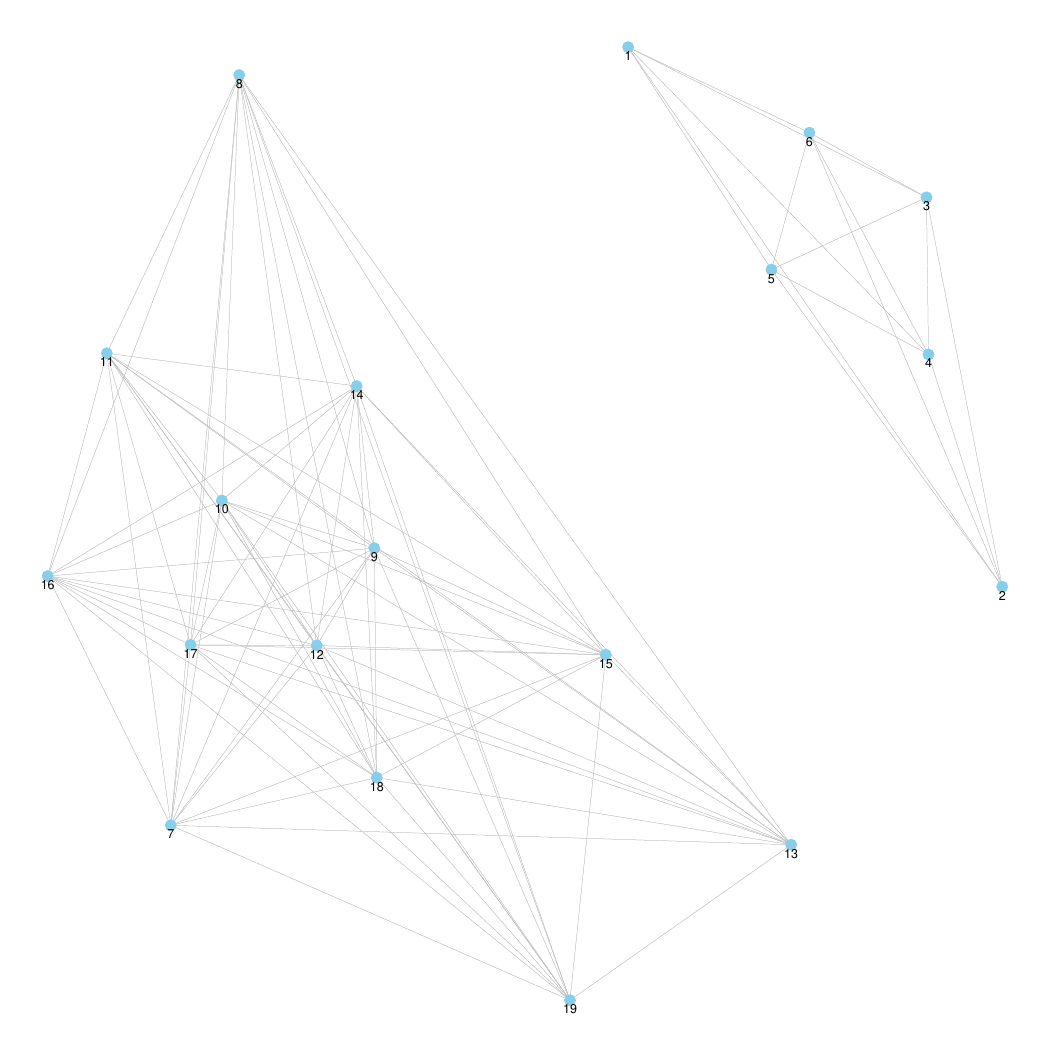}
\end{subfigure}
\hfill
\begin{subfigure}{0.4\textwidth}
\includegraphics[width = \textwidth]{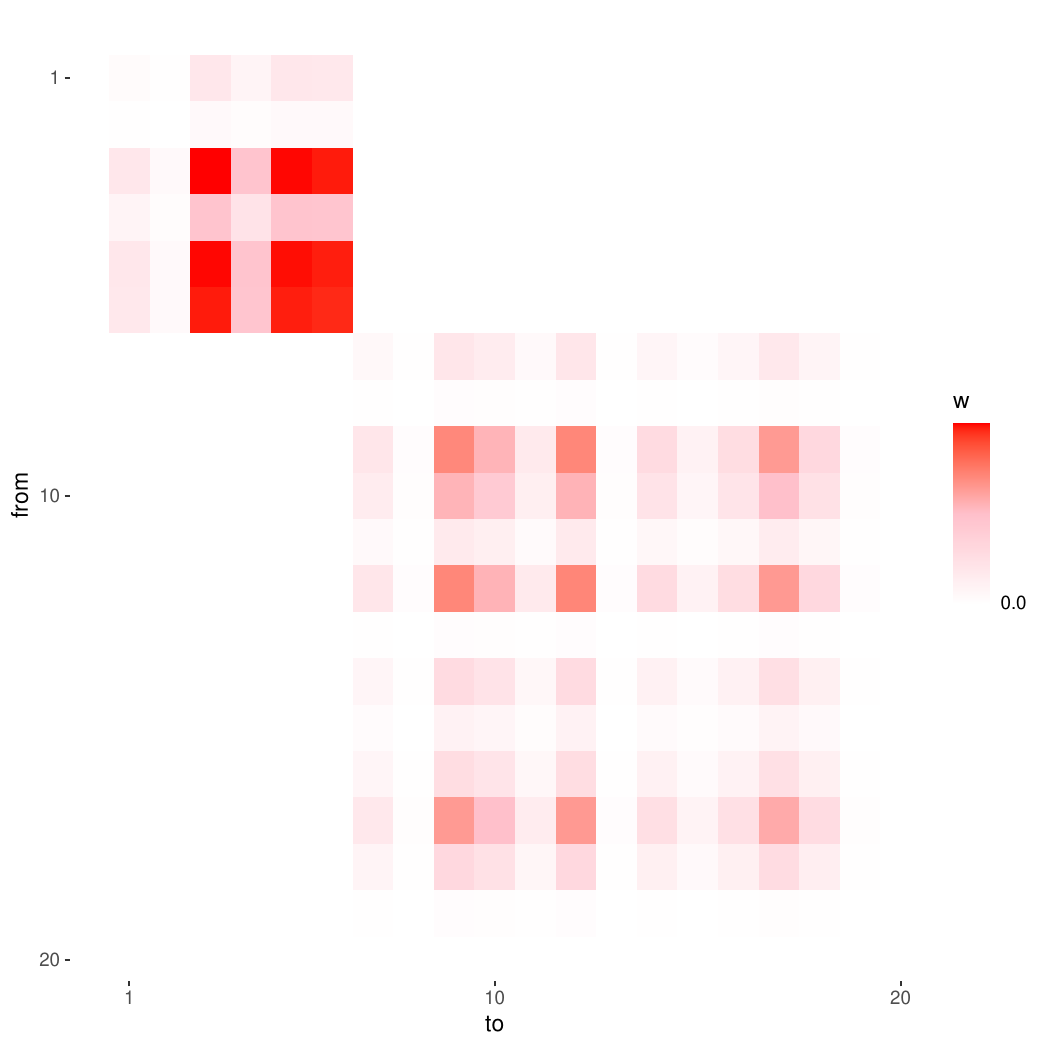}
\end{subfigure}
\hfill
\caption*{\footnotesize Notes: Examples for the weight matrix formed by individuals who belong to different groups. The left one is the network graph and the right one is the heat map of the adjacency matrix. There are two groups, one of which has $6$ members and the other has $14$ members. Only members within the same group have connections, and the connection strength is randomly simulated. } 
\label{fig:group}
\end{figure}

\begin{figure}[htbp]
\centering
\caption{Graphical representations of the simulated adjacency matrix with dominant units. }
\begin{subfigure}{0.4\textwidth}
\includegraphics[width = \textwidth]{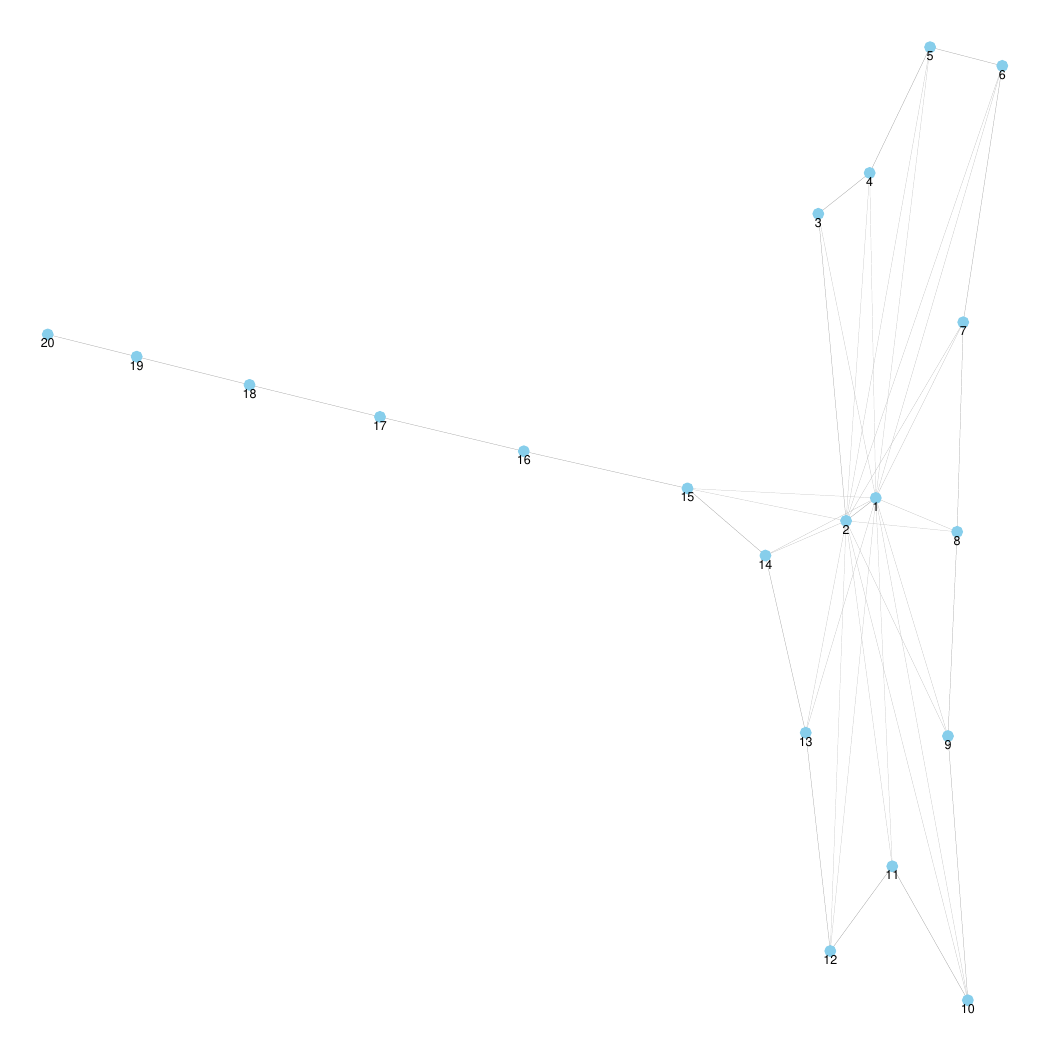}
\end{subfigure}
\hfill
\begin{subfigure}{0.4\textwidth}
\includegraphics[width = \textwidth]{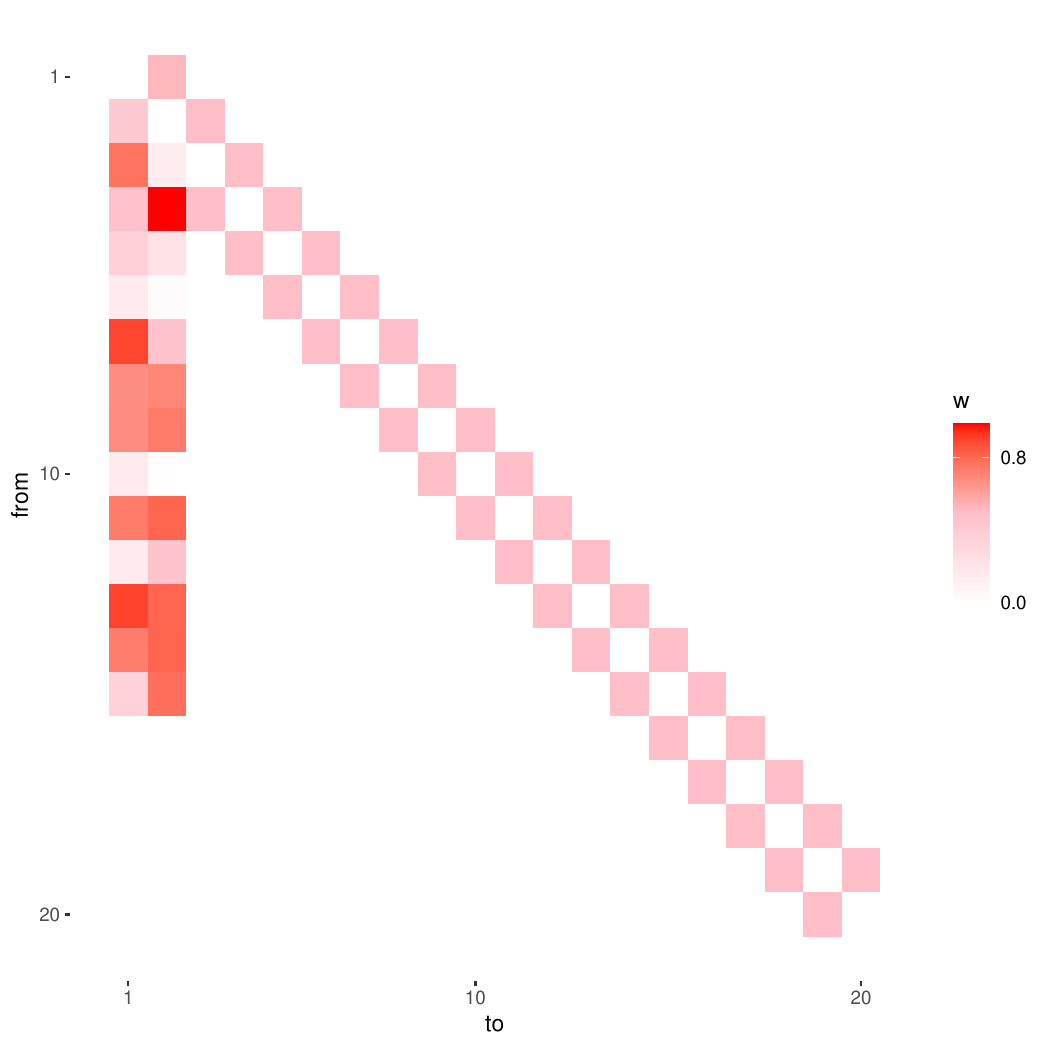}
\end{subfigure}
\hfill
\caption*{\footnotesize Notes: Examples for the weight matrix with dominant units. The left one is the network graph and the right one is the heat map of the adjacency matrix. There are $20$ individuals and each individual connects to its neighbors with strength set by $0.25$. The first one is a key player that influences 13 other individuals, and the influence strengths are generated by $\text{Uniform}(0, 1)$ distributed random variables. }
\label{fig:dominant}
\end{figure}

\subsection{Identification Examples}
\label{exampleid}
We now discuss how the identification assumption applies to the four 
examples of network structures mentioned in Section~\ref{set:Iden}. Complete network: In this example, $W_0=L_0+S_0$, where 
$L_0=\bm{c} \bm{c}^\top$, $\bm{c}=(c_1, \cdots, c_n)^\top$, and the 
sparse component $S_0$ is a diagonal matrix to guarantee that the 
diagonal elements of $W_0$ are zero. Note that 
$\operatorname{deg}_{\max}(S_0)=1$ and 
$\mbox{inc}(L_0)=\sqrt{\max_{1\leq i\leq n} c_i^2/\|\bm{c}\|_2^2}
\leq c_{\max}/(\sqrt{n}\,c_{\min})$. Hence, a sufficient condition 
for the identification condition 
$\mbox{inc}(L_0)\operatorname{deg}_{\max}(S_0)\leq 1/16$ of 
Assumption~\ref{Assump:W_structure} is 
$c_{\max}/c_{\min}\leq \sqrt{n}/16$.

Low-rank network: In this example, $W_0=L_0+S_0$, where $L_0={\bm a} {\bf 1}_n^\top + {\bm z} {\bm z}^\top$ is a matrix of rank no more than 2 with ${\bm a}=(a_1, \cdots, a_n)$ and ${\bm z}=(z_1, \cdots, z_n)$, and $S_0$ is a diagonal matrix whose $(i,i)$th element is $-L_{0,ii}$ to remove self-loops. For simplicity, we assume ${\bm z}$ is orthogonal to both ${\bm a}$ and ${\bm 1}_n$, and hence the left singular vectors are proportional to ${\bm a}$ and ${\bm z}$, while the right singular vectors are proportional to ${\bf 1}_n$ and ${\bm z}$. Since $\operatorname{deg}_{\max}(S_0)=1$, the 
identification assumption could be satisfied if 
$\mbox{inc}(L_0)=\sqrt{\max\left\{\max_{1\leq i\leq n}\left(
\frac{a_i^2}{\|\bm a\|_2^2}+\frac{z_i^2}{\|\bm z\|_2^2}\right),\ 
\max_{1\leq i\leq n}\left(\frac{1}{n}+\frac{z_i^2}{\|\bm z\|_2^2}
\right)\right\}}\leq 1/16$.

Group network: In this example, $W_0=L_0+S_0$, where $L_0=\bm{c}_a \bm{c}_a\T+\bm{c}_b \bm{c}_b\T$, $\bm{c}_a=(c_1, \cdots, c_6, 0, \cdots, 0)\T$, and $\bm{c}_b=(0, \cdots, 0, c_7, \cdots, c_{20})\T$, and $S_0$ is a diagonal matrix to ensure that the diagonal elements of $W_0$ are 0. Since $\operatorname{deg}_{\max }(S_0)=1$, the identification assumption could be satisfied if  $\mbox{inc}(L_0)=\sqrt{\max(\max_{1\leq i\leq 6} c_i^2/\sum_{i=1}^6 c_i^2, \max_{7\leq i \leq 20} c_i^2/\sum_{i=7}^{20} c_i^2)} \leq 1/16$. In the general case when $S_0$ has nonzero off-diagonal elements, the identification assumption then requires that $\operatorname{deg}_{\max }(S_0)\mbox{inc}(L_0)=\operatorname{deg}_{\max }(S_0)\sqrt{\max(\max_{1\leq i\leq 6} c_i^2/\sum_{i=1}^6 c_i^2, \max_{7\leq i \leq 20} c_i^2/\sum_{i=7}^{20} c_i^2)}\leq 1/16$.

Dominant units: In this example, \begin{equation*}
    W_0 = \begin{bmatrix}
            0 & w_{1,2} & 0 & \cdots & \cdots & \cdots &\cdots&\cdots& 0 \\
            w_{2,1} & 0 & w_{2,3} & \cdots & \cdots &\cdots&\cdots& \cdots & 0 \\
            \vdots & \vdots & \vdots & \vdots & \ddots & \vdots & \vdots &\vdots &\vdots\\
            w_{14,1} & 0 & \cdots  & w_{14,13} & 0 & w_{14,15} & \cdots & 0 & 0 \\
            0 & \cdots & \cdots & 0 & w_{15,14} & 0 & w_{15,16} & \cdots  & 0 \\
            \vdots & \vdots & \vdots & \vdots & \vdots & \vdots & \vdots &\vdots &\vdots\\
            0 & 0 & 0 & 0 & \cdots &\cdots&\cdots& w_{20,19} & 0 \\
        \end{bmatrix} \, .
\end{equation*}

When we treat it as a pure sparse matrix, the identification condition is automatically satisfied as $\mbox{inc}(L_0)=0$. Though it seems that one could write $W_0 = {\bf w}_1 {\bm e}\T_1 + W_{r}$, where $ {\bf w}_1$ is the first column of $W_0$,  ${\bm e}_1$ is the first column of the identity matrix $I_{n}$, and $W_{r}$ is the $n\times n$ partitioned matrix satisfying
 \begin{equation*}
 	\begin{bmatrix} 0 & {\bm s}\T_{1} \\
 		            \bm{0}_{(n-1)\times 1} & {\bm S}_{r} \end{bmatrix} \, 
 \end{equation*}
 with ${\bm s}_{1} = [w_{1,2},  \bm{0}_{1\times (n-2)}]\T $, and the $(i,j)$th element in ${\bm S_{r}}$ is identical to the $(i+1, j+1)$th element of $W_0$. In such a decomposition, the right singular vector of the low-rank part is $[1, \bm{0}_{1\times (n-1)}]\T$ and its incoherence index is $\mbox{inc}({\bm w_{1}} {\bm e}\T_1)=1$. Since $\operatorname{deg}_{\max }(W_{r})=2$, we have $\mbox{inc}({\bm w_{1}} {\bm e}\T_1)\operatorname{deg}_{\max }(W_{r})>1/16$. Therefore, we have to treat $W_0$ as a pure sparse matrix so that the identification Assumption holds.

    \end{appendices}

\end{document}